\documentclass[12pt,oneside,openany,article]{memoir}
\usepackage{mempatch}

\nouppercaseheads
\usepackage[amsthm,thmmarks,hyperref]{ntheorem}
\usepackage{etex}
\usepackage[english]{babel}

\usepackage{caption}

\usepackage[leqno]{mathtools}

\usepackage[scr]{rsfso}
\usepackage{upgreek}
\usepackage[compress,square,comma,numbers]{natbib}
\usepackage{hypernat} %required if hyperref is used as well

\usepackage{sfmath}
\usepackage{wrapfig}

\usepackage{microtype}

\usepackage{array}
\usepackage{tikz}
\usetikzlibrary{arrows, patterns,
decorations.pathreplacing, decorations.pathmorphing, shapes.symbols, spy}

\usepackage{hyperxmp}
\usepackage{xr-hyper}
\usepackage{nameref}
\usepackage[pdftex,bookmarks,pdfnewwindow,plainpages=false,unicode]{hyperref}

\usepackage{bookmark}
\usepackage{enumitem}
\usepackage{amssymb}

\hypersetup{
colorlinks=true,
linkcolor=black,
citecolor=black,
urlcolor=blue!80!cyan,
pdfauthor={Victor Ivrii},
pdftitle={Asymptotics of the ground state energy of heavy molecules and related topics},
pdfsubject={Sharp Spectral Asymptotics},
pdfkeywords={Microlocal Analysis,  Sharp Spectral Asymptotics},
bookmarksdepth={4}
}

\theoremstyle{plain}
\newtheorem{theorem}{Theorem}[section]

\newtheorem{proposition}[theorem]{Proposition}
\newtheorem{corollary}[theorem]{Corollary}
\newtheorem{Problem}[theorem]{Problem}

\theoremstyle{definition}

\theoremstyle{remark}
\newtheorem{remark}[theorem]{Remark}

\newtheorem{problem}[theorem]{Problem}

\makeatletter
\newtheoremstyle{plainfoot}%
  {\item[\hskip\labelsep \theorem@headerfont ##1\ ##2\,\footnotemark\theorem@separator]}%
  {\item[\hskip\labelsep \theorem@headerfont ##1\ ##2\ (##3)\, \footnotemark\theorem@separator]}
\makeatother

\theoremstyle{plainfoot}
\theoremseparator{.}
\newtheorem{theorem-foot}[theorem]{Theorem}
\newtheorem{lemma-foot}[theorem]{Lemma}
\newtheorem{proposition-foot}[theorem]{Proposition}
\newtheorem{corollary-foot}[theorem]{Corollary}
\newtheorem{conjecture-foot}[theorem]{Conjecture}
\newtheorem{condition-foot}[theorem]{Condition}

\theoremstyle{plainfoot}
\theoremseparator{.}
\theorembodyfont{}
\newtheorem{definition-foot}[theorem]{Definition}
\newtheorem{Problem-foot}[theorem]{Problem}

\theoremstyle{plainfoot}
\theoremseparator{.}
\theorembodyfont{}
\theoremheaderfont{\slshape}
\newtheorem{remark-foot}[theorem]{Remark}         % remark with no number
\newtheorem{example-foot}[theorem]{Example}
\newtheorem{problem-foot}[theorem]{Problem}

\numberwithin{equation}{section}

\newenvironment{phantomequation}[1][]{\refstepcounter{equation}}{}

\newenvironment{claim}[1][{\textup{(\theequation)}}]{\refstepcounter{equation}\vglue10pt
\begin{trivlist}
\item[{\hskip\labelsep#1}]}{\vglue10pt\end{trivlist}}

\newcommand{\cJ}{\mathcal{J}}
\newcommand{\cH}{\mathcal{H}}
\newcommand{\cI}{\mathcal{I}}

\newcommand{\const}{\mathsf{const}}
\newcommand{\TF}{{\mathsf{TF}}}

\newcommand{\W}{{\mathsf{W}}}
\newcommand{\E}{{\mathsf{E}}}
\newcommand{\N}{{\mathsf{N}}}

\newcommand{\Scott}{\mathsf{Scott}}
\newcommand{\Schwinger}{\mathsf{Schwinger}}
\newcommand{\Dirac}{\mathsf{Dirac}}
\newcommand{\T}{\mathsf{T}}

\newcommand{\I}{\mathsf{I}}
\newcommand{\Q}{\mathsf{Q}}
\newcommand{\x}{{\mathsf{x}}}
\newcommand{\y}{{\mathsf{y}}}

\newcommand{\bR}{\mathbb{R}}
\newcommand{\bC}{\mathbb{C}}
\newcommand{\bZ}{\mathbb{Z}}

\newcommand{\fH}{\mathfrak{H}}

\newcommand{\old}{{{\mathsf{old}}}}
\newcommand{\blangle}{{\boldsymbol{\langle}}}
\newcommand{\brangle}{{\boldsymbol{\rangle}}}
\newcommand{\cE}{\mathcal{E}}
\newcommand{\cN}{\mathcal{N}}

\newcommand{\cX}{\mathcal{X}}
\newcommand{\cY}{\mathcal{Y}}
\newcommand{\cZ}{\mathcal{Z}}
\newcommand{\cS}{\mathcal{S}}
\newcommand{\sC}{\mathscr{C}}
\newcommand{\cQ}{\mathcal{Q}}
\newcommand{\sL}{\mathscr{L}}

\newcommand{\D}{\mathsf{D}}
\newcommand{\sfH}{{\mathsf{H}}}
\newcommand{\dist}{\operatorname{dist}}
\newcommand{\Hess}{\operatorname{Hess}}
\newcommand{\supp}{\operatorname{supp}}

\newcommand{\mes}{\operatorname{mes}}
\newcommand{\Spec}{\operatorname{Spec}}
\newcommand{\Tr}{\operatorname{Tr}}
\newcommand{\tr}{\operatorname{tr}}
\newcommand{\vrai}{\operatorname{vrai}}

\newcommand{\boldupsigma}{{\boldsymbol{\upsigma}}}

\renewcommand{\Re}{\operatorname{Re}}

\addtopsmarks{headings}{}{%
\createmark{chapter}{right}{shownumber}{}{. \ }
}
\title{Asymptotics of the ground state energy of heavy molecules and related topics. II. External magnetic field}
\author{Victor Ivrii}

\begin{document}

\maketitle
{\abstract%
We consider asymptotics of the ground state energy of heavy atoms and molecules in the strong external magnetic field and derive it including Schwinger and Dirac corrections (if magnetic field is not too strong). We also consider related topics:  an excessive negative charge, ionization energy and excessive positive charge when atoms can still bind into molecules.
\endabstract}

\chapter{Introduction}
\label{sect-26-1}

In this paper we repeat analysis of the previous paper~\ref{book_new-sect-25} of \cite{futurebook} but in the case of the constant external magnetic field\footnote{\label{foot-26-1} Actually we need a magnetic field either sufficiently weak or close to a constant on the very small scale.}.

\section{Framework}
\label{sect-26-1-1}

Let us consider the following operator (quantum Hamiltonian)
\begin{gather}
\mathbf{H}=\mathbf{H}_N\coloneqq \sum_{1\le j\le N} H _{A,V,x_j}+
\sum_{1\le j<k\le N}|x_j-x_k| ^{-1}
\label{26-1-1}\\
\shortintertext{on}
\fH= \bigwedge_{1\le n\le N} \sfH, \qquad \sfH=\sL^2 (\bR^d, \bC^q) \label{26-1-2}\\
\shortintertext{with}
H_{V,A}=\bigl((i\nabla -A)\cdot \boldupsigma \bigr) ^2-V(x)
\label{26-1-3}
\end{gather}
describing $N$ same type particles in the external field with the scalar potential $-V$ and vector potential $A(x)$, and repulsing one another according to the Coulomb law.

Here $x_j\in \bR ^d$ and $(x_1,\ldots ,x_N)\in\bR ^{dN}$, potentials $V(x)$ and $A(x)$ are assumed to be real-valued. Except when specifically mentioned we assume that
\begin{equation}
V(x)=\sum_{1\le k\le M} \frac{Z_m }{|x-\y_m|}
\label{26-1-4}
\end{equation}
where $Z_m>0$ and $\y_m$ are charges and locations of nuclei. Here $\boldsymbol{\upsigma} =(\upsigma _1,\upsigma _2,\ldots,\upsigma _d)$,
$\upsigma _k$ are $q\times q$-Pauli matrices.

So far in comparison with the previous Chapter~\ref{book_new-sect-25} we only changed (\ref{book_new-25-1-3}) to (\ref{26-1-3}) introducing magnetic field. Now spin enters not only in the definition of the space but also into operator through matrices $\upsigma_k$. Since we need $d=3$ Pauli matrices it is sufficient to consider $q=2$ but we will consider more general case as well (but $q$ should be even).

\begin{remark}\label{rem-26-1-1}
In the case of the the constant magnetic field $\nabla \times A$
\begin{gather}
H_{A,V}=\bigl(-i\nabla -A(x)\bigr) ^2+ \boldupsigma \cdot (\nabla\times A) - V(x)
\label{26-1-5}\\
\intertext{In the case $d=2$ this operator downgrades to}
H_{A,V}=
\bigl(-i\nabla -A(x)\bigr) ^2+\upsigma _3 (\nabla \times A)-V(x)
\label{26-1-6}
\end{gather}
\end{remark}

Again, let us assume that
\begin{claim}\label{26-1-7}
Operator $\mathbf{H}$ is self-adjoint on ${\fH}$.
\end{claim}

As usual we will never discuss this assumption.

\section{Problems to consider}
\label{sect-26-1-2}

As in the previous Chapter we are interested in the \emph{ground state energy\/}\index{ground state energy} $\E=\E_N$ of our system i.e. in the lowest eigenvalue of the operator $\mathbf{H}=\mathbf{H}_N$ on $\fH$:
\begin{equation}
\E\coloneqq \inf \Spec \mathbf{H}\qquad \text{on\ \ }\fH;
\label{26-1-8}
\end{equation}
more precisely, we are interested in the asymptotics of
$\E_N=\E (\underline{\y};\underline{Z};N)$ as $V$ is defined by (\ref{26-1-4}) and $N\asymp Z\coloneqq Z_1+Z_2+\ldots+Z_M\to \infty$ and we are going to prove that\footnote{\label{foot-26-2} Under reasonable assumption
to the minimal distance between nuclei.} $\E$ is equal to \emph{Magnetic Thomas-Fermi energy\/}\index{Thomas-Fermi energy!magnetic} $\cE_B^\TF$, possibly with the Scott and Dirac-Schwinger corrections and with an appropriate error.

We are also interested in the asymptotics for the \emph{ionization energy\/}\index{ionization energy}
\begin{equation}
\I_N \coloneqq \E_{N-1}- \E_N
\label{26-1-9}
\end{equation}
and we also would like to estimate \emph{maximal excessive negative charge\/}\index{maximal excessive negative charge}
\begin{equation}
\max_{N\colon \I_N>0} (N-Z).
\label{26-1-10}
\end{equation}

All these questions so far were considered in the framework of the fixed positions $\y_1,\ldots, \y_M$ but we can also consider
\begin{gather}
\widehat{\E}\coloneqq \widehat{\E}_N= \widehat{\E}(\underline{\y};\underline{Z}; N)=
\E + U(\underline{\y};\underline{Z})\label{26-1-11}\\
\shortintertext{with}
U(\underline{\y};\underline{Z})\coloneqq \sum_{1\le m< m'\le M}\frac {Z_mZ_{m'}}{|\y_m-\y_{m'}|}\label{26-1-12}\\
\shortintertext{and}
\widehat{\E} (\underline{Z}; N)=\inf_{\y_1,\ldots,\y_M} \widehat{\E}(\underline{\y};\underline{Z}; N) \label{26-1-13}
\end{gather}
and replace $\I_N$ by $\widehat{I}_N = \widehat{\E}_{N-1}-\widehat{\E}_N$ and modify all our questions accordingly. We call these frameworks \emph{fixed nuclei model\/}\index{fixed nuclei model} and \emph{free nuclei model\/}\index{free nuclei model} respectively.

In the free nuclei model we can consider two other problems:
\begin{enumerate}[label=(\alph*), wide, labelindent=0pt]
\item\label{sect-26-1-2-a}
Estimate from below \emph{minimal distance between nuclei} i.e.
\begin{equation*}
\min_{1\le m < m'\le M} |\y_m-\y_{m'}|
\end{equation*}
for which such minimum is achieved.

\item\label{sect-26-1-2-b}
Estimate \emph{maximal excessive positive charge\/}\index{maximal excessive positive charge}
\begin{equation}
\max_N \bigl\{(Z-N)\colon \widehat{\E} <
\min_{\substack{\\[3pt]N_1,\ldots, N_M: \\[3pt] N_1+\ldots N_M=N}}\;
\sum_{1\le m\le M} \ \E (Z_m; N_m)\bigr\}
\label{26-1-14}
\end{equation}
for which molecule does not disintegrates into atoms.
\end{enumerate}

\section{Magnetic Thomas-Fermi theory}
\label{sect-26-1-3}

As in the previous Chapter~\ref{book_new-sect-25} the first approximation is the Hartree-Fock (or Thomas-Fermi) theory. Let us introduce the \emph{spacial density\/}\index{spacial density} of the particle with the state
$\Psi \in {\fH}$:
\begin{equation}
\rho (x)=\rho _\Psi (x)=N
\int |\Psi (x,x_2,\ldots ,x_N)| ^2 \,dx_2\cdots dx_N.
\label{26-1-15}
\end{equation} Let us write the Hamiltonian, describing the corresponding ``quantum liquid'':
\begin{gather}
\cE_B (\rho )=\int \tau_B (\rho (x)) \,dx- \int V(x)\rho (x)\, dx+\frac{1}{2}
\D(\rho ,\rho ),\label{26-1-16}\\
\shortintertext{with}
\D(\rho ,\rho )=\iint |x-y| ^{-1}\rho (x)\rho (y)\,dxdy\label{26-1-17}
\end{gather}
where $\tau _B$ is the energy density of a gas of noninteracting electrons:
\begin{equation}
\tau_B (\rho )=\sup _{w\ge 0} \bigl(\rho w-P _B(w)\bigr)
\label{26-1-18}
\end{equation}
is the Legendre transform of the \emph{pressure\/}\index{pressure} $P_B(w)$ given by the formula
\begin{equation}
P_B(w)=\varkappa _1 B\Bigl(\frac{1}{2} w_+ ^{\frac{d}{2}} +
\sum _{j\ge 1} (w-2jB)_+ ^{\frac{d}{2}}\Bigr)
\label{26-1-19}
\end{equation}
with $\varkappa _1=(2\pi)^{-1}q, (3\pi ^2)^{-1}q$ for $d=2,3$ respectively.

The classical sense of the second and the third terms in the right-hand
expression of (\ref{26-1-16}) is clear and the density of the kinetic energy is given by $\tau_B (\rho )$ in the semiclassical approximation (see remark \ref{rem-26-1-2}). So, the problem is
\begin{claim}\label{26-1-20}
Minimize functional $\cE_B(\rho)$ defined by (\ref{26-1-16}) under restrictions:
\begin{phantomequation}\label{26-1-21}\end{phantomequation}
\vglue-25pt
\begin{equation}
\rho \ge 0,\qquad \int \rho \,dx\le N.
\tag*{$\textup{(\ref*{26-1-21})}_{1,2}$}
\end{equation}
\end{claim}
\vglue-10pt
The solution if exists is unique because functional $\cE_B(\rho )$ is strictly convex (see below). The existence and the property of this solution denoted further by $\rho_B ^\TF$ is known in the series of physically important cases.

\begin{remark}\label{rem-26-1-2}
If $w$ is the negative potential then
\begin{equation}
\tr e(x,x,0)\approx P'_B(w)
\label{26-1-22}
\end{equation}
defines the density of all non-interacting particles with negative energies at point $x$ and
\begin{equation}
\int ^0_{-\infty} \tau \,d_\tau \tr e(x,x,\tau) dx \approx -\int P_B(w)\,dx
\label{26-1-23}
\end{equation}
is the total energy of these particles; here $\approx$ means ``in the semiclassical approximation''.
\end{remark}

We consider in the case of $d=3$ a large (heavy) molecule with potential (\ref{book_new-25-1-4}). It is well-known\footnote{\label{foot-26-3} Section IV of E.~H.~Lieb, J.~P.~Solovej and J.~Yngvarsson~\cite{LSY2}.} that

\begin{proposition}\label{prop-26-1-3}
\begin{enumerate}[label=(\roman*), wide, labelindent=0pt]
\item\label{prop-26-1-3-i}
For $V(x)$ given by \textup{(\ref{26-1-4})} minimization problem \textup{(\ref{26-1-20})} has a unique solution $\rho =\rho_B^\TF$; then denote $\cE^\TF_B \coloneqq \cE_B (\rho_B^\TF)$.

\item\label{prop-26-1-3-ii}
Equality in $\textup{(\ref{26-1-21})}_2$ holds if and only if
$N\le Z\coloneqq \sum_m Z_m$.
\item\label{prop-26-1-3-iii}
Further, $\rho^\TF$ does not depend on $N$ as $N\ge Z$.

\item\label{prop-26-1-3-iv}
Thus
\begin{equation}
\int \rho_B^\TF \,dx =\min (N,Z),\qquad Z\coloneqq \sum_{1\le m\le M} Z_m.
\label{26-1-24}
\end{equation}
\end{enumerate}
\end{proposition}

\section{Main results sketched and plan of the chapter}
\label{sect-26-1-4}

In the first half of this Chapter we derive asymptotics for ground state energy and justify Thomas-Fermi theory. As construction of Section~\ref{book_new-sect-25-2} works with minimal modifications (see Section~\ref{sect-26-6}) in the magnetic case as well we start immediately from magnetic Thomas-Fermi theory in Section~\ref{sect-26-2}.

We discover that there are three different cases: \emph{a moderate magnetic field case\/}\index{moderate magnetic field case} $B\ll Z^{\frac{4}{3}}$ when
$\cE^\TF _B \asymp Z^{\frac{5}{3}}$ and $\cE^\TF_B=\cE^\TF_0(1+o(1))$, \emph{a strong magnetic field case\/}\index{strong magnetic field case} $B\gg Z^{\frac{4}{3}}$ when $\cE^\TF _B \asymp B^{\frac{2}{5}} Z^{\frac{9}{5}}$ and $\cE^\TF_B=\bar{\cE}^\TF_B(1+o(1))$ where $\bar{\cE}^\TF_B$ is Thomas-Fermi potential derived as $P_B(w)= \frac{1}{2}\varkappa_1 w^{\frac{d}{2}}$ (cf. (\ref{26-1-19})), and \emph{an intermediate case}\index{intermediate case} $B\sim Z^{\frac{4}{3}}$.

Then we apply semiclassical methods (like in Section~\ref{book_new-sect-25-4}) albeit now our analysis is way more complicated due to two factors: the semiclassical theory of the magnetic Schr\"odinger operator is more difficult than the corresponding theory for the non-magnetic Schr\"odinger operator and also Thomas-Fermi potential $W^\TF$ is not very smooth in the magnetic case, so we need to approximate it by a smooth one (on a microscale).

We discover that both semiclassical methods and Thomas-fermi theory are relevant only if $B\ll Z^3$. The case of the superstrong magnetic field $B\gg Z^3$ was considered in E.~H.~Lieb, J.~P.~Solovej and J.~Yngvarsson~\cite{LSY1}.% and we hope to cover it in the Chapter~29\Note{Pending}.

First of all, in Section~\ref{sect-26-3} we consider the case $M=1$; then the Thomas-Fermi potential $W^\TF_B$ is non-degenerate and in this case we derive sharp spectral asymptotics.

Next, in Section~\ref{sect-26-4} we consider the case $M\ge 2$ but we analyze only zone $\{W^\TF_B+\nu \gtrsim B\}$ where $\nu$ is a chemical potential and $B$ is an intensity of the magnetic field. A certain weaker non-degeneracy condition is satisfied due to the Thomas-Fermi equation and we derive almost sharp spectral asymptotics.

Furthermore, in Section~\ref{sect-26-5} we analyze in the case $M\ge 2$ the \emph{boundary strip\/}
$\{W^\TF +\nu \lesssim B\}$ containing the boundary of $\supp (\rho^\TF_B)$; this is the most difficult case to analyze and our remainder estimates are not sharp unless $N\ge Z- CZ^{\frac{2}{3}}$.

Finally, in Section~\ref{sect-26-6} we derive asymptotics of the ground state energy. Their precision (or lack of it) follows from the precision of the corresponding semiclassical results; so our results in the case $M=1$ are sharp, but our results in the case $M\ge 2$ (especially if $N\le Z- CZ^{\frac{2}{3}}$) are not.

\smallskip
In the second half of this Chapter we consider related problems. In Section~\ref{sect-26-7} (cf. Section~\ref{book_new-sect-25-5}) we consider negatively charged systems ($N\ge Z$) and estimate both ionization energy $\I_N$ and excessive negative charge $(N-Z)_+$\,\footnote{\label{foot-26-4} In the (magnetic) Thomas-Fermi theory both answers are $0$.}.

In Section~\ref{sect-26-8} (cf. Section~\ref{book_new-sect-25-6}) we consider positively charged systems ($N\le Z$) and estimate the remainder $|\I_N+\nu|$ in the formula $\I_N \approx -\nu$; for $M\ge 2$ we also consider a free nuclei model and estimate from below the distance between nuclei and an excessive positive charge $(Z-N)_+$ when atoms can be bound into molecule\footref{foot-26-4}.

Appendices contain some auxiliary material, most notably, electrostatic inequalities in Appendix~\ref{sect-26-A-1} and also Zhislin's theorem (that system can bind at least $Z$ electrons) in Appendix~\ref{sect-26-A-4}--all in the case of magnetic field.

\chapter{Magnetic Thomas-Fermi theory}
\label{sect-26-2}

\section{Framework and existence}
\label{sect-26-2-1}

The Thomas-Fermi theory is well developed in the magnetic case as well albeit in the lesser degree than in the non-magnetic one. The most important source now is
Section IV of E.~H.~Lieb, J.~P.~Solovej and J.~Yngvarsson~\cite{LSY2}.

Again as in the previous Chapter~\ref{book_new-sect-25} to get the best lower estimate for the ground state energy (neglecting semiclassical errors) one needs to maximize functional $\Phi_{B,*}(W+\nu)$ defined by (\ref{book_new-25-3-1}) albeit with the pressure $P_B(w)$ given for $d=2,3$ by (\ref{26-1-19}). Formulae (\ref{book_new-25-3-2}) and (\ref{book_new-25-3-3}) also remain valid.

Further, to get the best upper estimate (neglecting semiclassical errors) one needs to minimize functional $\Phi_{B}^*(\rho',\nu)$ defined by (\ref{book_new-25-3-4}) where (\ref{book_new-25-3-4}) remains valid with $P$ replaced by $P_B$ and respectively $\tau(\rho')$ replaced by $\tau_B (\rho')$ which is Legendre transformation of $P_B$ (see~(\ref{26-1-18})).

Since $P_B$ is given by much more complicated expression (\ref{26-1-19}) rather than $\textup{(25.3.6)}_1$, and respectively
\begin{equation}
P'_B(w)=\frac{d}{2}\varkappa _1 B\Bigl(\frac{1}{2} w_+ ^{\frac{d}{2}-1} +
\sum _{j\ge 1} (w-2jB)_+ ^{\frac{d}{2}-1}\Bigr)
\label{26-2-1}
\end{equation}
(cf. $\textup{(25.3.6)}_2$), there is no explicit expression for $\tau_B$ similar to (\ref{book_new-25-3-7}).

\begin{remark}\label{rem-26-2-1}
\begin{enumerate}[label=(\roman*), wide, labelindent=0pt]
\item\label{rem-26-2-1-i}
$B(x)=|\nabla \times A(x)|$.

\item\label{rem-26-2-1-ii}
From now on we will assume that $d=3$.

\item\label{rem-26-2-1-iii}
$P_B$ is a strictly convex function and therefore $\tau_B$ is also a strictly convex function\footnote{\label{foot-26-5} As $d=2$, $P_B$ is a convex and piecewise linear function and therefore $\tau_B$ is also a convex function.}.

\item\label{rem-26-2-1-iv}
$P_B(w)\to P_0(w)$, $P'_B(w)\to P'_0(w)$ and
$\tau_B(\rho)\to \tau_0(\rho)$ as $B\to 0$ where (without subscript ``$0$'') the limit functions have been defined by $\textup{(25.3.6)}_{1,2}$ and (\ref{book_new-25-3-7}) respectively.
\end{enumerate}
\end{remark}

\begin{remark}\label{rem-26-2-2}
\begin{enumerate}[label=(\roman*), wide, labelindent=0pt]
\item\label{rem-26-2-2-i}
Alternatively we minimize $\cE_B(\rho)=\Phi ^*_B(\rho, 0)$ under assumptions
\begin{phantomequation}\label{26-2-2}\end{phantomequation}
\begin{equation}
\rho \ge 0,\qquad \int \rho \,dx\le N.
\tag*{$\textup{(\ref*{26-2-2})}_{1,2}$}\label{26-2-2-*}
\end{equation}

\item\label{rem-26-2-2-ii}
So far in comparison with the previous Chapter~\ref{book_new-sect-25} we changed only definition of $P_B(w)$ and $\tau_B(\rho)$ respectively. Note that $P_B(w)$ belongs to $\sC^{\frac{d}{2}+1}$ (as $d=2,3$) as function of $w$; this statement will be quantified later.

\item\label{rem-26-2-2-iii}
While not affecting existence (with equality in $\textup{(\ref{26-2-2})}_1$ iff $N\le Z$) and uniqueness of solution, it affects other properties, especially as $B\ge Z^{\frac{4}{3}}$.
\end{enumerate}
\end{remark}

\begin{proposition}\label{prop-26-2-3}
In our assumptions for any fixed $\nu \le 0$ Statements \ref{book_new-prop-25-3-1-i}--\ref{book_new-prop-25-3-1-viii} of Proposition~\ref{book_new-prop-25-3-1} hold.
\end{proposition}

\begin{proof} The proof is the same as of Proposition~\ref{book_new-prop-25-3-1}. The proof that threshold $\nu=0$ matches to $N=Z$ are theorems 4.9 and 4.10 of
Section IV of E.~H.~Lieb, J.~P.~Solovej and J.~Yngvarsson~\cite{LSY2}. \end{proof}

Note that (\ref{book_new-25-3-8})--(\ref{book_new-25-3-9}) and (\ref{book_new-25-3-10}) become
\begin{gather}
\rho =\frac{1}{4\pi}\Delta (W-V)=P'_B (W+\nu ), \label{26-2-3}\\
W=o(1)\qquad \text{as\ \ }|x|\to \infty \label{26-2-4}\\
\shortintertext{and}
{\cN}(\nu )=\int P'_B (W+\nu )\, dx
\label{26-2-5}
\end{gather}
respectively.

Similarly, Proposition~\ref{book_new-prop-25-3-2} remains true:

\begin{proposition}\label{prop-26-2-4}
For arbitrary $W$ the following estimates hold with absolute constants
$\epsilon _0>0$ and $C_0$:
\begin{multline}
\epsilon _0 \D(\rho -\rho ^\TF,\rho -\rho ^\TF)\le
\Phi_{B,*} (W ^\TF+\nu )-\Phi_{B,*} (W+\nu )\le \\
C_0 \D(\rho -\rho ',\rho -\rho ')\label{26-2-6}
\end{multline}
and
\begin{multline}
\epsilon _0 \D(\rho '-\rho ^\TF,\rho '-\rho ^\TF)\le
\Phi _B^*(\rho ,\nu )-\Phi_B ^*(\rho ^\TF,\nu )\le\\
C_0 \D(\rho -\rho ',\rho -\rho ')\label{26-2-7}
\end{multline}
with $\rho =\frac{1}{4\pi}\Delta (W-V)$, $\rho '=P'_B(W+\nu )$.
\end{proposition}

\begin{proof} This proof is rather obvious as well. \end{proof}

\section{Properties}
\label{sect-26-2-2}

\begin{proposition}\label{prop-26-2-5}
The solution of the magnetic Thomas-Fermi problem has the following scaling properties
\begin{multline}
W^\TF (x;\, \underline{Z};\, \underline{y};\, B;\,N;\, q)=\\
q^{\frac{2}{3}} N^{\frac{4}{3}} W^\TF (q^{\frac{2}{3}} N^{\frac{1}{3}}x;\,
N^{-1}\underline{Z};\, q^{\frac{2}{3}} N^{\frac{1}{3}}\underline{\y}; \,q^{-\frac{2}{3}}N^{-\frac{4}{3}}B;\,1; 1),
\label{26-2-8}
\end{multline}
\vglue-25pt
\begin{multline}
\rho^\TF (x;\, \underline{Z};\, \underline{y};\, B;\,N;\, q)=\\
q^2 N^2 \rho^\TF (q^{\frac{2}{3}} N^{\frac{1}{3}}x;\,
N^{-1}\underline{Z};\, q^{\frac{2}{3}} N^{\frac{1}{3}}\underline{\y}; \,q^{-\frac{2}{3}}N^{-\frac{4}{3}}B;\,1; 1),
\label{26-2-9}
\end{multline}
\vglue-25pt
\begin{align}
\cE^\TF (\underline{Z};\, \underline{y};\, B;\, N;\, q)&=
q^{\frac{2}{3}} N^{\frac{7}{3}} \cE^\TF
(N^{-1}\underline{Z};\, q^{\frac{2}{3}}N^{\frac{1}{3}}\underline{\y};\, q^{-\frac{2}{3}}N^{-\frac{4}{3}}B;\, 1;\, 1),
\label{26-2-10}\\[6pt]
\nu^\TF (\underline{Z};\, \underline{y};\,B;\, N;\, q) &=
q^{\frac{2}{3}} N^{\frac{4}{3}} \nu^\TF (N^{-1}\underline{Z};
\, q^{\frac{2}{3}} N^{\frac{1}{3}}\underline{\y};\, q^{-\frac{2}{3}}N^{-\frac{4}{3}}B;\, 1;\, 1)
\label{26-2-11}
\end{align}
where $\nu^\TF=\nu$ is the chemical potential; recall that $\underline{Z}=(Z_1,\ldots,Z_M)$ and $\underline{\y}=(\y_1,\ldots,\y_M)$ are arrays and parameter $q$ also enters into Thomas-Fermi theory.

In particular, $\nu^\TF$ and $B$ scale the same way.
\end{proposition}

\begin{proof}
Proof is trivial by scaling. \end{proof}

Now one can guess that there are two cases $B \ll Z^{\frac{4}{3}}$ and
$B \gg Z^{\frac{4}{3}}$ (recall that $N\asymp Z$) in which magnetic Thomas-Fermi theory looks very different (and also an intermediate case $B\sim Z^{\frac{4}{3}}$). To explain this difference let us consider one atom case:

%\begin{wrapfigure}[9]{l}[4pt]{6.5truecm}
% \begin{tikzpicture}[scale=1.65]
% \draw [thin,->] (0,0)--(3.3,0);
% \draw [thin,->] (0,0)--(0,2.8);
% \draw [blue,thick] plot [domain=0:3,smooth,samples=1000] function{10*(x**(1.5))*(1+ x**2)**(-2.25)} node [right] {$W^{\frac{3}{2}}r^3$} ;
% \draw [red,thick] plot [domain=0:3,smooth,samples=1000] function{7*(x**(.5))*(1+x**2)**(-4.75)} node [below] {$W^{\frac{5}{2}}r^3$};
% \draw [blue, dashed] (.7,0)--(.7,2.8);
% \draw [red, dashed] (.23,0)--(.23,2.8);
% \end{tikzpicture}
% \label{fig-26-1}
% \end{wrapfigure}
First of all recall that if $B=0$ and $N=Z$ theory (as $M=1$) has just one parameter and we can get rid off it by rescaling; $W^\TF \asymp Z \ell^{-1}$ as
$\ell\lesssim Z^{-\frac{1}{3}}$ \underline{and} $W^\TF \asymp \ell^{-4}$ as
$\ell\gtrsim Z^{-\frac{1}{3}}$.
Then
\begin{gather*}
\W^{\TF\, \frac{3}{2}}\ell^3 \asymp Z^{\frac{3}{2}} \ell^{\frac{3}{2}},\quad
\W^{\TF\, \frac{5}{2}}\ell^3 \asymp Z^{\frac{5}{2}} \ell^{\frac{1}{2}}\\ \intertext{\underline{and}}
\W^{\TF\, \frac{3}{2}}\ell^3 \asymp \ell^{-3},\quad
\W^{\TF\, \frac{5}{2}}\ell^3 \asymp \ell^{-7}
\end{gather*}
respectively where the first factors are spacial densities of the charge and (negative) Thomas-Fermi energy respectively and therefore zone $\ell \asymp Z^{-\frac{1}{3}}$ provides the main contributions into both.

Therefore, if in this \emph{main zone} $B\ll W^\TF \asymp Z^{\frac{4}{3}}$ we guess that the magnetic theory is similar to non-magnetic one, and actually it is true.

However, let us study an atomic case rigorously. Let $M=1$, $\y_m=0$ and $N\le Z$. Then
\begin{claim}\label{26-2-12}
$W_B^\TF $ is a spherically symmetric, and it is monotone non-increasing function of $|x|$; $W_B^\TF\to +0$ as $|x|\to \infty$;
\end{claim}
\vglue-25pt
\begin{equation}
W^\TF_B (x)\le -\nu \implies W^\TF_B = |x|^{-1}(Z-N).
\label{26-2-13}
\end{equation}
Indeed, (\ref{26-2-12}) is obvious and (\ref{26-2-13}) follows from it and Newton screening theorem.

Two propositions below treat cases $B\lesssim Z^{\frac{4}{3}}$ and
$B\gtrsim Z^{\frac{4}{3}}$ respectively; in the former case there is another fork: $B\lesssim (Z-N)_+^{\frac{4}{3}}$ and
$B\gtrsim (Z-N)_+^{\frac{4}{3}}$.

\begin{proposition}\label{prop-26-2-6}
Let $M=1$, $\y_m=0$, $N\asymp Z_m$ and $B\le Z^{\frac{4}{3}}$.

\begin{enumerate}[label=(\roman*), wide, labelindent=0pt]
\item\label{prop-26-2-6-i}
Then
\begin{gather}
W^\TF_B \le \min(Z|x|^{-1}, C|x|^{-4}) \label{26-2-14}\\
\shortintertext{and}
\rho^\TF_B \le
C\min(Z^{\frac{3}{2}}|x|^{-\frac{3}{2}}+ BZ^{\frac{1}{2}}|x|^{-\frac{1}{2}}, |x|^{-6}+ B |x|^{-2}).
\label{26-2-15}
\end{gather}
\item\label{prop-26-2-6-ii}
There exists
\begin{equation}
\bar{r}_m\asymp \min \bigl(B^{-\frac{1}{4}}, (Z-N)_+^{-\frac{1}{3}}\bigr)
\label{26-2-16}
\end{equation}
such that $W^\TF _B \gtrless -\nu$ as $|x|\lessgtr \bar{r}_m$ and then
$\rho^\TF_B= 0$ iff $x\ge \bar{r}_m$.

\item\label{prop-26-2-6-iii}
\textup{(\ref{26-2-14})} and \textup{(\ref{26-2-15})} become equivalencies ($\asymp$) as $|x|\le (1-\epsilon) \bar{r}_m$.

\item\label{prop-26-2-6-iv}
$B\le (Z-N)_+^{\frac{4}{3}}$ implies $\bar{r}_m \asymp (Z-N)_+^{-\frac{1}{3}}$, $\nu \asymp (Z-N)_+^{\frac{4}{3}}$ and
\begin{multline}
W^\TF +\nu \asymp (Z-N)_+^{\frac{5}{3}}(\bar{r}_m-|x|),\\
-\partial_{|x|} W^\TF \asymp (Z-N)_+^{\frac{5}{3}} \qquad\qquad\qquad
\text{if\ \ } (1-\epsilon)\bar{r}_m\le |x|\le \bar{r}_m.
\label{26-2-17}
\end{multline}
\item\label{prop-26-2-6-v}
$B\ge (Z-N)_+^{\frac{4}{3}}$ implies $\bar{r}_m \asymp B^{-\frac{1}{4}}$,
$\nu \asymp (Z-N)_+ B^{\frac{1}{4}}\lesssim B$ and
\begin{multline}
W^\TF +\nu \asymp
B^2 (\bar{r}_m-|x|)^4 + B^{\frac{1}{2}}(Z-N)_+(\bar{r}_m-|x|)\\
\shoveleft{-\partial_{|x|} W^\TF \asymp
B^2 (\bar{r}_m-|x|)^3 + B^{\frac{1}{2}}(Z-N)_+} \\
\text{as \ \ } (1-\epsilon)\bar{r}_m\le |x|\le \bar{r}_m.
\label{26-2-18}
\end{multline}
\end{enumerate}
\end{proposition}

\begin{proposition}\label{prop-26-2-7}
Let $M=1$, $\y_m=0$, $N\asymp Z_m$ and $B\ge Z^{\frac{4}{3}}$.
\enlargethispage{2\baselineskip}

\begin{enumerate}[label=(\roman*), wide, labelindent=0pt]
\item\label{prop-26-2-7-i}
Then
\begin{gather}
W^\TF_B \le Z|x|^{-1} \label{26-2-19}\\
\shortintertext{and}
\rho^\TF_B \le
C Z^{\frac{3}{2}}|x|^{-\frac{3}{2}}+ CB Z^{\frac{1}{2}}|x|^{-\frac{1}{2}}.
\label{26-2-20}
\end{gather}
\item\label{prop-26-2-7-ii}
There exist $\bar{r}_m $ and $\bar{r}'_m$,
\begin{equation}
\bar{r}_m\asymp B^{-\frac{2}{5}}Z^{\frac{1}{5}}, \qquad \bar{r}'_m\asymp B^{-1}Z_m,
\label{26-2-21}
\end{equation}
such that $W^\TF _B \gtrless B$ as $|x|\lessgtr \bar{r}_m$,
$W^\TF _B \gtrless -\nu$ as $|x|\lessgtr \bar{r}'_m$ and then $\rho^\TF_B= 0$ iff $x\ge \bar{r}_m$.
\item\label{prop-26-2-7-iii}
\textup{(\ref{26-2-19})}--\textup{(\ref{26-2-20})} become equivalencies ($\asymp$) as $|x|\le (1-\epsilon)\bar{r}_m$.

\item\label{prop-26-2-7-iv}
$\nu \asymp (Z-N) _+ B^{\frac{2}{5}}Z^{-\frac{1}{5}} \lesssim B$ and
\begin{gather}
W^\TF +\nu \asymp B^2 (\bar{r}_m-|x|)^4 + \bar{r}_m^{-2}(Z-N)_+(\bar{r}_m-|x|)
\label{26-2-22}\\
\shortintertext{and}
-\partial_{|x|} W^\TF \asymp B^2 (\bar{r}_m-|x|)^3 + \bar{r}_m^{-2}(Z-N)_+
\label{26-2-23}
\end{gather}
as $(1-\epsilon)\bar{r}_m\le |x|\le \bar{r}_m$.
\end{enumerate}
\end{proposition}

\begin{proof}[Proofs of Propositions~\ref{prop-26-2-6} and~\ref{prop-26-2-7}]
Proofs easily follow from equation and ``boundary conditions'' satisfied by $w(r)$ where $r=|x|$:
\begin{gather}
w'' + 2r^{-1}w =P'_B(w+\nu),\label{26-2-24}\\[2pt]
w= r^{-1}Z_m + O(1)\qquad \text{as\ \ } r\to 0,\label{26-2-25}\\[2pt]
w(\bar{r}_m)=-\nu,\qquad w'(\bar{r}_m)= \nu \bar{r}_m^{-1}\label{26-2-26}
\end{gather}
where $\nu = -(Z_m-N)_+\bar{r}_m^{-1}$. \end{proof}

\begin{corollary}\label{cor-26-2-8}
Let $M=1$, $\y_m=0$ and $N\asymp Z_m$. Then
\begin{enumerate}[label=(\roman*),wide, labelindent=0pt]
\item\label{cor-26-2-8-i}
$W^\TF_B\lesssim B$ if $|x|\ge \bar{r}'_m $ where $\bar{r}'_m\asymp B^{-1}Z_m$ as $B\ge Z_m^{\frac{4}{3}}$ and $\bar{r}'_m \asymp B^{-\frac{1}{4}}$ as
$B\le c Z_m^{\frac{4}{3}}$.

\item\label{cor-26-2-8-ii}
As $B\lesssim Z_m^{\frac{4}{3}}$ the main contribution to both the charge and the Thomas-Fermi energy is delivered by zone $\{x\colon |x|\asymp r^*_m \}$ with $r^*_m=Z_m^{-\frac{1}{3}}$; in particular, then
$\cE_B^\TF \asymp \cE^\TF \asymp Z_m^{\frac{7}{3}}$; further, in this case $W_B^\TF \asymp W^\TF$ in the zone $\{x\colon |x|\lesssim \epsilon \bar{r}_m\}$.

\item\label{cor-26-2-8-iii}
Further, $\cE_B^\TF \sim \cE^\TF$ as $B\ll Z_m^{\frac{4}{3}}$; furthermore, in this case $W_B^\TF \sim W^\TF$ in the zone $\{x\colon |x|\ll \bar{r}_m\}$.

\item\label{cor-26-2-8-iv}
On the other hand, as $B\ge Z^{\frac{4}{3}}$, the main contributions to the total charge and energy are delivered by $\{x\colon |x|\asymp \bar{r}_m\}$ and in particular $\rho _m\asymp B Z_m^{\frac{1}{2}}\bar{r}_m^{\frac{5}{2}}$ and
\begin{equation}
\cE^\TF_B \asymp B Z_m^{\frac{3}{2}}\bar{r}_m^{\frac{3}{2}}\asymp B^{\frac{2}{5}}Z_m^{\frac{9}{5}}.
\label{26-2-27}
\end{equation}
\end{enumerate}
\end{corollary}

Recall that $\bar{r}_m\asymp B^{-\frac{1}{4}}$ as $B\le Z_m^{\frac{4}{3}}$ and $\bar{r}_m\asymp B^{-\frac{2}{5}}Z_M^{\frac{1}{5}}$ as $B\ge Z_M^{\frac{4}{3}}$. Note that Proposition~\ref{book_new-prop-25-3-5} (comparing $W^\TF$ for molecule with the sum of those for single atoms) still holds. Therefore we conclude that

\begin{corollary}\label{cor-26-2-9}
\begin{enumerate}[label=(\roman*),wide, labelindent=0pt]
\item\label{cor-26-2-9-i}
Assume that
\begin{equation}
Z_m\asymp Z
\label{26-2-28}
\end{equation}
for all $m=1,\ldots, M$. Then all statements of corollary~\ref{cor-26-2-8} remain true for $M\ge 2$ with $|x|$ and $Z_m$ replaced by $\ell(x)$ and $Z$ and $\bar{r}_m$, $\bar{r}'_m$, $r^*_m$ by $\bar{r}$, $\bar{r}'$, $r^*$ respectively.

\item\label{cor-26-2-9-ii}
In the general case global statements remain true, pointwise statements remain true without modification only as $\ell(x)=\ell_m(x)\coloneqq |x-\y_m|$ with $Z_m\asymp Z$.
\end{enumerate}
\end{corollary}

\begin{remark}\label{rem-26-2-10}
\begin{enumerate}[label=(\roman*),wide, labelindent=0pt]
\item\label{rem-26-2-10-i}
Also holds Proposition~\ref{book_new-prop-25-3-13} as it uses only super-additivity of $\tau(\rho)$ and $\tau _B(\rho)$ is also super-additive (this follows from convexity of $\tau_B(\rho)$ and equality $\tau_B(0)=0$).

\item\label{rem-26-2-10-ii}
However there is a significant difference: if there is no magnetic field atoms really repulse one another on any distances and we can attribute it to either excessive positive charge as $N<Z$ or their infinite spatial size as $N=Z$. However with magnetic field atoms have a finite size even as $N=Z$ and they do not repulse one another on the large distances. In particular, Proposition~\ref{prop-26-2-11} below holds.
\end{enumerate}
\end{remark}

\begin{proposition}\label{prop-26-2-11}
Let $N=Z$ and
\begin{gather}
|\y_{m}-\y_{m'}|\ge \bar{r}_m +\bar{r}_{m'}\qquad \forall m\colon \ 1\le m<m'\le M.
\label{26-2-29}\\
\shortintertext{Then}
\cE^\TF_B (\underline{Z}, \underline{\y}, B,Z)=
\sum _{1\le m\le M} \cE^\TF_B (Z_m, \y_m, B,Z_m)\label{26-2-30}\\
\shortintertext{and}
\rho^\TF_B (x, \underline{Z}, \underline{\y}, B,Z)=
\sum _{1\le m\le M} \rho^\TF_B (x,Z_m, \y_m, B,Z_m).\label{26-2-31}
\end{gather}
\end{proposition}

\begin{proposition}\label{prop-26-2-12}
\begin{enumerate}[label=(\roman*), wide, labelindent=0pt]
\item\label{prop-26-2-12-i}
$\nu$ is monotone increasing function of $N$.

\item\label{prop-26-2-12-ii}
$W_B(x)$ is monotone non-increasing function of $N$.

\item\label{prop-26-2-12-iii}
$W_B(x)+\nu$ is monotone non-decreasing function of $N$ ; in particular
$\rho_B$ can only increase as $N$ increases.

\item\label{prop-26-2-12-iv}
$\nu$ is monotone non-increasing function of $Z_m$.

\item\label{prop-26-2-12-v}
$W_B(x)$ is monotone non-decreasing function of $Z_m$.
\end{enumerate}
\end{proposition}

\begin{proof}
\begin{enumerate}[label=(\roman*), wide, labelindent=0pt]
\item\label{pf-26-2-12-i}
Statement~\ref{prop-26-2-12-i} follows from the strict convexity of $\cE(\rho)$: consider two solutions with corresponding subscripts. Then $\cE(\rho)-\cE(\rho_j)> \nu_j(N-N_j)$ for any non-negative $\rho\ne \rho_j$ and $N=\int \rho\,dx$.

In particular,
$\cE(\rho_1)-\cE(\rho_2)> \nu_2(N_1-N_2)$ and
$\cE(\rho_2)-\cE(\rho_1)> \nu_1(N_2-N_1)$ and then $(\nu_1-\nu_2)(N_1-N_2)>0$.

\item\label{pf-26-2-12-ii}
Indeed, consider $N_1<N_2$ and in the definition of $W_2$ slightly decrease $Z_1,\ldots, Z_M$ thus replacing them by $Z'_1,\ldots, Z'_M$. Then $W_1> W_2$ for large $|x|$, $W_1-W_2\to +\infty$ as $x\to \y_m$ and therefore if Statement~\ref{prop-26-2-12-ii} fails, then $W_1-W_2$ reaches non-positive minimum at some regular point $\bar{x}$; at this point $W_1\le W_2$ and
\begin{equation*}
0\le \frac{1}{4\pi}\Delta (W_1-W_2) =P'(W_1+\nu_1) -P'(W_2+\nu_2).
\end{equation*}
This is possible only if at this point $W_2+\nu_2\le 0$ and $W_1+\nu_1<0$. Then in the small vicinity $\Delta (W_1-W_2)\le 0$ and $\bar{x}$ cannot be a point of minimum unless $W_1-W_2=\const$ there. Then any point of this vicinity is also a point of minimum and then due to standard analytic arguments $W_1-W_2=\const$ everywhere which is impossible.

So, $W_1(x;Z_1,\ldots,Z_M)> W_2(x;Z'_1,\ldots,Z'_M)$. Taking limit as
$Z'_m\to Z_m$ we arrive to $W_1(x;Z_1,\ldots,Z_M)\ge W_2(x;Z_1,\ldots,Z_M)$.

\item\label{pf-26-2-12-iii}
Proof of Statement~\ref{prop-26-2-12-iii} is similar but roles of $W_1$ and $W_2$ are played by $W_2+\nu_2$ and $W_1+\nu_1$ respectively.

\item\label{pf-26-2-12-iv}
Let $Z_{m,2}>Z_{m,1}$ for all $m$. Assume that $\nu_2>\nu_1$. Then similar arguments prove that $W_2+\nu_2\ge W_1+\nu_1$ and thus $\rho_2\ge \rho_1$ everywhere which is impossible unless there are just identical equalities as $W_2+\nu_2>0$, which is impossible.
\enlargethispage{\baselineskip}

\item\label{pf-26-2-12-v}
Finally, after Statement~\ref{prop-26-2-12-iv} was established, the same arguments prove Statement~\ref{prop-26-2-12-v}.
\end{enumerate} \end{proof}

As far as we know Theorem~1 of~R.~Benguria~\cite{benguria} (see Theorem~\ref{book_new-thm-25-3-8}) has not been proven in the case of magnetic field; however one can see easily that arguments of of~R.~Benguria's proof remain valid and we arrive to

\begin{theorem}\label{thm-26-2-13}
All Statements~\ref{book_new-thm-25-3-8-i}--\ref{book_new-thm-25-3-8-iii} of Theorem~\ref{book_new-thm-25-3-8} hold in the case of the constant magnetic field.
\end{theorem}

\begin{Problem}\label{Problem-26-2-14}
\begin{enumerate}[label=(\roman*), wide, labelindent=0pt]
\item\label{Problem-26-2-14-i}
Investigate how $\supp (\rho^\TF_B)$ depends on $B$ and on $Z$ in the atomic case $M=1$.

\item\label{Problem-26-2-14-ii}
More generally, investigate how $\supp (\rho^\TF _B)$ depends on $B$ and on $Z$ in the case $M\ge 2$.
\end{enumerate}
\end{Problem}

\section{Positive ions}
\label{sect-26-2-3}

In view of Remark~\ref{rem-26-2-10} we need to consider repulsion of positive ions in more details. Our purpose is to prove

\begin{theorem}\label{thm-26-2-15}
Let condition \textup{(\ref{26-2-28})} be fulfilled. Then the energy excess is estimated from below
\begin{equation}
\cQ\coloneqq \widehat{\cE}^\TF_B - \sum_{1\le m \le M} \cE^\TF_{B,m} \ge
\epsilon (Z-N)_+^2 a^{-1}.
\label{26-2-32}
\end{equation}
\end{theorem}

Note first that
\begin{multline}
\D\bigl(\rho^\TF_{B(\nu)}-\rho^\TF_{B,0},\,
\rho^\TF_{B(\nu)}-\rho^\TF_{B(0)}\bigl) +\\
\int \bigl( P'_B(W ^\TF_{B(\nu)}+\nu)-P'_B(W ^\TF_{B(0)})\bigr)
\bigl( W ^\TF_{B(\nu)}+\nu-W ^\TF_{B(0)})\bigr)\,dx=\\
\nu \int \bigl( P'_B(W ^\TF_{B(\nu)}+\nu)-P'_B(W ^\TF_{B(0)}+0)\bigr)\,dx
\label{26-2-33}
\end{multline}
with the right-hand expression equal $\nu (N-Z)\asymp (Z-N)^2\bar{r}^{-1}$ and due to monotonicity $P'_B(w)$ we conclude that

\begin{proposition}\label{prop-26-2-16}
Let condition \textup{(\ref{26-2-28})} be fulfilled. Then
\begin{equation}
\D\bigl(\rho^\TF_{B(\nu)}-\rho^\TF_{B(0)},\,\rho^\TF_{B,\nu}-\rho^\TF_{B(0)}\bigl) \le C (Z-N)^2\bar{r}^{-1}.
\label{26-2-34}
\end{equation}
\end{proposition}

\begin{proof}[Proof of Theorem~\ref{thm-26-2-15}]
\begin{enumerate}[label={\emph{Step \arabic*.\/}}, wide, labelindent=0pt]

\item\label{pf-26-2-15-1}
 Note first that due to non-negativity of the expression
\begin{equation}
\widehat{\cE}^\TF_B (\underline{Z},\underline{\y},N)-
\min _{N_1+N'=N} \bigl(\cE^\TF_B (Z_1,N_1) -
\widehat{\cE}^\TF_B (\underline{Z}',\underline{\y}',N')\bigr)
\label{26-2-35}
\end{equation}
(see proof of Proposition~\ref{book_new-prop-25-3-13} which persists even if there is constant magnetic field, see Remark~\ref{rem-26-2-10}) it is sufficient to prove theorem only for $M=2$. From now on we assume that $M=2$.

\item\label{pf-26-2-15-2}
According to Proposition~\ref{book_new-prop-25-3-13}
\begin{equation}
\D (\rho_B^\TF - \rho_{B,1}^\TF - \rho_{B,2}^\TF)\le C\cQ.
\label{26-2-36}
\end{equation}
Therefore due to superadditivity $\tau_B$
\pagebreak
\begin{multline}
\cQ\ge -\int V_1 \rho_{B,2}^\TF \,dx - \int V_2 \rho_{B,1}^\TF \,dx +\\
\D(\rho_{B,2}^\TF, \rho_{B,1}^\TF)+ Z_1Z_2a^{-1}- C\cQ
\label{26-2-37}
\end{multline}
and it is sufficient to prove the same estimate from below for the right-hand expression without the last term. However this is easy if
$a\ge \bar{r}_1+\bar{r}_2$ since $V_m= |x-\y_m|^{-1}Z_m$ and $\rho_{B,m}^\TF=\rho_{B,m}^\TF (|x-\y_m|)$ are spherically symmetric functions\footnote{\label{foot-26-6} However this is not true in general as $a<\bar{r}_1+\bar{r}_2$. Really, consider $N_m=Z_m$ and uniformly charged spheres. Then the right-hand expression of (\ref{26-2-37}) is $0$ as
$a\ge \bar{r}_1+\bar{r}_2$ and is negative and decays as $a$ decays from $\bar{r}_1+\bar{r}_2$ to $\max (\bar{r}_1,\bar{r}_2)$ and it increases again as $a$ decays from $\max (\bar{r}_1,\bar{r}_2)$ to $0$.}.

Therefore for $a\ge \bar{r}_1+\bar{r}_2$ inequality (\ref{26-2-32}) has been proven and in what follows we can assume that $a\le \bar{r}_1+\bar{r}_2$. Further, applying Theorem~\ref{thm-26-2-13} we conclude then that

\begin{claim}\label{26-2-38}
Inequality (\ref{26-2-32}) holds for $a\ge \epsilon\bar{r}$.
\end{claim}
\enlargethispage{2\baselineskip}

\item\label{pf-26-2-15-3}
 Recall that the bulk of electrons are in the zone
$\{\ell(x)\asymp r^*\}$\,\footnote{\label{foot-26-7} I.e. zone $\{c(\epsilon)^{-1}r^*\le \ell(x) \le c(\epsilon)r^*\}$ contains at least $(1-\epsilon)N$ electrons.}. Based on this one can prove easily that as
$a\le \epsilon\bar{r}$ the right-hand hand expression of (\ref{26-2-37}) is greater than $(1-\epsilon_1) a^{-1}Z_1Z_2$ and therefore

\begin{claim}\label{26-2-39}
As $B\ge Z^{\frac{4}{3}}$ and $a\le \epsilon r^*$ we have
$\cQ\ge (1-\epsilon_1) a^{-1}Z_1Z_2$
\end{claim}
and combining with (\ref{26-2-38}) we conclude that (\ref{26-2-32}) holds for $B\gtrsim Z^{\frac{4}{3}}$ and for $B\lesssim Z^{\frac{4}{3}}$ we need to consider the case $\epsilon_0 r^*\le a\le \epsilon \bar{r}$ with arbitrarily small constant $\epsilon$.

Replacing then $P_B$ by $P_0$ and noting that an error will not exceed
$C_0 \bar{r} B^2\le C_1\epsilon a^{-7}$ while $\cQ\ge \epsilon_0 a^{-7}$ for $B=0$ we conclude that (\ref{26-2-32}) holds as $\epsilon_0 r^*\le a \le c\bar{r}$ and $(Z-N)\le C_2 a^{-3}$.

Finally, as $(Z-N)\ge C_2a^{-3}$ we see that
$\bar{r}\le C_0(Z-N)^{-\frac{1}{3}}\le \epsilon a$ and (\ref{26-2-32}) holds again.
\end{enumerate}
\end{proof}

Even if we do not need it for our purposes we want to consider the repulsion of too close neutral atoms:

\begin{theorem}\label{thm-26-2-17}
Let condition \textup{(\ref{26-2-28})} be fulfilled and $N=Z$. Then as
$a\ge \epsilon\bar{r}$ the energy excess is estimated from below
\begin{gather}
\cQ \ge
\epsilon G^2\bar{r} \sum_{1\le m<m'\le M} (\bar{r}_m+\bar{r}_n-|\y_m-\y_{m'}|)_+^{12}\bar{r}^{-12}
\label{26-2-40}\\
\shortintertext{where}
G\coloneqq \left\{\begin{aligned}
&B\qquad&&\text{if\ \ } B\le Z^{\frac{4}{3}},\\
&Z^{\frac{4}{5}}B^{\frac{2}{5}} &&\text{if\ \ } B\ge Z^{\frac{4}{3}}.
\end{aligned}\right.
\label{26-2-41}
\end{gather}
and correspondingly
$G^2\bar{r}=\left\{\begin{aligned}
&B^{\frac{7}{4}}\qquad&&\text{if\ \ } B\le Z^{\frac{4}{3}},\\
&Z^{\frac{9}{5}}B^{\frac{2}{5}} &&\text{if\ \ } B\ge Z^{\frac{4}{3}}.
\end{aligned}\right.$
\end{theorem}

\begin{proof}
Again we need to consider case $M=2$. Since
\begin{gather}
\frac{1}{4\pi}\Delta W_B=\rho_B -\sum_{m=1,2} Z_m\updelta (x-\y_m)
\label{26-2-42}
\intertext{and $W_{B,1},\, W_{B,2}$ satisfy similar equations, (\ref{26-2-36}) implies that}
\|\nabla (W_B-W_{B,1}-W_{B,2})\| \le c\cQ^{\frac{1}{2}}.
\label{26-2-43}
\end{gather}
This inequality and the fact that $W_B=0$ as $\ell(x)\ge c\bar{r}$, and $W_{B,m}=0$ as $|x-\y_m|\ge \bar{r}_m$ imply that
\begin{equation}
\|(W_B-W_{B,1}-W_{B,2})\| \le c\bar{r}\cQ^{\frac{1}{2}}.
\label{26-2-44}
\end{equation}
Note that $\int (-\rho_B + \rho_{B,1}+\rho_{B,2})\,dx=0$ implies that
\begin{multline}
|\int \bigl((W_{B,1}+W_{B,2})^{\frac{1}{2}} -W_{B,1}^{\frac{1}{2}}-W_{B,1}^{\frac{1}{2}}\bigr)\,dx\le\\
\int |W_B^{\frac{1}{2}}-(W_{B,1}+W_{B,2})^{\frac{1}{2}}|\,dx.
\label{26-2-45}
\end{multline}
One can calculate easily that the left-hand expression has a magnitude
$(G\eta^4)^{\frac{1}{2}} \cdot\eta \bar{r}\cdot
(\eta ^{\frac{1}{2}}\bar{r})^2\asymp G^{\frac{1}{2}}\bar{r}^3 \eta^4$ where the first factor is a magnitude of an integrand as
$W_{B,1}\asymp W_{B,2}\asymp G\eta^4$, $\eta \bar{r}$ is a depth, and
$\eta^{\frac{1}{2}} \bar{r}$ the width of this zone.

On the other hand, consider the right hand expression. It consists of contributions of several zones:
\begin{enumerate}[label=(\alph*), wide, labelindent=0pt]
\item\label{pf-26-2-17-a}
Zone $\cY_t$ where $W_{B,1}+W_{B,2}\le G t^4$, $W_B\le 2Gt^4$. This contribution does not exceed
$CG^{\frac{1}{2}} t^2 \mes (\cY_t)\asymp CG^{\frac{1}{2}} \bar{r}^3 t^3$\,\footnote{\label{foot-26-8} Obviously,
$\mes (\cY_t\cup \cZ_t)\asymp \bar{r}^3 t$ and similarly
$\mes (\cX_\tau)\asymp \bar{r}^3 \tau$.}.

\item\label{pf-26-2-17-b}
Zone $\cZ_t$ where $W_{B,1}+W_{B,2}\le G t^4$, $W_B\ge 2Gt^4$. Its contribution does not exceed
\begin{equation*}
C\int_{\cZ_t} W_B^{\frac{1}{2}}\,dx \le
C\|W_B\|_\cZ ^{\frac{1}{2}} (\mes (\cZ_t))^{\frac{3}{4}} \le C\bar{r}^{\frac{11}{4}}\cQ^{\frac{1}{4}} t^{\frac{3}{4}}
\end{equation*}
since due to (\ref{26-2-44}) $\|W_B\|_{\cZ_t}\le c\bar{r}\cQ^{\frac{1}{2}}$.

\item\label{pf-26-2-17-c}
Zone where $W_{B,1}+W_{B,2}\asymp G\tau^4$. This contribution does not exceed
\begin{multline}
C G^{-\frac{1}{2}} \tau^{-2} \int |W_B-W_{B,1}-W_{B,2}|\,dx \le \\
CG^{-\frac{1}{2}} \tau^{-2} \times \|W_B-W_{B,1}-W_{B,2}\| \times (\mes (\cX_\tau))^{\frac{1}{2}}\asymp CG^{-\frac{1}{2}}\cQ^{\frac{1}{2}}\bar{r}^{\frac{5}{2}}\tau^{-\frac{3}{2}}.
\label{26-2-46}
\end{multline}
Integrating by $\tau^{-1}\,d\tau$ from $t$ we get (\ref{26-2-46}) calculated as $\tau=t$ (and capped by the same expression as $\tau=\eta$.
\enlargethispage{1.5\baselineskip}

So, the right-hand expression of (\ref{26-2-45}) does not exceed
\begin{equation*}
CG^{\frac{1}{2}} \bar{r}^3 t^3 +
C\cQ^{\frac{1}{4}}\bar{r}^{\frac{11}{4}} t^{\frac{3}{4}}+
CG^{-\frac{1}{2}}\cQ^{\frac{1}{2}}\bar{r}^{\frac{5}{2}}t^{-\frac{3}{2}};
\end{equation*}
optimizing with respect to $t=G^{-\frac{2}{9}}\cQ^{\frac{1}{9}}\bar{r}^{-\frac{1}{9}}$ we get all three terms equal to $CG ^{-\frac{1}{6}}Q^{\frac{1}{3}}\bar{r}^{\frac{8}{3}}$
comparing with $CG^{\frac{1}{2}}\bar{r}^3\eta^4$ we arrive to (\ref{26-2-40}).
\end{enumerate} \end{proof}

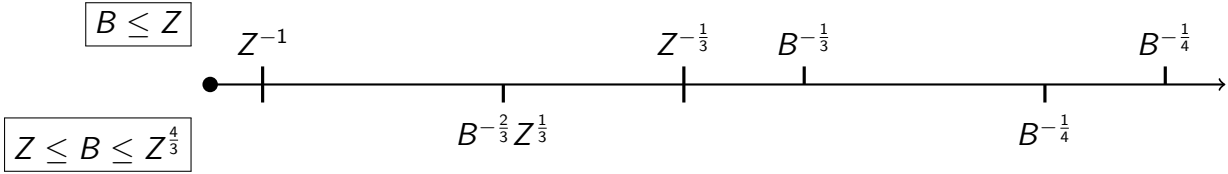
\begin{figure}[h]
\hskip-1in
\begin{tikzpicture}[scale=.8]
\draw [thick, *->] (0,0)--(17,0);
\draw [very thick] (1,-.3)--(1,.3) node[above]{$Z^{-1}$};
\draw [very thick] (8,-.3)--(8,.3) node[above]{$Z^{-\frac{1}{3}}$};
\draw [very thick] (10,-0)--(10,.3) node[above]{$B^{-\frac{1}{3}}$};
\draw [very thick] (16,0)--(16,.3) node[above]{$B^{-\frac{1}{4}}$};
\draw [very thick] (14,0)--(14,-.3) node[below]{$B^{-\frac{1}{4}}$};
\draw [very thick] (5,0)--(5,-.3) node[below]{$B^{-\frac{2}{3}}Z^{\frac{1}{3}}$};
\node[left] at (0,1) {$\boxed{B\le Z}$};
\node[left] at (0,-1) {$\boxed{Z\le B\le Z^{\frac{4}{3}}}$};
\end{tikzpicture}

\caption{\label{fig-26-1} Let $B\le Z^{\frac{4}{3}}$. Then at
$\{\ell \asymp Z^{-\frac{1}{3}}\}$ are contained both the bulk of charge and the bulk of energy, $\ell \asymp \min ((Z-N)_+^{-\frac{1}{3}}, B^{-\frac{1}{4}}\}$ is the border of $\supp (\rho^\TF_B)$; $\ell\asymp Z^{-1}$ is the Scott distance; here $h\asymp 1$.   Further, $\ell\asymp B^{-\frac{1}{3}}$ if $B\le Z$ and $\ell\asymp B^{-\frac{2}{3}}Z^{\frac{1}{3}}$ if $Z\le B\le Z^{\frac{4}{3}}$ separates $\{\mu \lesssim 1\}$ (on the left) and $\{\mu \gtrsim 1\}$ (on the right).}
\end{figure}

\begin{figure}[h]
\hskip-1in
\begin{tikzpicture}[scale=.8]
\draw [thick, *->] (0,0)--(17,0);
\draw [very thick] (1,-.3)--(1,.3) node[above]{$Z^{-1}$};
\draw [very thick] (16,0)--(16,.3) node[above]{$B^{-\frac{2}{5}}Z^{\frac{1}{5}}$};
\draw [very thick] (14,0)--(14,-.3) node[below]{$B^{-\frac{2}{5}}Z^{\frac{1}{5}}$};
\draw [very thick] (12,0)--(12,.3) node[above]{$B^{-1}Z$};
\draw [very thick] (10,0)--(10,-.3) node[below]{$B^{-1}Z$};
\draw [very thick] (5,0)--(5,.3) node[above]{$B^{-\frac{2}{3}}Z^{\frac{1}{3}}$};
\node[left] at (0,1) {$\boxed{Z^{\frac{4}{3}}\le B \le Z^2}$};
\node[left] at (0,-1) {$\boxed{Z^2\le B\le Z^{3}}$};
\end{tikzpicture}

\caption{\label{fig-26-2} Let $Z^{\frac{4}{3}}\le B\le Z^3$. Then at
$\{\ell \asymp Z^{-\frac{2}{5}}B^{\frac{1}{5}}\}$ are contained both the bulk of charge and the bulk of energy and it is also  the border of $\supp (\rho^\TF_B)$; $\ell\asymp Z^{-1}$ is the Scott distance; here $h\asymp 1$.  Further, if $Z^{\frac{4}{3}}\le B\le Z^2$
$\ell\asymp B^{-\frac{2}{3}}Z^{\frac{1}{3}}$ separates $\{\mu \lesssim 1\}$ (on the left) and $\{\mu \gtrsim 1\}$ (on the right) and $\ell\asymp B^{-1}Z$ separates  $\{\mu h\lesssim 1\}$ (on the left) and $\{\mu h\gtrsim 1\}$ (on the right).}
\end{figure}
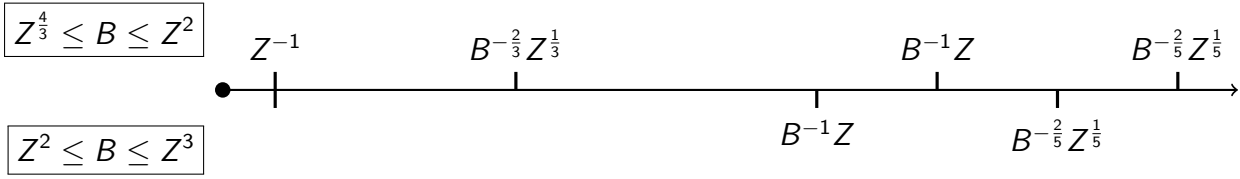
\enlargethispage{2\baselineskip}

\chapter{Applying semiclassical methods: $M=1$}
\label{sect-26-3}

\section{Heuristics}
\label{sect-26-3-1}

Let us consider first a mock proof of our main results; we deal here \emph{as if $W^\TF_B$ was very smooth\/} which it is not the case; however later we will show that its smoothness is sufficient to employ arguments of Chapter~\ref{book_new-sect-18} rather than those of Chapter~\ref{book_new-sect-13}. We also will deal as if non-degeneracy conditions were satisfied leaving them also to more rigorous arguments below.

It will allow us to establish our target remainder estimates which we will be able to prove rigorously for $M=1$ (in this section) while for $M\ge 2$ (in the next two sections) our results will be not that good.

\subsection{Total charge}
\label{sect-26-3-1-1}
Consider
\begin{equation}
\int e(x,x,\nu)\psi (x)\,dx,
\label{26-3-1}
\end{equation}
first with $\gamma$-admissible $\psi(x)$, where $\gamma \le \epsilon \ell$. Recall that $\ell(x)=\min_m|x-y_m|$ is the distance to the nearest nucleus.

\subsubsection{General arguments.}
\label{sect-26-3-1-1-1}
The main part of the semiclassical expression for (\ref{26-3-1}) is of magnitude
$h^{\prime\,-3} + \mu 'h ^{\prime\,-2}\asymp
\zeta^3\gamma^3+ B\zeta\gamma^3 $ with $h'= 1/(\zeta\gamma)$ and
$\mu'= B\gamma/\zeta$.

Indeed, let us rescale $x\mapsto x/\gamma$ and $\tau\mapsto \tau/\zeta^2$ which leads to $h=1\mapsto h'$ and $B\mapsto \mu'$. In particular, for $\gamma\asymp \ell$ we get
\begin{equation}
\zeta^3\ell^3+ B\zeta\ell^3 .
\label{26-3-2}
\end{equation}

Meanwhile, the remainder in the semiclassical expression for (\ref{26-3-1}) does not exceed $Ch^{\prime\,-2} + C\mu 'h^{\prime\,-1}\asymp
\zeta^2\gamma^2 + B\gamma^2 $ (gaining factor $h'$ in comparison to the main part; here we need the smoothness and if $\mu '\ge h^{\prime\,\delta-1}$ we also need the non-degeneracy); for $\gamma\asymp \ell$ we get
\begin{equation}
\zeta^2\ell^2+ B\ell^2 .
\label{26-3-3}
\end{equation}

Sure, we ignored the fact that $h'\le 1$ does not necessarily hold even if $\gamma \asymp \ell$ but we believe that the contributions to the main part and the remainder of these zones will be less than of zone where this inequality holds, provided $B\ll Z^3$.

Finally, let us sum expressions (\ref{26-3-2}) and (\ref{26-3-3}) with respect to $\ell$-partition.

\subsubsection{Moderate magnetic field.}
\label{sect-26-3-1-1-2}

Consider the case $B\le Z^{\frac{4}{3}}$ first. Then for $\ell \le Z^{-\frac{1}{3}}$ we plug $\zeta= Z^{\frac{1}{2}}\ell^{-\frac{1}{2}}$ into (\ref{26-3-2}) and (\ref{26-3-3}) resulting in
\begin{phantomequation}\label{26-3-4}\end{phantomequation}
\begin{equation}
Z^{\frac{3}{2}}\ell^{\frac{3}{2}}+ BZ^{\frac{1}{2}}\ell^{\frac{5}{2}}
\qquad\text{and}\qquad
Z\ell+ B\ell^2
\tag*{$\textup{(\ref*{26-3-4})}_{0,1}$}\label{26-3-4-*}
\end{equation}
in the main part and in the remainder respectively and the summation over zone
$\{x\colon \ell(x) \le Z^{-\frac{4}{3}}\}$ results in the same expressions with $\ell=Z^{-\frac{1}{3}}$, i. e. in $Z+BZ^{-\frac{1}{3}}\asymp Z$ and
$Z^{\frac{2}{3}}+BZ^{-\frac{2}{3}}\asymp Z^{\frac{2}{3}}$ respectively.

On the other hand, for $\ell \ge Z^{-\frac{1}{3}}$ we plug $\zeta=\ell^{-2}$ into
(\ref{26-3-2}) and (\ref{26-3-3}) resulting in
\begin{phantomequation}\label{26-3-5}\end{phantomequation}
\begin{equation}
\ell^{-3}+ B\ell
\qquad\text{and}\qquad
\ell^{-2}+ B\ell^2 ;
\tag*{$\textup{(\ref*{26-3-5})}_{0,1}$}\label{26-3-5-*}
\end{equation}
then summation over zone $\{x\colon Z^{-\frac{1}{3}}\le \ell(x) \le \bar{r}=B^{-\frac{1}{4}}\}$ results in $Z + B^{\frac{3}{4}}\asymp Z$ and
$Z^{\frac{2}{3}}+B^{\frac{1}{2}}\asymp Z^{\frac{2}{3}}$ respectively.

\subsubsection{Strong magnetic field.}
\label{sect-26-3-1-1-3}

Consider the case $B\ge Z^{\frac{4}{3}}$ now. Then the threshold $Z^{-\frac{1}{3}}$ disappears and we sum expressions \ref{26-3-4-*} over zone
$\{x\colon \ell (x) \le \bar{r}\coloneqq Z^{\frac{1}{5}}B^{-\frac{2}{5}}\}$, resulting in
$Z^{\frac{9}{5}}B^{-\frac{3}{5}}+ Z\asymp Z$ and
$Z ^{\frac{6}{5}}B^{-\frac{2}{5}}+Z^{\frac{2}{5}}B^{\frac{1}{5}}
\asymp Z^{\frac{2}{5}}B^{\frac{1}{5}}$ respectively.

Therefore, for both cases $B\lessgtr Z^{\frac{4}{3}}$ we arrive to
\begin{claim}\label{26-3-6}
The total charge is $\min(N,Z)$ (due to the choice of $\nu$) with the remainder estimate $O\bigl(\max (Z^{\frac{2}{3}},Z^{\frac{2}{5}}B^{\frac{1}{5}}))$ which is $O(Z^{\frac{2}{3}})$ if $B\le Z^{\frac{4}{3}}$ and $O(Z^{\frac{2}{5}}B^{\frac{1}{5}})$ if $Z^{\frac{4}{3}}\le B\le Z^3$.
\end{claim}

\begin{remark}\label{rem-26-3-1}
Remainder is less than the main part if
$Z^{\frac{2}{5}}B^{\frac{1}{5}}\lesssim Z$ i.e. $B\le Z^3$. It means exactly that $\zeta \ell \ge 1$ if $\ell=\bar{r}$ (in the case $B\ge Z^{\frac{4}{3}}$), or, in other words that $h \lesssim 1$. The same is true for all other semiclassical asymptotics below.

If $B\ll Z^3$ we arrive to asymptotics, ig $B\lesssim Z^3$ we have estimates and in the case of superstrong magnetic field $B\gg Z^3$ Thomas-Fermi theory is not valid for our main model.
\end{remark}
\enlargethispage{\baselineskip}

\subsection{Semiclassical $\D$-term}
\label{sect-26-3-1-2}
Consider now the semiclassical $\D$-term
\begin{equation}
\D \bigl( e(x,x, \nu) - \rho^\TF_B (x),\ e(x,x, \nu) - \rho^\TF_B(x)\bigr).
\label{26-3-7}
\end{equation}

\subsubsection{General arguments.}
\label{sect-26-3-1-2-1}
We do not have appropriate asymptotics for $e(x,x,\nu)$ in the case of the magnetic field\footnote{\label{foot-26-9} Unless we really assume that $W$ is smooth and apply results sections~\ref{book_new-sect-16-5}--\ref{book_new-sect-16-8}.} but we apply Fefferman--de Llave decomposition (\ref{book_new-16-3-1}):
\begin{multline}
|x-y|^{-1}_\gamma (x,y)\coloneqq |x-y|\varphi (\gamma^{-1}|x-y|)=\\
\gamma^{-4}\int \psi_{1,\gamma}(x,z)\psi_{2,\gamma}(y,z)\,dz
\label{26-3-8}
\end{multline}
where $\varphi\in \sC^\infty ([1,2])$.

Therefore contribution of $B(z,\gamma) \times B(z',\gamma)$ with
$3\gamma \le |z-z'|\le 4\gamma $, $\gamma \le \epsilon \ell(z)$ to such term does not exceed $C \bigl(\zeta^2 \gamma^2 +B\gamma^2)^2\gamma^{-1}$. There are
$\asymp\ell^3\gamma^{-3}$ of such pairs with $\ell(x)\asymp \ell$ and their total contribution does not exceed $C\bigl(\zeta^2 +B)^2\ell^3$.

Now we need to sum over $\gamma^{-1}\,d\gamma$ which does not look good because it leads to the logarithmic divergency but there is a simple remedy: we treat this way only pairs $t\ell\le |z-z'|\le \ell$ and apply for pairs with
$|z-z'|\le t\ell $ pointwise asymptotics; then we get
\begin{equation}
C(\zeta^2 +B)^2\ell^3 \bigl(1+ (\log B\ell/\zeta)_+\bigr);
\label{26-3-9}
\end{equation}
to get rid off this logarithmic factor we apply more delicate arguments similar to those of Subsection~\ref{book_new-sect-16-9-3}.

Thus, ignoring this logarithmic factor we conclude that the contribution of all pairs $(z,z')$ with $\ell(z)\asymp \ell(z') \asymp \ell$ does not exceed
$C(\zeta^2 +B)^2\ell^3$ while contribution of all pairs $(z,z')$ with $\ell(z)\asymp \ell_1 \not\asymp \ell(z') \asymp \ell_2$ does not exceed
$C(\zeta_1^2 +B)(\zeta_2^2+B)\ell_1^2\ell_2^2(\ell_1+\ell_2)^{-1}$.

Finally let us sum these expressions over partitions of unity.

\subsubsection{Moderate magnetic field.}
\label{sect-26-3-1-2-2}
Consider the case $B\le Z^{\frac{4}{3}}$. Then summation over zone
$\{\ell_1\le Z^{-\frac{1}{3}},\, \ell_2\le Z^{-\frac{1}{3}}\}$ results in $CZ^{\frac{5}{3}}$ and the same is also true for summation
over zone $\{ Z^{-\frac{1}{3}}\le \ell_1\le B^{-\frac{1}{4}},\,
Z^{-\frac{1}{3}}\le \ell_2\le B^{-\frac{1}{4}}\}$.

Obviously, in such estimates, if there is a fixed number of zones, we do not need to sum over ``mixed'' pairs when $z$ and $z'$ belong to different zones.

\subsubsection{Strong magnetic field.}
\label{sect-26-3-1-2-3}
Consider the case $B\ge Z^{\frac{4}{3}}$. Then summation over zone
$\{\ell_1\le Z^{\frac{1}{5}}B^{-\frac{2}{5}},\,
\ell_2\le Z^{\frac{1}{5}}B^{-\frac{2}{5}}\}$ results in
$CZ^{\frac{3}{5}}B^{\frac{4}{5}}$.\enlargethispage{2\baselineskip}

Therefore, for both cases $B\lessgtr Z^{\frac{4}{3}}$ we arrive to
\begin{claim}\label{26-3-10}
Term (\ref{26-3-7}) does not exceed $C\max(Z^{\frac{5}{3}},Z^{\frac{3}{5}}B^{\frac{4}{5}})$ which is $CZ^{\frac{5}{3}}$ if $B\le Z^{\frac{4}{3}}$ and $CZ^{\frac{3}{5}}B^{\frac{4}{5}}$ if $Z^{\frac{4}{3}}\le B\le Z^3$.
\end{claim}

\subsection{\texorpdfstring{$|\lambda_N-\nu|$}{|\textlambda\_N-\textmu|} and another $\D$-term}
\label{sect-26-3-1-3}

Consider two other non-trace terms in the upper estimate.

\subsubsection{Moderate magnetic field.}
\label{sect-26-3-1-3-1}

In the case $B\le Z^{\frac{4}{3}}$ we established the remainder in the total charge $O(Z^{\frac{2}{3}})$. Then using our standard arguments we conclude easily that $|\lambda_N-\nu|=O(Z)$ and then
\begin{equation}
|\lambda_N-\nu|\cdot |\N(\nu)-N|\le CZ^{\frac{5}{3}}
\label{26-3-11}
\end{equation}
and
\begin{multline}
\D\bigl( P_B'(W^\TF_B(x)+\lambda_N)-P_B'(W^\TF_B(x)+\nu),\\
P_B'(W^\TF_B(x)+\lambda_N)-P_B'(W^\TF_B(x)+\nu)\bigr)
\le CZ^{\frac{5}{3}};
\label{26-3-12}
\end{multline}
combining with the estimate of the previous subsubsection we conclude that
\begin{equation}
\D(\rho_\Psi -\rho_B^\TF, \rho_\Psi -\rho_B^\TF)\le CQ= O(Z^{\frac{5}{3}}),
\label{26-3-13}
\end{equation}
exactly as in (\ref{book_new-25-4-55}).

\subsubsection{Strong magnetic field.}
\label{sect-26-3-1-3-2}
Let now $Z^{\frac{4}{3}}\le B\le Z^3$. Then we established the remainder in the total charge $O(Z^{\frac{2}{5}}B^{\frac{1}{5}})$ and for the semiclassical $\D$-term we established estimate $O(Z^{\frac{3}{5}}B^{\frac{4}{5}})$. Therefore to estimate
\begin{gather}
|\lambda_N-\nu|\cdot |\N(\nu)-N|\le C Z^{\frac{3}{5}}B^{\frac{4}{5}}
\label{26-3-14}\\
\intertext{as well we want to prove that}
|\lambda_N-\nu|=O(Z^{\frac{1}{5}}B^{\frac{3}{5}}).
\label{26-3-15}
\end{gather}

Observe that
$|\nu|\lesssim Z\bar{r}^{-1}\asymp Z^{\frac{4}{5}}B^{\frac{2}{5}}\le CB$. Therefore if $|\lambda_N-\nu|\le \frac{1}{2}|\nu|$ we conclude that
\begin{multline}
|\int \bigl(P_B'(W^\TF_B(x)+\lambda_N)-P_B'(W^\TF_B(x)+\nu)\bigr)\,dx |\ge\\
\epsilon |\lambda_N-\nu| B \int (W+\nu)_+^{-\frac{1}{2}}\,dx
\label{26-3-16}
\end{multline}
with the integral taken over zone $\{x\colon W(x)+\nu \ge |\lambda_N-\nu|\}$.

One can see easily that as $|\lambda_N-\nu|\le \epsilon |\nu|$ the right-hand expression of (\ref{26-3-16}) is larger than
$\epsilon |\lambda_N-\nu| \cdot Z^{\frac{1}{5}}B^{-\frac{2}{5}}$
and it must be less than $CZ^{\frac{2}{5}}B^{\frac{1}{5}}$:
$|\lambda_N-\nu|Z^{\frac{1}{5}}B^{-\frac{2}{5}}\le CZ^{\frac{2}{5}}B^{\frac{1}{5}}$ which implies (\ref{26-3-15}).

Let us estimate the left-hand expression of (\ref{26-3-12}). For this, however, estimate (\ref{26-3-15}) is insufficient. We consider here only the atomic case. Then using (\ref{26-2-22})--(\ref{26-2-23}) one can prove easily that the right-hand expression of (\ref{26-3-16}) is of magnitude
\begin{equation*}
|\lambda_N-\nu|\cdot B\bar{r}^2 \times \underbracket{(|\nu|\bar{r}^{-1})^{-\frac{1}{3}}B^{-\frac{1}{3}}} \asymp
|\lambda_N-\nu| \cdot |\nu|^{-\frac{1}{3}} Z^{\frac{7}{15}}B^{-\frac{4}{15}}
\end{equation*}
provided $|\lambda-N|\le \epsilon \nu$, where the selected factor is just
$\int (B^2z^4+|\nu| \bar{r}^{-1})_+^{-\frac{1}{2}}\,dz$ (appearing due to
(\ref{26-2-22})--(\ref{26-2-23})). Comparing with $Z^{\frac{2}{5}}B^{\frac{1}{5}}$ we conclude that

\begin{enumerate}[label=(\alph*), wide, labelindent=0pt]

\item\label{sect-26-3-1-3-2-a}
If $|\nu|\ge C_1 Z^{-\frac{1}{10}}B^{\frac{7}{10}}$
($= C_1Z^{\frac{2}{5}}B^{\frac{3}{5}}\times Z^{-\frac{3}{10}}B^{\frac{1}{10}}$) then
\begin{equation}
|\lambda_N-\nu|\le C|\nu|^{\frac{1}{3}}Z^{-\frac{1}{15}}B^{\frac{7}{15}}
\label{26-3-17}
\end{equation}
which is less than $\epsilon |\nu|$ and coincides with (\ref{26-3-15}) as $(Z-N)_+\asymp Z$.

\item\label{sect-26-3-1-3-2-b}
If $|\nu|\ge C_1 Z^{-\frac{1}{10}}B^{\frac{7}{10}}$ then
$|\lambda_N-\nu|\le C_2 Z^{-\frac{1}{10}}B^{\frac{7}{10}}$.
\end{enumerate}
In the former case one can prove easily that the left-hand expression of (\ref{26-3-12}) does not exceed $CZ^{\frac{3}{5}}B^{\frac{4}{5}}$.

In the latter case (exactly as in Subsection~\ref{book_new-sect-25-4-2}) we consider Thomas-Fermi theory with $\nu=0$ i.e. $N=Z$ and also prove that that

\begin{claim}\label{26-3-18}
The left-hand expression of (\ref{26-3-15}) does not exceed $Q=CZ^{\frac{3}{5}}B^{\frac{4}{5}}$.
\end{claim}

In particular, we slightly improve estimate (\ref{26-3-14}) to
$|\nu|^{\frac{1}{3}}Z^{\frac{3}{5}}B^{\frac{2}{3}}$ as well (if $(Z-N)\ll Z$).

Therefore in our framework we estimated all non-trace terms in the upper estimate by $CZ^{\frac{3}{5}}B^{\frac{4}{5}}$ and therefore ``proved'' estimate
\begin{equation}
\D(\rho_\Psi -\rho_B^\TF, \rho_\Psi -\rho_B^\TF)\le CQ= O(Z^{\frac{3}{5}}B^{\frac{4}{5}}).
\label{26-3-19}
\end{equation}

\subsection{Trace}
\label{sect-26-3-1-4}
Consider now $\Tr((H_{A,W}-\nu)^-)$. This term is of magnitude
$\int (\zeta^5 + B\zeta^3)\,dx$ and one can see easily that it is
$\asymp Z^{\frac{7}{3}}$ for $B\le Z^{\frac{4}{3}}$ and
$\asymp B^{\frac{2}{5}}Z^{\frac{9}{5}}$ for $ Z^{\frac{4}{3}}\le B\le Z^3$.
\enlargethispage{2\baselineskip}

Meanwhile, consider the remainder. Again for simplicity consider only the atomic case. If $B\le Z$ the contribution of the zone
$\{x\colon \ell(x)\le Z^{-\frac{1}{3}}\}$ is $O(Z^{\frac{5}{3}})$ (we need to include Scott correction term in the main part) while the contribution of the zone $\{x\colon \ell(x)\ge Z^{-\frac{1}{3}}\}$ does not exceed
\begin{equation}
C\int (\zeta^3 + B\zeta)\ell^{-2}\,dx
\label{26-3-20}
\end{equation}
taken over this zone and it is $\asymp Z^{\frac{5}{3}}$ as well.

If $Z\le B\le B^2$ the contribution of the zone
$\{x\colon \ell(x)\le b\coloneqq B^{-\frac{2}{3}} Z^{\frac{1}{3}}\}$~is $O(b^{-\frac{1}{2}}Z^{\frac{3}{2}})= O(Z^{\frac{4}{3}}B^{\frac{1}{3}})$ and
we need to include Scott correction term. Meanwhile, the contribution of the zone $\{x\colon \ell(x)\ge b\}$ does not exceed integral (\ref{26-3-20}) taken over this zone which is
$\asymp Z^{\frac{4}{3}}B^{\frac{1}{3}} + Z^{\frac{3}{5}}B^{\frac{4}{5}}$ where the last term coincides with estimate for (\ref{26-3-7}) if
$B\ge Z^{\frac{4}{3}}$ and does not exceed $CZ^{\frac{5}{3}}$ if
$B\le Z^{\frac{4}{3}}$.

Finally, if $Z^2\le B\le B^3$ we need to reset $b=Z^{-1}$ because
$h =1/(\zeta\ell)$ becomes $\gtrsim 1$ inside. Then we do not need Scott correction term and the contributions of the zone $\{x\colon \ell(x)\le b\}$ to both the main part and the remainder do not exceed
$C\int (\zeta^5+B\zeta^3)\, dx \asymp Z^2+ B\asymp B$.

Further, the contribution of the zone $\{x\colon \ell(x)\ge b\}$ to the remainder does not exceed integral (\ref{26-3-20}) taken over this zone which results in
$CB + CZ^{\frac{3}{5}}B^{\frac{4}{5}}$ and the second term dominates due to assumption $B\ll Z^3$. Thus we arrive to

\begin{claim}\label{26-3-21}
The main therm in $\Tr ((H_{A,W}-\nu)^-)$ is of magnitude $Z^{\frac{7}{3}}$ for $B\le Z^{\frac{4}{3}}$ and $B^{\frac{2}{5}}Z^{\frac{9}{5}}$ for $Z^{\frac{4}{3}}\le B\le Z^3$, while the remainder estimate is $O(Z^{\frac{5}{3}})$ for $B\le Z$, $O(Z^{\frac{4}{3}}B^{\frac{1}{3}})$ for
 $Z\le B\le Z^{\frac{4}{3}}$, and
$O(Z^{\frac{4}{3}}B^{\frac{1}{3}} + Z^{\frac{3}{5}}B^{\frac{4}{5}})$ for
$Z^{\frac{4}{3}}\le B\le Z^3$.

If $B\le Z^{\frac{7}{4}}$ we need to include into main part Scott correction term.
\end{claim}

\subsection{Discussion}
\label{sect-26-3-1-5}
Now let us formulate our expectations:

\begin{remark}\label{rem-26-3-2}
We expect
\begin{enumerate}[label=(\roman*), wide, labelindent=0pt]
\item\label{rem-26-3-2-i}
Estimate (\ref{26-3-13}) for $B\le Z^{\frac{4}{3}}$ and estimate (\ref{26-3-19}) for $Z^{\frac{4}{3}}\le B\le Z^3$.

\item\label{rem-26-3-2-ii}
Furthermore, since for $B\le Z^{\frac{4}{3}}$ the main contribution to all terms needed to derive this estimate is delivered by the zone
$\{x\colon \ell(x)\approx Z^{-\frac{1}{3}}\}$ and the effective magnetic field is $\mu = B\ell/\zeta \approx BZ^{-1}$ we expect improved to ``$o$'' (or better) estimate (\ref{26-3-13}) if $B\ll Z$ and
$a\gg Z^{-\frac{1}{3}}$\,\footnote{\label{foot-26-10} Recall that
$a=\min_{1\le m< m'\le M}|\y_m-\y_{m'}|$ is the minimal distance between nuclei.}.\enlargethispage{\baselineskip}

\item\label{rem-26-3-2-iii}
 Statement, similar to~\ref{rem-26-3-2-ii} should be also true for the trace term; however then we need to include the Schwinger term.

\item\label{rem-26-3-2-iv}
The remainder estimate for the ground state energy is maximum of the remainder estimate for the non-trace and trace terms; therefore we expect the same remainder estimate as in (\ref{26-3-21}); Statement, similar to~\ref{rem-26-3-2-ii} should be also correct for the ground state energy. However then we need to include both Schwinger and Dirac terms.

\item\label{rem-26-3-2-v}
We expect the described remainder estimate of the trace term and the ground state energy if $a$ is large enough; otherwise it should contain term
$O(a^{-\frac{1}{2}}Z^{\frac{3}{2}})$ ifs $B\le Z^{\frac{7}{4}}$ and
$a\ge Z^{-1}$ (and in this case we include Scott correction term).
\end{enumerate}
\end{remark}

\begin{remark}\label{rem-26-3-3}
The other difference between cases $B\le Z^{\frac{4}{3}}$ and
$B\ge Z^{\frac{4}{3}}$ is that $\mu h = B\zeta^{-2} \lesssim 1$ in the former case if $\ell(x)\le \bar{r}$; however in the latter case it happens only if $\ell(x) \le B^{-1}Z$ but in the zone
$\{x\colon B^{-1}Z \le \ell(x)\le B^{-\frac{2}{5}}Z^{\frac{1}{5}}\}$ an opposite inequality holds.
\end{remark}

\section{Smooth approximation}
\label{sect-26-3-2}

An approach described in Subsection~\ref{sect-26-3-1} hits two obstacles: the non-smoothness of $W^\TF_B$ and its possible degeneration i.e. $\nabla W^\TF_B$ is not disjoint from $0$. However non-smoothness of $W^\TF_B$ is due to the non-smoothness of $P_B$. So we want to consider first the zone where we can just replace $P_B(W+\nu)$ by $P(W+\nu)$ and therefore $W^\TF_B$ by some smooth function $W$ which does not necessary coincides with $W^\TF$.

\subsection{Trivial arguments}
\label{sect-26-3-2-1}

Obviously we can do this as an effective magnetic field
$\mu=B\ell/\zeta \lesssim 1$. In this case we do not need assumption
$W+\nu \asymp \zeta^2$ and therefore we can take $\zeta= \ell^{-4}$ as $B\lesssim Z^{\frac{4}{3}}$ and $\ell\gtrsim Z^{-\frac{1}{3}}$ \underline{and}
$\zeta= Z^{\frac{1}{2}}\ell^{-\frac{1}{2}}$ in all other cases. Therefore zone in question is
\begin{multline}
\cX_1\coloneqq \{x\colon \ell(x)\le r_1\}\\
\text{with\ \ }
r_1= \left\{\begin{aligned}
&B^{-\frac{1}{3}} \qquad&&\text{if\ \ } 1\le B\lesssim Z,\\
&B^{-\frac{2}{3}}Z^{\frac{1}{3}} \qquad&&\text{if\ \ } Z\lesssim B\lesssim Z^2.
\end{aligned}\right.
\label{26-3-22}
\end{multline}

In this zone $\cX_1$ for such modified $W$ we can unleash the full power of the same smooth theory as in Section~\ref{book_new-sect-25-4} and prove easily the following

\begin{proposition}\label{prop-26-3-4}
Let $1\le B\le Z^2$. Then
\begin{enumerate}[label=(\roman*), wide, labelindent=0pt]
\item\label{prop-26-3-4-i}
A contribution of zone $\cX_1$ defined by \textup{(\ref{26-3-22})} to
\begin{gather}
\int \Bigl(e(x,x,\nu)-P'(W(x)+\nu)\Bigr)\, dx\label{26-3-23}\\
\intertext{does not exceed $CZ^{\frac{2}{3}}$ while its contribution to}
\D\Bigl( e(x,x,\nu)-P'(W(x)+\nu), e(x,x,\nu)-
P' (W(x)+\nu)\Bigr)\label{26-3-24}\\
\intertext{does not exceed $CZ^{\frac{5}{3}}$, and its contribution to}
\int \Bigl(e_1(x,x,\nu)+P (W(x)+\nu)\Bigr)\, dx-\Scott \label{26-3-25}
\end{gather}
does not exceed $CZ^{\frac{5}{3}}+Ca^{-\frac{1}{2}}Z^{\frac{3}{2}}+
CZ^{\frac{4}{3}}B^{\frac{1}{3}}$\,\footref{foot-26-10}, \footnote{\label{foot-26-11} If $a\lesssim Z^{-1}$ we skip $\Scott$ and reset $a=Z^{-1}$ in the remainder estimate which become $CZ^2$.}.
%\enlargethispage{\baselineskip}

\item\label{prop-26-3-4-ii}
Further, if $B\ll Z$ and $a\gg Z^{-\frac{1}{3}}$ we can recover for these contributions estimates $CZ^{\frac{2}{3}}\upsilon$, $CZ^{\frac{5}{3}}\upsilon$ and $CZ^{\frac{5}{3}}\upsilon$ respectively with
\begin{gather}
\upsilon \coloneqq Z^{-\delta}+ (aZ^{\frac{1}{3}})^{-\delta} + (BZ^{-1})^\delta
\label{26-3-26}\\
\intertext{where expression \textup{(\ref{26-3-25})} should be modified to}
\int \bigl(e_1(x,x,\nu)-P(W(x)+\nu)\bigr)\, dx-\Scott -\Schwinger.
\label{26-3-27}\\
\intertext{Furthermore in this case contribution of $\cX_1$ to}
\frac{1}{2} \int \tr \bigl( e^\dag (x,y,\nu)e(x,y,\nu)\bigr)\, dxdy-\Dirac
\label{26-3-28}
\end{gather}
does not exceed $CZ^{\frac{5}{3}-\delta}$.
\end{enumerate}
\end{proposition}

\begin{remark}\label{rem-26-3-5}
\begin{enumerate}[label=(\roman*), wide, labelindent=0pt]
\item\label{rem-26-3-5-i}
So far we should use $P(.)$ instead of $P_B(.)$ but we will prove that the same results would hold for $P_B$ as well.

\item\label{rem-26-3-5-ii}
In the next subsubsections we expand this zone to one defined by
$\mu \le h^{-\frac{1}{3}}$\,\footnote{\label{foot-26-12} Or even to
$\mu \le h^{-\frac{3}{5}}$ under non-degeneracy assumption (\ref{26-3-34}) with $\gamma=\ell$, in particular, in the atomic case.} but for trace term we still need a separate analysis as $\mu \lesssim 1$.

\item\label{rem-26-3-5-iii}
The same estimates hold if we replace in all expressions (\ref{26-3-23})--(\ref{26-3-27}) $P$ by $P_B$.

\item\label{rem-26-3-5-iv}
We assumed that $B\le Z^2$ since otherwise $h\gtrsim 1$ not only in $\cX_1$ but even in $\{ x\colon W^\TF (x) \ge B\}$.

\item\label{rem-26-3-5-v}
Note that if $r_1 \gtrsim (Z-N)_+^{-\frac{1}{3}}$ this zone (and the whole analysis) could be cut short since outside zone in question $W+\nu \ge 0$. From Chapter~\ref{book_new-sect-25} we already know how to deal with such irregularities.

\item\label{rem-26-3-5-vi}
We need to assume that $a\ge Z^{-\frac{1}{3}}$ and to include the second term $(aZ^{\frac{1}{3}})^{-\delta}$ in the definition of $\upsilon$ only as we estimate the trace term (\ref{26-3-25}).
\end{enumerate}
\end{remark}

\begin{remark}\label{rem-26-3-6}
\begin{enumerate}[label=(\roman*), wide, labelindent=0pt]
\item\label{rem-26-3-6-i}
If either $a\ll Z^{-\frac{1}{3}}$ or $B\gg Z$ we estimated (\ref{26-3-5}) by $Ca^{-\frac{1}{2}}Z^{\frac{3}{2}} + CZ^{\frac{4}{3}}B^{\frac{1}{3}}$. While the first term does not bother us since assumption $a\ll \min(Z^{-\frac{1}{3}},B^{-\frac{1}{4}})$ is unrealistic, the second term is troublesome. Let us assume that
$a\ge Z^{-\frac{1}{3}}$.

We can marginally improve this estimate of expression (\ref{26-3-25}) to $CZ^{\frac{5}{3}}+ o\bigl(Z^{\frac{4}{3}}B^{\frac{1}{3}}\bigr)$.

First, observe, that this term $CZ^{\frac{3}{2}}B^{\frac{1}{3}}$ appears as $b^{-\frac{1}{2}} Z^{\frac{3}{2}}$ with
$b= B^{-\frac{2}{3}}Z^{\frac{2}{3}}\ll 1$. Therefore we need to estimate this way only contribution of the zone
$\cY \coloneqq \{x\colon b^{\delta}\le \ell (x)Z^{\frac{1}{3}}\le b^{-\delta}\}$ and it is sufficient to investigate the corresponding classical dynamics in the zone $\cY_1 \coloneqq \{x\colon b^{2\delta}\le \ell (x)Z^{\frac{1}{3}}\le b^{-2\delta}\}$.

Indeed, to recover estimate we have now, we used a classical dynamics on $\Sigma\coloneqq \{(x,\xi)\colon H(x,\xi)=0\}$ for time
$T (x) =Z^{-1}\ell (x)^{\frac{3}{2}}$.

Further, one can see easily that along classical trajectories, starting in $\Sigma |_\cY$, $\ell(x)\le b^{-2\delta}$ for time $T= b^{-\sigma}T_0$ with $\sigma=\sigma(\delta)>0$.

On the other hand, the invariant measure of
$\Sigma_r=\{(x,\xi)\in \Sigma, \ell(x) \asymp r \}$ is $\asymp r^{2}Z$ and since the spatial speed there is $O(Z^{\frac{1}{2}}r^{-\frac{1}{2}})$ we conclude that

\begin{claim}\label{26-3-29}
The invariant measure of the points in
$\Sigma|_\cY$, such that the classical trajectories starting from them do not remain in $\cY_1$ for time $T= b^{-\sigma}T_0$, does not exceed $b^{2+\sigma}Z^{\frac{1}{3}}$.
\end{claim}

\item\label{rem-26-3-6-ii}
Now it is sufficient to explore the classical dynamics with the Hamiltonian, corresponding to the Coulomb potential and constant magnetic field, and to prove that

\begin{claim}\label{26-3-30}
The invariant measure of the periodic points $\Sigma$ is $0$.
\end{claim}
\enlargethispage{2\baselineskip}

To do so, we need to prove that there are non-periodic trajectories, which do not hit an origin. It is sufficient to consider trajectories belonging to the plane $\{z=0\}$; we assume that magnetic intensity is $(0,0,B)$. See Part~\ref{rem-26-3-6-iii}.

\item\label{rem-26-3-6-iii}
To improve this estimate further we need to investigate the classical dynamics in more details, and it seems to be a daunting, if not impossible task. Indeed, while in $2\D$ the system is completely integrable\footnote{\label{foot-26-13} And easily solvable in the polar coordinates since $E=\frac{1}{2}\dot{r}^2 + V^*(r)$ with effective potential
$V^*(r)=\frac{1}{2}r^2(M_z r^{-2}-B/2)^2-r^{-1}$, and the corrected angular momentum $M_z=r^2 \dot{\theta} + \frac{1}{2}Br^2$. One can see easily that $V^*(r)\to +\infty$ as $r\to +0$ or $r\to +\infty$ and has a single nondegenerate minimum. Therefore along each trajectory $r$ oscillates between $r_{\min }$ and $r_{\max }$. If all trajectories on the energy level $E$ were periodic, then the number of oscillations was constant for increment $\theta$ equal to $2\pi n$ with some $n\in \bZ^+$. But this is definitely not the case since the number of oscillations tends to $\infty$ as
$M_z/B\to +\infty$.}, it does not seem so in $3\D$ as we know only two first integrals, energy $E$ and $M_z = (x\dot{y}-y\dot{x}) + \frac{1}{2}B(x^2+y^2)$.
\end{enumerate}
\end{remark}

\subsection{Formal expansion}
\label{sect-26-3-2-2}

Now we want to expand zone $\cX_1$. Note first that
\begin{gather}
P'_B(W+\nu) -P'(W+\nu)=O( B^{\frac{3}{2}})
\label{26-3-31}\\
\shortintertext{and}
P_B(W+\nu)-P (W+\nu)-\kappa_1 B^2 (W+\nu)^{\frac{1}{2}}=
O(B^{\frac{5}{2}}).
\label{26-3-32}
\end{gather}
Really, one can consider $P'_B(w)$ and $P_B(w)$ as Riemann sums for integrals $P'(w)$ and $P(w)$ respectively; see Appendix~\ref{sect-26-A-3} for details.

However under non-degeneracy assumption $|\nabla W|\asymp \zeta^2 \ell^{-1}$ we can do better with the integrated expressions.

\begin{proposition}\label{prop-26-3-7}
Assume that in $B(z,\gamma)$
\begin{gather}
|\nabla^\alpha W|\le C_\alpha \zeta^2\gamma^{-|\alpha|}\qquad
\forall \alpha:|\alpha|\le n,
\label{26-3-33}\\[3pt]
|\nabla W|\ge \epsilon \zeta^2 \gamma^{-1},
\label{26-3-34}\\
\shortintertext{and}
B\le \zeta^2.
\label{26-3-35}
\end{gather}
Then
\begin{gather}
\int \phi (x)\Bigl(P'_B(W(x)+\nu) -\tilde{P}'_B(W(x)+\nu)\Bigr)\,dx = O(B^2\zeta^{-1}\gamma^3),
\label{26-3-36}\\
\int \phi (x)\Bigl(P_B(W(x)+\nu) -\tilde{P}'_B(W(x)+\nu)\Bigr)\,dx = O(B^5\zeta^{-5}\gamma^3)
\label{26-3-37}
\end{gather}
and
\begin{multline}
\D\Bigl(\phi (x)\bigl(P'_B(W(x)+\nu) -\tilde{P}'_B (W(x)+\nu)\bigr),\\
\phi (x)\bigl(P'_B(W(x)+\nu) -\tilde{P}'_B (W(x)+\nu)\bigr)\Bigr)=
O(B^5\zeta^{-4}\gamma^5)\label{26-3-38}
\end{multline}
with
\begin{equation}
\tilde{P}_B (w)\coloneqq P(w)+ \bigl(\kappa_1 P'' (w) B^2 + \kappa_2 P^{IV} B^4\bigr) \cdot \bigl(1-\varphi (w/B)\bigr)
\label{26-3-39}
\end{equation}
where $\varphi\in \sC^\infty ([-2,2])$, $\varphi=1$ on $[-1,1]$.
\end{proposition}

\begin{proof}
Rescaling $x\mapsto x\gamma^{-1}$, $w\mapsto w\zeta^{-2}$ and therefore
$B \mapsto\beta =B\zeta^{-2}$ one can reduce the case to $\gamma=\zeta=1$, $\beta\le 1$\,\footnote{\label{foot-26-14} Recall that $\beta =\mu h $ with $\mu =B\gamma/\zeta$ and $h=1/\zeta\gamma$.}. Then estimates (\ref{26-3-36}) and (\ref{26-3-37}) are trivially proven by (multiple) integration by parts which integrates $P_\beta$ on each step increasing its smoothness\footnote{\label{foot-26-15} In fact one can prove then estimates $O(\beta^s)$ but adding correction terms
$\sum \kappa_k \beta^{2k}$. However this improvement is not carried on to (\ref{26-3-38}) in full.}.

To prove estimate (\ref{26-3-38}) we apply decomposition (\ref{26-3-8}). Integration by parts shows that (\ref{26-3-36}) with $t$-admissible function $\phi$ is $O(\beta^3 t^{\frac{3}{2}})$ if $t\ge \beta$ and therefore the contribution of the zone $\{(x,y)\colon |x-y|\asymp t\}$ is
$O(\beta^6 t^{3}\times t^{-4})$. Then the total contribution of the zone $\{(x,y)\colon |x-y|\ge \beta \}$ is $O(\beta^5)$. Meanwhile a total contribution of the zone $\{(x,y)\colon |x-y|\le \beta \}$ is
$O(\beta^3 \times \beta^2)$. \enlargethispage{\baselineskip}
\end{proof}

Therefore we expect that the zone $\cX_1$ defined by $\mu \lesssim 1$ could be expanded to the zone $\cX_1 '$ defined by
$\mu \lesssim h^{-\frac{1}{3}}$\,\footref{foot-26-12} or even larger\footnote{\label{foot-26-16} We do not need for each $\ell$ have a sharp remainder estimates but need only them to sum to a sharp estimate.}; furthermore, under assumption $|\nabla W|\asymp \zeta^2\ell^{-1}$ we can define $\cX_1 '$ by $\mu \lesssim h^{-\frac{3}{5}}$ or even larger\footref{foot-26-16}.

\subsection{Expansion: justification}
\label{sect-26-3-2-3}

Now however we need to deal with $e(x,x,\nu)$ rather than $P'_B(W(x)+\nu)$ (etc).

\begin{proposition}\label{prop-26-3-8}
Assume that in $B(z,\gamma)$ conditions \textup{(\ref{26-3-33})},
$\zeta \gamma \ge 1$ and
\begin{equation}
B\le c\zeta^2 (\zeta \gamma)^{-\delta}
\label{26-3-40}
\end{equation}
are fulfilled. Then for $\gamma$-admissible $\phi$
\begin{gather}
\int \phi (x)\Bigl(e(x,x,\nu) -\tilde{P}'_B(W(x)+\nu)\Bigr)\,dx = O(\zeta^2\gamma^2),
\label{26-3-41}\\
\int \phi (x)\Bigl(e_1(x,x,\tau) -\tilde{P}_B (W(x)+\nu)\Bigr)\,dx = O(\zeta^{3}\gamma)
\label{26-3-42}
\end{gather}
and
\begin{multline}
\D\Bigl(\phi (x)\bigl(e(x,x,\nu) -\tilde{P}'_B (W(x)+\nu)\bigr),\\
\phi (x)\bigl(e(x,x,\nu) -\tilde{P}'_B (W(x)+\nu)\bigr)\Bigr)=
O(\zeta^{4}\gamma^3).\label{26-3-43}
\end{multline}
\end{proposition}

\begin{proof}
Estimates (\ref{26-3-41}) and (\ref{26-3-42}) are due to Chapter~\ref{book_new-sect-13}. Really, rescale $x\mapsto x\gamma^{-1}$, $\tau\mapsto \tau\zeta^{-2}$ and $h=1\mapsto h= \gamma^{-1}\zeta^{-1}$, $B\mapsto \mu = B\gamma \zeta^{-1}$.

To prove (\ref{26-3-43}) let us apply decomposition (\ref{26-3-8}); then according to (\ref{26-3-41}) $J(t)$ (defined as expression (\ref{26-3-41}) with $t\gamma$-admissible $\phi_t$) does not exceed $C\zeta^2\gamma^2 t^2$ as long as $t\zeta\gamma \ge 1$; therefore contribution of zone
$\{(x,y)\colon |x-y|\asymp t\}$ into the left-hand expression of (\ref{26-3-43}) does not exceed $C(\zeta^2\gamma^2 t^2)^2 \times t^{-4}\gamma^{-1}\asymp
C\zeta^4 \gamma^3$.

Then summation over $t\ge \mu^{-1}= B^{-1}\gamma^{-1}\zeta$ returns
$C\zeta^4 \gamma^3\int t^{-1}\,dt \asymp C\zeta^4 \gamma^3 \log \mu$ (we assume that $\mu\ge 2$; case $\mu \lesssim 2$ has been covered already). So, the total contribution of zone $\{(x,y)\colon |x-y|\ge \mu^{-1}\}$ does not exceed
$C\zeta^4 \gamma^3\log \mu$.

Let us get rid off the logarithmic factor. Returning back to $B(z,t)$ stretched to $B(0,1)$ one can see easily that conditions of Proposition~\ref{book_new-prop-13-6-25} are fulfilled as well with
$T=\min \bigl(t ^{-\delta },h ^{-\delta } t ^\delta \bigr)$ and thus
$|J(t)|\le C(h t^{-1}) ^{-2}T ^{-1}\le
C h ^{-2}t ^2\bigl(t ^\delta +h ^\delta t ^{-\delta }\bigr)$. Plugging
into (\ref{26-3-8}) we get
\begin{equation*}
Ch ^{-4}\gamma^{-1}
\int _{\mu^{-1}} ^1 t ^{-1}\bigl(t ^\delta +h ^\delta t ^{-\delta }\bigr)\,dt \asymp Ch ^{-4}\gamma^{-1}= C\zeta^4 \gamma^3.
\end{equation*}

On the other hand, in zone $t\le \mu^{-1}$ we use the trivial estimate \begin{equation*}
e(x,x,\nu)-P'(W(x)+\nu)= O(\mu \zeta^2\gamma^2 )
\end{equation*}
(due to simple rescaling $x\mapsto \mu x$) and its contribution to the left-hand expression of (\ref{26-3-43}) does not exceed $C(\mu \zeta^2\gamma^2 )^2 \times \mu^{-2}\gamma^{-1}\asymp C\zeta^4 \gamma^3$. \end{proof}

Combining with estimates (\ref{26-3-31}) and (\ref{26-3-32}) we arrive to Statement~\ref{cor-26-3-9-i} below; combining with Proposition~\ref{26-3-7} to Statement~\ref{cor-26-3-9-ii}:

\begin{corollary}\label{cor-26-3-9}
Assume that in $B(z,\gamma)$ conditions \textup{(\ref{26-3-33})} and
$\zeta \gamma \ge 1$ are fulfilled. Let $\phi$ be $\gamma$-admissible function.
\begin{enumerate}[label=(\roman*), wide, labelindent=0pt]
\item\label{cor-26-3-9-i}
Let
\begin{gather}
B\le c\zeta ^{\frac{4}{3}}\gamma^{-\frac{2}{3}};
\label{26-3-44}
\shortintertext{then}
\int \phi (x)\Bigl(e(x,x,\nu) -\tilde{P}'_B(W(x)+\nu)\Bigr)\,dx = O(\zeta^2\gamma^2),
\label{26-3-45}\\
\int \phi (x)\Bigl(e_1(x,x,\tau) -\tilde{P}_B (W(x)+\nu)\Bigr)\,dx = O(\zeta^{3}\gamma)
\label{26-3-46}
\end{gather}
and
\begin{multline}
\D\Bigl(\phi (x)\bigl(e(x,x,\nu) -\tilde{P}_B' (W(x)+\nu)\bigr),\\
\phi (x)\bigl(e(x,x,\nu) -\tilde{P}_B' (W(x)+\nu)\bigr)\Bigr)=
O(\zeta^{4}\gamma^3).
\label{26-3-47}
\end{multline}
\item\label{cor-26-3-9-ii}
Let assumption \textup{(\ref{26-3-34})} be fulfilled and
\begin{equation}
B\le c\zeta ^{\frac{8}{5}}\gamma^{-\frac{2}{5}}.
\label{26-3-48}
\end{equation}
Then \textup{(\ref{26-3-45})}--\textup{(\ref{26-3-47})} hold.
\end{enumerate}
\end{corollary}

\section{Rough approximation}
\label{sect-26-3-3}

Unless our analysis has been cut short with $r_1\gtrsim (Z-N)_+^{-\frac{1}{3}}$, we need to consider the zone $\{x\colon \ell(x)\ge r_1\}$ with redefined $r_1$, so that this zone is described by $\mu \gtrsim h^{-\frac{1}{3}}$ or
$\mu \gtrsim h^{-\frac{3}{5}}$ in the general or non-degenerate (i.e. satisfying assumption (\ref{26-3-34})) cases respectively.

In this zone the replacement of $P_B$ by $P$ and thus $W_B^\TF$ by some smooth function leads to the error which is too large. Therefore instead in this zone we consider $\varepsilon\ell$-mollification of $W^\TF_B$ with
$\varepsilon\ll 1$ (after rescaling $x\mapsto x/\ell$). In contrast to potentials considered in Chapter~\ref{book_new-sect-18} function $W^\TF_B$ is more regular.

\subsection{Properties of mollification}
\label{sect-26-3-3-1}

First, recall regularity properties of $W^\TF_B$:

\begin{proposition}\label{prop-26-3-10}
$W_B^\TF$ have the following properties:
\begin{equation}
|\nabla ^\alpha W_B ^\TF(x)|\le c_\alpha \zeta (x) ^2\ell (x) ^{-|\alpha |}
\qquad \forall \alpha : |\alpha |\le 2,
\label{26-3-49}
\end{equation}
\begin{multline}
|\nabla ^\alpha \bigl(W_B ^\TF(x)-W_B ^\TF (y)\bigr)|\le\\
\shoveright{c_0 B\ell (x) ^{-\frac{5}{2}} |x-y| ^\frac{1}{2}+
c_0\zeta (x) ^2\ell (x) ^{-3}|x-y|}\\[3pt]
\qquad \forall |\alpha |=2\quad
\forall x,y\colon |x-y|\le \epsilon \ell (x)
\label{26-3-50}
\end{multline}
where we recall
\begin{align}
&\zeta (x)= \min \bigl(Z^{\frac{1}{2}}\ell(x)^{-\frac{1}{2}}, \ell(x)^{-2}\bigr)
\qquad&&\text{if\ \ } B\le Z^{\frac{4}{3}},\label{26-3-51}\\[3pt]
&\zeta (x)= Z^{\frac{1}{2}}\ell(x)^{-\frac{1}{2}}
\qquad&&\text{if\ \ } B\ge Z^{\frac{4}{3}};\label{26-3-52}
\end{align}
\end{proposition}

\begin{proof} This proof is rather obvious corollary of the Thomas-Fermi equation (\ref{26-2-3}). See also arguments below. \end{proof}

Let us consider $B(z, \ell(z))$ with $\zeta^2 \gtrsim B$ and rescale
$x\mapsto x\ell^{-1}$, $W\mapsto w=\zeta^{-2}(W+\nu)$ (where we included $\nu$ for a convenience). After such rescaling $w\in \sC^{\frac{5}{2}}$ uniformly, but there is more: Thomas-Fermi equation (\ref{26-2-3}) translates into
\begin{equation}
\frac{1}{4\pi}\Delta w = \ell^2 P'_\beta (w)= \ell^2 \zeta P' (w)+
\ell^2 \zeta \bigl(P'_\beta (w)-P'(w)\bigr)
\label{26-3-53}
\end{equation}
with $\beta = B\zeta^{-2}$; observe that $P'_B(W)$ is positively homogeneous of degree $3$ with respect to $(W,B)$.

Note that parameter $\eta \coloneqq \zeta \ell^2\lesssim 1$ and $\eta \asymp 1$ if and only if $B\lesssim Z^{\frac{4}{3}}$ and $\ell \gtrsim Z^{-\frac{4}{3}}$ (in which case $\zeta \asymp \ell^{-2}$).

Also note that the first term and the second terms in the right-hand expression of (\ref{26-3-53}) belong to $\sC^{\frac{5}{2}}$ and
$\beta\eta \sC^{\frac{1}{2}}$ respectively uniformly\footnote{\label{foot-26-17} I.e. norms do not depend on any parameters.} and
\begin{equation}
\beta\eta = \beta \zeta \ell^2= B \zeta^{-1}\ell^2,\qquad \eta \coloneqq \zeta \ell^2\qquad \text{if\ \ }\beta\lesssim 1.
\label{26-3-54}
\end{equation}
Because of this $w \in \sC^{\frac{9}{2}}\oplus \beta\eta \sC^{\frac{5}{2}}$ again uniformly. Iterating, we conclude that
$w \in \sC^n \oplus \beta\eta \sC^{\frac{5}{2}}$ with arbitrarily large exponent $n$.

On the other hand, if $B\gtrsim \zeta^2$ (i.e. $\beta\gtrsim 1$) without invoking $P'_B$ one can prove easily that $w \in \eta \sC^{\frac{5}{2}}$ with
\begin{equation}
\eta\coloneqq \beta \zeta \ell^2= B \zeta^{-1}\ell^2\qquad
\text{if\ \ }\beta\gtrsim 1.
\tag*{$\textup{(\ref*{26-3-54})}'$}\label{26-3-54-'}
\end{equation}
Therefore we have proven

\begin{claim}\label{26-3-55}
$w \in \sC^n \oplus \beta\eta \sC^{\frac{5}{2}}$ with arbitrarily large exponent $n$ as $\beta \lesssim 1$ and
$w \in \eta \sC^{\frac{5}{2}}$ as $\beta \gtrsim 1$
\end{claim}
and one can see easily that
\begin{claim}\label{26-3-56}
Parameter $\eta=B \zeta^{-1}\ell^2$ is $O(1)$ and $\eta \asymp 1$ iff
\underline{either} $B\le Z^{\frac{4}{3}}$ and $\ell\asymp B^{-\frac{1}{4}}$ \underline{or} $B\ge Z^{\frac{4}{3}}$ and $\ell\asymp B^{-\frac{2}{5}}Z^{\frac{1}{5}}$ (i.e. near border of $\supp (\rho^\TF_B)$, uncut by $\nu$).
\end{claim}

\begin{remark}\label{rem-26-3-11}
It may seem strange to define $\eta$ differently as $\beta \lesssim 1$ and $\beta \gtrsim 1$ but there is a good reason for this when we consider the case of $M\ge 2$. Anyway, $\eta$ is the magnitude of the right-hand expression of (\ref{26-3-53}).
\end{remark}

\begin{proposition}\label{prop-26-3-12}
\begin{enumerate}[label=(\roman*), wide, labelindent=0pt]
\item\label{prop-26-3-12-i}
Let $w_\varepsilon$ be a $\varepsilon$-mollification of $w$ with $\varepsilon \lesssim \min(\beta, h^\delta)$ (recall that $h=1/(\zeta\ell)$). Then if $\beta \lesssim 1$ the following estimates hold:
\begin{align}
&|\nabla ^\alpha (w-w_\varepsilon)|\le
c_\alpha \beta\eta \varepsilon^{\frac{5}{2}-|\alpha|}
\qquad &&\forall \alpha:|\alpha|\le 2,
\label{26-3-57}\\[3pt]
&|P_\beta (w)- P_\beta(w_\varepsilon)|\le
c\beta\eta \varepsilon^{\frac{5}{2}}
\label{26-3-58}\\
\shortintertext{and}
&|P'_\beta (w)- P'_\beta(w_\varepsilon)|\le
c \beta^{\frac{3}{2}}\eta^{\frac{1}{2}} \varepsilon^{\frac{5}{4}}+
c \beta \eta \varepsilon^{\frac{5}{2}};
\label{26-3-59}
\end{align}

\item\label{prop-26-3-12-ii}
On the other hand, if $\beta \gtrsim 1$ the right-hand expressions of
\textup{(\ref{26-3-57})}--\textup{(\ref{26-3-59})} should be replaced by the similar expressions albeit without $\beta$:
\begin{align}
&|\nabla ^\alpha (w-w_\varepsilon)|\le
c_\alpha \eta \varepsilon^{\frac{5}{2}-|\alpha|}
\qquad &&\forall \alpha:|\alpha|\le 2,
\tag*{$\textup{(\ref*{26-3-57})}'$}\label{26-3-57-'}\\[3pt]
&|P_\beta (w)- P_\beta(w_\varepsilon)|\le
c\eta \varepsilon^{\frac{5}{2}}
\tag*{$\textup{(\ref*{26-3-58})}'$}\label{26-3-58-'}\\
\shortintertext{and}
&|P'_\beta (w)- P'_\beta(w_\varepsilon)|\le
c \eta^{\frac{1}{2}} \varepsilon^{\frac{5}{4}}.
\tag*{$\textup{(\ref*{26-3-59})}'$}\label{26-3-59-'}
\end{align}

\item\label{prop-26-3-12-iii}
Further, under assumption $|\nabla w|\asymp 1$ in both cases
\begin{gather}
|\int \phi (x)\bigl(P_\beta (w)- P_\beta(w_\varepsilon)\bigr)\, dx |\le
c\eta \varepsilon^{\frac{7}{2}},
\label{26-3-60}\\
|\int \phi (x)\bigl(P'_\beta (w)- P'_\beta(w_\varepsilon)\bigr)\, dx |\le
c\eta \varepsilon^{\frac{9}{4}}
\label{26-3-61}\\
\shortintertext{and}
\D\bigl( \phi (P'_\beta (w)-P'_\beta (w_\varepsilon)),\,
\phi (P'_\beta (w)-P'_\beta (w_\varepsilon))\bigr)\le
c\eta^2 \varepsilon^{\frac{9}{2}}.
\label{26-3-62}
\end{gather}
\end{enumerate}
\end{proposition}

\begin{proof}
Proof of Statement~\ref{prop-26-3-12-i} is trivial; in particular, we observe that
$\eta \varepsilon^{\frac{5}{2}}\lesssim \beta $.

Proof of Statement~\ref{prop-26-3-12-iii} is also easy since then $w_\varepsilon$ is different from $w$ on the set of measure
$\asymp \beta^{-1}\varepsilon$ if $\beta\le C_0$ and on the set of measure $\asymp \varepsilon$ if
$\beta \ge C_0$. Actually $w$ is uniformly smooth if $\beta \gtrsim 1$ and $\ell(x)\le \epsilon \bar{r}$ and we do not need any mollification here.

One definitely can improve estimates (\ref{26-3-60})--(\ref{26-3-62}) but we do not need it. \end{proof}

Consider now the analytical expressions and estimate the semiclassical errors.

\begin{remark}\label{rem-26-3-13}
\begin{enumerate}[label=(\roman*), wide, labelindent=0pt]
\item\label{rem-26-3-13-i}
From now on until the end of this Section we assume that $M=1$ to avoid possible degenerations.

\item\label{rem-26-3-13-ii}
Recall that we can reduce operator with mollified potential to a canonical form provided
$\varepsilon \ge C(\mu ^{-1}h )^{\frac{1}{2}}|\log \mu|$ (see Section~\ref{book_new-sect-18-7}). However here we will have a much better estimate since we will take $\varepsilon \ge h^{\frac{2}{3}-\delta}$.
\end{enumerate}
\end{remark}

\subsection{Charge term}
\label{sect-26-3-3-2}

Let us consider the \index{charge term}{\emph{charge term\/} i.e. expression $\int e(x,x,\nu)\,dx= (\Tr \uptheta (\nu-H))$.

\subsubsection{Regular zone.}
\label{sect-26-3-3-2-1}
Then the results of Section~\ref{book_new-sect-18-9} implies that as
\begin{gather}
W+\nu \asymp \zeta^2
\label{26-3-63}\\
\shortintertext{and}
|\nabla W| \asymp \zeta^2\ell^{-1}
\label{26-3-64}
\end{gather}
contribution of the ball $B(x,\ell(x))$ to expression (\ref{26-3-23}) does not exceed $C(1+\mu h) h^{-2}\asymp C\zeta^2 \ell^2+ CB \ell^2$ exactly as in the mock proof.

Then summation with respect to $\ell$-partition in this zone results in
$CB^{\frac{2}{3}}$ as $B\le Z$,
$CZ^{\frac{2}{3}}$ as $Z\le B\le Z^{\frac{4}{3}}$ and $CB^{\frac{1}{5}}Z^{\frac{2}{5}}$ as $Z^{\frac{4}{3}}\le B\le Z^3$.

\begin{remark}\label{rem-26-3-14}
\begin{enumerate}[label=(\roman*), wide, labelindent=0pt]
\item\label{rem-26-3-14-i}
Condition (\ref{26-3-63}) is fulfilled as $\ell(x)\le \epsilon \bar{r}$.
\item\label{rem-26-3-14-ii}
Further, since $M=1$ both conditions (\ref{26-3-63}) and (\ref{26-3-64}) are fulfilled if $|x| \le (1-\epsilon)\bar{r}_m$ (we pick up $\y_m=0$ and $\bar{r}_m$ exact radius of $\supp(\rho^\TF_B)$).
\end{enumerate}
\end{remark}

\subsubsection{Border strip.}
\label{sect-26-3-3-2-2}
Now we need to consider the contribution of the border strip
$\cY\coloneqq \{x\colon \gamma(x) \le \epsilon \}$ with
$\gamma (x) = \epsilon (\bar{r} -|x|)\bar{r}^{-1}$ and
$\bar{r}\coloneqq \bar{r}_m$. Here
$\ell \asymp \bar{r}$, $\zeta \asymp \bar{\zeta}$ with
\begin{align}
&\bar{r}\asymp \left\{\begin{aligned}
&(Z-N)_+^{-\frac{1}{3}}\qquad &&\text{if\ \ } B\le (Z-N)_+^{\frac{4}{3}},\\
&B^{-\frac{1}{4}}\qquad
&&\text{if\ \ } (Z-N)_+^{\frac{4}{3}}\le B\le Z^{\frac{4}{3}},\\
&Z^{\frac{1}{5}}B^{-\frac{2}{5}} \qquad
&&\text{if\ \ } Z^{\frac{4}{3}}\le B\le C Z^3
\end{aligned}\right.\label{26-3-65}\\
\shortintertext{and}
&\bar{\zeta}\asymp \left\{\begin{aligned}
&(Z-N)_+^{\frac{2}{3}}\qquad
&&\text{if\ \ } B\le (Z-N)_+^{\frac{4}{3}},\\
&B^{\frac{1}{2}}\qquad
&&\text{if\ \ } (Z-N)_+^{\frac{4}{3}}\le B\le Z^{\frac{4}{3}},\\
&Z^{\frac{2}{5}}B^{\frac{1}{5}} \qquad
&&\text{if\ \ } Z^{\frac{4}{3}}\le B\le C Z^3
\end{aligned}\right.\label{26-3-66}
\end{align}
and scaling we get $\mu = B\bar{r}\bar{\zeta}^{-1}$ and $h=\bar{\zeta}^{-1}\bar{r}^{-1}$ here.

Let us consider first the case $\nu=0$. Then both conditions (\ref{26-3-63}) and (\ref{26-3-64}) are fulfilled albeit with $\ell_1=\gamma(x)\bar{r}$ and
$\varsigma (x) = \bar{\zeta} \gamma(x)^2$ instead of $\ell$ and $\zeta$.

Thus if $\varsigma\ell_1\ge 1$
(i.e. $\gamma\ge \bar{\gamma}\coloneqq h^{\frac{1}{3}}$), the contribution of the ball $B(x, \gamma(x)\bar{r})$ to the remainder does not exceed
$C\mu h^{-1} \gamma^2$\,\footnote{\label{foot-26-18} Really, after additional rescaling $x\mapsto x\gamma^{-1}$, $w\mapsto w\gamma^{-2}$ we have $\mu_1=\mu \gamma^{-1}$, $h_1=h\gamma^{-3}$ and $\mu_1h_1^{-1}=\mu h^{-1}\gamma^2$.} and therefore the total contribution of zone
$\cY_1\coloneqq \{x\colon \bar{\gamma} \le \gamma(x) \le \epsilon \bar{r} \}$ to the remainder does not exceed
\begin{equation}
C\mu h^{-1}\int \gamma(x)^{-1}\,dx \asymp C\mu h^{-1} |\log h|=
CB\bar{r}^2 |\log h|
\label{26-3-67}
\end{equation}
which is $O(Z^{\frac{2}{3}})$ as long as $B\le Z^{\frac{4}{3}}(\log Z)^{-2}$.

Further, the same approach works if
$|\nu | \lesssim \bar{\zeta}^2\bar{\gamma}^3\asymp \bar{\zeta}^2h \asymp \bar{\zeta}\bar{r}^{-1}$ which is equivalent to $(Z-N)_+\le \bar{\zeta}$ (then $|\nabla W|\asymp \varsigma^2 \ell_1^{-1}$ if $\gamma(x)\ge \bar{\gamma}$) \underline{and} also if this condition is violated but $|\nu|\le \bar{\zeta}^2$; in the latter case we need to pick up
$\bar{\gamma}= \bar{\gamma}_1\coloneqq |\nu|^{\frac{1}{3}}\bar{\zeta}^{-\frac{2}{3}}$.

To get rid off the logarithmic factor let us consider propagation. Recall that it goes along magnetic lines i.e. that $(x_1,x_2)$ remains constant. Let us consider propagation in the direction in which $|x_3|$ increases
(i.e. $\gamma(x)$ decays); we do not need to consider zone
$\cY _1\cap \{|x_3|\le Z^{-\delta}\bar{r}\}$ since contribution of this zone (\ref{26-3-67}) is $o(B\bar{r}^2)$.

One can see easily that we can follow dynamics which does not return for a time
$T_1^*(x) \coloneqq T_1(x) (\gamma(x)/\bar{\gamma})^\delta$ where
$T(x)\asymp \ell_1 ^{-1}\varsigma^{-1}\asymp \bar{r}\bar{\zeta}^{-1}\gamma^{-1}$ is a time required for this dynamics to pass though $B(x,\ell_1(x))$. Therefore one can replace (\ref{26-3-67}) by
\begin{equation}
C\mu h^{-1}\int_{\{ x\colon \gamma(x)\ge \bar{\gamma}\}}
\bar{\gamma}^\delta \gamma(x) ^{-1-\delta}\,dx \asymp
C\mu h^{-1} = CB\bar{r}^2.
\label{26-3-68}
\end{equation}

Further, as $|\nu |\ge B\bar{r}$ we need also to consider zone
$\cY_0\coloneqq \{x\colon \gamma(x)\le \bar{\gamma}\}$. In this zone we take $\ell_1=\bar{\ell}_1=\bar{\gamma}\bar{r}$ and
$\varsigma= \bar{\varsigma}=(|\nu|\bar{\gamma})^{\frac{1}{2}}$ with $\ell_1\varsigma\ge 1$ and since
$|\nabla W|\asymp \varsigma^2\ell_1^{-1}$, contribution of $B(x,\ell_1(x))$
to the remainder does not exceed $CB\bar{\ell}_1^2$ and the total contribution of $\cY_2$ does not exceed $CB \bar{r}^2$ which what exactly we achieved for zone $\cY_1$ after we got rid off logarithm. We take mollification parameter
$\varepsilon =\bar{\varsigma}^{-1}Z^\delta$\,\footnote{\label{foot-26-19} One can see easily that the resulting errors in the expressions(\ref{26-3-23}) and (\ref{26-3-24})--(\ref{26-3-25}) will not violate our claims.}.

Furthermore, zone $\cY_3=\{x\colon |x|\ge \bar{r}+\bar{\ell}_1\}$ is classically forbidden. So we can take here
\begin{equation}
\ell_1(x)=\epsilon (|x|-\bar{r}),\qquad \varsigma(x)= \min\bigl(\bar{\varsigma}\bar{\ell}_1^{-\frac{1}{2}}\ell_1(x)^{\frac{1}{2}},
|\nu|^{\frac{1}{2}}\bigr)
\label{26-3-69}
\end{equation}
and prove easily that its contribution also does not exceed $CB\bar{r}^2$.

Returning to the case $|\nu|\lesssim \bar{\zeta}$ we see that the contribution of zone $\cY_2 $ to the remainder does not exceed $CB\bar{r}^2$ because effective semiclassical parameter here is $h_1= 1$ and non-degeneracy condition is of no concern for us. We take mollification parameter
$\varepsilon =\bar{\varsigma}^{-1}Z^\delta$\,\footref{foot-26-19}.

Moreover, we can modify $W$ in $\cY_2$ (make it negative there) so that this zone would be classically forbidden with $\ell_1$, $\varsigma$ defined by (\ref{26-3-69}) with $|\nu|$ replaced by $\bar{\zeta}$.

Finally in the case $B\le |\nu|$ (i. e. $B\le C(Z-N)_+^{\frac{4}{3}}$ we can apply the above arguments with $\bar{\gamma}=1$ and arrive to the same result.
Therefore we proved in all cases

\begin{claim}\label{26-3-70}
If $M=1$ the total contribution of the border strip $\cY$ to the remainder in the charge term is $O(B\bar{r}^2)$ which does not exceed $CB^{\frac{1}{2}}$ as $B\le Z^{\frac{4}{3}}$ and $CB^{\frac{1}{5}}Z^{\frac{2}{5}}$ if $Z^{\frac{4}{3}}\le B\le Z^3$.
\end{claim}

\subsubsection{Conclusion.}
\label{sect-26-3-3-2-3}
If $Z^2\le B\le Z^3$ we need to estimate also contribution of the \index{inner core}\emph{inner core\/} $\cX_0\coloneqq \{x\colon \ell(x)\le CZ^{-1}\}$. By means of variational methods we will prove (see Corollary~\ref{cor-26-A-5})

\begin{claim}\label{26-3-71}
If $Z^2\le B\le Z^3$ the contributions of $\cX_0$ to both $\int e(x,x,\nu)\,dx$ and $\int P'_B(W(x)+\nu)\,dx$ do not exceed $CBZ^{-2}$.
\end{claim}

Then we arrive to the following

\begin{proposition}\label{prop-26-3-15}
Let $M=1$. Then

\begin{enumerate}[label=(\roman*), wide, labelindent=0pt]
\item\label{prop-26-3-15-i}
For constructed above potential $W$ expression \textup{(\ref{26-3-23})} does not exceed $CZ^{\frac{2}{3}}+ CB^{\frac{1}{5}}Z^{\frac{2}{5}}$.

\item\label{prop-26-3-15-ii}
If $B\le Z$ expression \textup{(\ref{26-3-23})} does not exceed
$C(B+1)^{\delta} Z^{\frac{2}{3}-\delta}$.
\end{enumerate}
\end{proposition}

\subsection{Trace term}
\label{sect-26-3-3-3}

Let us consider the \index{trace term}\emph{trace term\/} i.e. expression
$\int e_1 (x,x,\nu)\,dx=\Tr ((H-\nu)^-)$.

\subsubsection{Regular zone.}
\label{sect-26-3-3-3-1}

Here again let us consider first zone where $|x|\le (1-\epsilon)\bar{r}$. Then the contribution of $B(x,\ell(x))$ to the Tauberian remainder\footnote{\label{foot-26-20} We will consider a bit later transition from the Tauberian expression to the magnetic Weyl expression.} does not exceed
$C\zeta^2 (h^{-1}+\mu) \asymp C\zeta^3\ell + CB \zeta \ell$ as in the mock proof and the summation over zone results in $CZ^{\frac{5}{3}}+ CZ^{\frac{4}{3}}B^{\frac{1}{3}}+ CZ^{\frac{3}{5}}B^{\frac{4}{5}}$.

\subsubsection{Border strip.}
\label{sect-26-3-3-3-2}
Again in zone $\cY_1$ contribution of $B(x,\gamma(x))$ does not exceed $CB\varsigma \ell_1$ and the summation over this zone returns
\begin{equation}
CB\int \varsigma \ell_1^{-2}\,dx
\label{26-3-72}
\end{equation}
and plugging $\ell_1=\bar{r}\gamma$ and $\varsigma =\bar{\zeta}\gamma^2$ results in $CB^{\frac{5}{4}}$ as $B\lesssim Z^{\frac{4}{3}}$ and $CB^{\frac{4}{5}}Z^{\frac{3}{5}}$ otherwise. The analysis of zone $\cY_0$ if there $\cY_2=\emptyset$ is also easy.

Consider zones $\cY_2$ and $\cY_0$. The same arguments as before imply that their contributions to the remainder do not exceed
$CB\bar{r}^2\bar{\varsigma} \bar{\ell}_1^{-1}$ which is what we got before.

\subsubsection{Justification: from Tauberian to magnetic Weyl expression.}
\label{sect-26-3-3-3-3}

\paragraph{Case $\mu h \le C_0$.}
\label{sect-26-3-3-3-3-1}
We need to prove that with the announced error we can replace the Tauberian expression by magnetic Weyl one. Note that the canonical form of $\zeta^{-2} H_{A,W}$ as described in Sections~\ref{book_new-sect-13-2} and \ref{book_new-sect-18-7} is
\begin{multline}
\cH=\cH_0 +\\
\shoveright{\mu^{-2} \omega_1 (x_1, \mu^{-1}h D_1,x_3) +
\mu^{-2} \omega_2 (x_1, \mu^{-1}h D_1,x_3)
(x_2^2+\mu^2h^2 D_2^2 )}\\[3pt]
+\mu^{-1}h \omega_3 (x_1, \mu^{-1}h D_1,x_3) + O\bigl( \mu^{-3} h (\gamma+\mu^{-1})^{-\frac{3}{2}}+ \mu^{-4}\bigr)
\label{26-3-73}
\end{multline}
with
\begin{gather}
\cH_0= h^2 D_3^2 - (x_2^2+\mu^2h^2 D_2^2\pm \mu h) +
w(x_1, \mu^{-1}h D_1,x_3)
\label{26-3-74}\\
\shortintertext{and}
\gamma = \epsilon \min_j |w-2j\mu h|
\label{26-3-75}
\end{gather}
where we used the fact that $w\in \mu h \sC^{\frac{5}{2}}+\sC^n$, $\mu^{-3}h = \mu^{-4}\cdot \mu h$. Here we have signs ``$+$'' and ``$-$'' on $q/2$ of the diagonal elements equally.

Then the Tauberian expression is
\begin{multline}
\const \cdot \mu h^{-2}
\int \sum_{j\ge 0}
\bigl( w-2j \mu h -
\mu^{-2}\omega_1 - 2j\mu^{-1}h \omega_2 \bigr)_+^{\frac{3}{2}} \times \\
\bigl(\psi + \mu^{-2}\psi_1 + 2j \mu^{-1}h\psi_2\bigr)\, dx
\label{26-3-76}
\end{multline}
where term with $j=0$ enters with the weight $\frac{1}{2}$ and an error does not exceed
\begin{equation*}
Ch^{-3}
\Bigl(\mu ^{-4} + \mu^{-3}h \int (\gamma+\mu^{-1})^{-\frac{3}{2}}\, dx\Bigr) \asymp
C\mu^{-\frac{7}{2}}h^{-3}
\end{equation*}
because an integral does not exceed $C\mu ^{\frac{1}{2}} (\mu h)^{-1}$; since $\mu \ge h^{-\frac{3}{5}}$ this error does not exceed $Ch^{-\frac{9}{10}}$ which is better than $O(h^{-1})$.

On the other hand, if we consider the difference between (\ref{26-3-76}) and the same expression with $\omega_1=\omega_2=\psi_1=\psi_2=0$ and consider it as a Riemannian sum and replace it by an integral we get $G\mu^{-2}h^{-3}$ with an error not exceeding $C\mu^{-4}(\mu h)^{\frac{1}{2}}h^{-3}$ which is even less. Therefore (\ref{26-3-76}) becomes
\begin{equation*}
\int P_{\mu h} (w)\psi\,dx + G\mu^{-2}h^{-3}
\end{equation*}
and comparing with the result if $\mu \asymp h^{-\frac{3}{5}}$ when we get the same answer albeit with $G=0$ we conclude that $G$ \emph{must\/} be $0$. This concludes the justification in $\cX_2$.

\paragraph{Case $\mu h \ge C_0$.}
\label{sect-26-3-3-3-3-2}

In this case we need a simplified version of (\ref{26-3-73})
$\cH=\cH_0 + O(\mu^{-1}h )$ and we need to consider only $j=0$ and replacing $\cH$ by $\cH_0$ brings and error
$C\mu h^{-2}\times \mu^{-1}h =O(h^{-1}) $. This takes care of $\cX_2$ and after scaling of $\cY$.

\subsubsection{Conclusion.}
\label{sect-26-3-3-3-4}
As $Z^2\le B\le Z^3$ we need to estimate also contribution of
$\cX_0=\{x\colon \ell(x)\le CZ^{-1}\}$. By means of variational methods we will prove (see Corollary~\ref{cor-26-A-5})

\begin{claim}\label{26-3-77}
For $Z^2\le B\le Z^3$ the contributions of $\cX_0$ to both
$\int e_1(x,x,\nu)\,dx$ and $\int P'_B(W(x)+\nu)\,dx$ do not exceed $CB$.
\end{claim}

Then we arrive to the following

\begin{proposition}\label{prop-26-3-16}
Let $M=1$. Then
\begin{enumerate}[label=(\roman*), wide, labelindent=0pt]
\item\label{prop-26-3-16-i}
For constructed above potential $W$ expression \textup{(\ref{26-3-25})} does not exceed $CZ^{\frac{5}{3}}+ CZ^{\frac{4}{3}}B^{\frac{1}{3}}+ CZ^{\frac{3}{5}}B^{\frac{4}{5}}$.

\item\label{prop-26-3-16-ii}
If $B\le Z$ expression \textup{(\ref{26-3-25})} does not exceed $C(B+1)^{\delta} Z^{\frac{5}{3}-\delta}$ (but one should subtract a Schwinger term from the trace).
\end{enumerate}
\end{proposition}

\subsection{Semiclassical $\D$-term: local theory}
\label{sect-26-3-3-4}
Unfortunately, we do not have any non-smooth theory (cf. Section~\ref{book_new-sect-16-7}) here so far but actually we almost do not need it since singularities are rather rare. Let us introduce a scaling function (\ref{26-3-75}) and consider
\begin{equation}
J_\lambda (z)=\int \phi_{z,\lambda} (x)
\bigl(e(x,x,\tau) - P_\beta (w(x)+\tau)\bigr)\,dx
\label{26-3-78}
\end{equation}
with $\phi_{z,\lambda} (x)= \phi (\lambda^{-1}(x-z))$ and
$\lambda \le \gamma(z)$. Scaling $x\mapsto \lambda^{-1}(x-z)$ we have
$\mu \mapsto \mu' = \lambda \mu$ and $h\mapsto h'= \lambda^{-1}h$.

Then, according to Section~\ref{book_new-sect-13-4}
\begin{equation}
|J_\lambda (z)|\le Ch^{\prime\, -2}(1+\mu'h')\asymp C\lambda^2 h^{-2} (1+\mu h)
\label{26-3-79}
\end{equation}
as long as $\lambda \ge h$.

Really, a transition from the Tauberian decomposition to magnetic Weyl one in this case is easy: skipping all perturbation terms $O(\mu^{-2}+\mu^{-1}h )$ in (\ref{26-3-73}) and also setting $\psi_1=\psi_2=0$ results in an error
$O\bigl( \mu^{-2}h ^{-3} + h ^{-1} \bigr)$ in
(\ref{26-3-76})-like expression albeit with the power $\frac{1}{2}$ rather than $\frac{3}{2}$ and without integration:
\begin{multline}
\const \cdot \mu h ^{-2}
\sum_{j\ge 0}
\bigl( w-2j \mu h -
\mu^{-2}\omega_1 - 2j\mu^{-1}h \omega_2 \bigr)_+^{\frac{1}{2}} \times \\
\bigl(\psi + \mu^{-2}\psi_1 + 2j \mu^{-1}h \psi_2\bigr);
\label{26-3-80}
\end{multline}
scaling produces expression smaller than (\ref{26-3-79}).

Let us apply this estimate (\ref{26-3-79}) to the Fefferman--de Llave decomposition (\ref{26-3-8}).

\paragraph{Case $\mu \le C_0h ^{-1}$.}
\label{sect-26-3-3-4-0-1}
\begin{enumerate}[label=(\roman*), wide, labelindent=0pt]
\item\label{sect-26-3-3-4-0-1-i}
Consider first a pair $(z,z')$ such that $|z'-z''|\le \epsilon_0 \gamma(z')$; then also $|z'-z''|\le \epsilon_0 \gamma(z'')$ and we take
$\lambda =\epsilon |z'-z''|$.

Then in the virtue of (\ref{26-3-79}) the total contribution to $\D$-term of all such pairs belonging to $B (z, \gamma(z))$, and with $|z'-z''|\asymp \lambda$ does not exceed
\begin{equation}
C\gamma^3\lambda^{-3} \times \lambda^{-1} \times
\lambda^2 h ^{-4} (1+\mu h )^2 \asymp C\gamma^3 h ^{-4} (1+\mu h )^2
\label{26-3-81}
\end{equation}
where $C\gamma^3\lambda^{-3}$ estimate the number of such pairs, $\lambda^{-1}$ the inverse distance between them, and $C\lambda^2 h ^{-2} (1+\mu h )$ is the right-hand expression of (\ref{26-3-79}).

\smallskip
Then summation over $\lambda \in(\mu^{-1},\gamma)$ results in
$C\gamma^3(1+\mu h )^2h ^{-4}|\log (\mu \gamma)|$.

\smallskip
Further, summation over all balls
$B(z, \gamma)\subset B(0,1)$ with $\gamma(z) \asymp \gamma$ results in
$C(\mu h )^{-1}\gamma h ^{-4} |\log (\mu \gamma)|$ since there are
$\asymp (\mu h )^{-1}\gamma^{-2}$ such balls due to non-degeneracy assumption
$|\nabla w|\asymp 1$. Summation over $\gamma\in ( \mu^{-1},\mu h ) $ results in $Ch ^{-4} |\log (\mu ^2h )|$.

As $\lambda \le \mu^{-1}$ we can apply standard non-magnetic methods without Fefferman--de Llave decomposition (\ref{26-3-8}). Coefficients are smooth after scaling as long as $\varepsilon \ge \mu^{-1}$.

\item\label{sect-26-3-3-4-0-1-ii}
Consider \emph{disjoint\/} pairs $(z',z'')$ with
$|z'-z''|\ge \max\bigl(\gamma(z'),\gamma(z'')\bigr)$. Here estimate (\ref{26-3-79}) is not sufficient and it should be replaced by
\begin{gather}
|J_\gamma (z)|\le
C\lambda^3 h ^{-2} (1+\mu h )\label{26-3-82}\\
\shortintertext{as long as}
\gamma \ge h ^{\frac{2}{3}-\delta}.
\label{26-3-83}
\end{gather}
Really, the shift for time $T$ with respect to $\xi_3$ is $\asymp T $ provided $|\nabla_{x_3}w|\asymp 1$ and this shift is observable if
$T \times \gamma \gtrsim h ^{1-\delta}$. Similarly, in the canonical form the shift for time $T$ with respect to $\mu^{-1}\xi_i$ is $\asymp \mu^{-1}T$ provided $|\nabla_{x_i}w|\asymp 1$ and this shift is observable if
$\mu^{-1}T \times \gamma \gtrsim \mu^{-1}h ^{1-\delta}$. In both cases shift with $T\in ( \gamma^{\frac{1}{2}},\epsilon_0)$ is observable under assumption (\ref{26-3-83}) and therefore we can extend $T\asymp \gamma^{\frac{1}{2}}$ to $T\asymp 1$.

Note that for $\varepsilon \ge h ^{\frac{2}{3}-\delta}$ assumption (\ref{26-3-83}) is fulfilled automatically. Then contribution of each disjoint pair to $\D$-term does not exceed
\begin{gather*}
Ch ^{-4}(1+\mu h )^2 \gamma(z')^3\gamma(z'')^3|z'-z''|^{-1}\\
\intertext{and the total contribution does not exceed}
Ch ^{-4}(1+\mu h )^2 \iint |z'-z''|^{-1}\,dz'dz''\asymp
Ch ^{-4}(1+\mu h )^2.
\end{gather*}

\item\label{sect-26-3-3-4-0-1-iii}
To shed off logarithm in \ref{sect-26-3-3-4-0-1-i}
 we need a slightly better estimate than (\ref{26-3-79}). The same arguments as in Part~\ref{sect-26-3-3-4-0-1-ii} result in
\begin{equation}
|J_\lambda (z)|\le
C\lambda^2 h ^{-2} (1+\mu h )\cdot(1+\lambda \gamma/h )^{-\delta}.
\label{26-3-84}
\end{equation}
Really, we just advance from time $T\asymp \lambda$ to
$T\asymp \lambda (1+\lambda \gamma/h )^{\delta}$.

Then the same factor is acquired by the right-hand expression of (\ref{26-3-81}) and the summation with respect to $\lambda \in (h \gamma^{-1}, \gamma)$ results in $C\gamma^3(1+\mu h )^2h ^{-4}$ but summation with respect to
$\lambda \in (\mu^{-1},h \gamma^{-1})$ results
\begin{equation*}
C\gamma^3(1+\mu h )^2h ^{-4}\bigl(1+|\log (\mu h \gamma^{-1})|\bigr).
\end{equation*}
Further, summation over all balls $B(z, \gamma)\subset B(0,1)$ with
$\gamma(z) \asymp \gamma$ results in
$C(\mu h )^{-1}\gamma (1+\mu h )^2h ^{-4}
\bigl(1+|\log (\mu h \gamma^{-1})|\bigr)$ and, finally, summation over
$\gamma \lesssim \mu h $ results in $C h ^{-4}(1+\mu h )^2$.
\end{enumerate}

Note that in all cases perturbation terms in (\ref{26-3-73}) and (\ref{26-3-80}) result in the error not exceeding the announced one.

\paragraph{Case $\mu \ge C_0h ^{-1}$.}
\label{sect-26-3-3-4-0-2}

So far factor $(1+\mu h)$ was for a compatibility only. Now it is important.

Exactly the same arguments work as $\mu \ge C_0h ^{-1}$ with a minor modifications:
\begin{enumerate}[label=(\alph*),wide, labelindent=0pt]
\item\label{sect-26-3-3-4-0-2-a}
$\gamma (x)$ now is defined by (\ref{26-3-75}) with $j=0$ and its upper bound is $1$ rather than $\mu h $.

\item\label{sect-26-3-3-4-0-2-b}
Also the number of $\gamma$-balls is $\asymp \gamma^{-2}$ rather than $\asymp (\mu h)^{-1}\gamma^{-2}$;

\item\label{sect-26-3-3-4-0-2-c}
$\lambda$ now runs from $h $ to $\gamma$ in~\ref{sect-26-3-3-4-0-1-i} and~\ref{sect-26-3-3-4-0-1-iii}.

\item\label{sect-26-3-3-4-0-2-d}
We need to estimate contribution of pairs $(z',z'')$ with
$|z'-z''|\le h $. One can see easily that $e(x,x,\tau)\le \mu h ^{-2}$ and therefore the total contribution of these pairs does not exceed
$C\mu^2h ^{-4} \iint |z'-z''|^{-1}\,dz'dz''\asymp C\mu^2h ^{-4} \times h ^2 \asymp C\mu^2h ^{-2} $.
\end{enumerate}

Therefore we have the following

\begin{proposition}\label{prop-26-3-17}
As $|\nabla w|\asymp 1$ and $\varepsilon \ge h ^{\frac{2}{3}-\delta}$ in $B(0,1)$ and $\phi \in \sC^\infty (B(0,\frac{1}{2})$
\begin{multline}
\D\Bigl(\phi \bigl(e(x,x,\tau)-P'_\beta (w(x)+\tau)\bigr),
\phi \bigl(e(x,x,\nu)-P'_\beta (w(x)+\tau)\Bigr)\le \\
C(1+\mu h )^2h ^{-4}.
\label{26-3-85}
\end{multline}
\end{proposition}

\begin{remark}\label{rem-26-3-18}
One can see easily that one can select
$\varepsilon \ge h ^{\frac{2}{3}-\delta}$ such that expressions (\ref{26-3-60}), (\ref{26-3-61}) and (\ref{26-3-61}) will be respectively $O(h ^{2+\delta})$, $O(h ^{1+\delta})$ and $O(h ^{2+\delta})$.
\end{remark}

\subsection{Semiclassical $\D$-term: global theory}
\label{sect-26-3-3-5}

\subsubsection{Regular zone.}
\label{sect-26-3-3-5-1}

The above results allow us to consider a total contribution of zone $\cX_2$ into semiclassical $\D$-term. As before let us consider $\ell$-admissible partition of unity there and apply it to Fefferman--de Llave decomposition (\ref{26-3-8}). Then the total contribution of the elements which are not disjoint does not exceed

\begin{equation}
\sum_n \ell_n^{-1} \bigl(1+B\zeta_n^{-2}\bigr)^2 \ell_n^4 \zeta_n^4 \asymp
\int \underbracket{ \bigl(\zeta^4+B^2\bigr)\ell^3} \,\ell^{-1}d\ell
\label{26-3-86}
\end{equation}
where $\bigl(1+B\zeta_n^{-2}\bigr)^2$ and $\ell_n^4 \zeta_n^4$ are
$(1+\mu h )$ and $h ^{-4}$ respectively and $\ell_n^{-1}$ is a scaling factor.

Then if $\zeta^2 = Z\ell^{-1}$, an integral equals to the value of the selected expression as $\ell$ reaches its maximum, i.e. for $\ell=Z^{-\frac{1}{3}}$ for $B\le Z^{\frac{4}{3}}$ and $\ell= Z^{\frac{1}{5}}B^{-\frac{2}{5}}$ for
$ Z^{\frac{4}{3}}\le B\le Z^3$ and we arrive to $CZ^{\frac{5}{3}}$ and $CZ^{\frac{3}{5}}B^{\frac{4}{5}}$ respectively.

On the other hand, if $\zeta^2=\ell^{-4}$ an integral equals to the value of the selected expression as $\ell$ reaches its minimum, i.e. for $\ell=Z^{-\frac{1}{3}}$ and only in the case $B\le Z^{\frac{4}{3}}$ and we arrive to $CZ^{\frac{5}{3}}$ again.

Furthermore, the total contribution of the disjoint elements does not exceed
\begin{multline}
\sum_{n,p} |z_n-z_p|^{-1} \bigl(1+B\zeta_n^{-2}\bigr) \bigl(1+B\zeta_p^{-2}\bigr) \ell_n^2 \zeta_n^2 \ell_p^2 \zeta_p^2 \asymp\\
\iint \underbracket{ (\ell+\ell')^{-1} \bigl(\zeta^2+B\bigr)\ell^2
\bigl(\zeta^{\prime\, 2}+B\bigr)\ell^{\prime\,2}}\, \ell^{-1}d\ell \,\ell^{\prime\,-1}d\ell'.
\label{26-3-87}
\end{multline}
Then if $\zeta^2 = Z\ell^{-1}$ and $\zeta^{\prime\,2} = Z\ell^{\prime\,-1}$ an integral equals to the value of the selected expression as both $\ell$ and $\ell'$ reach their maxima, and we arrive to $CZ^{\frac{5}{3}}$ and $CZ^{\frac{3}{5}}B^{\frac{4}{5}}$ respectively.

On the other hand, if $\zeta^2=\ell^{-4}$ and $\zeta^{\prime\,2}=\ell^{\prime\,-4}$ (we do not need to consider mixed pair) an integral equals to the value of the selected expression as both $\ell$ and $\ell'$ reach their minima, and we arrive to $CZ^{\frac{5}{3}}$.

Therefore (combining with Proposition~\ref{cor-26-A-5} as $Z^2\le B\le Z^3$) we arrive to

\begin{proposition}\label{prop-26-3-19}
Let $M=1$. Then
\begin{enumerate}[label=(\roman*), wide, labelindent=0pt]
\item\label{prop-26-3-19-i}
The total contribution of the zone
$\{x\colon \ell (x)\le (1-\epsilon)\bar{r}\}$ to the semiclassical $\D$-term does not exceed $CZ^{\frac{5}{3}}$ and $CZ^{\frac{3}{5}}B^{\frac{4}{5}}$ for
$B\le Z^{\frac{4}{3}}$ and for $Z^{\frac{4}{3}}\le B\le Z^3$ respectively.

\item\label{prop-26-3-19-ii}
If $B\le Z$ this contribution does not exceed
$C(B+1)^{\delta} Z^{\frac{5}{3}-\delta}$.
\end{enumerate}
\end{proposition}

\subsubsection{Border strip.}
\label{sect-26-3-3-5-2}

Border strip
$\cY=\{x\colon (1-\epsilon)\bar{r}\le \ell(x) \le (1+\epsilon)\bar{r}\}$ is more subtle. Here we need to use the same $\ell_1(x)= \epsilon_0(\bar{r}-|x|)$ partition as before.

\begin{remark}\label{rem-26-3-20}
$\cY$ is already covered by our arguments if
$\bar{r}\asymp (Z-N)_+^{-\frac{1}{3}}$.
\end{remark}

\paragraph{Close elements.}
\label{sect-26-3-3-5-2-1}
Consider first contribution of elements which are not disjoint. It is given by the left-hand expression of (\ref{26-3-86}) with $\ell$, $\zeta$ replaced by
$\ell_1(x)=\bar{r}\gamma (x)$ and $\varsigma(x)=\bar{\zeta}\gamma(x)^2$ respectively. However since the layer $\{x\colon \gamma(x)\asymp \gamma\}$ contains $\asymp \gamma^{-2}$ elements the right-hand expression should be replaced by
\begin{equation*}
\int \underbracket{ B^2\bar{r}^3 \gamma } \,\gamma^{-1}d\gamma \asymp B^2\bar{r}^3
\end{equation*}
since $\varsigma^2\le B$; so we arrive to $O\bigl(\max(B^{\frac{5}{4}},Z^{\frac{3}{5}}B^{\frac{4}{5}}\bigr)$.

Meanwhile for $\cY_2$ we have $\gamma (x)=\bar{\gamma}\le 1$ and
$\varsigma(x)= \bar{\varsigma}= \bar{\zeta}\bar{\gamma}^2$ and its contribution does not exceed what we got for $\cY_1$.

\paragraph{Disjoint elements.}
\label{sect-26-3-3-5-2-2}
Consider contribution of the disjoint elements. It is given by the left-hand expression of (\ref{26-3-87}) with $\ell$, $\zeta$ replaced by $\gamma$ and $\varsigma$ respectively. Note that
$\sum_{n,p} |z_n-z_p|^{-1} \asymp \bar{r}^{-1} \gamma^{-2} \gamma^{\prime\,-2}$ where we sum with respect to all pairs with $\gamma_n\asymp \gamma$ and $\gamma_p\asymp \gamma'$. Therefore the right-hand expression should be replaced by
\begin{equation}
\int \underbracket{ \bar{r}^3 B^2 } \,\gamma^{-1}d\gamma\, \gamma^{\prime\,-1}d\gamma'
\label{26-3-88}
\end{equation}
which leads to $C\bar{r}^3B^2 |\log (\bar{r}^{-1}\bar{\gamma})|^2$ which differs from what we got before by a logarithmic factor. To get rid off it we will use exactly the same trick as in Paragraph~\emph{\ref{sect-26-3-3-2-2}.2.2. \nameref{sect-26-3-3-2-2}\/} proving Proposition~\ref{prop-26-3-15} because considering disjoint pairs we consider the same objects as there. Then instead of (\ref{26-3-88}) we arrive to
\begin{equation*}
\int \underbracket{ \bar{r}^3 B^2 \gamma^{-\delta}\gamma^{\prime\,-\delta}\bar{\gamma}^{2\delta}} \,\gamma^{-1}d\gamma\, \gamma^{\prime\,-1}d\gamma'
\end{equation*}
which results in $C\bar{r}^3B^2$.

Meanwhile for $\cY_2$ we have $\gamma (x)=\bar{\gamma}\le \bar{r}$ and
$\varsigma(x)= \bar{\varsigma}= \bar{\zeta}\bar{\gamma}^2$ and its contribution does not exceed what we got for $\cY_1$.
\enlargethispage{2\baselineskip}

\paragraph{Conclusion.}
\label{sect-26-3-3-5-2-3}
Finally, analysis in the outer zone is trivial. Therefore we arrive

\begin{proposition}\label{prop-26-3-21}
Let $M=1$. Then for constructed above potential $W$
\begin{enumerate}[label=(\roman*), wide, labelindent=0pt]

\item\label{prop-26-3-21-i}
Expression \textup{(\ref{26-3-24})} does not exceed
$CZ^{\frac{5}{3}}+ CZ^{\frac{3}{5}}B^{\frac{4}{5}}$.

\item\label{prop-26-3-21-ii}
If $B\le Z$ expression \textup{(\ref{26-3-24})} does not exceed
$C(B+1)^{\delta} Z^{\frac{5}{3}-\delta}$.
\end{enumerate}
\end{proposition}

\chapter{Applying semiclassical methods: \texorpdfstring{$M\ge 2$}{M\textge 2}}
\label{sect-26-4}

Let us consider now the molecular case ($M\ge 2$). The major problem is that the non-degeneracy condition $|\nabla w|\asymp 1$ is not necessarily fulfilled. Therefore we need to find an alternative approach to the zone
where $\mu \ge h^{-\frac{1}{3}}$) (with $\mu = B\ell \zeta^{-1}$ and $h=1/\ell\zeta$). Recall that it consists of three smaller zones:
zone $\cX_2\coloneqq \{\mu h \le C_0,\ W^\TF_B\ge \epsilon_0 \zeta^2\}$\,\footnote{\label{foot-26-21} Only if $B\le C_1 Z^2$; this zone disappears for $C_1Z^2\le B\lesssim Z^3$.}, zone
$\cX_3 \coloneqq \{\mu h \ge C_0,\ W^\TF_B\ge \epsilon_0 \zeta^2\}$\,\footnote{\label{foot-26-22} Only if
$Z^{\frac{4}{3}}\lesssim B \lesssim Z^3$; this zone disappears for
$B\lesssim Z^{\frac{4}{3}}$.}, and the (most difficult) \emph{boundary strip\/}
$\cY=\{ W^\TF_B \le \epsilon_0\zeta^2\}$, which we leave for the next Section~\ref{sect-26-5}.
\enlargethispage{1.5\baselineskip}

\section{Scaling functions in zone \texorpdfstring{$\cX_2$}{X\texttwoinferior}}
\label{sect-26-4-1-}

\subsubsection{Step 1.}\label{sect-26-4-1-1-1}
We will use the scaling method in this zone; the good news is that $W^\TF_B$ is sufficiently regular after a proper rescaling and also sufficiently non-degenerate. Recall that after we rescale
$x\mapsto \bar{\ell}^{-1}(x-\bar{x})$, $\tau\mapsto \bar{\zeta}^{-2}\tau$ in the ball $B(\bar{x}, \frac{1}{2} \bar{\ell})$ with $\bar{\ell}= \ell(\bar{x})$,
$\bar{\zeta}= Z^{\frac{1}{2}}\bar{\ell}^{-1}$, the rescaled potential $w=\bar{\zeta}^{-2}W^\TF_B$ satisfies in $B(0,2)$ equation
\begin{gather}
\frac{1}{4\pi}\Delta w = \eta P'_\beta (w)\qquad \text{with\ \ }
\eta= \bar{\zeta}\bar{\ell}^2\lesssim 1,\quad \beta=\mu h = B\bar{\zeta}^{-2}\le 1
\label{26-4-1}\\
\intertext{and therefore in $B(0,1)$}
w = -\frac{1}{4\pi} \int |x-z|^{-1}\eta P'_\beta (w(z))\phi(z)\, dz + w'
\label{26-4-2}
\end{gather}
with $\phi\in \sC_0^\infty (B(0,\frac{5}{6}))$ and
$w'\in \sC_0^\infty (B(0,\frac{3}{4}))$.

\begin{remark}\label{rem-26-4-1}
\begin{enumerate}[label=(\roman*), wide, labelindent=0pt]
\item\label{rem-26-4-1-i}
If $\zeta^2 \asymp Z\ell^{-1}$ then
$\eta =Z ^{\frac{1}{2}}\ell^{\frac{3}{2}}\le 1$; this happens for
$B\le Z^{\frac{4}{3}}$, $\ell\lesssim Z^{-\frac{1}{3}}$ and for
$B\ge Z^{\frac{4}{3}}$, $\ell\lesssim B^{-\frac{1}{3}}$.

\item\label{rem-26-4-1-ii}
If $\zeta^2 \asymp \ell^{-2}$ then $\eta \asymp 1$; this happens only for
$B\le Z^{\frac{4}{3}}$, $\ell\gtrsim Z^{-\frac{1}{3}}$.
\end{enumerate}
\end{remark}
\enlargethispage{\baselineskip}

Let us introduce a function
\begin{equation}
\gamma_0(x)= \Bigl(\min_j | w-2j\beta|^{3} +|\nabla w| ^{4}
+|\nabla^2 w|^{6} + |\nabla^3 w'|^{12} \Bigr)^{\frac{1}{12}}.
\label{26-4-3}
\end{equation}

\begin{remark}\label{rem-26-4-2}
We cannot replace $w'$ by $w$ in the last term because
$w\in \sC^{\frac{5}{2}}$ only rather than $\sC^3$.
\end{remark}

\begin{proposition}\label{prop-26-4-3}
$\gamma_0(x)$ is a scaling function i.e. $|x-y|\le \epsilon \gamma_0(x)\implies \gamma_0(y)\asymp \gamma_0(x)$.
\end{proposition}

\begin{proof}
\begin{enumerate}[label=(\alph*), wide, labelindent=0pt]
\item\label{pf-26-4-3-a}
If $w$ belonged to $\sC^4$ (and we would put $w$ instead of $w'$ in the last term of (\ref{26-4-3})) then we would just prove that
$|\nabla \gamma_0|\le c$. Here we should be more subtle. We need to prove that if \begin{phantomequation}\label{26-4-4}\end{phantomequation}
\pagebreak
\begin{align}
&\min_j | w-2j\beta|\le \gamma_0^4, &&|\nabla w|\le \gamma_0^3, \tag*{$\textup{(\ref*{26-4-4})}_{1,2}$}\label{26-4-4-1}\\
&|\nabla ^2 w|\le \gamma_0^2,&& |\nabla^3 w'|\le \gamma_0
\tag*{$\textup{(\ref*{26-4-4})}_{3,4}$}\label{26-4-4-3}
\end{align}
at point $x$, then at point $y$ the same inequalities hold with $\gamma_0$ replaced by $c\gamma_0$\,\footnote{\label{foot-26-23} We need to prove a bit of converse as well; see Part~\ref{pf-26-4-3-c}.}. Definitely this is true for $\textup{(\ref{26-4-4})}_4$ since $w'$ is smooth.

\item\label{pf-26-4-3-b}
Consider $|\nabla^2w|$. Consider $|\Lambda_\alpha w(y)-\Lambda_\alpha w(x)|$ with $\Lambda_\alpha = \nabla^\alpha -\frac{1}{3} \delta_{ij} \Delta$, $\alpha=(i,j)$. Then due to (\ref{26-4-2})
\begin{multline}
|\Lambda_{\alpha,y} w(y)-\Lambda_{\alpha,x} w(x)|\le \\
|\eta \int \bigl(\Lambda_{\alpha,x}|x-z|^{-1}-\Lambda_{\alpha,y}|y-z|^{-1}\bigr) P'_\beta( w(z)) \phi(z)\,dz| +\epsilon_1\gamma_0^2
\label{26-4-5}
\end{multline}
where the last term estimates
$|\Lambda_{\alpha,y} w'(y)-\Lambda_{\alpha,x} w'(x)|$ and we used $\textup{(\ref{26-4-4})}_{4}$. Integrals here are understood in the sense of the principal value ($\vrai$) and $\epsilon_1=\epsilon_1(\epsilon)\to +0$ as $\epsilon\to +0$.

Note that the integral in expression (\ref{26-4-5}) does not change if we add to
$P'_\beta (w)$ any constant with respect to $z$.

Let us consider first this integral over $\{z\colon |x-z|\ge 2\gamma_0 \}$,
provided $\gamma_0 \ge |x-y|$, and note that this integral does not exceed
\begin{multline*}
\eta |\int |x-z|^{-4}|x-y|
\times \\
\bigl(|\nabla w(x)|\cdot |x-z| +
|\nabla^2 w(x)|\cdot |x-z|^2 +|x-z|^{\frac{5}{2}}\bigr)^{\frac{1}{2}}\,dz|\le
C\epsilon_1\eta.
\end{multline*}
Consider now integral over zone $\{z\colon |x-z|\le 2\gamma_0\}$:
\begin{equation*}
|\eta \int \Lambda_{\alpha,x}|x-z|^{-1}\cdot
\bigl( P'_\beta( w(z))-P'_\beta( w(x))\bigr) \,dz|
\end{equation*}
and note that it also does not exceed $\epsilon_1\eta$. Further, t same arguments work for this integral with $x$ replaced by $y$ but still integrated over zone
$\{z\colon |x-z|\le 2\gamma_0\}$ (one needs to remember that
$|\nabla w(x)-\nabla w(y)|= O(\gamma_0)$ and
$|\nabla^2 w(x)-\nabla^2 w(y)|= O(\gamma_0^{\frac{1}{2}})$).

Furthermore, the same arguments work also for these expressions integrated over $\{z\colon |y-z|\le 2\gamma_0\}$ and we are left with
\pagebreak
\begin{equation*}
|\eta \int \omega (x,y,z)
\bigl( P'_\beta( w(x))-P'_\beta( w(y))\bigr) \,dz|
\end{equation*}
integrated over zone $\{z\colon |x-z|\asymp \gamma_0\}$ and $\omega \asymp \gamma_0^{-3}$ and one can estimate it by $\epsilon_1\eta$ easily in the same way. Therefore since $|\Delta w(x)-\Delta w(y)|$ does not exceed $\epsilon_1\eta$ we conclude that $|\nabla ^2w(x)-\nabla ^2w(y)|$ does not exceed $\epsilon_1(\eta + \gamma_0^2)$.

However (\ref{26-4-2}) implies that $\eta \le c\gamma_0^2$ since
$P'_\beta(w)\asymp 1$ in $\cX_2$ and therefore
$|\nabla ^2w(x)-\nabla ^2w(y)|\le \epsilon_1\gamma_0^2$ in $B(x,\gamma_0)$.

Finally, combining this inequality with $\textup{(\ref{26-4-4})}_{2}$ we conclude that $|\nabla w(x)-\nabla w(y)|\le \epsilon_1\gamma_0^3$ in $B(x,\gamma_0)$; finally, combining with $\textup{(\ref{26-4-4})}_{1}$ we conclude that $| w(x)- w(y)|\le \epsilon_1\gamma_0^4$ in $B(x,\gamma_0)$.

\item\label{pf-26-4-3-c}
Therefore $\textup{(\ref{26-4-4})}_{1-4}$ are fulfilled in
$y\in B(x, \epsilon\gamma_0)$ with $\gamma_0$ replaced by
$\gamma_0 (1+C\epsilon_1)$. Further, if we redefine $\gamma_0$ as the minimal scale such that inequalities $\textup{(\ref{26-4-4})}_{1-4}$ are fulfilled in $x$, then $\textup{(\ref{26-4-4})}_{1-4}$ fail in $y\in B(x, \epsilon \gamma_0)$ with $\gamma_0$ replaced by $\gamma_0 (1-C\epsilon_1)$. Therefore with appropriate $\epsilon>0$ we conclude that
$\frac{1}{2}\le \gamma_0(x)/\gamma_0(y)\le 2$.

Obviously, $\gamma_0 \asymp \gamma_{0\old}$ where $\gamma_{0\old}$ was defined by (\ref{26-4-3}) and therefore the same conclusion also holds for $\gamma_{0\old}$.
\end{enumerate}
\end{proof}

Now let us reintroduce the scaling function
\begin{multline}
\gamma_0(x)= \epsilon \Bigl(\min_j | w-2j\beta|^{3} +|\nabla w| ^{4}
+|\nabla^2 w|^{6} + |\nabla^3 w'|^{12} \Bigr)^{\frac{1}{12}}+ \\C_0h ^{\frac{1}{3}}+C_0\eta^{\frac{1}{2}}.
\tag*{$\textup{(\ref*{26-4-3})}^*$}\label{26-4-3-*}
\end{multline}
Then
\begin{claim}\label{26-4-6}
$|x-y|\le 2\gamma_0 (x)\implies \gamma_0(y)\asymp \gamma_0(x)$ and
$\textup{(\ref{26-4-4})}_{1-4}$ hold (with some constant factor in the right-hand expression).
\end{claim}

Consider $B(\bar{x},\bar{\gamma}_0)$, $\bar{\gamma}_0=\gamma_0(\bar{x})$, and scale again $x\mapsto \bar{\gamma}_0^{-1}(x-\bar{x})$,
$\tau \mapsto \bar{\gamma}_0^{-4}\tau$ and respectively
$w\mapsto w_1=\bar{\gamma}_0^{-4}(w-2\bar{\jmath}\beta))$,
$h \mapsto h _1=h \bar{\gamma}_0^{-3}$. Further, since after rescaling
$|\Delta w_1|=O(\eta \bar{\gamma}_0^{-2})$ we set $\eta\mapsto \eta_1= \eta  \bar{\gamma}_0^{-2}$.

Due to cut-off in the end of \ref{26-4-3-*} $h_1\lesssim 1$ and
$\eta_1\lesssim 1$. If $\bar{\gamma}_0\asymp h ^{\frac{1}{3}}$ then
$h _1\asymp 1$ and we are done. If $\bar{\gamma}_0\asymp \eta^{\frac{1}{2}}$ then $\eta_1\asymp 1$ and we proceed to \emph{\nameref{sect-26-4-1-1-3}}.

\subsubsection{Step 2.}
\label{sect-26-4-1-1-2}

So, let $\eta_1\ll 1$. Let us introduce a scaling function in $B(0,1)$ obtained after the previous rescaling:\pagebreak\enlargethispage{\baselineskip}
\begin{equation}
\gamma_1(x)= \epsilon \Bigl(\min_j | w_1+2(\bar{\jmath}-j)\beta \bar{\gamma}_0^{-4}|^{2} +|\nabla w_1| ^{3} +|\nabla^2 w_1|^{6} \Bigr)^{\frac{1}{6}}.
\label{26-4-7}
\end{equation}
Then
\begin{phantomequation}\label{26-4-8}\end{phantomequation}
\begin{align}
&\min_j | w_1+2(\bar{\jmath}-j)\beta \bar{\gamma}_0^{-4}|\le C_0\gamma_1^3,
&&|\nabla w_1|\le C_0\gamma_1^2, \tag*{$\textup{(\ref*{26-4-8})}_{1-2}$}\label{26-4-8-1}\\
&|\nabla ^2 w_1|\le C_0\gamma_1
\tag*{$\textup{(\ref*{26-4-8})}_{3}$}\label{26-4-8-3}.
\end{align}

\begin{remark-foot}\footnotetext{\label{foot-26-24} Cf. Remark~\ref{rem-26-4-2}.}\label{rem-26-4-4}
Since now we do not have the third derivative in (\ref{26-4-7}), we do not need in $w'_1$ in the definition of $\gamma_1$, only in the proof of Proposition~\ref{prop-26-4-5} below.
\end{remark-foot}

\begin{proposition}\label{prop-26-4-5}
\begin{enumerate}[label=(\roman*), wide, labelindent=0pt]
\item\label{prop-26-4-5-i}
$\gamma_1(x)$ is a scaling function:
$|x-y|\le 2\gamma_1(x)\implies \gamma_1(y)\asymp \gamma_1(x)$.

\item\label{prop-26-4-5-ii}
If $\eta_1\le \epsilon_0$ then
\begin{equation}
|\nabla^2 w_1(x)- \nabla^2 w'_1(x)|\le \epsilon_2 \gamma_1,\qquad w'_1=(w'-\bar{\jmath}\beta)\gamma_0^{-4}.
\label{26-4-9}
\end{equation}
\end{enumerate}
\end{proposition}

\begin{proof}
Proof is similar but simpler than one of Proposition~\ref{prop-26-4-3}; it is based on the rescaled version of (\ref{26-4-1})--(\ref{26-4-2}):
\begin{gather}
\frac{1}{4\pi}\Delta w_1 = \eta_1 P_\beta (w_1 \bar{\gamma}_0^4 + 2\bar{\jmath}\beta ),\qquad \eta_1= \eta \bar{\gamma}_0^{-2}\le 1
\label{26-4-10}\\
\intertext{and therefore in $B(0,1)$}
w_1 = -\frac{1}{4\pi} \int |x-z|^{-1}\eta_1
P'_\beta (w_1 (z)\bar{\gamma}_0^4 + 2\bar{\jmath}\beta )\phi(z)\, dz + w'_1.
\label{26-4-11}
\end{gather}
\ \end{proof}

Now let us reintroduce the scaling function
\begin{equation}
\gamma_1(x)= \epsilon
\Bigl(\min_j | w_1+2(\bar{\jmath}-j)\beta\bar{\gamma}_0^{-4} |^{2} +
|\nabla w_1| ^{3} +|\nabla^2 w_1|^{6} \Bigr)^{\frac{1}{6}}+
C_0h _1^{\frac{2}{5}}.
\tag*{$\textup{(\ref*{26-4-7})}^*$}\label{26-4-7-*}
\end{equation}
Then
\begin{claim}\label{26-4-12}
$|x-y|\le 2\gamma_0 (x)\implies \gamma_0(y)\asymp \gamma_0(x)$ and
$\textup{(\ref{26-4-8})}_{1-3}$ hold (with some constant factor in the right-hand expression).
\end{claim}

Let us consider $\bar{x}\in B(0,1)$ (it is a new point), $B(\bar{x},\bar{\gamma}_1)$, $\bar{\gamma}_1=\gamma_1(\bar{x})$, and scale again
$x\mapsto \bar{\gamma}_1^{-1}(x-\bar{x})$,
$\tau \mapsto \bar{\gamma}_1^{-3}\tau$ and respectively
$w_1\mapsto w_2=\bar{\gamma}_1^{-3}w_1$,
$h _1\mapsto h _2=h _1\bar{\gamma}_1^{-\frac{5}{2}}$.

If $\bar{\gamma}_1\asymp h _1^{\frac{2}{5}}$ then $h _2\asymp 1$ and we are done. If $|\nabla w_2|\asymp 1$ we are done as well.

\subsubsection{Step 3.}
\label{sect-26-4-1-1-3}
So, consider the remaining case $|\nabla^2 w_2|\asymp 1$. Then we introduce the scaling function (now, we have no doubt that this is a scaling function):
\begin{equation}
\gamma_2(x)= \epsilon \Bigl(\min_j | w_2+2(\bar{j}-j)\beta \bar{\gamma}_0^{-4}\bar{\gamma}_1^{-3} |+|\nabla w_2| ^{2}\Bigr)^{\frac{1}{2}}+ C_0h _2^{\frac{1}{2}}.
\label{26-4-13}
\end{equation}
Let us consider $\bar{x}\in B(0,1)$ (it is a new point), $B(\bar{x},\bar{\gamma}_2)$, $\bar{\gamma}_2=\gamma_2(\bar{x})$, and scale again
$x\mapsto \bar{\gamma}_2^{-1}(x-\bar{x})$,
$\tau \mapsto \bar{\gamma}_2^{-2}\tau$ and respectively
$w_2\mapsto w_3=\bar{\gamma}_2^{-2}w_2$,
$h _2\mapsto h _3=h _2\bar{\gamma}_2^{-2}$.

If $\bar{\gamma}_2\asymp h _1^{\frac{1}{2}}$ then $h _3\asymp 1$ and we are done. If $|\nabla w_3|\asymp 1$ we are done as well.
\enlargethispage{\baselineskip}

\subsubsection{Step 4.}
\label{sect-26-4-1-1-4}
Finally, introduce
\begin{equation}
\gamma_3(x)= \epsilon
\min_j | w_3+2(\bar{j}-j)\beta \bar{\gamma}_0^{-4}\bar{\gamma}_1^{-3}\bar{\gamma}_2^{-2} | +
Ch _3^{\frac{2}{3}}.
\label{26-4-14}
\end{equation}

\section{Zone \texorpdfstring{$\cX_2$}{X\texttwoinferior}: Semiclassical $\N$-term}
\label{sect-26-4-2}

Now we apply scaling the arguments using scaling functions $\gamma_{1-3}$ constructed above.

We revert our steps. While we call $\gamma_{1-3}$
\emph{relative scaling functions\/}\index{scaling function!relative} let us introduce \emph{absolute scaling functions\/}\index{scaling function!absolute} $\alpha_0(x)=\gamma_0(x)$, $\alpha_1(x)=\gamma_0(x)\gamma_1(x)$, $\alpha_2(x)=\gamma_0(x)\gamma_1(x)\gamma_2(x)$, and $\alpha_3(x)=\gamma_0(x)\gamma_1(x)\gamma_2(x)\gamma_3(x)$\,\footnote{\label{foot-26-25} So far we ignore the very first scaling
$x\mapsto (x-\bar{x})\ell^{-1}$. Therefore really absolute scaling functions would be $\alpha_j\ell$.}.

We need first

\begin{proposition}\label{prop-26-4-6}
Consider $B(0,1)$ and assume that in it
\begin{gather}
|\nabla w|\asymp \theta,\label{26-4-15}\\
\shortintertext{and}
|\nabla^2 w|\le c\theta.
\label{26-4-16}\\
\shortintertext{Let}
\gamma(x)\coloneqq \epsilon \min_j |w-2\mu h j|\theta^{-1}+\hbar ^{\frac{2}{3}}
\label{26-4-17}\\
\shortintertext{and}
\varepsilon \ge \hbar^{\frac{2}{3}-\delta}, \qquad
\hbar\coloneqq h \theta^{-\frac{1}{2}}.
\label{26-4-18}
\end{gather}
Let $\varphi \in \sC^\infty_0([-\epsilon_0,\epsilon_0])$. Then for
$\alpha\le \bar{\gamma}\coloneqq \gamma(\bar{x})$,
\begin{multline}
|\int \phi_\alpha (x)
\Bigl( e_\varphi (x,x,\tau) -P'_{\mu h ,\varphi} \bigl(w(x)+\tau\bigr)\Bigr)\, dx|\le\\
C\mu h ^{-1}\alpha^3+ C\mu h ^{-1}\alpha^3 \bar{\gamma}^{-1}
(\hbar \bar{\gamma}^{-\frac{1}{2}}\alpha^{-1})^s
\label{26-4-19}
\end{multline}
with
\begin{gather}
e_\varphi (x,y,\tau)\coloneqq
\varphi \bigl(h ^2 D^2_{x_3} (\mu h)^{-1}\bigr) e(x,y,\tau)
\label{26-4-20}\\
\intertext{and Weyl expression}
P'_{\beta,\varphi} (w+\tau)=
\const \sum_j \beta (w+\tau-2j\mu h ) ^{\frac{1}{2}}
\varphi (w+\tau-2j\mu h )
\label{26-4-21}
\end{gather}
with the standard constant where for each $x$ only one term is present in this sum. Here we take large $s>0$ as $\hbar \le \bar{\gamma}\alpha$ and $s=0$ otherwise.
\end{proposition}

\begin{proof}
The proof is standard and based on the standard reduction to the canonical form, standard estimates for $U(x,y,t)$ a Schwartz kernel of propagator
$e^{i\hbar^{-1}\theta ^{-1}t H}$:
\begin{equation}
|F_{t\to \hbar^{-1}\tau}
\int \bar{\chi}_T(t) \phi_\alpha(x) U_\varphi (x,x,t) \,dx| \le
C\mu h ^{-1}\alpha^3
\label{26-4-22}
\end{equation}
for $T\asymp 1$ and
\begin{multline}
|F_{t\to \hbar ^{-1}\tau}
\int \phi_\alpha (x) \bigl(\bar{\chi}_T(t) - \bar{\chi}_{\bar{T}}(t)\bigr) U_\varphi (x,x,t) \,dx|\le \\
C\mu h ^{-1} \alpha^3 (\hbar \bar{\gamma}^{-\frac{1}{2}}\alpha^{-1})^s
\label{26-4-23}
\end{multline}
with $\bar{T}=\epsilon\gamma^{\frac{1}{2}}$ where $U_\varphi$ is defined similarly to (\ref{26-4-20}).

Here obviously we can skip in (\ref{26-4-19}) all perturbation terms in the argument and in $\phi_{\alpha}$ transformed. \end{proof}

Then plugging into (\ref{26-4-19}) $\alpha=\gamma$ ($=\bar{\gamma}$), we have factor $(h \gamma ^{-\frac{3}{2}})^s$ in the second term.

There are two cases: $\theta \le \mu h $ and $\theta \ge \mu h $.

In the former case $\theta \le \mu h $, taking the sum over $\gamma$-partition of $1$-element we estimate the same expressions with $\phi_1$ instead of $\phi_{\gamma}$ by their right-hand expressions integrated over
$\gamma ^{-3}\,d\gamma $ which returns $C\mu h ^{-1}$.

On the other hand, in the latter case $\theta \ge \mu h $, let
$\lambda=\mu h\theta^{-1}$. Taking the sum over $\gamma_3$-partition of $\lambda$-element $\phi_\lambda$ by the right-hand expressions which returns $C\mu h _2^{-1}\lambda^2$. In this case summation over $\lambda$-partition return $C\theta h ^{-2}$.

In both cases we arrive to the following estimate:
\begin{equation}
|\int \phi (x)
\Bigl( e_\varphi (x,x,\tau) -P'_{\mu h ,\varphi} \bigl(w(x)+\tau\bigr)\Bigr)\,dx| \le
C\theta h ^{-2}+C\mu h ^{-1}.
\label{26-4-24}
\end{equation}

Applying this estimate after $\alpha_2$-scaling we conclude that the left-hand expression with $\phi=\phi_{\alpha_2}$ (in the non-scaled settings) does not exceed $C\theta h ^{-2}\alpha_2^2 + C\mu h ^{-1}\alpha_2^2$. Here the first term is $O(h ^{-2}\alpha_2^3)$ and the summation over $1$-element returns
$O(h ^{-2})$.

Consider the second term $C\mu h ^{-1}\alpha_2^2=
C\mu h ^{-1}\gamma_0^2\gamma_1^2\gamma_2^2$. Then summation over $\alpha_2$-partition of $\alpha_1$-element returns
\begin{equation*}
C\mu h ^{-1}\gamma_0^2\gamma_1^2
\int \gamma_2^2\times \gamma_2^{-3}\,dx \asymp
C\mu h ^{-1}\gamma_0^2\gamma_1^2(1+|\log \hat{\gamma}_2|)
\end{equation*}
where $\hat{\gamma}_k$ is a minimal value of $\gamma_k$ over $\gamma_{k-1}$-element. However in fact there will be no logarithmic factor because in virtue of equation (\ref{26-4-1}) there is a positive eigenvalue of $\Hess w_2$ of the maximal size (cf. Section \ref{book_new-sect-5-2-1}). Therefore, in fact, we have $C\mu h ^{-1}\gamma_0^2\gamma_1^2$.

Now summation over $\alpha_1$-partition of $\alpha_0$-element returns
\begin{equation*}
C\mu h ^{-1}\gamma_0^2 \int \gamma_1^2\times \gamma_1^{-3}\,dx \asymp
C\mu h ^{-1}\gamma_0^2(1+|\log \hat{\gamma}_1|).
\end{equation*}

Finally, summation over $\alpha_0$-partition of $1$-element returns
\begin{equation*}
C\mu h ^{-1}\int \gamma_0^2 (1+|\log \hat{\gamma}_1|)
\times \gamma_0^{-3}\, dx \asymp
C\mu h ^{-1}\hat{\gamma}_0^{-1} (1+|\log \check{\gamma}_1|),
\end{equation*}
where $\check{\gamma}_k$ is an absolute minimum of $\gamma_k$. However $\gamma_0^2\ge \eta$ and $\gamma_1\ge \eta$ and therefore expression above does not exceed $C\mu h ^{-1} \eta^{-\frac{1}{2}}(1+|\log \eta|)$.

\begin{remark}\label{rem-26-4-7}
Recall that we estimated only the cut-off expression. To calculate the full expression we need to calculate also the contribution of the zone
$\{\xi_3^3\ge \mu h\}$. However this is easy.

Really, instead of $\varphi (h ^2D_3^2 /(\mu h ))$ consider
$\varphi' (h ^2D_3^2 /\theta)$ with $\varphi'\in \sC^\infty ([1,4])$ and $\epsilon \mu h \le \theta \le 1$. Without any scaling one can prove easily that such modified expression \textup{(\ref{26-4-24})} does not exceed
$C\theta h ^{-2}$. We leave easy details to the reader.
\end{remark}

Therefore plugging $\theta =2^n\mu h $ and taking a sum over
$n=0,\ldots, \lfloor |\log_2 \mu h |\rfloor$ we get the required expressions. Also note that in such expressions we need to consider perturbed argument
$w+\mu^{-2}\omega_1+ j \mu^{-1}h \omega_2$ (all other terms which are $O(\mu^{-4}+\mu^{-\frac{3}{2}}h )$ could be skipped and also a perturbed function transformed).

\begin{remark}\label{rem-26-4-8}
\begin{enumerate}[label=(\roman*), wide, labelindent=0pt]
\item\label{rem-26-4-8-i}
However we need to get rid off these perturbations for
$\theta \le \mu h ^{1-\delta}$ only. Indeed, for
$\theta \ge \mu h ^{1-\delta}$ we need canonical form \emph{only\/} to study propagation and calculations could be performed without it. But then getting rid off the perturbation is trivial provided this perturbation does not exceed
$C\mu h ^{1+\delta}$ which is the case if
\begin{equation}
\mu \ge h ^{-\frac{1}{3}-\delta}.
\label{26-4-25}
\end{equation}

\item\label{rem-26-4-8-ii}
Note that in the smooth approximation contributions of $\cX_1$ is always less than $CZ^{\frac{2}{3}-\delta_1}$,
$C\max (Z^{\frac{5}{3}}, Z^{\frac{3}{5}}B^{\frac{4}{5}})Z^{-\delta_1}$ or
$C\max (Z^{\frac{5}{3}}, Z^{\frac{3}{5}}B^{\frac{4}{5}})Z^{-\delta_1}+ CB^{\frac{1}{2}}Z^{\frac{7}{6}-\delta_1}$ respectively with an exception of the first two and only in the case of the threshold value
$\ge Z^{-\frac{1}{3}-\delta_2}$. However in this case $\eta \ge Z^{-\delta_3}$ and the errors of the smooth approximation approach in fact are less than
$CZ^{\frac{2}{3}-\delta_4}$, $C Z^{\frac{5}{3}-\delta_4}$ as well. Therefore there are in fact no exception.

\item\label{rem-26-4-8-iii}
It is important to have $\varepsilon \le \mu h$ and with $\varepsilon=\hbar^{\frac{2}{3}-\delta}$, $\hbar=h\theta^{-\frac{1}{2}}$ it means $\mu \ge h ^{-\frac{1}{3}-\delta}\theta^{\frac{1}{3}-\frac{1}{2}\delta}$ which is due to (\ref{26-4-25}).
\end{enumerate}
\end{remark}

Therefore we conclude that in the completely non-scaled settings with
$\phi =\phi_\ell (x)$
\begin{multline}
|\int \phi (x)
\Bigl( e (x,x,\tau) -P'_B \bigl(W(x)+\tau\bigr)\Bigr)|\le \\
C\zeta^2\ell^2 + CB\ell \zeta^{-\frac{1}{2}} (1+|\log \ell^2\zeta|)
\label{26-4-26}
\end{multline}
where the first term is $Ch ^{-2}$ and the second term is
$C\mu h ^{-1}\eta^{-\frac{1}{2}}(1+|\log \eta|)$; recall that
$h ^{-1}\asymp \ell\zeta $, $\mu \asymp B\ell\zeta^{-1}$ and
$\eta \asymp \ell^2\zeta$. In comparison with the non-degenerate case $|\nabla W^\TF_B|\asymp \zeta^2\ell^{-1}$ we acquired the last term.

Assume first that condition (\ref{26-2-28}) is fulfilled. Then
\begin{enumerate}[label=(\roman*), wide, labelindent=0pt]
\item\label{sect-26-4-2-i}
If $B\le Z^{\frac{4}{3}}$, $\ell\le Z^{-\frac{1}{3}}$ we have
$\zeta=Z^{\frac{1}{2}}\ell^{-\frac{1}{2}}$ and the right-hand expression of (\ref{26-4-26}) returns $CZ\ell + CB\ell^{\frac{5}{4}}Z^{-\frac{1}{4}}$ and the summation with respect to $\ell$ results in its value as $\ell=Z^{-\frac{1}{3}}$ i.e. $CZ^{\frac{2}{3}}+CBZ^{-\frac{2}{3}}$ with the dominating first term.

\item\label{sect-26-4-2-ii}
If $B\le Z^{\frac{4}{3}}$, $\ell\ge Z^{-\frac{1}{3}}$ we have $\zeta=\ell^{-2}$ and the right-hand expression of (\ref{26-4-26}) returns
$C\ell^{-2}+CB\ell^2$. We need to sum as long as $\mu h\le 1$ i.e.
$Z^{-\frac{1}{3}}\le \ell \le B^{-\frac{1}{4}}$ and the summation returns
$CZ^{\frac{2}{3}} + CB^{\frac{1}{2}}$ with the dominating first term.

\item\label{sect-26-4-2-iii}
If $ Z^{\frac{4}{3}}\le B \le Z^2$,
$\ell\le B^{-1}Z$ we have $\zeta=Z^{\frac{1}{2}}\ell^{-\frac{1}{2}}$ and the right-hand expression of (\ref{26-4-26}) returns
$CZ\ell + CBZ^{-\frac{1}{4}}\ell^{\frac{5}{4}}$. Then summation results in
$CZ^2B^{-1} + CZ^2B^{-\frac{1}{4}}Z\lesssim Z^{\frac{2}{3}}$.
\end{enumerate}

If assumption (\ref{26-2-28}) is not fulfilled we can estimate in the first term of the right-hand expression of (\ref{26-4-26}) parameter $\zeta$ from above by
$\min (Z^{\frac{1}{2}}\ell^{-\frac{1}{2}},\ell^{-2})$ and in the second term from below by
$\zeta_m=\min (Z_m^{\frac{1}{2}}\ell_m^{-\frac{1}{2}},\ell_m^{-2})$ if
$\ell = \ell_m\coloneqq |x-\y_m|$ and repeat all above arguments.

Therefore we arrive to the Statement~\ref{prop-26-4-9-i}   of Proposition~\ref{prop-26-4-9} below. Furthermore, note that for  $B\le Z$ the zone $\cX_2$ is contained in the zone
$\{x\colon \ell(x)\ge B^{-\frac{3}{10}}\ge Z^{-\frac{3}{10}}\}$ (really,
$\mu \ge h^{-\frac{1}{3}}$ in $\cX_2$) and we arrive to the Statement~\ref{prop-26-4-9-ii}  below.\enlargethispage{\baselineskip}

\begin{proposition}\label{prop-26-4-9}
\begin{enumerate}[label=(\roman*), wide, labelindent=0pt]
\item\label{prop-26-4-9-i}
For $B \le Z^2$ the contribution of  zone $\cX_2$ to the expression
\begin{equation}
\int \bigl( e(x,x,\nu)-P'_B(W(x)+\nu)\bigr)\,dx
\label{26-4-27}
\end{equation}
does not exceed $CZ^{\frac{2}{3}}$.
%\pagebreak
\item\label{prop-26-4-9-ii}
For $B \le Z$ the contribution of zone $\cX_2$ to the expression \textup{(\ref{26-4-27})} does not exceed $CZ^{\frac{2}{3}-\delta}$.
\end{enumerate}
\end{proposition}

\section{Zone \texorpdfstring{$\cX_2$}{X\texttwoinferior}: Semiclassical $\D$-term}
\label{sect-26-4-3}

Further, we need to estimate the \emph{semiclassical $\D$-term}
\begin{equation}
\D \Bigl(\phi_\alpha [ e(x,x,\nu) - P'_B(W(x)+\nu)],
\phi_\alpha [ e(x,x,0) - P'_B(W(x)]\Bigr)
\label{26-4-28}
\end{equation}
where $\phi_\alpha (x)$ is an $\alpha$-admissible function. Again we revert our steps.

Consider $B(x,\bar{\alpha}_3)$ and apply Fefferman--de Llave decomposition (\ref{26-3-8}). Then in the framework of Proposition~\ref{prop-26-4-6} contribution of pairs
$B(x,\alpha)$ and $B(y,\alpha)$ with $3\alpha \le |x-y|\le 4 \alpha$ does not exceed the right-hand expression of (\ref{26-4-19}) squared and multiplied by $\alpha^{-4}$, where $C\alpha^{-3}$ estimates the number of the pairs and $\alpha^{-1}$ is the inverse distance. At this moment we discuss a cut-off version of (\ref{26-4-28}) i.e. with $e_\varphi(.,.,.)$ and $P'_{\mu h ,\varphi}(.)$.
So, we have
\begin{equation*}
C\mu^2h ^{-2}
\bigl(1+ \gamma_3^{-1}(h _2\gamma_3^{-\frac{1}{2}}\alpha^{-1})^s\bigr)^2 \alpha^4.
\end{equation*}
Then integrating this expression with respect to $\alpha^{-1}\,d\alpha$ with
$\alpha\le \gamma_3$ we arrive to
$C\mu^2h ^{-2} \bigl(\gamma_3^4+ h _2^2 \gamma_3\bigr)$.

Therefore we conclude that
\begin{claim}\label{26-4-29}
A cut-off version of expression (\ref{26-4-28}) with $\alpha=\alpha_3$ does not exceed
$C\mu^2h ^{-2} \bigl(\gamma_3^4+ h _2^2 \gamma_3\bigr) \alpha_2^3$
\end{claim}
The first term here $C\mu^2h ^{-2} \gamma_3^4\alpha_2^3$ does not exceed $C\mu^2h ^{-2}\alpha_3^3$ (recall that $\alpha_j=\alpha_{j-1}\gamma_j$) and the summation over $\alpha_3$-partition of $1$-element returns $C\mu^2h ^{-2}$.

Consider the second term $C\mu^2h ^{-2} h _2^2 \gamma_3 \alpha_2^3$; its summation with respect to $\alpha_3$-partition of $\alpha_2$-element returns
$C\mu^2h ^{-2}\alpha_2^3h _2^2\int \gamma_3^{-2}\,d\gamma_3\lesssim
C\mu^2h ^{-2}\alpha_2^3$ (really, recall that according to (\ref{26-4-14}) $\gamma_3\ge h _3^{\frac{2}{3}}$) and then the summation over $\alpha_2$-partition of $1$-element returns $C\mu^2h ^{-2}$.

Consider $B(x,\bar{\alpha}_2)$ and apply Fefferman--de Llave decomposition (\ref{26-3-8}). There are two kinds of pairs:
\begin{enumerate}[label=(\alph*), wide, labelindent=0pt]
\item\label{sect-26-4-3-a}
those with $|x-y|\ge \epsilon (\alpha_3(x)+\alpha_3(y))$ for all $(x,y)$ and
\item\label{sect-26-4-3-b}
those with $|x-y|\le \min(\alpha_3(x),\alpha_3(y))$ for all $(x,y)$.
\end{enumerate}

The total contribution of the pairs of the second type (i.e. summation is taken over \emph{all\/} pairs of $\alpha_3$-elements in $B(0,1)$) as we already know is $O(\mu^2h ^{-2})$. Meanwhile according to the analysis in the previous Subsection~\ref{sect-26-4-2} a contribution of one pair of kind \ref{sect-26-4-3-a} does not exceed
\begin{multline*}
C \underbrace{\bigl(h ^{-2} +
\mu h ^{-1}\gamma_0^{-1}\gamma_1^{-1}\gamma_2^{-1}
\gamma_3^{-1-\delta}\bar{\gamma}_3^\delta\bigr)\alpha_3^3}_{\text{at } x} \times\\
\underbrace{\bigl(h ^{-2} +
\mu h ^{-1}\gamma_0^{-1}\gamma_1^{-1}\gamma_2^{-1}\gamma_3^{-1-\delta}
\bar{\gamma}_3^\delta\bigr)\alpha_3^3}_{\text{at } y} \times |x-y|^{-1}
\end{multline*}
where each of two first factors is just an estimate of the integral (\ref{26-4-24}) calculated over corresponding domain. If we take the first term in the first factor and sum over $\alpha_3$-partition of $1$-element we get only the second factor multiplied by $\mu h ^{-1}$ and then summation was done in the previous subsection. Similarly we can deal with the first term in the second factor. On the other hand, if we take only second factors and sum over pairs of $\alpha_3$-subelements of the same $\alpha_2$-element we get
\begin{equation*}
C\mu^2h ^{-2} \gamma_0^{-2}\gamma_1^{-2}\gamma_2^{-2} \alpha_2^5\asymp
C\mu^2h ^{-2}\alpha_2^3.
\end{equation*}
Then summation with respect to $\alpha_2$-partition of $1$-element returns $C\mu^2h ^{-2}$.

Consider now $B(\bar{x},\bar{\alpha}_1)$ and apply here Fefferman-de Llave decomposition (\ref{26-3-8}). There are two kinds of pairs:
\begin{enumerate}[label=(\alph*), wide, labelindent=0pt]
\item\label{sect-26-4-3-a'}
those with $|x-y|\ge \epsilon (\alpha_2(x)+\alpha_2(y))$ for all $(x,y)$ and
\item\label{sect-26-4-3-b'}
those with $|x-y|\le \min(\alpha_2(x),\alpha_2(y))$ for all $(x,y)$.
\end{enumerate}
According to the above analysis the total contribution of the pairs of the second type (i.e. the summation is taken over \emph{all\/} pairs of $\alpha_2$-elements in $B(0,1)$) as we already know is $O(\mu^2h ^{-2})$. Meanwhile according to the analysis in the previous Subsection~\ref{sect-26-4-2} a contribution of one pair of kind \ref{sect-26-4-3-a'} does not exceed
\begin{equation*}
C \underbrace{\bigl(h ^{-2} +
\mu h ^{-1}\gamma_0^{-1}\gamma_1^{-1}\gamma_2^{-1}\bigr)\alpha_3^3}
_{\text{at } x} \times
\underbrace{\bigl(h ^{-2} +
\mu h ^{-1}\gamma_0^{-1}\gamma_1^{-1}\gamma_2^{-1}\bigr)\alpha_3^3}
_{\text{at } y} \times |x-y|^{-1}
\end{equation*}
and here again we can ``forget'' about the first terms in each factor. Then the summation with respect to pairs of $\alpha_2$-subelements of the same $\alpha_1$-element results in
$C\mu^2h ^{-2} \gamma_0^{-2}\gamma_1^{-2}\alpha_1^5\asymp
C\mu^2h ^{-2} \alpha_1^3$ where we avoid logarithmic factor in virtue of the same positive eigenvalue of $\Hess w$. Summation with respect to $\alpha_1$-admissible partition of $1$-element returns $C\mu^2h ^{-2}$.

Consider now $B(\bar{x},\bar{\alpha}_0)$ and apply here Fefferman--de Llave decomposition. Again there are two kinds of pairs and the total contributions of the pairs of the second kind we already calculated and contribution of the pairs of $\alpha_1$-subelements of the same $1$-element does not exceed
\begin{equation*}
C\mu ^2h ^{-2}\gamma_0^{-2}(1+|\log \hat{\gamma}_1|)^2\alpha_0^5\lesssim
C\mu ^2h ^{-2}(1+|\log \eta|)^2\alpha_0^3
\end{equation*}
and the summation with respect to $\alpha_0$-partition of $1$-element returns\newline
$C\mu ^2h ^{-2}(1+|\log \eta|)^2$.

Finally, consider $B(\bar{x},1)$ and apply here Fefferman-de Llave decomposition. Again there are two kinds of pairs and the total contributions of the pairs of kind (b) we already estimated while the total contribution of the pairs of kind (a) does not exceed
$Ch ^{-4}+ C\mu ^2h ^{-2}\eta^{-1}(1+|\log \eta|)^2$ where we recalled the forgotten terms.

Again, this is estimate for cut-off expression. Going to uncut expression we repeat the same trick as before but as we deal with $\D$-term we need to consider ``mixed'' pairs when one ``factor'' comes with $\theta$ and another with $\theta'$ but then contribution of such pair does not exceed
$C(\nu h ^{-4})^{\frac{1}{2}} (\nu' h ^{-4})^{\frac{1}{2}} $. Easy details are left to the reader.

Therefore returning to the original scale we conclude that the contribution of $\ell$-layer to (\ref{26-4-28}) does not exceed
\begin{equation}
C\zeta^4\ell^3 + CB^2\ell \zeta ^{-1}(1+|\log \ell^2 \zeta|)^2
\label{26-4-30}
\end{equation}
which is exactly the right-hand expression of (\ref{26-4-26}) squared and multiplied by $\ell^{-1}$ due to scaling.

\begin{remark}\label{rem-26-4-10}
In comparison with the non-degenerate case $|\nabla W^\TF_B|\asymp \zeta^2\ell^{-1}$ we acquired the last term.
\end{remark}

Assume first that condition (\ref{26-2-28}) is fulfilled. Then

\begin{enumerate}[label=(\roman*), wide, labelindent=0pt]
\item\label{sect-26-4-3-i}
For $B\le Z^{\frac{4}{3}}$, $\ell\le Z^{-\frac{1}{3}}$ we have
$\zeta=Z^{\frac{1}{2}}\ell^{-\frac{1}{2}}$ and expression (\ref{26-3-41}) returns $CZ^2\ell + CB^2\ell^{\frac{3}{2}}Z^{-\frac{1}{2}}$ and the summation with respect to $\ell$ results in its value as $\ell=Z^{-\frac{1}{3}}$ i.e. $CZ^{\frac{5}{3}}+CB^2Z^{-1}$ with the dominating first term.

\item\label{sect-26-4-3-ii}
For $B\le Z^{\frac{4}{3}}$, $\ell\ge Z^{-\frac{1}{3}}$ we have $\zeta=\ell^{-2}$ and expression (\ref{26-3-41}) returns
$C\ell^{-5}+CB^2\ell^3$. We need to sum as long as $\mu h\le 1$ i.e.
$Z^{-\frac{1}{3}}\le \ell \le B^{-\frac{1}{4}}$ and the summation returns
$CZ^{\frac{5}{3}} + CB^{\frac{5}{4}}$ with the dominating first term.

\item\label{sect-26-4-3-iii}
For $ Z^{\frac{4}{3}}\le B \le Z^2$,
$\ell\le B^{-1}Z$ we have $\zeta=Z^{\frac{1}{2}}\ell^{-\frac{1}{2}}$ and expression (\ref{26-3-41}) returns
$CZ^2\ell + CB^2Z^{-\frac{1}{2}}\ell^{\frac{3}{2}}$. Then the summation results in $CZ^3B^{-1}+ CB^{\frac{1}{2}} Z \lesssim Z^{\frac{3}{5}}B^{\frac{4}{5}}$.
\end{enumerate}

Sure, we need to consider also mixed pairs of the layers and their contributions are
\begin{equation*}
C\bigl (\zeta^2\ell^2 +C B\ell \zeta^{-\frac{1}{2}}(1+|\log \ell^2\zeta|)\bigr)
\times
\bigl (\zeta^{\prime\, 2}\ell^{\prime\,2} +
C B\ell' \zeta^{\prime\,-\frac{1}{2}}(1+|\log \ell^{\prime\,2}\zeta'|)\bigr) \times(\ell +\ell')^{-1}
\end{equation*}
and the summation with respect to $\ell$ and $\ell'$ returns the same expression as above.

If assumption (\ref{26-2-28}) is not fulfilled we use the same trick as in the previous Subsection~\ref{sect-26-4-2}. Therefore we arrive to the Statement~\ref{prop-26-4-11-i} of Proposition~\ref{prop-26-4-11} below. Applying the same arguments as in the proof of Proposition~\ref{prop-26-4-9} we arrive to the Statement~\ref{prop-26-4-11-ii}:

\begin{proposition}\label{prop-26-4-11}
\begin{enumerate}[label=(\roman*), wide, labelindent=0pt]
\item\label{prop-26-4-11-i}
For $B \le Z^2$ the contribution of the zone $\cX_2\times \cX_2$ to expression \textup{(\ref{26-4-28})} does not exceed
$C\max(Z^{\frac{5}{3}}, Z^{\frac{3}{5}}B^{\frac{4}{5}})$.
\enlargethispage{2\baselineskip}

\item\label{prop-26-4-11-ii}
For $B \le Z$ contribution of the zone $\cX_2\times \cX_2$ to expression \textup{(\ref{26-4-28})} does not exceed $CZ^{\frac{5}{3}-\delta}$.
\end{enumerate}
\end{proposition}

\section{Semiclassical $\T$-term}
\label{sect-26-4-4}

\subsection{Semiclassical $\T$-term: zone \texorpdfstring{$\cX_1$}{X\textoneinferior} extended}
\label{sect-26-4-4-1}
First let us cover zone $\cX_1$ extended.

\subsubsection{What is \texorpdfstring{$\cX_1$}{X\textoneinferior} extended?}
\label{sect-26-4-4-1-1}
To define this zone $\cX_1'\coloneqq\{x\colon \ell (x)\le r\}$, where we define
$W $ using $P$ rather than $P_B$ let first us analyze the precise extension in the framework of $\N$- and $\D$-terms. For $\N$-term we have approximation
error and corresponding $\D$-term not exceeding respectively
$C(\mu h )^2 \eta^{-\frac{1}{2}}h ^{-3} = CB^2\zeta^{-\frac{3}{2}}\ell^2$
and this expression squared and multiplied by $\ell^{-1}$ i.e. $CB^4\zeta^{-3}\ell^3$. Finally, both expressions are summed to their values as $\ell=r$. Recall that either $\zeta=Z^{\frac{1}{2}}\ell^{-\frac{1}{2}}$ or $\zeta=\ell^{-2}$.

\begin{enumerate}[label=(\roman*), wide, labelindent=0pt]
\item\label{sect-26-4-4-1-1-i}
Consider first  $B\le Z^{\frac{4}{3}}$. Then we want these errors not to exceed respectively
$CZ^{\frac{2}{3}}$ and $CZ^{\frac{5}{3}}$. Obviously, if $r\ge Z^{-\frac{1}{3}}$ the first condition is more restrictive. In this case plugging  $\zeta=r^{-2}$  and we set $CB^2 r^5=Z^{\frac{2}{3}}$ i.e.
$r= B^{-\frac{2}{5}}Z^{\frac{2}{15}}$. Then $r\ge Z^{-\frac{1}{3}}$ as long as
$B\le Z^{\frac{7}{6}}$.

Then $\mu = Br^3= B^{-\frac{1}{5}}Z^{\frac{6}{15}}\ge Z^{\frac{1}{6}}$ and
$h =B^{-\frac{2}{5}}Z^{\frac{2}{15}}\ge Z^{-\frac{1}{3}}$; and one can see easily that $\mu \gtrsim h ^{-\frac{1}{2}}$ provided $\ell(x)\ge r$.

\item\label{sect-26-4-4-1-1-ii}
Consider next $Z^{\frac{7}{6}}\le B\le Z^{\frac{4}{3}}$.  Then for
$r\le Z^{-\frac{1}{3}}$ we have $\zeta=Z^{\frac{1}{2}}r^{-\frac{1}{2}}$. In this case the second requirement is more restrictive and we set
$B^4 Z^{-\frac{3}{2}}r^{\frac{9}{2}}= Z^{\frac{5}{3}}$, i.e.
$r= B^{-\frac{8}{9}}Z^{\frac{19}{27}}$. Then
$\mu = B^{-\frac{1}{3}}Z^{\frac{5}{9}}$ and
$h = Z^{-\frac{1}{2}}r^{-\frac{1}{2}}= B^{\frac{4}{9}}Z^{-\frac{23}{27}}$
and $\mu \ge h ^{-\frac{3}{7}}$; this is better than
$h ^{-\frac{1}{3}}$.

However, we can do better than this: observe that $\mu h \le \eta $ if and only if $r\ge B^2Z^{-3}$ i.e. $B\le Z^{\frac{50}{39}}$, which is greater than $Z^{\frac{7}{6}}$  but  less than $Z^{\frac{4}{3}}$, so we test
$\mu $ and $h $ in this case: $\mu = Z^{\frac{5}{39}}$ and
$h = Z^{\frac{11}{39}}$ and $\mu \ge h ^{-\frac{5}{11}}$ provided
$\ell(x)\ge r$.

If
$B\ge Z^{\frac{5}{11}}$ we will use another estimate for $\D$-term: namely it does not exceed $(\mu h )^3 h ^{-6}r^{-1}= B^3r^5$  and we want it not to exceed  $Z^{\frac{5}{3}}$, so $r=Z^{\frac{1}{3}}B^{-\frac{3}{5}}$  (which is still less than $Z^{-\frac{1}{3}}$) and
$\mu = B^{\frac{1}{10}}$ and $h = B^{\frac{3}{10}}Z^{-\frac{2}{3}}$ and we test it as $B=Z^{\frac{4}{3}}$ when $\mu =B^{\frac{1}{10}}$ and
$h =B^{-\frac{1}{5}}$, so exponent $-\frac{5}{11}$ fits again.

\item\label{sect-26-4-4-1-1-iii}
Finally, if $Z^{\frac{4}{3}}\le B\le Z^3$ then the error $\D$-term does not exceed $B^3r^5$ and we want it not to exceed $B^{\frac{4}{5}}Z^{\frac{3}{5}}$.
So, we pick up $r= B^{-\frac{11}{25}}Z^{\frac{3}{25}}$,
$\mu = B^{\frac{17}{50}}Z^{-\frac{8}{25}}$,
$h = B^{\frac{11}{50}}Z^{-\frac{14}{25}}$ and exponent $-\frac{5}{11}$ fits again.
\end{enumerate}

\subsubsection{When we can use the same method for $\T$-term?}
\label{sect-26-4-4-1-2}
As far as semiclassical $\T$-expression is concerned an approximation error of such approach in the localized and scaled settings is
$C(\mu h )^{\frac{5}{2}}h ^{-3}$\,\footnote{\label{foot-26-26} If instead of $P_B(W)$ we use $P(W)  +\frac{1}{2}P''(W)B^2$ rather than $P(W)$. This modification does not affect our previous arguments.}  which is $O(h ^{-1})$ only if $\mu \ge h ^{-\frac{1}{5}}$. One can extend it to
$\mu \ge h ^{-\frac{1}{5}-\delta}$ using the same trick as in Remark~\ref{rem-26-4-8} but we need to do better than this.

On the other hand, observe that in fact an approximation error does not
exceed\footref{foot-26-26}  $C(\mu h )^3 \eta^{-\frac{1}{2}}h ^{-3}\asymp
CB^3 \ell^2 \zeta^{-\frac{7}{2}} $ in the localized scaled settings. The simple proof is left to the reader. This is translated into
$CB^3 \ell^2 \zeta^{-\frac{3}{2}}$ into unscaled settings. Summation with respect to $\ell\le r$ returns its value as $\ell=r$.

So we get $CB^3 r^5$ as $B\le Z^{\frac{4}{3}}$ and $r\ge Z^{-\frac{1}{3}}$. Consider first  $B\le Z$. In this case we want $CB^3 r^5\le CZ^{\frac{5}{3}}$ and we pick up $r= B^{-\frac{3}{5}}Z^{\frac{1}{3}}$ which is greater than $Z^{-\frac{1}{3}}$ provided $B\le Z^{\frac{10}{9}}$. Then
$\mu = Br^3 = B^{-\frac{4}{5}}Z\ge Z^{\frac{1}{5}}$ and
$h = r=B^{-\frac{3}{5}}Z^{\frac{1}{3}}\ge Z^{-\frac{4}{15}}$ and
$\mu \ge h ^{-\frac{3}{4}}$.

If $Z\le B\le Z^{\frac{4}{3}}$ but still $r\ge Z^{-\frac{1}{3}}$  we want
$CB^3 r^5\le CZ^{\frac{4}{3}}B^{\frac{1}{3}}$\,\footnote{\label{foot-26-27} Because the semiclassical remainder estimate is not better than this. Actually, due to Remark~\ref{rem-26-3-6} we can do marginally better than this, but we leave this analysis to the reader.} and we pick up
$r= B^{-\frac{8}{15}}Z^{\frac{4}{15}}$ and we want it to be greater than
$Z^{-\frac{1}{3}}$ i.e. $B\le Z^{\frac{9}{8}}$. Then
$\mu = Br^3 = B^{-\frac{3}{5}}Z^{\frac{4}{5}}\ge Z^{\frac{1}{8}}$ and
$h =B^{-\frac{8}{15}}Z^{\frac{4}{15}}\ge B^{-\frac{1}{3}}$ and
$\mu \ge h ^{-\frac{3}{8}}$. It is not as good as $\mu \ge h ^{-\frac{3}{7}}$.

Then we use the smooth canonical form. In the operator perturbation terms have factors $\mu^{-2}$, $\mu^{-4}$ etc and we can use the standard approach to get rid off $\mu^{-4}\le \mu h $, so we need to consider only $\mu^{-2}$.

However let before scaling the second derivative of $W$ be of magnitude $\theta$; then after scaling it becomes of magnitude $\theta'=\gamma_0^2\gamma^2_1 \theta$ and then the perturbation is of magnitude $\theta \mu^{-2}$ but contribution of the error will be (after we compare the true Riemann sum and the corresponding integral and their difference $Ch ^{-3}\nu \mu^{-2}(\mu h )^2 (\theta')^{-\frac{1}{2}}\times \alpha_1^3\le
C \theta^{\frac{1}{2}}h ^{-1} \gamma_0^{-1}\gamma_1^{-1} \alpha_1^3\le
Ch ^{-1}\gamma_1^{-\frac{1}{2}}\alpha_1^3$ where we used that
$\theta \le C\gamma_0^2\gamma_1$.
Then summation over $\alpha_1$-partition of $\alpha_0$ element returns
$Ch ^{-1} \alpha_0^3$ and the summation over $\gamma_0$-partition returns
$Ch ^{-1}$ as desired. Therefore we covered zone $\cX_1$ for $\T$-term.

\subsection{Semiclassical $\T$-term: zone \texorpdfstring{$\cX_2$}{X\texttwoinferior}}
\label{sect-26-4-4-2}

\subsubsection{Tauberian estimate.}
\label{sect-26-4-4-2-1}

Tauberian estimate for cut-off expression is rather simple:
\begin{equation*}
C\mu h ^{-1} \gamma_0^{-1}\gamma_1^{-1}\gamma_2^{-1}\times h \gamma_0^{-3}\gamma_1^{-\frac{5}{2}}\gamma_1^{-2}\times \gamma_0^4 \gamma_1^3 \gamma_2^2 \alpha_2^3\asymp
C\mu \gamma_1^{-\frac{1}{2}}\gamma_2^{-1}\alpha_2^3
\end{equation*}
which nicely sums to $C\mu $ without logarithm due to the same positive eigenvalue arguments as before; for $\theta$-cut-off with $\theta \ge \mu h $ we get the same albeit with $\gamma_j$ defined by the same formula albeit with $(w_j - 2j\mu h s_j^{-1})$ replaced by $\theta s_j^{-1}$ where $s_j^{-1} $ means the scale; and this should be multiplied by $\theta/(\mu h )$. The result nicely sums to $Ch ^{-1} $. This is what was required.

\subsubsection{Magnetic Weyl expression.}
\label{sect-26-4-4-2-2}
Now we will get the same answer albeit $C\mu^{-4}$ term will be supplemented by $C\mu^{-\frac{3}{2}}h $ which in cut-off sum adds $C\mu^{-\frac{3}{2}}h \times \mu h ^{-2}\le C h ^{-1}$.

We can use the standard approach, with an error
$C\mu^{-\frac{3}{2}}h \times \theta/(\mu h ) \times \mu h ^{-2}\asymp C\theta \mu^{-\frac{3}{2}}h ^{-2}$ which means that we can take
$\theta = \mu ^{\frac{3}{2}}h $ which is sufficient to deal with with
$\theta \ge C\mu ^{\frac{3}{2}}h $; in particular, for
$\mu \ge h ^{\frac{2}{3}}$ we are done. But for
$\theta \ge \mu h ^{1-\delta}$ we can apply the weak magnetic field approach, which is sufficient. So we arrive to inequality
\begin{equation}
|\int^\tau_{-\infty} \int \phi(x)
\Bigl( e_\varphi (x,x,\tau) -
P'_{\beta,\varphi} \bigl(w(x)+\tau\bigr)\Bigr)\, dx\,d\tau |\le Ch ^{-1}
\label{26-4-31}
\end{equation}
and therefore we arrive to

\begin{proposition}\label{prop-26-4-12}
\begin{enumerate}[label=(\roman*), wide, labelindent=0pt]
\item\label{prop-26-4-12-i}
If $B \le Z^2$ the contribution of zone $\cX_2$ to the expression
\begin{equation}
\int^\tau_{-\infty} \int \phi(x)
\Bigl( e_\varphi (x,x,\tau) -
P'_{B,\varphi} \bigl(W(x)+\tau\bigr)\Bigr)\, dx\,d\tau
\label{26-4-32}
\end{equation}
does not exceed
$C\max\bigl((Z+B)^{\frac{1}{3}}Z^{\frac{4}{3}}, Z^{\frac{3}{5}}B^{\frac{4}{5}}\bigr)$.

\item\label{prop-26-4-12-ii}
If $B \le Z$ the contribution of zone $\cX_2$ to expression \textup{(\ref{26-4-32})} does not exceed $CZ^{\frac{5}{3}-\delta}$.
\end{enumerate}
\end{proposition}

\subsubsection{Mollification errors.}
\label{sect-26-4-4-2-3}

Further, we need to estimate
\begin{gather}
\int \phi(x) \Bigl(P'_B(W(x)+\tau))- P'_B(W^\TF_B(x)+\tau))\Bigr)\,dx ,
\label{26-4-33}\\
\int \phi(x) \Bigl(P_B(W(x)+\tau))- P_B(W^\TF_B(x)+\tau))\Bigr)\,dx,
\label{26-4-34}
\end{gather}
\vspace{-15pt}
\begin{multline}
\D \Bigl( \phi (x) \bigl(P'_B(W(x)+\tau))- P'_B(W^\TF_B(x)+\tau))\bigr),\\
\phi (x) \bigl(P'_B(W(x)+\tau))- P'_B(W^\TF_B(x)+\tau))\bigr)\Bigr)
\label{26-4-35}
\end{multline}
and
\begin{equation}
\| \phi(x) \nabla \bigl(W(x)-W^\TF_B(x)\bigr)\|^2.
\label{26-4-36}
\end{equation}
We start from local versions (so in fact we dealing with $w$ and $w^\TF_\beta$).

Obviously after all rescalings
$\hbar = h \gamma_0^{-3}\gamma_1^{-\frac{5}{2}}\gamma_2^{-2}$ and therefore $\varepsilon =\hbar ^{\frac{3}{2}}= h ^{\frac{3}{2}} (\gamma_0^{-3}\gamma_1^{-\frac{5}{2}}\gamma_2^{-2})^{-\frac{3}{2}}$ where we set $\delta=0$ but we will show that we have a reserve to set it as $\delta>0$ if we want to estimate (\ref{26-4-33}) by $h $ and (\ref{26-4-34})--(\ref{26-4-36}) by $h ^2$.

We claim that
\begin{gather}
|w -w^\TF_\beta|\le
C\varsigma \coloneqq
C\beta \eta
\bigl(\gamma_0^4\gamma_1^3\gamma_2^2 \varepsilon\bigr)^{\frac{5}{2}}.
\label{26-4-37}\\
\shortintertext{and}
|\nabla(w -w^\TF_\beta)|\le
C\varsigma_1 \coloneqq
C\beta \eta
\bigl(\gamma_0^4\gamma_1^3\gamma_2^2 \varepsilon\bigr)^{\frac{3}{2}}.
\label{26-4-38}
\end{gather}
Indeed, it follows from equation (\ref{26-4-2}).

Then the contribution of $\alpha_2$-element to (\ref{26-4-34}) does not exceed
$C\varsigma \varepsilon \alpha_2^3$ as measure of zone of $\alpha_2$-element where $w\ne w^\TF_\beta$ is $O(\varepsilon \alpha_2^3)$. One can see easily that $\varsigma \varepsilon= O(h ^{\frac{7}{3}})$ and therefore
$C\varsigma \varepsilon \alpha_2^3=O(h ^{\frac{7}{3}}\alpha_2^3)$ and the summation over $\alpha_2$-partition of $1$-element returns
$O(h ^{\frac{7}{3}})$.

Modulo above calculations the contribution of $\alpha_2$-element to (\ref{26-4-33}) does not exceed
$C\beta \varsigma^{\frac{1}{2}}\varepsilon \alpha_3^2$. One can check easily that $\varsigma \varepsilon=O(h ^{\frac{3}{2}}\gamma_2^{-\frac{1}{2}})$ and therefore $C\varsigma^{\frac{1}{2}} \varepsilon \alpha_2^3 =O(h ^{\frac{7}{3}}\alpha_1^3\gamma_2^{\frac{5}{2}})$ and the summation over $\alpha_2$-partition of $\alpha_1$-element returns $O(h ^{\frac{3}{2}}\alpha_1^3)$ and then the summation over $\alpha_1$-partition of $1$-element returns $O(h ^{\frac{3}{2}})$.

Similarly, expression (\ref{26-4-35}) with $\phi=\phi_{\alpha_2}$ does not exceed
$C\varsigma\varepsilon_2 \alpha_2^5\le Ch ^3 \alpha_2^4$ and the summation over $\alpha_2$-partition of $1$-element returns $O(h ^3)$. However we need to consider disjoint pairs of $\alpha_2$-elements belonging to given $\alpha_1$-element and their contribution does not exceed
\begin{equation*}
Ch ^3 \int \gamma_{2x}^{-\frac{1}{2}}\gamma_{2y}^{-\frac{1}{2}}|x-y|^{-1}\,dxdy
\le Ch ^2\alpha_1^5
\end{equation*}
and then summation over $\alpha_1$-partition of $1$-element returns
$O(h ^{\frac{3}{2}})$. We need also to consider disjoint pairs of $\alpha_1$-elements belonging to given $1$-element and their contribution does not exceed $Ch ^3 \int |x-y|^{-1}\,dxdy=O(h ^3)$.

Finally, contribution of $\alpha_2$-element to (\ref{26-4-36}) does not exceed
$C\varsigma_1^2\varepsilon \alpha_2^3$ and one can check easily that this does not exceed $C\alpha_1^3\gamma_2^{\frac{8}{3}}h ^{\frac{8}{3}}$ and the summation over $\alpha_2$-partition of $\alpha_1$-element returns $C\alpha_1^3h ^{\frac{8}{3}}$; then summation over $1$-partition of $1$-element returns
$O(h ^{\frac{8}{3}})$.

So, the scaled versions of (\ref{26-4-33}) and (\ref{26-4-34})--(\ref{26-4-36}) do not exceed $Ch $ and $Ch ^2$ respectively. Then the original versions of (\ref{26-4-33}), (\ref{26-4-34}), (\ref{26-4-34}), and (\ref{26-4-36}) do not exceed respectively
$C\zeta^3\ell^3 \times (\zeta\ell)^{-1}= C\zeta^2\ell^2$,
$C\zeta^5\ell^3 \times (\zeta\ell)^{-2}= C\zeta^3\ell$,
$C\zeta^6\ell^5 \times (\zeta\ell)^{-2}= C\zeta^4\ell^3$, and
$C\zeta^4\ell \times (\zeta\ell)^{-2}= C\zeta^2\ell^{-1}\le C\zeta^4\ell^3$.

Leaving the easy details to the reader we arrive to

\begin{proposition}\label{prop-26-4-13}
\begin{enumerate}[label=(\roman*), wide, labelindent=0pt]
\item\label{prop-26-4-13-i}
Contribution of zone $\cX_2$ to the mollification error \textup{(\ref{26-4-33})} does not exceed $CZ^{\frac{2}{3}}$.

\item\label{prop-26-4-13-ii}
Contribution of zone $\cX_2$ to the mollification error \textup{(\ref{26-4-34})} does not exceed
$CZ^{\frac{5}{3}} +o(Z^{\frac{4}{3}}B^{\frac{1}{3}})$.

\item\label{prop-26-4-13-iii}
Contributions of zone $\cX_2$ to the mollification errors \textup{(\ref{26-4-35})} and \textup{(\ref{26-4-36})} do not exceed $CZ^{\frac{5}{3}}$.
\end{enumerate}
\end{proposition}
and

\begin{proposition}\label{prop-26-4-14}
Let $B\le Z$. Then
\begin{enumerate}[label=(\roman*), wide, labelindent=0pt]
\item\label{prop-26-4-14-i}
Contribution of zone $\cX_2$ to the mollification error \textup{(\ref{26-4-33})} does not exceed $CB^{\delta}Z^{\frac{2}{3}-\delta}$.

\item\label{prop-26-4-14-ii}
Contributions of zone $\cX_2$ to the mollification errors \textup{(\ref{26-4-35})}--\textup{(\ref{26-4-35})} do not exceed $CB^{\delta}Z^{\frac{5}{3}-\delta}$.
\end{enumerate}
\end{proposition}

\begin{remark}\label{rem-26-4-15}
Consider the mollification parameter in ``absolute'' scale (i.e. $\ell$-scale): $\varepsilon= \gamma_0\gamma_1\gamma_2\gamma
(h /\gamma_0^3\gamma_1^{\frac{5}{2}}\gamma_2^2)^{\frac{2}{3}-\delta}$. One can see easily that $\varepsilon \ge
h ^{\frac{2}{3}-\delta}\ge (\mu^{-1}h )^{\frac{1}{2}-\delta_1}$ which makes reduction possible.
\end{remark}

\begin{remark}\label{rem-26-4-16}
All statements of Propositions~\ref{prop-26-4-13} and~\ref{prop-26-4-14} are valid for semiclassical errors as well except statements, concerning $T$-term; tose should include also terms $Ca^{-\frac{1}{2}}Z^{\frac{3}{2}}$ for
$a\le Z^{-\frac{1}{3}}$ and $Ca^{-\delta}Z^{\frac{5}{3}+\frac{1}{3}\delta}$ for
$a\ge Z^{-\frac{1}{3}}$.
\end{remark}

\section{Zone \texorpdfstring{$\cX_3$}{X\textthreeinferior}}
\label{sect-26-4-5}

Zone $\cX_3$ defined by $\mu h \ge C_0$, $h \le 1$,
$\{x\colon \ell(x)\le \epsilon_0 \bar{r}\}$ appears only as
$Z^{\frac{4}{3}}\le B\le Z^3$.  In this zone $W^\TF_B$ is smooth and no mollification is necessary. Further, in this zone the canonical form contains only one number $j=0$ and
$|D^\alpha W|\le C_\alpha \zeta^2\ell^{-|\alpha|}$ and $W\asymp \zeta^2$.

Therefore we have non-degeneracy condition fulfilled and applying the standard theory we conclude that in the scaled version contribution of $B(0,1)$ to the semiclassical errors in $\N$- and $\T$-terms and into $\D$-term are
$C\mu h ^{-1}$, $C\mu$ and $C\mu^2h ^{-2}$ respectively.

In the unscaled version they become $CB\ell^2$,
$CB\ell\zeta\le CBZ^{\frac{1}{2}}\ell^{\frac{1}{2}}$ and $CB^2\ell^3$ and after summation (where for $\D$-term we need to consider mixed contribution of different layers) we arrive to the same expressions calculated as $\ell=\bar{r}=B^{-\frac{2}{5}}Z^{\frac{1}{5}}$ i.e.
$CB^{\frac{1}{5}}Z^{\frac{2}{5}}$,
$CB^{\frac{4}{5}}Z^{\frac{3}{5}}$ and $CB^{\frac{4}{5}}Z^{\frac{3}{5}}$ respectively. Thus we have proven

\begin{proposition}\label{prop-26-4-17}
Let $Z^{\frac{4}{3}}\le B\le Z^3$. Then
\begin{enumerate}[label=(\roman*), wide, labelindent=0pt]
\item\label{prop-26-4-17-i}
Contribution of zone $\cX_3$ to the $\N$-error does not exceed $CZ^{\frac{2}{5}}B^{\frac{1}{5}}$.

\item\label{prop-26-4-17-ii}
Contributions of zone $\cX_3$ to the $\T$-error and $\D$-term do not exceed $CZ^{\frac{3}{5}}B^{\frac{4}{5}}$.
\end{enumerate}
\end{proposition}

\chapter{Semiclassical analysis in the boundary strip for \texorpdfstring{$M\ge 2$}{M\textge 2}}
\label{sect-26-5}
To finish our analysis we need to get the same estimates as before in the \emph{boundary strip\/}
\begin{gather}
\cY\coloneqq \{x\colon W(x)+\nu \le \epsilon G,\,
\epsilon \bar{r}\le \ell (x)\le c\bar{r}\}
\label{26-5-1}
\shortintertext{with}
G\coloneqq \left\{\begin{aligned}
&(Z-N)_+^{\frac{4}{3}}
&&\text{for\ \ } B\le (Z-N)_+^{\frac{4}{3}},\\
&B\qquad
&&\text{for\ \ } (Z-N)_+^{\frac{4}{3}}\le B\le Z^{\frac{4}{3}},\\
&Z^{\frac{4}{5}}B^{\frac{2}{5}} &&\text{for\ \ } B\ge Z^{\frac{4}{3}}.
\end{aligned}\right.
\label{26-5-2}
\end{gather}
which coincides with (\ref{26-2-41}) as $B\ge (Z-N) ^{\frac{4}{3}}$. Recall that
$\bar{r}= (Z-N)_+^{-\frac{1}{3}}$, $\bar{r}= B^{-\frac{1}{4}}$ and
$\bar{r}= B^{-\frac{2}{5}}Z^{\frac{1}{5}}$ in these three cases respectively.
Analysis of the \emph{external zone\/}\index{zone!external}
$\cX_4\coloneqq \{x\colon \ell(x)\ge C_1\bar{r}\}$ will be trivial and \emph{inner zone\/} $\{x\colon W(x)+\nu \ge \epsilon G\}$ has been covered already.
\enlargethispage{\baselineskip}

\section{Properties of \texorpdfstring{$W_B ^\TF$}{W\_B\^{TF}} if $N=Z$}
\label{sect-26-5-1}
Let us explore properties of $W_B ^\TF$ in $\cY$ if $N=Z$\,\footnote{\label{foot-26-28} I.e. $\nu=0$ and $G= B\bar{r}^{-4}$.}
Let us rescale $x\mapsto x'=x\bar{r} ^{-1}$, $W\mapsto w=G^{-1}W$ and define $h = G^{-\frac{1}{2}}\bar{r}^{-1}$, $\mu= G^{-\frac{1}{2}}B\bar{r}$. Then

\begin{enumerate}[label=(\alph*), wide, labelindent=0pt]
\item\label{sect-26-5-1-a}
In the case $ B\le Z ^{\frac{4}{3}}$ we need to rescale
$w(x')=B ^{-1}W_B ^\TF(x'\bar{r})$ and take $h =B ^{-\frac{1}{4}}\le 1$,
$\mu =B ^{\frac{1}{4}}\ge 1$, $\mu h =1$.

\item\label{sect-26-5-1-b}
On the other hand, for $B\ge Z ^{\frac{4}{3}}$ one should set
$w(x')=\bar{r}Z^{-1}W_B ^\TF(x'\bar{r})$ and
$h =(Z\bar{r}) ^{-\frac{1}{2}}=(B Z^{-3}) ^{\frac{1}{5}}\le 1$,
$\mu =B Z^{-\frac{1}{2}}\bar{r} ^{\frac{3}{2}}=
B ^{\frac{2}{5}} Z^{-\frac{1}{5}}\ge 1$,
$\mu h = B ^{\frac{3}{5}}Z ^{-\frac{4}{5}}\ge 1$
($\mu h \asymp 1$ iff $B\lesssim Z ^{\frac{4}{3}}$, $h \asymp 1$ iff
$B\asymp Z ^3$).
\end{enumerate}

We will use now only rescaled coordinates unless the opposite is specified. Then in $\cY$ rescaled
\begin{equation}
\Delta w= \kappa w_+ ^{\frac{1}{2}},\qquad \kappa =12,\qquad
w\to \theta =\nu \bar{\zeta} ^{-2}\quad \text{as\ \ }|x|\to \infty,
\label{26-5-3}
\end{equation}
with $\bar{\zeta}\coloneqq G^{\frac{1}{2}}$ where one can always get
$\kappa =12$ after rescaling $w\mapsto 144\kappa ^{-2}w$.

\begin{proposition}\label{prop-26-5-1}
Let $Z= N$. Then in $\cY$ after rescaling
\begin{gather}
|D ^\alpha w|\le C_\alpha w\gamma ^{-|\alpha |}
\qquad \forall \alpha
\label{26-5-4}\\
\intertext{with the scaling function $\gamma =w ^{\frac{1}{4}}$ and}
|\nabla w ^{\frac{1}{4}}|\le 1+Cw ^t
\label{26-5-5}
\end{gather}
with some constant $C$ and exponent $t>0$.
\end{proposition}

\begin{proof}
\begin{enumerate}[label=(\alph*), wide, labelindent=0pt]
\item\label{pf-26-5-1-a}
Rescaling $x\mapsto x\bar{r} ^{-1}$ we get an equations
$\textup{(\ref{26-5-3})}_0=\textup{(\ref{26-5-3})}$ with $\theta =0$. We know that $W=0$ for $\ell (x)\ge c\bar{r}$; so after rescaling $w=0$ for $\ell (x)\ge c$. On the other hand, $w\asymp 1$ as $\ell (x)\le \epsilon $ (uniformly with
respect to all the parameters).

Let us consider solution of the equation
\begin{equation}
\Delta w_s= 12 w_+ ^s
\label{26-5-6}
\end{equation}
in $\Omega =\{w\le \epsilon ,\ell \le c\}$ with the boundary condition
$w_s=w$ at $\partial \Omega $; $s>\frac{1}{2}$.

Note first that $w_s\ge 0$. Really, $w_s$ is the solution of the variational problem to minimize
\begin{equation}
\|\nabla w\| ^2+24(s+1) ^{-1} \int w_+ ^{s+1}\,dx
\label{26-5-7}
\end{equation}
 and one makes this functional only less replacing $w$ by $w_+$.

Further, the standard maximum principle arguments show that $w_s\searrow$ as $s\searrow$\,\footnote{\label{foot-26-29} If $\Delta w_i=f_i(w_i)$ in $\Omega $, $f_i(w)\nearrow$ as $w\nearrow$ and $f_1(w)\ge f_2(w)$ then $\Delta (w_1-w_2)>0$ as $w_1>w_2$ and then $w_1-w_2$ does not reach maximum inside $\Omega $.}. Obviously $w_s\searrow w$ and
$w_s\to w$ in $\sC ^\infty $ in $\{x\colon w(x)>0\}$ as $s\searrow \frac{1}{2}$.

We claim that
\begin{claim}\label{26-5-8}
$w_s\in \sC ^{4s+2}$.
\end{claim}

To prove (\ref{26-5-8}) note first that $w\in \sC ^{2-\delta '}$ uniformly with respect to all the parameters for any $\delta '>0$. Then
$w _s^s\in \sC ^{s-\delta }$ and then
(\ref{26-5-6}) yields that $w_s\in \sC ^{2+s-\delta }$ as soon as
$s-\delta \notin \bZ$. Then since $w_s\ge 0$ we get
$|\nabla w_s|\le cw_s ^{\frac{1}{2}}$ and so $w_s ^s\in \sC ^{s-\frac{1}{2}}$.
Then equation (\ref{26-5-6}) again yields that $w_s\in \sC ^{s+\frac{3}{2}}$.

\pagebreak
Now we need more subtle arguments. First, for $|y|=1$
\begin{equation}
w_s(x+ty)=w_s(x)+t(\nabla w_s)_x\cdot y +
\frac{1}{2} (\nabla ^2w_s)_x(y)t ^2+ O\bigl(t ^3\bigr).
\label{26-5-9}
\end{equation}
Then the lowest eigenvalue $\varsigma $ of $\nabla ^2w_s$ at $x$ should be greater than $-Cw_s ^{\frac{1}{3}}$. Indeed, otherwise we can take $y$ as
the corresponding eigenvector and $t$ with $|t|=\epsilon \varsigma $ and
with a sign making second term non-positive and get $w_s(x+ty)<0$.

This lower estimate for eigenvalues of $\nabla ^2w_s$ and equation (\ref{26-5-6}) yield that $|\nabla ^2w_s|\le Cw_s ^{\frac{1}{3}}$. But then
$|\nabla w_s|\le Cw _s^{\frac{2}{3}}$. Really, otherwise picking
$y=|\nabla w_s| ^{-1}\nabla w_s$ with $|t|=\epsilon |\nabla w_s| ^{\frac{1}{2}}$ and an appropriate sign we would get $w_s(x+ty)<0$.

These estimates yield that $w_s(x')\asymp w_s(x)$ in
$B\bigl(x,\gamma (x)\bigr)$ with $\gamma (x)=\epsilon w _s^{\frac{1}{3}}$.
Then $w_s ^s\in \sC ^{\frac{3}{2}}$. In fact, let us consider
$f=w_s ^{s-1}\nabla w$ and $|f(x)-f(x')|$. Let us consider first
$|x-x'|\ge {\frac{1}{3}}\bigl(\gamma (x)+\gamma (x')\bigr)$; since
$|f(x)|\le \gamma (x) ^{\frac{1}{2}}$ at each point we get that
$|f(x)-f(x')\le |x-x'| ^{\frac{1}{2}}$.

On the other hand, for $|x-x'|\le \frac{1}{3}\bigl(\gamma (x)+\gamma (x')\bigr)$ we conclude that $\gamma (x)\asymp \gamma (x')$ and
$|f(x)-f(x')|\le |\nabla f|\cdot |x-x'|\le |x-x'| ^s$ due to inequality
$|\nabla f|\le
|\nabla ^2w_s|w_s ^{s-1}+|\nabla w_s| ^2w_s ^{s-2}\le C\gamma ^{s-1}$.

Therefore $w_s ^s\in \sC ^{s+1}$ and equation (\ref{26-5-6}) yields that
$w_s\in \sC ^{3+s}$.

\item\label{pf-26-5-1-b}
In the next round we assume that $w_s\in \sC ^{4+s-\delta }$ with \emph{some\/} $\delta \in (0,1)$. Then
\begin{multline}
w_s(x+ty)\le \\
w_s(x)+t(\nabla w_s)_x\cdot y+\frac{1}{2}(\nabla ^2w_s)_x(y)t ^2+
\frac{1}{6}(\nabla ^3w_s)_x(y)t ^3 +C|t|^p
\label{26-5-10}
\end{multline}
with $p=\min(4,4+s-\delta )$.

We claim now that the lowest possible eigenvalue $\varsigma $ of
$\bigl(\nabla ^2w_s\bigr)_x$ is greater than $-C w_s ^{(p-2)/ p }$. Really, otherwise let us pick up $y$ as the corresponding eigenvector, $t$ with
$|t|=\epsilon |\varsigma | ^{1/(p-2) }$ and with a sign making expression
\begin{equation*}
t(\nabla w_s)_x\cdot y+ \frac{1}{6} t ^3 (\nabla ^3w_s)_x(y)
\end{equation*}
non-positive and get $w_s(x+ty)<0$ again. Now equation (\ref{26-5-6}) yields
that inequality\begin{phantomequation}\label{26-5-11}\end{phantomequation}
\pagebreak
\begin{equation}
|\nabla ^kw_s|\le Cw_s ^{(p-k) / p }
\tag*{$\textup{(\ref*{26-5-11})}_k$}\label{26-5-11-k}
\end{equation}
holds with $k=2$.

Further, we claim that this inequality holds with $k=1,3$. Indeed, if one or both of these inequalities are violated then let us take corresponding $y$ and $t$ with
\begin{equation*}
|t|=\epsilon \Bigl(|\nabla w_s| ^{1/(p-1)}+|\nabla ^3w_s(y)|^{1/(p-3)}\Bigr)
\end{equation*}
(calculated on $y$); replacing $\epsilon $ by $2\epsilon $ if necessary we get
\begin{equation*}
|t(\nabla w_s)_x\cdot y +
\frac{1}{6}t ^3 (\nabla ^3w_s)_x(y)\bigr|\ge
\epsilon _0 |t(\nabla w_s)_x\cdot y| +
|\frac{1}{6}t ^3 (\nabla ^3w_s)_x(y)|
\end{equation*}
and choosing an appropriate sign of $t$ we get $w(x+ty)<0$.

\emph{Therefore inequalities $\textup{(\ref{26-5-11})}_{1-3}$ hold\/}. The same arguments as above with $\gamma =w_s^{1/p}$ lead us to $w ^s\in \sC ^{ps}$ and then equation (\ref{26-5-6}) yields that $w_s\in \sC ^{ps+2}$. So, now we came back with $\delta $ replaced by $\delta '=2+s-ps$ and one can see easily that if
$\delta >s$ then $\delta ' =s+(2-4s)+(\delta -s)s$ and after few repeats
$\delta <s$. Then we get (\ref{26-5-8}). Unfortunately, constants depend on $s$
due to the fact that $\Delta w\in \sC ^2$ fails to yield $w\in \sC ^4$.

\item\label{pf-26-5-1-c}
Now we are going to finish the proof of (\ref{26-5-4}). Let us consider $w_s$ again and let $\gamma =\gamma _{s,\delta }=w_s ^{1/(4-\delta )}$. Due to the previous inequalities $\gamma \in \sC ^1$. We claim that $|\nabla \gamma |$ is bounded uniformly with respect to $s,\delta $. Note first that
$\Delta \gamma ^{4-\delta }=\gamma ^{(4-\delta )s}$ implies that
\begin{equation}
a|\nabla \gamma | ^2+b\gamma \Delta \gamma =\gamma ^\sigma
\label{26-5-12}
\end{equation}
with $a=\frac{1}{12}(4-\delta )(3-\delta )$, $b=\frac{1}{12}(4-\delta )$, and $\sigma =4s-2+(1-s)\delta $. Let $\psi =|\nabla \gamma | ^2$; obviously $\psi $
is uniformly bounded at $\partial \Omega $. Let us consider maximum of $\psi $ reached inside $\Omega $. At the point of maximum
\begin{equation}
\sum_i \gamma _{x_ix_j}\gamma _{x_i}=0
\label{26-5-13}\\
\end{equation}
and
\begin{multline*}
\frac{1}{2}\Delta \psi =\sum_{i,j}\gamma _{x_ix_j} ^2+
\sum_i\gamma _{x_i}\bigl(\Delta \gamma \bigr)_{x_i}=\\
\sum_{i,j}\gamma _{x_ix_j} ^2+b ^{-1}\sum_i\gamma _{x_i}
\Bigl(\gamma ^{-1}\bigl(\gamma ^\sigma -a|\nabla \gamma | ^2\bigr)\Bigr) _{x_i}.
\end{multline*}
Due to (\ref{26-5-12}) and due to (\ref{26-5-13}) this expression is equal to
\begin{equation*}
\sum_{i,j}\gamma _{x_ix_j} ^2-
b ^{-1}\gamma ^{-2}|\nabla \gamma |^2
\bigl(\gamma ^\sigma -a|\nabla \gamma | ^2\bigr)
+b ^{-1}\sigma \gamma ^{\sigma -2}|\nabla \gamma | ^2
\end{equation*}
and therefore at an inner point of minimum
$a|\nabla \gamma| ^2\le \gamma ^\sigma $. So, $|\nabla \gamma |\le C$ is proven and for $s\searrow {\frac{1}{2}}$,
$\delta \searrow 0$ we get that $|\nabla w ^{\frac{1}{4}}|\le C$.

Let us pick $\gamma (x)=\epsilon' w ^{\frac{1}{4}}(x)$; then
$|\nabla w|\le {\frac{1}{2}}$ and $w(x)\asymp w(\x)$ in
$B\bigl(x,\gamma (\x)\bigr)$. This and equation (\ref{26-5-6}) easily yield (\ref{26-5-4}).

To prove inequality (\ref{26-5-5}) let us consider $w_s$ again and let us take now $\psi =|\nabla \gamma | ^2-F\gamma ^{2t}$ with $t>0$; obviously $\psi $ is
non-positive at $\partial \Omega $ for sufficiently large $F$. Let us
consider maximum of $\psi $ reached inside $\Omega $. At the point of
maximum
\begin{equation}
\sum_i \gamma _{x_ix_j}\gamma _{x_i}-Ft\gamma ^{2t-2}\gamma _{x_j}=0
\tag*{$\textup{(\ref*{26-5-13})}'$}\label{26-5-13-'}
\end{equation}
and the same arguments as before (plus inequality $|\nabla \gamma |\le C_0$)
show that at an inner point of maximum
$a|\nabla \gamma | ^2\le \gamma ^\sigma +CtF\gamma ^{2t}$ where $C$ does not
depend on $F$ and small $t>0$. Then at this point $\psi \le 1$ for
small enough $t>0$ and as $s\to \frac{1}{2}$ and $\delta \to 0$ we get (\ref{26-5-5}). \end{enumerate}\end{proof}

The following statement heavily uses estimate (\ref{26-5-5}):

\begin{proposition}\label{prop-26-5-2}
The following estimate holds
\begin{equation}
\D(\gamma ^{-1+s},\gamma ^{-1+s})\le Cs ^{-2}
\label{26-5-14}
\end{equation}
with some constant $C$ which does not depend on $s\in (0,1)$ where we set
$\gamma^{-1+s}\coloneqq w_+^{\frac{1}{4}(-1+s)}$ (i.e. it is $0$ as $w\le 0$).
\end{proposition}

\begin{proof}
As in the notations of the proof of Proposition~\ref{prop-26-5-1} $\delta=0$ and $s=\frac{1}{2}$ we have (\ref{26-5-12}) with $a=1$, $b=3$ and $\sigma=0$:
\begin{gather}
\frac{1}{3}\gamma \Delta \gamma +|\nabla \gamma | ^2=1.
\label{26-5-15}\\
\shortintertext{Then}
\gamma ^{-1+s}=\gamma ^{-1+s}|\nabla \gamma | ^2+
\frac{1}{3}\gamma ^s\Delta \gamma =
(1-\frac{s}{3})\gamma ^{-1+s}|\nabla \gamma | ^2+
\frac{1}{3(1+s)}\Delta \gamma ^{1+s}
\notag
\end{gather}
and
\begin{multline*}
\D(\gamma ^{-1+s},\gamma ^{-1+s})\le
(1-\frac{s}{3})\D(\gamma ^{-1+s}|\nabla \gamma | ^2,\gamma ^{-1+s})+C\le\\
(1-\frac{s}{3})\D(\gamma ^{-1+s},\gamma ^{-1+s})+
C\D(\gamma ^{-1+t+s},\gamma ^{-1+s})+C
\end{multline*}
due to (\ref{26-5-5}) and this yields
\begin{equation*}
\D(\gamma ^{-1+s},\gamma ^{-1+s})\le
Cs ^{-2}\D(\gamma ^{-1+t+s},\gamma ^{-1+t+s})+ Cs ^{-1}.
\end{equation*}
Substituting $s+mt$ instead of $s$, $0\le m\le Ct ^{-1}$ we recover
(\ref{26-5-14}). \end{proof}

\section{Analysis in the boundary strip $\cY$ for \texorpdfstring{$N\ge Z$}{N\textge Z}}
\label{sect-26-5-2}

We consider now the case of  if $N\ge Z$ (i.e. $\nu=0$ and $G= B\bar{r}^{-4}$).

It is really easy to construct the proper potential in this case: we just take
\begin{equation}
w_\varepsilon =w\phi _\varepsilon ,\qquad
\phi _\varepsilon =f (w\varepsilon ^{-4})
\label{26-5-16}
\end{equation}
with $f\in \sC  ^\infty ((\frac{1}{2},\infty ))$, $\supp (f)\subset (\frac{1}{2},\infty )$, $0\le f\le 1$, $f(t)=1$ for
$t>1$. Note that due to (\ref{26-5-3})
\begin{align*}
&\D(\gamma ^{-1}\phi _\varepsilon ,
\gamma ^{-1}\varphi _\varepsilon )\,\le
C\varepsilon ^{-2s}\D(\gamma ^{s-1},\gamma ^{s-1})\,\le
Cs ^{-2}\varepsilon ^{-2s},\\[3pt]
&\D(1-\phi _\varepsilon , 1-\phi _\varepsilon )\le
C\varepsilon ^{2-2s}\D(\gamma ^{s-1},\gamma ^{s-1})\le
Cs ^{-2}\varepsilon ^{2-2s};
\end{align*}
then minimizing with respect to $s$ ($=|\log \varepsilon|^{-1}$) the right-hand expression we conclude that
\begin{gather}
\D\bigl(\gamma ^{-1}\phi _\varepsilon ,
\gamma ^{-1}\varphi _\varepsilon \bigr)+
\varepsilon ^{-2}\D\bigl(1-\phi _\varepsilon , 1-\phi _\varepsilon )\le
C\bigl(1+|\log \varepsilon |\bigr)^2
\label{26-5-17}\\
\intertext {and therefore}
\int \gamma ^{-1}\phi _\varepsilon \, dx + \varepsilon ^{-1}\int (1-\phi _\varepsilon )\, dx
\le C\bigl(1+|\log \varepsilon |\bigr).
\label{26-5-18}
\end{gather}

\begin{remark}\label{rem-26-5-3}
\begin{enumerate}[label=(\roman*), wide, labelindent=0pt]
\item\label{rem-26-5-3-i}
Recall that all these integrals are taken over domain $\{x\colon w(x) >0\}$. To avoid possible troubles we pick $\varepsilon =h ^{\frac{1}{3}}$ and set in the zone $\{x\colon w (x)\le C_0 h ^{\frac{4}{3}}\}$
\begin{gather}
\gamma (x)=\dist (x, \{w\ge 2C_0 h ^{\frac{4}{3}}\}),\notag\\[3pt]
w_\varepsilon =\left\{\begin{aligned}
- \gamma ^4\phi '_\varepsilon \quad
&\text{for\ \ }\gamma \le \varepsilon ,\\
-\varepsilon ^4 \quad &\text{for\ \ }\gamma \ge \varepsilon
\end{aligned}\right.
\tag*{$\textup{(\ref*{26-5-16})}'$}\label{26-5-16-'}
\end{gather}
with $\phi '_\varepsilon =f(\gamma \varepsilon ^{-1})$ and then in the complemental domain $\{x\colon w(x)\le -\varsigma ^2\}$ our assumptions are fulfilled with $\varsigma =\varepsilon ^2$ and
$\varsigma \gamma =\gamma ^3\ge h $.

\item\label{rem-26-5-3-ii}
Further, for $\varepsilon =h ^{\frac{1}{3}-\delta }$ with sufficiently small exponent $\delta >0$ it does not break estimate for mollification error in $\T$-term.

\item\label{rem-26-5-3-iii}
Furthermore, for $t>\varepsilon $
\begin{gather*}
\mes (\{x\colon \gamma(x) \le t\})\le
Ct ^3\varepsilon ^{-3}\mes (\{x\colon \gamma (x)\le \epsilon \varepsilon \})
\intertext{and therefore}
h ^s\int \gamma ^{-1-s}\varsigma ^{-s}\, dx\le
C\varepsilon ^{-1}\mes (\{x\colon \gamma (x)\le \epsilon \varepsilon \})\le CL\coloneqq
C(1+|\log h |)
\end{gather*}
for sufficiently large $s$.
\end{enumerate}
\end{remark}

Using these estimates and the last remark we can prove easily

\begin{proposition}\label{prop-26-5-4}
Let $N\ge Z$. Then
\begin{enumerate}[label=(\roman*), wide, labelindent=0pt]
\item\label{prop-26-5-4-i}
Contribution of $\cY\cup \cX_4$ with \emph{external zone\/}\index{zone!external} $\cX_4\coloneqq \{x\colon w(x)=0\}$ to mol\-li\-fication and semiclassical  errors in $\N$-term do not exceed $CT_0\varepsilon^3 (1+|\log \varepsilon|)$ and $R_0(1+|\log \varepsilon|)$ respectively with
\begin{phantomequation}\label{26-5-19}\end{phantomequation}
\begin{multline}
T_0=B ^{\frac{3}{4}},\qquad R_0=B ^{\frac{1}{2}}, \qquad
T=B ^{\frac{7}{4}},\qquad R=B ^{\frac{5}{4}}\\
\text{for\ \ }
B\le Z ^{\frac{4}{3}}
\tag*{$\textup{(\ref*{26-5-19})}_1$}\label{26-5-19-1}
\end{multline}
and
\begin{multline}
T_0=Z,\qquad R_0=B ^{\frac{1}{5}} Z ^{\frac{2}{5}} , \qquad
T=Z ^{\frac{9}{5}}B ^{\frac{2}{5}},\qquad R=Z ^{\frac{3}{5}}B ^{\frac{4}{5}}\\
\text{for\ \ }
Z ^{\frac{4}{3}}\le B\le Z^3.
\tag*{$\textup{(\ref*{26-5-19})}_2$}\label{26-5-19-2}
\end{multline}
\item\label{prop-26-5-4-ii}
Contribution of $\cY\cup \cX_4$ to mollification and semiclassical $\D$-terms do not exceed $CT \varepsilon^6 (1+|\log \varepsilon|)^2$ and
$R(1+|\log \varepsilon|)^2$ respectively.
\enlargethispage{\baselineskip}

\item\label{prop-26-5-4-iii}
Contribution of $\cY\cup \cX_4$ to both mollification and semiclassical errors in $\T$-term do not exceed $CT \varepsilon^7 (1+|\log \varepsilon|)$ and $CR$ respectively.
\end{enumerate}
\end{proposition}

\begin{proof}
Really, estimates for mollification errors and terms immediately follow from the inequality
\begin{equation}
\mes (\{x\colon w(x)\le \varepsilon^4\}) \le C\varepsilon (1+|\log \varepsilon|)
\label{26-5-20}
\end{equation}
which is due to (\ref{26-5-18}).\pagebreak

Let us consider semiclassical errors and terms.
\begin{enumerate}[label=(\roman*), wide, labelindent=0pt]
\item\label{pf-26-5-4-i}
Let us consider $\N$-term first. Let us consider all possible balls and
their contributions: the contribution of each ball
$B\bigl(x,\gamma (x)\bigr)$ to the semiclassical error does not exceed
$C\mu h ^{-1} \gamma ^2\asymp CB\bar{r}^2 \gamma^2$
and the total contribution does not exceed
\begin{equation}
CR_0\int \gamma (x)^{-1}\, dx\le CR_0\big(1+|\log \varepsilon |\bigr)
\label{26-5-21}
\end{equation}
where $R_0=B\bar{r}^2$; recall that $\gamma(x)\ge \varepsilon$.

\item\label{pf-26-5-4-ii}
Consider semiclassical $\D$-term. Let us consider all possible balls and
their contributions: the similar arguments with the analysis of disjoint balls of different types and with analysis of the intersecting balls (of the same type) lead us to the proper estimate of the contribution of ${\cY}_4\cup \cX_4$ to semiclassical $\D$-term: namely, it does not exceed
$CR_0 ^2\bar{r}^{-1}\big(1+|\log \varepsilon |\bigr) ^2$ (i.e. expression (\ref{26-5-21}) squared and multipled by $C\bar{r}^{-1}$) where
$R_0 ^2\bar{r} ^{-1}\asymp R$.
\enlargethispage{\baselineskip}

\item\label{pf-26-5-4-iii}
Consider $\T$-term. Let us consider all possible balls and
their contributions. Contribution of each ball $B\bigl(x,\gamma (x)\bigr)$ to the semiclassical error does not exceed
$C\bar{\zeta}^2 \mu \varsigma ^2\gamma \asymp
CB\bar{\zeta}^2 \bar{r}\varsigma \gamma^2 $ and the total
contribution does not exceed
\begin{equation}
CR\int \varsigma (x)\gamma (x)^{-2}\, dx\asymp CR
\label{26-5-22}
\end{equation}
where $R=B\bar{\zeta}^2\bar{r}$ and $\varsigma(x)\asymp \gamma(x)^2$.
\end{enumerate}
\end{proof}

Then picking appropriate $\varepsilon =h ^{\frac{1}{3}}$ we arrive to

\begin{corollary}\label{cor-26-5-5}
Let $N\ge Z$. Then
\begin{enumerate}[label=(\roman*), wide, labelindent=0pt]
\item\label{cor-26-5-5-i}
Contributions of $\cY\cup \cX_4$ to all errors in $\N$-terms do not exceed $CR_0L$ with $L=(1+|\log BZ^{-3}|)$.

\item\label{cor-26-5-5-ii}
Contribution of $\cY\cup \cX_4$ to all $\D$-terms do not exceed $CRL^2$.

\item\label{cor-26-5-5-iii}
Contribution of $\cY\cup \cX_4$ to all errors in $\T$-terms do not exceed $CR$.
\end{enumerate}
\end{corollary}

We will sum contributions of all zones to errors in Propositions~\ref{prop-26-5-14} and~\ref{prop-26-5-17} below.

\begin{remark}\label{rem-26-5-6}
Could we get rid off the logarithmic factors i.e. make $L=1$ as it was in the case $M=1$?

\begin{enumerate}[label=(\roman*), wide, labelindent=0pt]
\item\label{rem-26-5-6-i}
With the mollification errors we need to replace (\ref{26-5-20}) by
\begin{equation}
\mes (\{x\colon w(x)\le \varepsilon^4\}) \le C\varepsilon;
\label{26-5-23}
\end{equation}
\item\label{rem-26-5-6-ii}
With the semiclassical terms our arguments here are insufficient even if we established (\ref{26-5-23}); we need extra propagation arguments in the direction of decaying $w$ along magnetic lines--exactly as in the case $M=1$. Surely there could be points where such arguments do not work; f.e. consider $M=2$ and nuclei so that $|\y_1-\y_2|$ is slightly less than $\bar{r}_1+\bar{r}_2$ where $\bar{r}_{1,2}$ are precise radii of support. Then $w$ reaches its minimum at $\cY$.

So, we need to prove that the measure of such points is sufficiently small (f.e. less than $C|\log BZ^{-3}|^{-1}$).
\end{enumerate}
\end{remark}

Unfortunately, we do not know how to make the above remark work and we suggest
\enlargethispage{\baselineskip}

\begin{Problem}\label{Problem-26-5-7}
Follow through the discussed plan. For $M=2$ it could be easier due to the rotational symmetry of the potential $W^\TF_B$.
\end{Problem}

\section{Analysis in the boundary strip $\cY$ for $N< Z$}
\label{sect-26-5-3}

Now let us consider the case of $N<Z$ (i.e. $\nu <0$).

\subsection{Case \texorpdfstring{$B\ge (Z-N)_+^{\frac{4}{3}}$}{B \textge (Z-n)\textplusinferior\^{4/3}}}
\label{sect-26-5-3-1}

We start from the case $B\ge (Z-N)_+^{\frac{4}{3}}$ when $\bar{r}=\min(B^{-\frac{1}{4}}, Z^{\frac{1}{5}}B^{-\frac{2}{5}})$ matching cases $B\lesssim Z^{\frac{4}{3}}$ and $Z^{\frac{4}{3}}\lesssim B\lesssim Z^3$.

\begin{remark}\label{rem-26-5-8}
\begin{enumerate}[label=(\roman*), wide, labelindent=0pt]
\item\label{rem-26-5-8-i}
The results of the previous Subsection~\ref{sect-26-4-2} remain true as long as
$|\nu|G^{-1}\le C_0h ^{\frac{4}{3}}$; in other words, as
$(Z-N)_+\le C_0 G \bar{r}h ^{\frac{4}{3}}$. Plugging $\bar{r}$, $G$ and $h $, we rewrite it as
\begin{equation}
(Z-N)_+\le C_0 \min\bigl( B^{\frac{5}{12}},\, Z^{\frac{1}{5}}B^{\frac{4}{15}}\bigr)
\label{26-5-24}
\end{equation}
matching cases $B\lesssim Z^{\frac{4}{3}}$ and
$Z^{\frac{4}{3}}\lesssim B\lesssim Z^3$.

\item\label{rem-26-5-8-ii}
Therefore in this Subsection we assume that condition (\ref{26-5-24}) fails. Let
$\theta =|\nu | G ^{-1}\asymp (Z-N)_+\cdot \max\bigl(B^{-\frac{3}{4}},\,Z^{-1}\bigr)$,
also matching cases
$B\lesssim Z^{\frac{4}{3}}$ and $Z^{\frac{4}{3}}\lesssim B\lesssim Z^3$.
\end{enumerate}
\end{remark}

\begin{proposition}\label{prop-26-5-9}
Consider dependence of $W^\TF_{B}=W^\TF_{B(\nu)} (x) $ on $\nu $. Then
\begin{enumerate}[label=(\roman*), wide, labelindent=0pt]
\item\label{prop-26-5-9-i}
$W^\TF_{B(\nu)} (x) +\nu$ is non-decreasing with respect to $\nu$ at each point $x$.

\item\label{prop-26-5-9-ii}
 $W^\TF_{B(\nu)} (x)$ is non-increasing with respect to $\nu$ at each point $x$.

\item\label{prop-26-5-9-iii}
In particular, $W^\TF_{B(\nu)} (x) +\nu \nearrow W^\TF_{B(0)} (x)$ and
$W^\TF_{B(\nu)} (x) \searrow W^\TF_{B(0)} (x)$ at each point $x$ as
$\nu \nearrow 0$.
\end{enumerate}
\end{proposition}

\begin{proof}
\begin{enumerate}[label=(\roman*), wide, labelindent=0pt]
\item\label{pf-26-5-9-i}
Consider $W_j \coloneqq W^\TF_{B(\nu_j)}+\nu_j$ with $0>\nu_1>\nu_2$. One can prove easily that $W^\TF_B-V$ is a continuous function and since
\begin{equation}
W_1 -W_2 \to \nu_1-\nu_2 \qquad \text{as\ \ } \ell (x)\to \infty
\label{26-5-25}
\end{equation}
with $\nu_1-\nu_2>0$ we conclude that $W_1\ge W_2$ at each point $x$ (which is exactly our Statement~\ref{prop-26-5-9-i}) unless $W_1-W_2$ achieves a negative minimum at some point~$x^*$:

\begin{enumerate}[label=(\alph*), wide, labelindent=0pt]
\item\label{pf-26-5-9-ia}
Let $x^*\ne \y_m$; then $\Delta (W_1-W_2) (x^*)=P'_B(W_1)-P'_B(W_2)\le 0$ because $W_1<W_2$ at $x^*$ and therefore $x^*$ cannot be such point.

\item\label{pf-26-5-9-ib}
Let $x^*=\y_m$. From Thomas-Fermi equations for $W_{1,2}$ one can prove easily that
\begin{multline*}
(W_1-W_2)(x)=(W_1-W_2)(\y_m) + L_m(x-\y_m) + \\
\kappa_m |x-y_m| ^{\frac{3}{2}} (W_1-W_2)(\y_m)+O(|x-\y_m|^2)
\end{multline*}
near $\y_m$ where $L_m(x)$ is a linear function and $\kappa_m>0$ and therefore if $(W_1-W_2)(\y_m)<0$, $\y_m$ cannot be a minimum point either.
\end{enumerate}

\item\label{pf-26-5-9-ii}
So, $W_1\ge W_2$ and therefore $\Delta (W_1-W_2) (x^*)=P'_B(W_1)-P'_B(W_2)\ge 0$
and  $W_1-W_2$ is a subharmonic function.  Then due to (\ref{26-5-25})  we conclude that $W_1-W_2\le \nu_1-\nu_2$ i.e. $W^\TF_{B,(\nu_1)} \le W^\TF_{B,(\nu_2)}$ at each point.

\item\label{pf-26-5-9-iii}
Statement~\ref{prop-26-5-9-iii} follows from Statements~\ref{prop-26-5-9-i} and~\ref{prop-26-5-9-ii}.
\end{enumerate}
\end{proof}

From Statements~\ref{prop-26-5-9-i} and~\ref{prop-26-5-9-iii} we conclude immediately that\enlargethispage{2\baselineskip}

\begin{corollary}\label{cor-26-5-10}
\begin{enumerate}[label=(\roman*), wide, labelindent=0pt]
\item\label{cor-26-5-10-i}
$\rho^\TF_{B(\nu)}(x)$ is non-decreasing with respect to $\nu$ at each point $x$.

\item\label{cor-26-5-10-ii}
$\rho^\TF_{B(\nu)}(x) \nearrow \rho^\TF_{B(0)} (x)$ at each point $x$ as
$\nu \nearrow 0$.
\end{enumerate}
\end{corollary}
\pagebreak

Therefore in the zone $\{x\in \cY\colon W^\TF_{B(\nu)}\ge (1+\epsilon)|\nu|\}$ we can apply the same $(\gamma,\varsigma)$ scaling with $\varsigma=\gamma^2$ defined for $\nu =0$. Indeed, we know that there
$W^\TF _{B(\nu)}+\nu \asymp W^\TF _{B(0)} \asymp \varsigma^2$ and $\varsigma=\gamma^2$.

Then Thomas-Fermi equation (\ref{26-2-3}) implies that
\begin{equation}
|\nabla^\alpha W^\TF_{B(\nu)}|\le C_\alpha \varsigma^2\gamma^{-|\alpha|}
\qquad \forall \alpha
\label{26-5-26}
\end{equation}
and then we arrive to the Statement~\ref{prop-26-5-11-i} in Proposition~\ref{prop-26-5-11} below.
On the other hand, in the zone
$\{x \colon W^\TF_{B(\nu)}\le (1-\epsilon)|\nu|\}$ we can apply the same arguments but this zone is classically forbidden and we arrive to  Statement~\ref{prop-26-5-11-ii} below. In both cases
$\varsigma \gamma \ge h $ (where in the latter case $\gamma$ is the distance from $x$ to $W^\TF_{B(\nu)}$ (scaled) and $\varsigma=|\theta|^{\frac{1}{2}}$ in virtue of Remark~\ref{rem-26-5-8}.

\begin{proposition}\label{prop-26-5-11}
Let \underline{either} $B\le Z^{\frac{4}{3}}$ and
$|\nu|^{\frac{3}{4}}\ge Z^{\frac{2}{3}}$
\underline{or} $ Z^{\frac{4}{3}}\le B\le Z^3$ and
$|\nu|^{\frac{3}{4}}\ge B^{\frac{1}{2}}$. Then

\begin{enumerate}[label=(\roman*), wide, labelindent=0pt]
\item\label{prop-26-5-11-i}
Contributions of zone $\{x\colon W^\TF _B(x)\ge (1+\epsilon_0) |\nu|\}$ to the semiclassical errors in $\N$- and $\T$-terms and into semiclassical $\D$-term do not exceed $CR_0L$, $CR$ and $CRL^2$ respectively.

\item\label{prop-26-5-11-ii}
Contributions of zone $\{x\colon W^\TF _B(x)\le (1-\epsilon_0) |\nu|\}$ to the semiclassical errors in $\N$- and $\T$-terms and into semiclassical $\D$-term do not exceed $CR_0L$, $CR$ and $CRL^2$ respectively.\enlargethispage{2\baselineskip}
\end{enumerate}
\end{proposition}

\begin{remark}\label{rem-26-5-12}
Here actually we can replace $L$ by $L_*= 1+|\log \theta|$ with
\begin{equation}
\theta=|\nu| G^{-1}\asymp
\left\{\begin{aligned}
&(Z-N)_+ B^{-\frac{3}{4}} \qquad
&&\text{if\ \ } (Z-N)_+^{\frac{4}{3}}\le B\le Z^{\frac{4}{3}},\\
&(Z-N)_+Z^{-1} \qquad
&&\text{if\ \ } Z^{\frac{4}{3}}\le B\le Z^3;
\end{aligned}\right.
\label{26-5-27}
\end{equation}
\end{remark}

Therefore we need to explore the following zone
\begin{equation*}
\cY^*\coloneqq \{x\colon (1-\epsilon_0) |\nu| \le W^\TF _B(x)\le (1+\epsilon_0) |\nu|\}
\end{equation*}
in the framework of Proposition~\ref{prop-26-5-11}. In virtue of Remark~\ref{rem-26-5-8} $\hbar \lesssim 1$ where
\begin{equation}
\hbar = h \theta^{-\frac{3}{4}}\asymp
\left\{\begin{aligned}
& (Z-N)_+^{-\frac{3}{4}}B^{\frac{5}{16}} \qquad
&&\text{if\ \ } B\le Z^{\frac{4}{3}},\\
& (Z-N)_+^{-\frac{3}{4}}Z^{\frac{3}{20}}B^{\frac{1}{5}} \qquad
&&\text{if\ \ } Z^{\frac{4}{3}}\le B\le Z^3.
\end{aligned}\right.
\label{26-5-28}
\end{equation}

Let us rescale the ball $B(.,\alpha)$ to $B(.,1)$ by $x\mapsto x\alpha^{-1}$ with $\alpha=\theta^{\frac{1}{4}}$ (after we already rescaled
$x\mapsto x\bar{r}^{-1}$). After this let us introduce scaling function $\gamma_0$ by (\ref{26-4-3}). Then let us introduce consequently scaling functions $\gamma_1$ by (\ref{26-4-7}), $\gamma_2$ by (\ref{26-4-13}) and $\gamma_3$ by (\ref{26-4-14})\,\footnote{\label{foot-26-30} With $j=\bar{\jmath}=0$ and corrected as in \ref{26-4-3-*} and \ref{26-4-7-*}.}.

Consider contributions of different balls in this hierarchy into semiclassical and approximation errors in $\N$- and $\T$-terms and into $\D$ semiclassical and approximation $\D$-terms.\enlargethispage{\baselineskip}

\begin{enumerate}[label=(\roman*), wide, labelindent=0pt]
\item\label{sect-26-5-3-1-i}
Consider first semiclassical error in $\N$-term. Due to Chapter~\ref{book_new-sect-18} the contribution of $\alpha_j$ element  does not exceed
$CB\ell_j^2=CB\bar{r}^2 \alpha_j^2$ for $j=3,2$, where
recall that $\alpha_j=\gamma_0\cdots \gamma_j$.

Then for $j=2$ we have
$CB\bar{r}^2 \alpha_2^2=CB\bar{r}^2 \alpha_1^2\gamma_1^2$ and therefore we estimate the contribution of $\alpha_1$ element by
$CB\bar{r}^2 \alpha_1^2\int \gamma_2^{-1}\,dx$\,\footnote{\label{foot-26-31} With the integral calculated in the scaled coordinates.}, which also results in $CB\bar{r}^2 \alpha_1^2$ but with the logarithmic factor. However we can get rid of this factor due to a simple observation:

\begin{claim}\label{26-5-29}
If $\gamma_2 \le \epsilon$ then $\Hess (w_1)$ has at least two eigenvalues of magnitude $1$ due to $|\Delta w_1|\le \epsilon_1$.
\end{claim}

Then the contribution of $\alpha_0$ element does not exceed
$CB\bar{r}^2 \alpha_0^2\int \gamma_1^{-1}\,dx$\,\footref{foot-26-31}; we claim that it is $CB\bar{r}^2 \alpha_0^2$. Indeed,  we need to consider only points with $\gamma_1\le \epsilon$ and there we use a similar observation:

\begin{claim}\label{26-5-30}
If $\gamma_1 \le \epsilon$ then $|\nabla^3 w'|\asymp 1$ and also
$|\nabla^3 w'-e \otimes e \otimes e|\ge c^{-1}$ for any $e\in \mathbb{R}^3$ due to $|\partial_j\Delta w_1|\le \epsilon_1$; here $\nabla^3 w'$ is a $3$-tensor of the third derivatives of $w'$.
\end{claim}

Further, the contribution of $\alpha$ element does not exceed
$CB\bar{r}^2 \alpha^2\int \gamma_0^{-1}\,dx$\,\footref{foot-26-31}; since $\gamma\ge \bar{\gamma}_0 =\hbar ^{\frac{1}{3}}$, we estimate it by
$CB\bar{r}^2 \alpha^2 \hbar^{-\frac{1}{3}}$.

Finally, since $\bar{r}^{-1}\cY^*$  is covered by no more than
$CL_* \alpha^{-2}$ such elements\footnote{\label{foot-26-32} Indeed, due to Subsection~\ref{sect-26-5-1} $\mes (\bar{r}^{-1}\cY^*)\le C\alpha L_*$.}, we conclude that

\begin{claim}\label{26-5-31}
The total contribution of $\cY^*$ into the semiclassical (and also approximation) errors in $\N$-term does not exceed
$CB\bar{r}^2 \hbar ^{-\frac{1}{3}}L_*$, where
$L_*\coloneqq (1+|\log \theta|)$.
\end{claim}

Plugging values of $\hbar $ and $\theta$, we arrive to expression (\ref{26-5-33}) in Proposition~\ref{prop-26-5-13}\ref{prop-26-5-13-i} below.\pagebreak

\item\label{sect-26-5-3-1-ii}
Similarly, in virtue of Subsubsection~\emph{\ref{sect-26-3-1-2}.2. \nameref{sect-26-3-1-2}\/} we know the that the contribution of the non-disjoint pair  of $\alpha_j$-elements to the semiclassical $\D$-term does not exceed $CB^2 \ell^3=CB^2\bar{r}^3\alpha_j^3$ for $j=3,2$.

Therefore the contribution of all non-disjoint pairs of $\alpha_2$ subelements to the same expression for $\alpha_1$ element does not exceed $CB^2\bar{r}^3\alpha_1^3$. Adding all disjoint pairs, we get\footref{foot-26-31}
\begin{equation}
CB^2\bar{r}^3\alpha_1^3\iint |x-y|^{-1}\gamma_2(x)^{-1}\gamma_2(y)^{-1}\,dxdy.
\label{26-5-32}
\end{equation}
Then using the results of Part~\ref{sect-26-5-3-1-i} together with observation (\ref{26-5-29}) we arrive to $CB^2\bar{r}^3\alpha_1^3$. So, contribution of the non-disjoint pair  of $\alpha_1$-elements to the semiclassical $\D$-term does not exceed $CB^2 \ell^3=CB^2\bar{r}^3\alpha_1^3$.

Further, continuing in the same manner, we estimate the contribution of the non-disjoint pair of $\alpha_0$-elements by $CB^2\bar{r}^3\alpha_0^3$.

Furthermore, in the same manner we estimate the contribution of the non-disjoint pair of $\alpha$-elements by expression (\ref{26-5-32}) with $\gamma_2$ replaced by $\gamma_0$, which does not exceed
$CB^2\bar{r}^3\alpha^3\hbar ^{-\frac{2}{3}}$.\enlargethispage{\baselineskip}

Finally, adding contribution of all non disjoint pairs and using results of Part~\ref{sect-26-5-3-1-i}, we conclude that the total contribution of $\cY^*\times \cY^*$ into the semiclassical (and also approximation) $\D$-terms does not exceed the final expression we recovered there, squared and multiplied by $\bar{r}^{-1}$, i.e.  $CB^2\bar{r}^3 \hbar ^{-\frac{2}{3}}L_*^2$.

Plugging values of $\hbar$ and $\theta$ we arrive to expression (\ref{26-5-34}) in Proposition~\ref{prop-26-5-13}\ref{prop-26-5-13-ii} below.

\item\label{sect-26-5-3-1-iii}
Due to Chapter~\ref{book_new-sect-18} the contribution of $\alpha_j$ element to the semiclassical error in $\T$-term does not exceed $CB\ell_* \zeta_*$ as $j=3,2$. Note that $\zeta=
G^{\frac{1}{2}}\gamma_0^2\gamma_1^{\frac{3}{2}}\gamma_2 \gamma_3^{\frac{1}{2}}$ and $\zeta= G^{\frac{1}{2}}\gamma_0^2\gamma_1^{\frac{3}{2}}\gamma_1$ for $j=3,2$. Here we took $\theta=\alpha=1$ thus covering the whole zone $\cY$.

Then the contribution of $\alpha_2$-element does not exceed
$CBG^{\frac{1}{2}}\bar{r}\gamma_0^3\gamma_1^{\frac{5}{2}}\gamma_2^2$.
Further, the contribution of $\alpha_1$-element does not exceed
$CBG^{\frac{1}{2}}\bar{r}\gamma_0^3\gamma_1^{\frac{5}{2}}
\int \gamma_2^{-1}\,dx$\,\footref{foot-26-31}, resulting in $CBG^{\frac{1}{2}}\bar{r}\gamma_0^3\gamma_1^{\frac{5}{2}}$ in virtue of the same observation (\ref{26-5-29}).

Further, the contribution of $\alpha_0$-element does not exceed $CBG^{\frac{1}{2}}\bar{r}\gamma_0^3\int \gamma_1^{-\frac{1}{2}}$ resulting in $CBG^{\frac{1}{2}}\bar{r}\gamma_0^3$ in virtue of the same observation (\ref{26-5-30}).

Finally, the total contribution of $\cY$ does not exceed
$CBG^{\frac{1}{2}}\bar{r}=CB^2\bar{r}^3=\max (B^{\frac{5}{4}}, Z^{\frac{3}{5}} B^{\frac{4}{5}})$.
\end{enumerate}

Therefore we arrive to

\begin{proposition}\label{prop-26-5-13}
In the framework of Proposition~\ref{prop-26-5-13} there exists potential $W_\varepsilon$ such that
\begin{enumerate}[label=(\roman*), wide, labelindent=0pt]
\item\label{prop-26-5-13-i}
Contributions of $\cY^*$ to both semiclassical and approximation errors for $\N$-term do not exceed
\begin{equation}
(Z-N)_+^{\frac{1}{4}}L\times
(B^{\frac{19}{48}};\, Z^{\frac{7}{20}}B^{\frac{2}{15}}),
\label{26-5-33}
\end{equation}
where here and below we list different values for
$(Z-N)_+^{\frac{4}{3}}\le B\le Z^{\frac{4}{3}}$ and for
$Z^{\frac{4}{3}}\le B\le Z^3$.

\item\label{prop-26-5-13-ii}
Contributions of $\cY^*\times \cY^*$ to both semiclassical and approximation $\D$-terms do not exceed
\begin{equation}
(B^{\frac{25}{24}}; \, Z^{\frac{1}{2}}B^{\frac{2}{3}}).
\label{26-5-34}
\end{equation}

\item\label{prop-26-5-13-iii}
Contributions of $\cY^*$ to both semiclassical and approximation errors for $\T$-term do not exceed
\begin{equation}
(B^{\frac{5}{4}};\, Z^{\frac{3}{5}} B^{\frac{4}{5}}).
\label{26-5-35}
\end{equation}
\end{enumerate}
\end{proposition}

\subsection{Case $B\le (Z-N)_+^{\frac{4}{3}}$}
\label{sect-26-5-3-2}

Now let us consider the case $B\le (Z-N)_+^{\frac{4}{3}}$. In this case the \emph{boundary strip\/}
\begin{gather}
\cY\coloneqq \{x\colon |W(x)+\nu|\le \epsilon |\nu|\}
\label{26-5-36}\\
\intertext{consists of two subzones}
\cY_1 \coloneqq \{x\colon \epsilon B\le |W(x)+\nu|\le \epsilon |\nu|\}
\label{26-5-37}\\
\shortintertext{and}
\cY^* \coloneqq \{x\colon |W(x)+\nu|\le \epsilon B \}.
\label{26-5-38}
\end{gather}

Applying arguments of Section~\ref{sect-26-4} (more precisely, analysis in zones $\cX_1$, $\cX_1$ extended and $\cX_2$) one can prove easily that

\begin{proposition}\label{prop-26-5-14}
Let $B\le (Z-N)_+^{\frac{4}{3}}$. Then
\begin{enumerate}[label=(\roman*), wide, labelindent=0pt]
\item\label{prop-26-5-14-i}
Contributions of $\cY_1$ into semiclassical and approximation errors in $\N$-term do not exceed $C|\nu|\bar{r}^2\asymp C(Z-N)_+^{\frac{2}{3}}$.

\item\label{prop-26-5-14-ii}
Contributions of $\cY_1\times \cY_1$ into semiclassical and approximation $\D$-terms do not exceed $C|\nu|^2\bar{r}^3\asymp C(Z-N)_+^{\frac{5}{3}}$.
\pagebreak
\item\label{prop-26-5-14-iii}
Contribution of $\cY_1$ into semiclassical and approximation errors in $\T$-term do not exceed
$C|\nu|^{\frac{3}{2}}\bar{r}\asymp C(Z-N)_+^{\frac{5}{3}}$.
\end{enumerate}
\end{proposition}

\begin{proof}
We leave easy details to the reader. \end{proof}

On the other hand, applying arguments of the previous Subsubsection~\emph{\ref{sect-26-5-3-1}.1. \nameref{sect-26-5-3-1}\/} with $\theta=1$,
$\hbar =|\nu|^{-\frac{1}{2}}\bar{r}^{-1}\asymp (Z-N)_+^{-\frac{1}{3}}$ one can prove easily the following

\begin{proposition}\label{prop-26-5-15}
Let $B\le (Z-N)_+^{\frac{4}{3}}$. Then
\begin{enumerate}[label=(\roman*), wide, labelindent=0pt]
\item\label{prop-26-5-15-i}
Contributions of $\cY_2$ into semiclassical and approximation errors in $\N$-term do not exceed $C(Z-N)_+^{-\frac{5}{9}}B$.

\item\label{prop-26-5-15-ii}
Contributions of $\cY_2\times \cY_2$ into semiclassical and approximation $\D$-terms do not exceed $C(Z-N)_+^{-\frac{7}{9}}B^2$.

\item\label{prop-26-5-15-iii}
Contribution of $\cY_2$ into semiclassical and approximation errors in $\T$-term do not exceed $C(Z-N)_+^{\frac{5}{3}}$.
\end{enumerate}
\end{proposition}

\begin{proof}
We leave easy details to the reader.\end{proof}

\section{Summary}
\label{sect-26-5-4}
Adding contributions of all other zones we arrive to
\begin{proposition}\label{prop-26-5-16}
Let $M\ge 2$. Then for the constructed potential $W$
\begin{enumerate}[label=(\roman*), wide, labelindent=0pt]
\item\label{prop-26-5-16-i}
Total semiclassical and approximation errors in $\N$-term do not exceed
\begin{equation}
C\left\{\begin{aligned}
&CZ^{\frac{2}{3}}+(Z-N)_+^{-\frac{5}{9}}B \quad
&&\text{if\ \ } B\le (Z-N)_+^{\frac{4}{3}},\\
&Z^{\frac{2}{3}}+B^{\frac{1}{2}}L + (Z-N)_+^{\frac{1}{4}}B^{\frac{19}{48}}L_* \quad
&&\text{if\ \ } (Z-N)_+^{\frac{4}{3}}\le B\le Z^{\frac{4}{3}},\\
&Z^{\frac{2}{5}}B^{\frac{1}{5}}L +
(Z-N)_+^{\frac{1}{4}}Z^{\frac{7}{20}}B^{\frac{2}{15}}L_* \quad
&&\text{if\ \ } Z^{\frac{4}{3}}\le B\le Z^3
\end{aligned}\right.
\label{26-5-39}
\end{equation}
where $L_*=(1+|\log \theta|)$ with
$\theta=|\nu| G^{-1}=(Z-N)_+\cdot \max(B^{-\frac{3}{4}},Z^{-1})$ and
$L= (1+|\log BZ^3|)$.

\item\label{prop-26-5-16-ii}
Both semiclassical and approximation $\D$-terms do not exceed
\begin{equation}
C\left\{\begin{aligned}
&Z^{\frac{5}{3}}+(Z-N)_+^{-\frac{7}{9}}B^2 \quad
&&\text{if\ \ } B\le (Z-N)_+^{\frac{4}{3}},\\
&Z^{\frac{5}{3}}+B^{\frac{5}{4}}L^2+ (Z-N)_+^{\frac{1}{2}}B^{\frac{25}{24}}L_*^2 \quad
&&\text{if\ \ } (Z-N)_+^{\frac{4}{3}}\le B\le Z^{\frac{4}{3}},\\
&Z^{\frac{3}{5}}B^{\frac{4}{5}}L^2 +
(Z-N)_+^{\frac{1}{2}}Z^{\frac{1}{2}}B^{\frac{2}{3}}L_*^2\quad
&&\text{if\ \ } Z^{\frac{4}{3}}\le B\le Z^3.
\end{aligned}\right.
\label{26-5-40}
\end{equation}

\item\label{prop-26-5-16-iii}
Total approximation error in $\T$-term does not exceed
\begin{equation}
CQ\coloneqq C\max\bigl(Z^{\frac{5}{3}},\, Z^{\frac{3}{5}}B^{\frac{4}{5}}\bigr) = C\left\{\begin{aligned}
&Z^{\frac{5}{3}}\quad
&&\text{if\ \ } B\le Z^{\frac{4}{3}},\\
&Z^{\frac{3}{5}}B^{\frac{4}{5}} \quad
&&\text{if\ \ } Z^{\frac{4}{3}}\le B\le Z^3.
\end{aligned}\right.
\label{26-5-41}
\end{equation}

\item\label{prop-26-5-16-iv}
Total semiclassical error in $\T$-term does not exceed
\begin{equation}
CQ + CZ^{\frac{4}{3}}B^{\frac{1}{3}} + CZ^{\frac{3}{2}}a^{-\frac{1}{2}}
\label{26-5-42}
\end{equation}
provided $a\ge Z^{-1}$; for $a\le Z^{-1}$ the last term should be replaced by $CZ^2$.
\end{enumerate}
\end{proposition}

Also we arrive to
\begin{proposition}\label{prop-26-5-17}
Let $M\ge 2$, $B\le Z$ and $a\ge Z^{-\frac{1}{3}}$. Then for the constructed potential $W$
\begin{enumerate}[label=(\roman*), wide, labelindent=0pt]
\item\label{prop-26-5-17-i}
Total semiclassical and approximation errors in $\N$-term do not exceed
$CZ^{\frac{2}{3}}\bigl( (BZ^{-1})^\delta + (aZ^{\frac{1}{3}})^{-\delta} + Z^{-\delta}\bigr)$.

\item\label{prop-26-5-17-ii}
Both semiclassical and approximation $\D$-terms and semiclassical and approximation errors in $\T$-term do not exceed
$CZ^{\frac{5}{3}}\bigl( (BZ^{-1})^\delta + (aZ^{\frac{1}{3}})^{-\delta} + Z^{-\delta}\bigr)$.
\end{enumerate}
\end{proposition}

\chapter{Ground state energy}
\label{sect-26-6}

\section{Lower estimates}
\label{sect-26-6-1}
Now the lower estimates for the ground state energy $\E_N$ are already proven: in virtue of the analysis given in Subsection~\ref{book_new-sect-25-2-1} we know that
\begin{equation}
\E_N \ge \Phi_*(W) + \Bigl(\Tr (H_{A,W}-\nu)^- +\int P_B(W +\nu) \,dx\Bigr)
-\nu N
\label{26-6-1}
\end{equation}
for \emph{arbitrary\/} potential $W$ and $\nu\le 0$; picking Thomas-Fermi potential $W=W^\TF_B$ and chemical potential $\nu$, we arrive to estimate
(\ref{26-6-2}) below with $W=W^\TF_B$ and $Q=0$.

However we use slightly different potential $W$ and arrive to estimate (\ref{26-6-2}) below where $CQ$, defined by (\ref{26-5-42}), estimates an approximation error; replacing $\T$-term by its semiclassical approximation and applying Proposition~\ref{prop-26-5-16}\ref{prop-26-5-16-iii} and~\ref{prop-26-5-17}\ref{prop-26-5-16-ii}, we arrive to estimates (\ref{26-6-3})--(\ref{26-6-6}) below:

\begin{proposition}\label{prop-26-6-1}
Let $B\le Z^3$. Then
\begin{enumerate}[label=(\roman*), wide, labelindent=0pt]
\item\label{prop-26-6-1-i}
The following estimate holds with an approximate potential $W$ we constructed:
\begin{equation}
E^\TF \ge
\cE^\TF + \Bigl(\Tr ((H_{A,W}-\nu)^-) +\int P_B(W +\nu) \,dx\Bigr) - CQ
\label{26-6-2}
\end{equation}
with $Q$ defined by \textup{(\ref{26-5-41})}; further, for $W=W^\TF_B$ this estimate holds with $Q=0$.

\item\label{prop-26-6-1-ii}
The following estimates hold for $M=1$ and $M\ge 2$ respectively
\begin{align}
&E^\TF \ge
\cE^\TF + \Scott - CQ-CZ^{\frac{4}{3}}B^{\frac{1}{3}}
\label{26-6-3}\\
\shortintertext{and}
&E^\TF \ge
\cE^\TF + \Scott - CQ - CZ^{\frac{4}{3}}B^{\frac{1}{3}} -
CZ^{\frac{3}{2}}a^{-\frac{1}{2}}
\label{26-6-4}
\end{align}
provided $a\ge Z^{-1}$\,\footnote{\label{foot-26-33} Recall that $a$ is the minimal distance between nuclei.} and $B\le Z^2$; on the other hand, if
$a\le Z^{-1}$, we can skip $\Scott$ and replace the last term in (\ref{26-4-2}) by $CZ^2$.

\item\label{prop-26-6-1-iii}
As $B\le Z$ the following estimates hold for $M=1$ and
$M\ge 2$, $a\ge Z^{-\frac{1}{3}}$ respectively
\begin{equation}
E^\TF \ge
\cE^\TF + \Scott + \Dirac +\Schwinger -
C Z^{\frac{5}{3}}\bigl(Z^{-\delta}+ (BZ^{-1})^\delta\bigr)
\label{26-6-5}
\end{equation}
and
\begin{multline}
E^\TF \ge
\cE^\TF + \Scott + \Dirac +\Schwinger -\\
C Z^{\frac{5}{3}}\bigl(Z^{-\delta}+ (BZ^{-1})^\delta+(aZ^{\frac{1}{3}})^{-\delta}\bigr).
\label{26-6-6}
\end{multline}
\end{enumerate}
\end{proposition}

\section{Upper estimate: general scheme}
\label{sect-26-6-2}

On the other hand, the upper estimate is more demanding. Recall that, according to Subsection~\ref{book_new-sect-25-2-2}, for the upper estimate in addition to the trace we need to estimate also $|\lambda_N-\nu|$ where $\lambda_N<0$ is $N$-th eigenvalue of $H_{A,W}$ and $\lambda_N=0$ if $H_{A,W}$ has less than $N$ negative eigenvalues, and the product
\begin{equation}
|\lambda_N-\nu|\cdot |\N(H_{A,W}) -N|
\label{26-6-7}
\end{equation}
and also three $\D$-terms: two of them are semiclassical: \begin{phantomequation}\label{26-6-8}\end{phantomequation}
\begin{gather}
\D\Bigl( e(x,x,\lambda) - P'_B(W(x)+\lambda), \,
e(x,x,\lambda) - P'_B(W(x)+\lambda)\Bigr)
\tag*{$\textup{(\ref*{26-6-8})}_{1,2}$}\label{26-6-8-*}\\
\intertext{with $\lambda=\nu$ and $\lambda=\lambda_N$ and also}
\D\Bigl( P'_B(W(x)+\lambda_N) - P'_B(W(x)+\nu), \,
P'_B(W(x)+\lambda_N) - P'_B(W(x)+\nu)\Bigr).
\label{26-6-9}
\end{gather}
For this purpose our tool will be semiclassical estimates for two semiclassical $\N$-terms
\begin{phantomequation}\label{26-6-10}\end{phantomequation}
\begin{gather}
\int \Bigl( e(x,x,\lambda) - P'_B(W(x)+\lambda)\Bigr)\,dx
\tag*{$\textup{(\ref*{26-6-10})}_{1,2}$}\label{26-6-10-*}\\
\intertext{also with $\lambda=\nu$ and $\lambda=\lambda_N$ and also estimate from below for the third $\N$-term}
|\int \Bigl( P'_B(W(x)+\lambda_N) - P'_B(W(x)+\nu)\Bigr)\, dx|.
\label{26-6-11}
\end{gather}

\section{Upper estimate as $M=1$}
\label{sect-26-6-3}

\subsection{Estimate for $|\lambda_N-\nu|$}
\label{sect-26-6-3-1}

We start from the easier case $M=1$. Exactly as in Subsection~\ref{book_new-sect-25-2-2} we have two cases: in the first case $|\nu|$ is small enough so we construct $W^\TF$ with $\nu=0$ and estimate $|\lambda_N|$, and in the second case  we prove that $\lambda_N\asymp \nu$ and estimate
$|\lambda_N-\nu|\le \epsilon |\nu|$.

\begin{proposition}\label{prop-26-6-2}
Let $M=1$, $B\le Z^3$.
\begin{enumerate}[label=(\roman*), wide, labelindent=0pt]
\item\label{prop-26-6-2-i}
Assume first that
\begin{gather}
(Z-N)_+ \le K\coloneqq C_0\max\bigl(Z^{\frac{2}{3}},\,Z^{\frac{2}{5}}B^{\frac{1}{5}}\bigr)
\label{26-6-12}\\
\intertext{and let us construct $W$ as if $\nu=0$ i.e. $N=Z$. Then}
|\lambda_N|\le
C_1\max\bigl(Z^{\frac{8}{9}},\, B^{\frac{2}{3}}\bigr).
\label{26-6-13}
\end{gather}

\item\label{prop-26-6-2-ii}
Assume now that
\begin{equation}
(Z-N)_+ \ge K= C_0\max\bigl(Z^{\frac{2}{3}},\,Z^{\frac{2}{5}}B^{\frac{1}{5}}\bigr)
\label{26-6-14}
\end{equation}
with sufficiently large $C_0$. Then $\lambda_N\asymp \nu$ and
\begin{equation}
|\lambda_N-\nu |\le C_1\max\bigl( Z^{\frac{2}{3}}, \, B^{\frac{1}{2}}\bigr)|\nu|^{\frac{1}{4}}.
\label{26-6-15}
\end{equation}
\end{enumerate}
\end{proposition}

\begin{proof}
\begin{enumerate}[label=(\roman*), wide, labelindent=0pt]
\item\label{pf-26-6-2-i}
In the framework of Statement~\ref{prop-26-6-2-i} assume first that
 $B\ge (Z-N)_+^{\frac{4}{3}}$.  One can see easily that then
\begin{claim}\label{26-6-16}
Expression (\ref{26-6-11}) is
\begin{equation*}
\asymp B |\lambda_N|^{\frac{1}{2}} \times
\underbracket{\Bigl(\frac{|\lambda_N|}{G}\Bigr)^{\frac{1}{4}} \bar{r}^3}
\asymp
|\lambda_N|^{\frac{3}{4}} \min \bigl(1,\, B^{-\frac{3}{10}}Z^{\frac{2}{5}}\bigr)
\end{equation*}
\end{claim}
\vspace{-10pt}
where $(|\lambda_N|/G)^{\frac{1}{4}}\bar{r}$ is a width of the zone where
$0<W\le -\lambda_N$ and the selected factor is the volume of this zone. Indeed,  $W\asymp (\bar{r}-|x|)_+^4 \bar{r}^{-4} G$ for $|x|\asymp \bar{r}$.

However this expression (\ref{26-6-11}) should be less than
$C\max \bigl(Z^{\frac{2}{3}},\, Z^{\frac{2}{5}}B^{\frac{1}{5}}\bigr)$ which is exactly an error estimate in the semiclassical expression for $N$. Thus
\begin{equation}
|\lambda_N|^{\frac{3}{4}} \min \bigl(1,\, Z^{\frac{2}{5}}B^{-\frac{3}{10}}\bigr) \le C\max \bigl(Z^{\frac{2}{3}},\, Z^{\frac{2}{5}}B^{\frac{1}{5}}\bigr)
\label{26-6-17}
\end{equation}
where everywhere the first and the second cases are as $B\le Z^{\frac{4}{3}}$ and $Z^{\frac{4}{3}}\le B \le Z^3$ respectively. The last inequality is equivalent to (\ref{26-6-13}).

On the other hand, if  $B\le (Z-N)_+^{\frac{4}{3}}$, inequality (\ref{26-6-17}) is replaced by
$ |\lambda_N|^{\frac{3}{4}}\le CZ^{\frac{2}{3}}$ which coincides with (\ref{26-6-17}) with $B$ reset to $(Z-N)_+^{\frac{4}{3}}$ and also with the same inequality derived for $B=0$ in Subsection~\ref{book_new-sect-25-4-2}; therefore (\ref{26-6-13}) holds in this case as well.

\item\label{pf-26-6-2-ii}
One can prove easily that
\begin{claim}\label{26-6-18}
If condition (\ref{26-6-14}) is fulfilled, and expression (\ref{26-6-11}) does not exceed the semiclassical error estimate
$C_0\max(Z^{\frac{2}{3}}, Z^{\frac{2}{5}}B^{\frac{1}{5}})$, then $\lambda_N\asymp \nu$ and, furthermore, expression (\ref{26-6-11}) is
\begin{gather}
\asymp B |\lambda_N-\nu | \int P''_B(W+\nu)\,dx
\asymp B |\lambda_N-\nu|\int (W+\nu)_+^{-\frac{1}{2}}\,dx,
\label{26-6-19}\\
\intertext{which for $B\ge (Z-N)_+^{\frac{4}{3}}$ is }
\asymp
B \bar{r}^3 |\lambda_N-\nu|\cdot|\nu|^{-\frac{1}{2}}
\Bigl(\frac{|\nu|}{G}\Bigr)^{\frac{1}{4}} \asymp
|\lambda_N-\nu|\cdot|\nu|^{-\frac{1}{4}}
\min \bigl(1,\, B^{-\frac{3}{10}}Z^{\frac{2}{5}}\bigr)
\label{26-6-20}
\end{gather}
\end{claim}
\vspace{-10pt}
and this should be less than
$C\max \bigl(Z^{\frac{2}{3}},\, Z^{\frac{2}{5}}B^{\frac{1}{5}}\bigr)$, and this implies (\ref{26-6-15}).

On the other hand, if $B\le (Z-N)_+^{\frac{4}{3}}$, then the right-hand expression of (\ref{26-6-19}) is $\asymp |\lambda_N-\nu|\bar{r}\asymp |\lambda_N-\nu|(Z-N)_+^{-\frac{1}{3}}$ and this should be less than $CZ^{\frac{2}{3}}$, and this implies (\ref{26-6-15}) in this case as well.
\end{enumerate}\end{proof}

Proposition~\ref{prop-26-6-2} immediately implies

\begin{corollary}\label{cor-26-6-3}
In the frameworks of Proposition~\ref{prop-26-6-2}\ref{prop-26-6-2-i}, \ref{prop-26-6-2-ii},
\begin{equation}
|\lambda_N-\nu| \cdot \N ([\lambda_N, \nu]) \le CQ,
\label{26-6-21}
\end{equation}
where $\N(\lambda_N,\nu)$ is the number of (non-zero) eigenvalues on interval $[\lambda_N,\nu]$ or $[\nu,\lambda_N]$\,\footnote{\label{foot-26-34} Recall, that the frameworks of Proposition~\ref{prop-26-6-2}\ref{prop-26-6-2-i} we pick up $\nu=0$.}.
\end{corollary}

\subsection{Estimate for $\D$-terms}
\label{sect-26-6-3-2}

\begin{proposition}\label{prop-26-6-4}
In the frameworks of Proposition~\ref{prop-26-6-2}\ref{prop-26-6-2-i},\ref{prop-26-6-2-ii} expressions
\begin{gather}
\D \bigl(e(x,x,\lambda)-P'_B(W+\lambda),\,
e(x,x,\lambda)-P'_B(W+\lambda)\bigr) \label{26-6-22}\\
\intertext{with $\lambda=\nu$\,\footref{foot-26-34} and with $\lambda=\lambda_N$ and}
\D \bigl(P'_B(W+\nu)-P'_B(W+\lambda),\,
P'_B(W+\nu)-P'_B(W+\lambda_N)\bigr)\label{26-6-23}
\end{gather}
with $\lambda=\lambda_N$ do not exceed
$C\max(Z^{\frac{5}{3}}, Z^{\frac{3}{5}}B^{\frac{4}{5}})$.
\end{proposition}

\begin{proof}
Recall that we already derived in Section~\ref{sect-26-3} this estimate for $\D$-term (\ref{26-6-22}) with $\lambda=\nu$. Further, the same estimate for this term with $\lambda=\lambda_N$ can be proven exactly in the same way; we leave easy details to the reader.

Furthermore, one can derive the same estimate for $\D$-term (\ref{26-6-23})  using Proposition~\ref{prop-26-6-2}; again we leave easy details to the reader.
\end{proof}

\begin{remark}\label{rem-26-6-5}
Let $B\le Z$. Then in (\ref{26-6-12})--(\ref{26-6-15}) and therefore also in (\ref{26-6-19}) and in Proposition~\ref{prop-26-6-4} one can replace $C_0$ and $C$ by $C_0\varepsilon$ and $C\varepsilon$ respectively with the small parameter $\varepsilon$: $\max \bigl(Z^{-\delta}, (BZ^{-1})^\delta\bigr)\le \varepsilon \le 1$.
\end{remark}

\subsection{Summary}
\label{sect-26-6-3-3}

Then following the scheme of Subsection~\ref{book_new-sect-25-4-4} we arrive to upper estimates in Theorem~\ref{thm-26-6-6} below (lower estimates have been proven in Proposition~\ref{prop-26-6-1}). Furthermore, based on estimates (\ref{26-6-2}) and (\ref{26-6-24}) and the fact, that the left-hand term in (\ref{26-6-28}) should fit into the ``gap'' between them (see Section~\ref{book_new-sect-25-2}), we also arrive to Theorem~\ref{thm-26-6-7} below:\enlargethispage{\baselineskip}

\begin{theorem}\label{thm-26-6-6}
Let $M=1$, $B\le Z^3$. Then
\begin{enumerate}[label=(\roman*), wide, labelindent=0pt]
\item\label{thm-26-6-6-i}
The following estimate holds:
\begin{equation}
E^\TF \le
\cE^\TF + \Bigl(\Tr ((H_{A,W}-\nu)^-) +\int P_B(W^\TF(x) +\nu) \,dx\Bigr) + CQ
\label{26-6-24}
\end{equation}
with $Q= \max \bigl(Z^{\frac{5}{3}}, Z^{\frac{3}{5}}B^{\frac{4}{5}}\bigr)$.

\item\label{thm-26-6-6-ii}
The following estimate holds:
\begin{equation}
E^\TF \le
\cE^\TF + \Scott + CQ+CZ^{\frac{4}{3}}B^{\frac{1}{3}}.
\label{26-6-25}
\end{equation}
Here for $Z^2\le B\le Z^3$ one can skip $\Scott$.

\item\label{thm-26-6-6-iii}
If $B\le Z$, then
\begin{equation}
E^\TF \ge
\cE^\TF + \Scott + \Dirac +\Schwinger +
C Z^{\frac{5}{3}}\bigl(Z^{-\delta}+ (BZ^{-1})^\delta\bigr).
\label{26-6-26}
\end{equation}
\end{enumerate}
\end{theorem}

\begin{theorem}\label{thm-26-6-7}
Let $M=1$, $B\le Z^3$. Then
\begin{enumerate}[label=(\roman*), wide, labelindent=0pt]
\item\label{thm-26-6-7-i}
The following estimate holds:
\begin{equation}
\D( \rho_\Psi-\rho_B^\TF, \rho_\Psi-\rho_B^\TF) \le CQ.
\label{26-6-27}
\end{equation}
\item\label{thm-26-6-7-ii}
If $B\le Z$, then
\begin{equation}
\D( \rho_\Psi-\rho_B^\TF, \rho_\Psi-\rho_B^\TF) \le
C Z^{\frac{5}{3}}\bigl(Z^{-\delta}+ (BZ^{-1})^\delta\bigr).
\label{26-6-28}
\end{equation}
\end{enumerate}
\end{theorem}

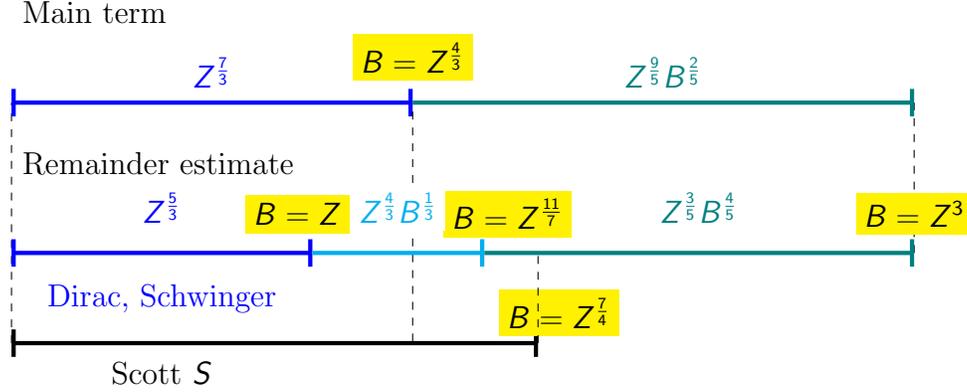
\begin{figure}[h]
\centering
\begin{tikzpicture}[scale=4]
\node[right] at (0,.8) {Main term};
\draw[dashed] (0,-.3)--(0,.5);
\draw[dashed] (3,0)--(3,.5);

\draw [ultra thick, blue, |-|](0,.5)--(4/3,.5);
\draw[dashed] (4/3,.47)--(4/3,-.3);
\node at (4/3,.65) {\colorbox{yellow}{$B=Z^{\frac{4}{3}}$}};
\node at (2/3,.6) {\color{blue}$Z^{\frac{7}{3}}$};
\draw [ultra thick, teal, -|](4/3,.5)--(3,.5);
\node at (13/6,.6) {\color{teal}$Z^{\frac{9}{5}}B^{\frac{2}{5}}$};
\node[right] at (0,.3) {Remainder estimate};
\draw [ultra thick, teal, -|](11/7,0)--(3,0);
\node at (3,.13) {\colorbox{yellow}{$B=Z^3$}};
\node at (16/7,.15) {\color{teal}$Z^{\frac{3}{5}}B^{\frac{4}{5}}$};
\draw [ultra thick,blue, |-|](0,0)--(1,0);
\node at (.95,.13) {\colorbox{yellow}{$B=Z$}};
\node at (.5,.15) {\color{blue}$Z^{\frac{5}{3}}$};
\node at (.5,-.15) {\color{blue}Dirac, Schwinger};
\draw [ultra thick,cyan, -|](1,0)--(11/7,0);
\node at (1.65,.13) {\colorbox{yellow}{$B=Z^{\frac{11}{7}}$}};
\node at (9/7,.15) {\colorbox{white}{\color{cyan}$Z^{\frac{4}{3}}B^{\frac{1}{3}}$}};
\draw [ultra thick, |-|](0,-.3)--(7/4,-.3);
\node at (1.82,-.2) {\colorbox{yellow}{$B=Z^{\frac{7}{4}}$}};
\draw [dashed] (7/4,-.3)--(7/4,0);
\node at (.5,-.4) {Scott $S$};
\end{tikzpicture}
\caption{\label{fig-26-3} This figure illustrates the remainder estimate for $\E_N$. Thresholds $B=Z^\star$  are shown in the yellow boxes.}
\end{figure}

\section{Upper estimate as $M\ge 2$}
\label{sect-26-6-4}

\subsection{Estimate for $|\lambda_N-\nu|$}
\label{sect-26-6-4-1}

Again we need to consider two cases: \emph{almost neutral molecules\/} (systems) when $(Z-N)_+\le C _0K$ with $K$ slightly redefined below and we can set $\nu=0$ in the definition of Thomas-Fermi potential and establish estimate for $|\lambda_N|$ (and for optimal $\nu$ we have the same estimate for both $|\nu|$ and $\lambda_N$) and \emph{not almost neutral molecules\/} (systems) when $(Z-N)_+\ge C_0K$ and we can prove that $|\lambda_N|\asymp |\nu|$ and estimate $|\lambda_N-\nu|$.

\begin{proposition-foot}\label{prop-26-6-8}\footnotetext{\label{foot-26-35}
Cf. Proposition~\ref{prop-26-6-2}.}
Let $M\ge 1$, $B\le Z^3$ and condition \textup{(\ref{26-2-28})} be fulfilled.

\begin{enumerate}[label=(\roman*), wide, labelindent=0pt]
\item\label{prop-26-6-8-i}
Assume first that
\begin{align}
(Z-N)_+ \le K\coloneqq C_0&\left\{\begin{aligned}
&Z^{\frac{2}{3}}+B^{\frac{1}{2}}L\qquad
&&\text{if\ \ }B\le Z^{\frac{4}{3}},\\
&Z^{\frac{1}{5}}B^{\frac{4}{15}} L\qquad
&&\text{if\ \ } Z^{\frac{4}{3}}\le B\le Z^3
\end{aligned}\right.
\label{26-6-29}\\
\intertext{and let us construct $W$ as if $\nu=0$ i.e. $N=Z$. Then}
|\lambda_N|\le C_1&\left\{\begin{aligned}
&Z^{\frac{8}{9}}+B^{\frac{2}{3}}L^{\frac{4}{3}}\qquad
&&\text{if\ \ }B\le Z^{\frac{4}{3}},\\
&Z^{\frac{1}{5}}B^{\frac{4}{15}} L^{\frac{4}{3}}\qquad
&&\text{if\ \ } Z^{\frac{4}{3}}\le B\le Z^3;
\end{aligned}\right.
\label{26-6-30}
\end{align}
recall that $L=|\log BZ^{-3}|$.

\item\label{prop-26-6-8-ii}
Assume now that
\begin{equation}
(Z-N)_+ \ge K
\label{26-6-31}
\end{equation}
with sufficiently large $C_0$ in the definition of $K$. Then
$\lambda_N\asymp \nu$ and moreover
\begin{gather}
|\lambda_N-\nu|\le C\max\bigl(Z^{\frac{2}{3}},B^{\frac{1}{2}}L_1\bigr)|\nu|^{\frac{1}{4}},
\label{26-6-32}\\
\shortintertext{where}
L_1=\left\{\begin{aligned}
&1\qquad
&&\text{if\ \ } B\le (Z-N)_+^{\frac{4}{3}},\\
&|\log ((Z-N)_+/B^{\frac{3}{4}}|\qquad
&&\text{if\ \ } (Z-N)_+^{\frac{4}{3}}\le B\le Z^{\frac{4}{3}},\\
&|\log (Z-N)_+/B^{\frac{4}{5}}Z^{\frac{3}{5}}|) \qquad
&&\text{if\ \ } Z^{\frac{4}{3}}\le B\le Z^3.
\end{aligned}\right.
\label{26-6-33}
\end{gather}

\item\label{prop-26-6-8-iii}
For $M=1$ one can take $L=L_1=1$.
\end{enumerate}
\end{proposition-foot}

\begin{proof}
We will apply arguments slightly more sophisticated than the obvious ones, used in the proof of Proposition~\ref{prop-26-6-2}. These better arguments will allow us to derive slightly better estimates for $|\lambda_N-\nu|$ as $(Z-N)_+\ge CK$, and for threshold $K$ itself.

Recall that estimates for $|\lambda_N-\nu|$ are derived by comparison of expression (\ref{26-6-11}) and the semiclassical errors for the number of eigenvalues below $\lambda=\nu$ and $\lambda=\lambda_N$: expression (\ref{26-6-11}) should be less than the sum of these semiclassical errors.

Consider contribution of each ball
\begin{equation}
B(x,\ell(x))\subset \cY= \{x\colon \min_m |x-\y_m|\ge \epsilon \bar{r}\}
\label{26-6-34}
\end{equation}
to semiclassical errors as $\lambda=\nu$ and $\lambda=\lambda_N$ and compare it with its contribution to (\ref{26-6-11}):

\begin{enumerate}[label=(\alph*), wide, labelindent=0pt]
\item\label{pf-26-6-8-a}
Each ball contributes no more than $CB\ell^2$ to the first error (with $\lambda=\nu$) where due to our choice $\zeta \ell \ge 1$.

\item\label{pf-26-6-8-b}
Further, each ball with $\zeta\ge C_1|\lambda_N-\nu|^{\frac{1}{2}}$ contributes no more than $CB\ell^2$. On the other hand, each ball with
$\zeta\le C_1|\lambda_N-\nu|^{\frac{1}{2}}$ contributes no more than
$CB\ell^{3-\sigma}|\lambda_N-\nu|^{\sigma/2}$ to the second error (with $\lambda=\lambda_N$); here $\sigma=\frac{1}{3}$ is due to rescaling.

\item\label{pf-26-6-8-c}
Meanwhile,  each ball with $\zeta\ge C_1|\lambda_N-\nu|^{\frac{1}{2}}$ contributes no less than $\epsilon_0 B|\lambda_N-\nu| \zeta^{-1}\ell^3$, and each ball with $\zeta\ge C_1|\lambda_N-\nu|^{\frac{1}{2}}$ contributes no less than $\epsilon_0 |\lambda_N-\nu|^{\frac{1}{2}}\ell^3$ to expression (\ref{26-6-11}) and it is larger than the contributions of this ball to each of semiclassical errors (multiplied by $C$) as long as \begin{phantomequation}\label{26-6-35}\end{phantomequation}
\begin{equation}
\zeta^2 \ge |\lambda_N-\nu|\ge C_2\zeta\ell^{-1},\qquad
|\lambda_N-\nu|\ge C_2\ell^{-2}.
\tag*{$\textup{(\ref*{26-6-35})}_{1,2}$}\label{26-6-35-*}
\end{equation}
\end{enumerate}

Obviously in Statements~\ref{prop-26-6-8-i}, \ref{prop-26-6-8-ii} we can assume that
\begin{claim}\label{26-6-36}
Inequalities (\ref{26-6-30}) and (\ref{26-6-32}) respectively (with $C$ replaced by arbitrarily large $C_3$) are violated.
\end{claim}

\begin{enumerate}[label=(\roman*), wide, labelindent=0pt]
\item[(i)(a)]\label{pf-26-6-35-ia}
Assume first that $(Z-N)_+^{\frac{4}{3}}\le B \le Z^3$. Then in the framework of assumption $\zeta = B^2\ell^4$ with minimal
$\ell= B^{-\frac{1}{3}}$ and therefore \ref{26-6-35-*}  are fulfilled for
$\ell \le C_2 B^{-1}|\lambda_N|$. Therefore we need to account for the semiclassical errors contributed by an inner shell (not exceeding $C\max(Z^{\frac{2}{3}}, B^{\frac{1}{2}})$) and by zone
$\cY\cap \{\ell \ge C_2B^{-1}|\lambda_N|\}$; there
$\zeta \ge C_1|\lambda_N|^{\frac{1}{2}}$ and therefore its contribution does not exceed $CB\int \ell(x)^{-1}\,dx $ with integral over this zone and it does not exceed $CB\bar{r}^2 L$.

So, these truncated semiclassical errors do not exceed
$C\max(Z^{\frac{2}{3}} , B\bar{r}^2 L)$. Meanwhile, expression (\ref{26-6-11}) is no less than $CB ^{\frac{1}{2}}\bar{r}^2|\lambda_N|^{\frac{3}{4}}$. Therefore comparing these two expressions as $B\le Z^{\frac{4}{3}}$ and as $Z^{\frac{4}{3}}\le B\le Z^3$ we arrive to (\ref{26-6-30}).

\item[(b)]\label{pf-26-6-35-ib}
Consider the remaining case $B\le (Z-N)_+^{\frac{4}{3}}$. Semiclassical arguments remain valid while estimate of (\ref{26-6-11}) from below by $\epsilon_0|\lambda_N|^{\frac{3}{4}}$ also could be proven easily.

\item[(ii)(a)]\label{pf-26-6-35-iia}
Again, assume first that $(Z-N)_+^{\frac{4}{3}}\le B\le Z^3$. Again, in the calculation of the truncated semiclassical errors we integrate over zone
$\{\ell \ge C_2 B^{-1}|\lambda_N-\nu|\}$ where
$\zeta \ge C_1|\lambda_N|^{\frac{1}{2}}$ and therefore its contribution does not exceed $CB\int \ell(x)^{-1}\,dx $ with integral over this zone and it does not exceed $CB\bar{r}^2 L_1$\,\footnote{\label{foot-26-36} In Statement~\ref{prop-26-6-8-i} this leads only to insignificant improvement.}.

Again, expression (\ref{26-6-21}) is larger than the expressions afterwards and comparing with the semiclassical error estimate we arrive to (\ref{26-6-32}).

\item[(b)]\label{pf-26-6-35-iib} Consider the remaining case
$B\le (Z-N)_+^{\frac{4}{3}}$. Semiclassical arguments remain valid while estimate of (\ref{26-6-11}) from below by $\epsilon_0|\lambda_N-\nu|\cdot |\lambda_N|^{-\frac{1}{4}}$ also could be proven easily.

\item[(iii)] Recall that for $M=1$ the semiclassical error estimate hold with $L=L_1=1$.
\end{enumerate}
\end{proof}

Then we arrive immediately to

\begin{corollary}\label{cor-26-6-9}
In the framework of Proposition~\ref{prop-26-6-8}
$|\lambda_N-\nu|\cdot \N ([\lambda_N,\nu])$ does not exceed expression \textup{(\ref{26-5-40})}.
\end{corollary}

\subsection{Estimate for $\D$-terms for almost neutral systems}
\label{sect-26-6-4-2}

We need to estimate the semiclassical error $\D$-term (\ref{26-6-22}) with $\lambda=\lambda_N$ because for $\lambda=\nu$ we already estimated it, and also we need to estimate another $\D$-term (\ref{26-6-23}). We start from the latter one.
Recall that under assumption (\ref{26-6-29}) we take $\nu=0$. The trivial estimate is based on
\begin{gather}
|P'_B(W)-P'_B(W+\lambda)|\le CW^{\frac{1}{2}}|\lambda|+ CBW^{-\frac{1}{4}}|\lambda|^{\frac{3}{4}},
\label{26-6-37}\\
\shortintertext{leading to}
J\le C \D(W^{\frac{1}{2}},\,W^{\frac{1}{2}}) |\lambda|^2 +
CB^2 |\lambda|^{\frac{3}{2}} \D(W^{-\frac{1}{4}}\theta,\,W^{-\frac{1}{4}}\theta)
\label{26-6-38}
\end{gather}
where here and below $J$ is expression (\ref{26-6-23}), $\theta$ is a characteristic function of the domain
$\{x\colon \gamma (x)\ge h ^{\frac{1}{3}}\}$ and we can ignore the contribution of the zone $\{x\colon \gamma (x)\le h ^{\frac{1}{3}}\}$. Really, the contribution of this zone does not exceed a semiclassical error estimate
$R\coloneqq C\max (Z^{\frac{5}{3}},\, Z^{\frac{3}{5}}B^{\frac{4}{5}}L^2)$.

Note that even without assumption (\ref{26-6-29})
\begin{equation}
\D(W^{\frac{1}{2}},\, W^{\frac{1}{2}})\asymp
(B^{-\frac{1}{4}};\, B^{-\frac{8}{5}}Z^{\frac{9}{5}})\qquad\text{for\ \ }
B\le Z^{\frac{4}{3}}, \, Z^{\frac{4}{3}} \le B\le Z^3
\label{26-6-39}
\end{equation}
respectively and (\ref{26-6-30}) implies that the first term in the right-hand expression of (\ref{26-6-38}) is much less than $R$.

Meanwhile, under assumption (\ref{26-6-29})
\begin{equation}
\D(W^{-\frac{1}{4}}\theta,W^{-\frac{1}{4}}\theta)\asymp
B^{-1}\D(\ell^{-1}\theta,\ell^{-1}\theta) \asymp B^{-1}\bar{r}^3 \D(\gamma^{-1}\theta,\gamma^{-1}\theta)
\label{26-6-40}
\end{equation}
where in the right-hand expression $\D$ and $\gamma, \theta$ are in the scale
$x\mapsto x\bar{r}^{-1}$ and then
$\D(\gamma^{-1}\theta,\gamma^{-1}\theta)\asymp L^2$ so the second term in (\ref{26-6-38}) does not exceed $CB\bar{r}^3 |\lambda_N|^{\frac{3}{2}}L^2$ which due to (\ref{26-6-30}) does not exceed
\begin{equation}
R\coloneqq C\max \bigl(Z^{\frac{5}{3}},\, Z^{\frac{3}{5}}B^{\frac{4}{5}}L^4\bigr).
\label{26-6-41}
\end{equation}

Consider now term (\ref{26-6-22}) with $\lambda=\lambda_N$. Let us consider zones $\Omega_1\coloneqq \{x\colon |\lambda-\nu|\lesssim \zeta \ell^{-1}\}$
and $\Omega_2\coloneqq \{x\colon |\lambda-\nu|\gtrsim \zeta \ell^{-1}\}$.

Note that the contribution to the term in question of each pair of balls contained in $\Omega_1 \times \Omega_1$ does not exceed estimate for the same term with $\lambda=\nu$; really, after rescaling $x\mapsto x/\ell$ and
$\tau\mapsto \tau/\zeta^2$ we conclude that the difference between energy levels does not exceed local semiclassical parameter $C/(\zeta\ell)$.

Therefore the total contribution of $\Omega_1\times \Omega_1$ to this term does not exceed
$C\max \bigl(Z^{\frac{5}{3}},\, Z^{\frac{3}{5}}B^{\frac{4}{5}}L^2\bigr)$.

On the other hand, the contribution to the term in question of each pair of balls contained in $\Omega_2 \times \Omega_2$ does not exceed its contribution to (\ref{26-6-25}) and therefore the total contribution of
$\Omega_2\times \Omega_2$ to this term does not exceed expression (\ref{26-6-41}). Thus, term (\ref{26-6-22}) with $\lambda=\lambda_N$ does not exceed (\ref{26-6-41}).

Therefore we arrive immediately to

\begin{theorem}\label{thm-26-6-10}
Let $M\ge 2$, $B\le Z^3$ and condition \textup{(\ref{26-2-28})} be fulfilled. Then under assumption \textup{(\ref{26-6-29})}
\begin{enumerate}[label=(\roman*), wide, labelindent=0pt]

\item\label{thm-26-6-10-i}
The following estimate holds:
\begin{multline}
E^\TF \le
\cE^\TF + \Bigl(\Tr ((H_{A,W}-\nu)^-) +\int P_B(W^\TF +\nu) \,dx\Bigr) + \\
C\max \bigl(Z^{\frac{5}{3}},\, Z^{\frac{3}{5}}B^{\frac{4}{5}}L^4\bigr).
\label{26-6-42}
\end{multline}
\item\label{thm-26-6-10-ii}
If $a\ge Z^{-1}$ then the following estimate holds:
\begin{equation}\\[2pt]
E^\TF \le
\cE^\TF + \Scott + C\max \bigl(Z^{\frac{5}{3}},\, Z^{\frac{3}{5}}B^{\frac{4}{5}}L^4\bigr)+
CZ^{\frac{4}{3}}B^{\frac{1}{3}}+
Ca^{-\frac{1}{2}}Z^{\frac{3}{2}};
\label{26-6-43}\\[2pt]
\end{equation}
if $a\le Z^{-1}$ one should replace the last term in the right-hand expression by $CZ^2$ and skip $\Scott$.

\item\label{thm-26-6-10-iii}
If $B\le Z$ and $a\ge Z^{-\frac{1}{3}}$
\begin{multline}
E^\TF \le
\cE^\TF + \Scott + \Dirac +\Schwinger +\\
C Z^{\frac{5}{3}}
\bigl(Z^{-\delta}+ (BZ^{-1})^\delta + (aZ^{\frac{1}{3}})^{-\delta}\bigr).
\label{26-6-44}
\end{multline}
\end{enumerate}
\end{theorem}

Here proof of Statement~\ref{thm-26-6-10-iii} is due to the same arguments as in the case $B=0$. Combining with the estimate from below we also conclude that

\begin{theorem}\label{thm-26-6-11}
\begin{enumerate}[label=(\roman*), wide, labelindent=0pt]
\item\label{thm-26-6-11-i}
In the framework of Theorem~\ref{thm-26-6-10} the following estimate holds:
\begin{equation}
\D\bigl(\rho_\Psi -\rho^\TF,\,\rho_\Psi -\rho^\TF\bigr)\le
C\max \bigl(Z^{\frac{5}{3}};\, Z^{\frac{3}{5}}B^{\frac{4}{5}}L^4\bigr).
\label{26-6-45}
\end{equation}
\item\label{thm-26-6-11-ii}
In the framework of Theorem~\ref{thm-26-6-10}\ref{thm-26-6-10-iii} the following estimate holds:
\begin{equation}
\D\bigl(\rho_\Psi -\rho^\TF,\,\rho_\Psi -\rho^\TF\bigr)\le
C Z^{\frac{5}{3}}
\bigl(Z^{-\delta}+ (BZ^{-1})^\delta+ (aZ^{\frac{1}{3}})^{-\delta}\bigr).
\label{26-6-46}
\end{equation}
\end{enumerate}
\end{theorem}

\subsection{Estimate for $\D$-terms for positively charged systems}
\label{sect-26-6-4-3}

Let assumption (\ref{26-6-31}) be fulfilled. Let $W=W_\nu$ and $\ell=\ell_\nu$ be a potential and a scaling function (used to derive semiclassical remainder estimates) for this $\nu<0$ (and $N<Z$) while $W_0$ and $\ell_0$ be a potential and a scaling function for $\nu=0$ (and $N=Z$).

Let us start from rather trivial arguments. Note that
\begin{equation}
|P'_B(W+\lambda)-P'_B(W+\nu)|\le CW^{\frac{1}{2}}|\lambda-\nu|\theta_1
+ CB|\lambda-\nu|^{\frac{1}{2}}\theta_2,
\label{26-6-47}
\end{equation}
where $\theta_1$ and $\theta_2$ are characteristic functions of
$\cY_1=\{x\colon W(x)+\nu \ge C_0|\nu|\}$ and
$\cY_2= \{x\colon 0< W(x)+\nu \le C_0|\nu|\}$ respectively. Let
\begin{gather*}
J_k\coloneqq
\D\bigl( [P'_B(W+\lambda_N)-P'_B(W+\nu)]\theta_k,\, [P'_B(W+\lambda_N)-P'_B(W+\nu)]\theta_k\bigr).\\
\intertext{Then in virtue of (\ref{26-6-47})}
J_1\le C\D(W^{\frac{1}{2}} \theta_1, W^{\frac{1}{2}} \theta_1)
|\lambda_N-\nu|^2.\\
\shortintertext{Note that\footnotemark}
\D(W^{\frac{1}{2}} \theta_1,\, W^{\frac{1}{2}} \theta_1)
\asymp \bigl((Z-N)_+^{-\frac{1}{3}};\, B^{-\frac{1}{4}};\, Z^{\frac{9}{5}}B^{-\frac{8}{5}}\bigr).
\end{gather*}
\footnotetext{\label{foot-26Z-28} In our three cases
$B\le (Z-N)_+^{\frac{4}{3}}$, $(Z-N)_+^{\frac{4}{3}}\le B\le Z^{\frac{4}{3}}$, and $Z^{\frac{4}{3}}\le B\le Z^3$ respectively.} Then, using inequality (\ref{26-6-32}) one can prove easily that $J_1\le CZ^{\frac{5}{3}}$.

However estimate for a contribution of zone $\cY_2$ is much worse:
\begin{equation}
J_2\le CB^2\D(\theta_2, \theta_2)|\lambda_N-\nu|\le
CB^2\bar{r}^3 (|\nu|/B^2)^{\frac{1}{2}}L_1^2|\lambda_N-\nu|,\label{26-6-48}
\end{equation}
where for $B\le (Z-N)_+^{\frac{1}{3}}$ we should replace
$(|\nu|/B^2)^{\frac{1}{2}}L_1^2$ by $\bar{r}^2$.
Then, using (\ref{26-6-32}), we conclude that for $(Z-N)_+^{\frac{4}{3}}\lesssim B\lesssim Z^3$
\pagebreak
\begin{equation*}
J_2\le
CB\bar{r}^3 |\nu|^{\frac{3}{4}}L_1^2\max(Z^{\frac{2}{3}}, B^{\frac{1}{2}}) \asymp CB(Z-N)_+^{\frac{3}{4}} \bar{r}^{\frac{9}{4}} \max(Z^{\frac{2}{3}}, B^{\frac{1}{2}}L_1)L_1^2 ,
\end{equation*}
and therefore we arrive to the last two cases below; the first case is proven similarly:
\begin{equation}
J_2 \le
C\left\{\begin{aligned}
&(Z-N)_+^{-\frac{4}{3}}Z^{\frac{2}{3}}B^2 ,\\
&(Z-N)_+^{\frac{3}{4}} B^{\frac{7}{16}}
\max (Z^{\frac{2}{3}}, B^{\frac{1}{2}}L_1)L_1^2 ,\\
&(Z-N)_+^{\frac{3}{4}}Z^{\frac{9}{20}} B^{\frac{3}{5}}L_1^3
\end{aligned}\right.
\label{26-6-49}
\end{equation}
in our three cases.

This is really shabby estimate. To improve it let us observe that
\begin{claim}\label{26-6-50}
If estimate $|\lambda_N-\nu|\le C\max(B^{\frac{2}{3}},\, (Z-N)_+^{\frac{8}{9}})$ holds, then $J_2$ does not exceed (\ref{26-5-40})
\end{claim}
and therefore we can assume that
\begin{equation}
|\lambda_N-\nu|\ge C\max((Z-N)_+^{\frac{8}{9}};\,B^{\frac{2}{3}}).
\label{26-6-51}
\end{equation}

Let us estimate the truncated semiclassical error\footnote{\label{foot-26-38} I.e. contribution to such error of the zone, where it exceeds the contribution to the principal part.}.

\begin{proposition}\label{prop-26-6-12}
\begin{enumerate}[label=(\roman*), wide, labelindent=0pt]
\item\label{prop-26-6-12-i}
Let $(Z-N)_+^{\frac{4}{3}}\le B\le Z^3$ and
\begin{equation}
C_0B^{\frac{2}{3}} \le |\lambda_N-\nu| \le
C_1 B^{\frac{1}{2}}|\nu|^{\frac{1}{4}}.
\label{26-6-52}
\end{equation}
Then the truncated semiclassical error  in $\N$-term does not exceed
\begin{gather}
F\coloneqq CZ^{\frac{2}{3}}+
C B \bar{r}^2 (|\nu|/B^2)^{\frac{1}{4}}L \times (B^{-1}|\lambda_N-\nu|)^{-1}.
\label{26-6-53}
\end{gather}

\item\label{prop-26-6-12-ii}
Let $(Z-N)_+^{\frac{4}{3}}\le B\le Z^3$ and
\begin{equation}
C_1 B^{\frac{1}{2}}|\nu|^{\frac{1}{4}}\le |\lambda_N-\nu| \le
C_1 B^{\frac{1}{2}}|\nu|^{\frac{1}{4}}L.
\label{26-6-54}
\end{equation}
Then the truncated semiclassical error does not exceed
\begin{equation}
F\coloneqq C Z^{\frac{2}{3}} + C B \bar{r}^2 L.
\label{26-6-55}
\end{equation}

\item\label{prop-26-6-12-iii}
Let $B\le (Z-N)_+^{\frac{4}{3}}$ and
\begin{equation}
(Z-N)_+^{\frac{8}{9}} \le |\lambda_N-\nu| \le C_1 (Z-N)_+.
\label{26-6-56}
\end{equation}
Then the truncated semiclassical error does not exceed
\begin{equation}
F\coloneqq C Z^{\frac{2}{3}}+ C (Z-N)_+^{\frac{5}{3}} |\lambda_N-\nu|^{-1}.
\label{26-6-57}
\end{equation}

\item\label{prop-26-6-12-iv}
Let $B\le (Z-N)_+^{\frac{4}{3}}$ and
\begin{equation}
C_0(Z-N)_+ \le |\lambda_N-\nu| \le C_1 Z^{\frac{2}{3}}(Z-N)_+^{\frac{1}{3}}.
\label{26-6-58}
\end{equation}
Then the truncated semiclassical error does not exceed
$F\coloneqq C Z^{\frac{2}{3}}$.
\end{enumerate}
\end{proposition}

\begin{proof}
The easy proof, which  uses arguments of the proof of Proposition~\ref{prop-26-6-8}, is left to the reader. \end{proof}

\begin{proposition}\label{prop-26-6-13}
In the framework of Proposition~\ref{prop-26-6-12}\ref{prop-26-6-12-i}--\ref{prop-26-6-12-iv} term \textup{(\ref{26-6-25})} does not exceed
\begin{equation*}
CF^{\frac{5}{3}} (B|\lambda_N-\nu|^{\frac{1}{2}})^{\frac{1}{3}} +\textup{(\ref{26-5-40})}
\end{equation*}
 with $F$ defined in the corresponding cases in Proposition~\ref{prop-26-6-12}.
\end{proposition}

\begin{proof}
Using Proposition~\ref{prop-26-6-8} one can prove easily that

\begin{claim}\label{26-6-59}
Contribution of $\{x\colon \ell(x)\coloneqq \min_m |x-\y_m|\le \epsilon\bar{r}\}$ to (\ref{26-6-25}) does not exceed $CZ^{\frac{5}{3}}$.
\end{claim}

Now we need to estimate the excess of expression (\ref{26-6-25}) over semiclassical $\D$-term (with $\lambda=\nu$), which has been estimated by (\ref{26-5-40}). To do so we need to estimate
\begin{equation}
\D\bigl(
[P_B(W+\nu)-P_B(W+\lambda)]\theta,\, [P_B(W+\nu)-P_B(W+\lambda)]\theta\bigr)
\label{26-6-60}
\end{equation}
which is the contribution of the domain
$\Omega' \coloneqq \{x\colon \ell(x)\le CB^{-1}|\lambda-\nu|\}$ where $\theta$ is the characteristic function of $\Omega'$. Recall that in the complimentary domain $|P_B(W+\nu)-P_B(W+\lambda)|\le C\ell^{-1}$. Let us consider
\begin{gather}
\D\bigl( [P_B(W+\nu)-P_B(W+\lambda)]\theta_0 ,\, [P_B(W+\nu)-P_B(W+\lambda)]\theta_0\bigr)
\label{26-6-61}
\shortintertext{and}
\D\bigl( [P_B(W+\nu)-P_B(W+\lambda)]\theta_t,\, [P_B(W+\nu)-P_B(W+\lambda)]\theta_{t'}\bigr),
\label{26-6-62}
\end{gather}
where $\theta_0$ is a characteristic function of
\begin{equation*}
\Omega'_0\coloneqq \{x\colon \ell(x) \le t_0\coloneqq (|\lambda_N-\nu|B^{-2})^{\frac{1}{4}}\}
\end{equation*}
and $\theta_t $ is a characteristic function of
$\Omega'_t\coloneqq \{x\colon t\le \ell(x) \le 2t\}$ with $t\ge t'\ge t_0$.

Observe that, when calculating expression (\ref{26-6-11}), the contribution of $\Omega'_0$ is $\asymp B|\lambda_N-\nu|^{\frac{1}{2}} \mes(\Omega'_0)$, and therefore due to Proposition~\ref{prop-26-6-12}
\begin{gather}
\mes(\Omega'_0)\le CF \bigl(B|\lambda_N-\nu|^{\frac{1}{2}}\bigr)^{-1},
\label{26-6-63}\\
\intertext{while term (\ref{26-6-61}) is}
\asymp B^2 |\lambda_N-\nu|\D(\theta_0,\theta_0) \le
C B^2 |\lambda_N-\nu| \bigl(\mes (\Omega_0)\bigr)^{\frac{5}{3}}\le
CF^{\frac{5}{3}} \bigl(B |\lambda_N-\nu|^{\frac{1}{2}}\bigr)^{\frac{1}{3}},
\notag\\
\intertext{where the middle inequality}
\D(\chi_G,\chi_G)\le C\bigl(\mes (G)\bigr)^{\frac{5}{3}}
\label{26-6-64}
\end{gather}
is well known\footnote{\label{foot-26-39} Really, among uniform solids of equal mass and density the ball has the least potential energy; then $C=\frac{1}{5}(12\pi)^{\frac{1}{3}}$.} and the last one is due to (\ref{26-6-63}); $\chi_G$ denotes characteristic function of $G$.
\enlargethispage{\baselineskip}

Similarly, when calculating expression (\ref{26-6-11}), one can see easily that the contribution of $\Omega'_t$ is
$\asymp B|\lambda_N-\nu| (B^2t^4)^{-1} \mes(\Omega'_t)$ and therefore
\begin{equation}
\mes(\Omega'_t)\le CF |\lambda_N-\nu|^{-1} t^2,
\label{26-6-65}
\end{equation}
while term (\ref{26-6-62}) is
$\asymp |\lambda_N-\nu|^2 t^{-2} t^{\prime\, -2}\D(\theta_t,\theta_{t'})$, which does not exceed
\begin{gather}
C |\lambda_N-\nu|^2 t^{-2} t^{\prime\,-2}
\mes (\Omega_t) \mes (\Omega_{t'})
\bigl[ \max\bigl(\mes (\Omega_t),\, \mes (\Omega_{t'})\bigr)\bigr]^{-\frac{1}{3}} \label{26-6-66} \\
\shortintertext{due to inequality}
\D(\chi_G,\chi_{G'})\le C \mes (G) \mes (G')
\bigl[ \max\bigl(\mes (G),\, \mes (G')\bigr)\bigr]^{-\frac{1}{3}},
\label{26-6-67}
\end{gather}
which trivially follows from the obvious inequality
$\D(\chi_G,\updelta_z )\le C(\mes (G))^{\frac{2}{3}}$, where $\updelta_z(x)=\updelta(x-z)$.

Due to (\ref{26-6-65}) expression (\ref{26-6-66}) does not exceed
$CF^{\frac{5}{3}}t^{-\frac{2}{3}}$; recall that $t\ge t'$. Since summation with respect to $t\ge t'$ and then with respect to $t'\ge t_0$ returns $CF^{\frac{5}{3}}t_0^{-\frac{2}{3}}$, we conclude that term (\ref{26-6-61}) with $\theta_0$ replaced by $\theta''$ (the characteristic function of
$\{x\colon \ell(x)\ge t_0\}$) also does not exceed
$CF^{\frac{5}{3}} (B|\lambda_N-\nu|^{\frac{1}{2}})^{\frac{1}{3}}$. \end{proof}

So, we have now two estimates for an excess of expression (\ref{26-6-25}) over (\ref{26-5-40}): one estimate is
\begin{equation}
CF^{\frac{5}{3}} (B|\lambda_N-\nu|^{\frac{1}{2}})^{\frac{1}{3}}
\label{26-6-68}
\end{equation}
with $F=F(|\lambda_N-\nu|)$ derived in Proposition~\ref{prop-26-6-12} and another one is due to (\ref{26-6-48}). Let us consider the best of them. Note that estimate (\ref{26-6-68}) consists of two terms each due to the corresponding term in the definition of $F$. The second term in the framework of Proposition~\ref{prop-26-6-12}\ref{prop-26-6-12-i} is
\begin{equation*}
C\bigl(B^{\frac{3}{2}}\bar{r}^2 |\nu|^{\frac{1}{4}}L |\lambda_N-\nu|^{-1} \bigr)^{\frac{5}{3}} \bigl(B|\lambda_N-\nu|^{\frac{1}{2}}\bigr)^{\frac{1}{3}}
\asymp
B^{\frac{17}{6}}\bar{r}^{\frac{10}{3}}|\nu|^{\frac{5}{12}}L^{\frac{5}{3}}
|\lambda_N-\nu|^{-\frac{3}{2}}.
\end{equation*}
Then, taking minimum of this expression and
$CB\bar{r}^3 |\nu|^{\frac{1}{2}}L^2 |\lambda_N-\nu|$, we see that this minimum does not exceed
\begin{multline*}
C\bigl(B^{\frac{17}{6}}\bar{r}^{\frac{10}{3}}|\nu|^{\frac{5}{12}}L^{\frac{5}{3}}
\bigr)^{\frac{2}{5}}\bigl(B\bar{r}^3 |\nu|^{\frac{1}{2}}L^2\bigr)^{\frac{3}{5}}
\asymp CB^{\frac{5}{3}}\bar{r}^{\frac{8}{3}}(Z-N)_+^{\frac{7}{15}}L^{\frac{28}{15}}
\asymp\\
C\left\{\begin{aligned}
&(Z-N)_+^{\frac{7}{15}}BL^{\frac{28}{15}} \qquad
&&\text{if\ \ } (Z-N)_+^{\frac{4}{3}}\le B\le Z^{\frac{4}{3}},\\
&(Z-N)_+^{\frac{7}{15}}Z^{\frac{8}{15}}B^{\frac{3}{5}}L^{\frac{28}{15}}\qquad &&\text{if\ \ } Z^{\frac{4}{3}}\le B\le Z^3,
\end{aligned}\right.
\end{multline*}
which is achieved for $|\lambda_N-\nu| \asymp B^{\frac{11}{15}}\bar{r}^{\frac{2}{15}}|\nu|^{-\frac{1}{30}}L^{-\frac{2}{15}}$. One can see easily that this expression does not exceed (\ref{26-5-40}).

Therefore in the framework of Proposition~\ref{prop-26-6-12}\ref{prop-26-6-12-i}\ref{prop-26-6-12-ii} we can select $F= (Z^{\frac{2}{3}}+B\bar{r}^2L)$ according to (\ref{26-6-55}), arriving to
\begin{equation*}
C(Z^{\frac{2}{3}}+B\bar{r}^2L)^{\frac{5}{3}}B ^{\frac{1}{3}}
|\lambda_N-\nu|^{\frac{1}{6}} \le
C(Z^{\frac{2}{3}}+B\bar{r}^2L)^{\frac{5}{3}}B ^{\frac{1}{3}}
(Z^{\frac{2}{3}}+B^{\frac{1}{2}}L_1)^{\frac{1}{6}}|\nu|^{\frac{1}{24}},
\end{equation*}
which we can rewrite (slightly increasing powers of logarithms) as two last cases in expression
\begin{equation}
C\left\{\begin{aligned}
& (Z-N)_+^{\frac{1}{18}}Z^{\frac{11}{9}}B^{\frac{1}{3}}\qquad
&&\text{if\ \ } B\le (Z-N)_+^{\frac{4}{3}},\\
& (Z-N)_+^{\frac{1}{24}} (Z^{\frac{11}{9}}+ B^{\frac{11}{12}}L^4) B^{\frac{11}{32}}\qquad
&&\text{if\ \ } (Z-N)_+^{\frac{4}{3}}\le B\le Z^{\frac{4}{3}},\\
& (Z-N)_+^{\frac{1}{24}} Z^{\frac{79}{120}}B^{\frac{23}{30}}\qquad
&&\text{if\ \ } Z^{\frac{4}{3}}\le B\le Z^3.
\end{aligned}\right.
\label{26-6-69}
\end{equation}

In the framework of Proposition~\ref{prop-26-6-12}\ref{prop-26-6-12-iii} one should replace $(\nu/B^2)^{\frac{1}{4}}$ by $\bar{r}$ and $L$ by $1$, so $B^{\frac{3}{2}}\bar{r}^2 |\nu|^{\frac{1}{4}}L |\lambda_N-\nu|^{-1} \mapsto
B^2\bar{r}^3 |\lambda_N-\nu|^{-1}$; further, one should preserve
$B|\lambda_N-\nu|^{\frac{1}{2}}$ and therefore the second term becomes
\begin{equation*}
\bigl(B^2\bar{r}^3 |\lambda_N-\nu|^{-1}\bigr)^{\frac{5}{3}}
\bigl(B|\lambda_N-\nu|^{\frac{1}{2}}\bigr)^{\frac{1}{3}}\asymp
B^{\frac{11}{3}}\bar{r}^5|\lambda_N-\nu|^{-\frac{3}{2}}
\end{equation*}
and taking minimum of it and (\ref{26-6-48}) we again get a term lesser than (\ref{26-5-40}).

Meanwhile, the first term becomes $Z^{\frac{10}{9}}B^{\frac{1}{3}}|\lambda_N-\nu|^{\frac{1}{6}}\le
(Z-N)_+^{\frac{1}{18}}Z^{\frac{11}{9}}B^{\frac{1}{3}}$ occupying the first line in (\ref{26-6-69}).

Therefore we have proven

\begin{proposition}\label{prop-26-6-14}
If $M\ge 2$, $B\le Z^3$ all three $\D$-terms do not exceed $\textup{(\ref{26-5-40})}+\textup{(\ref{26-6-69})}$.
\end{proposition}

\subsection{Summary}
\label{sect-26-6-4-4}

Therefore all error terms in the upper estimate do not exceed (\ref{26-5-40}) and we arrive to

\begin{theorem}\label{thm-26-6-15}
Let $M\ge 2$, $B\le Z^3$. Then
\begin{enumerate}[label=(\roman*), wide, labelindent=0pt]
\item\label{thm-26-6-15-i}
The following estimate holds:
\begin{multline}
E^\TF \le
\cE^\TF + \Bigl(\Tr ((H_{A,W}-\nu)^-) +\int P_B(W^\TF +\nu) \,dx\Bigr) + \\
\textup{(\ref{26-5-40})}+\textup{(\ref{26-6-69})}.
\label{26-6-70}
\end{multline}

\item\label{thm-26-6-15-ii}
The following estimate holds for $a\ge Z^{-1}$:
\begin{equation}
E^\TF \le
\cE^\TF + \underbracket{\Scott + CZ^{\frac{4}{3}}B^{\frac{1}{3}}+ a^{-\frac{1}{2}}Z^{\frac{3}{2}}}+
\textup{(\ref{26-5-40})}+\textup{(\ref{26-6-69})};
\label{26-6-71}
\end{equation}
for $a\le Z^{-1}$ one should replace selected terms by $CZ^2$.

\item\label{thm-26-6-15-iii}
If $B\le Z$ and $a\ge Z^{-\frac{1}{3}}$
\begin{multline}
E^\TF \ge
\cE^\TF + \Scott + \Dirac +\Schwinger +\\
C Z^{\frac{5}{3}}\bigl(Z^{-\delta}+ (BZ^{-1})^\delta+(aZ^{\frac{1}{3}})^{-\delta}\bigr).
\label{26-6-72}
\end{multline}
\end{enumerate}
\end{theorem}

We also arrtive to

\begin{theorem}\label{thm-26-6-16}
\begin{enumerate}[label=(\roman*), wide, labelindent=0pt]
\item\label{thm-26-6-16-i}
In the framework of Theorem~\ref{thm-26-6-15}\ref{thm-26-6-15-i} the following estimate holds:
\begin{equation}
\D\bigl(\rho _\psi -\rho ^\TF,\rho _\psi -\rho ^\TF\bigr)\le
\textup{(\ref{26-5-40})}+\textup{(\ref{26-6-69})}.
\label{26-6-73}
\end{equation}

\item\label{thm-26-6-16-ii}
In the framework of Theorem~\ref{thm-26-6-15}\ref{thm-26-6-15-iii} (albeit without assumption $a\ge Z^{-\frac{1}{3}}$) the following estimate holds:
\begin{equation}
\D\bigl(\rho _\psi -\rho ^\TF,\rho _\psi -\rho ^\TF\bigr)\le CQ\coloneqq
C Z^{\frac{5}{3}}\bigl(Z^{-\delta}+ (BZ^{-1})^\delta).
\label{26-6-74}
\end{equation}
\end{enumerate}
\end{theorem}

\begin{remark}\label{rem-26-6-17}
In virtue of Remark~\ref{rem-26-3-6} we can replace term $CZ^{\frac{4}{3}}B^{\frac{1}{3}}$ to $o(Z^{\frac{4}{3}}B^{\frac{1}{3}})$. This is also true in the case of the better estimates $M=1$.
\end{remark}

We leave to the reader the following easy problem:

\begin{problem}\label{problem-26-6-17}
Investigate conditions to $(Z-N)_+$ so that terms (\ref{26-5-40}) and (\ref{26-6-69}) do not spoil the upper estimate for $\E_N$ or
$\D(\rho_\Psi-\rho^\TF_B,\,\rho_\Psi-\rho^\TF_B)$.
\end{problem}\enlargethispage{\baselineskip}

\chapter{Negatively charged systems}
\label{sect-26-7}

In this section we following Section~\ref{book_new-sect-25-5} consider the case $N\ge Z$ and provide upper estimates for the excessive negative charge $(N-Z)$ if $\I_N>0$ and for the ionization energy $\I_N$.

\section{Estimates of the correlation function}
\label{sect-26-7-1}
First of all we provide some estimates which will be used for both negatively and positively charged systems. Let us consider the ground-state function
$\Psi (x_1,\varsigma_1;\ldots;x_N,\varsigma_N)$ and the corresponding density $\rho _\Psi (x)$. Again the crucial role play estimates\footnote{\label{foot-26-40} Namely, estimate (\ref{26-6-27})  of Theorem~\ref{thm-26-6-7} if $M=1$, and similar estimates
(\ref{26-6-45}) of Theorem~\ref{thm-26-6-11} and (\ref{26-6-73}) of Theorem~\ref{thm-26-6-16} if $M\ge 2$. For $B\le Z$ and $a\ge Z^{-\frac{1}{3}}$, we use estimate (\ref{26-6-74}) in all cases.}
\begin{equation}
\D\bigl(\rho _\psi -\rho ^\TF,\rho _\psi -\rho ^\TF\bigr)\le \bar{Q}
\label{26-7-1}
\end{equation}
where $\bar{Q}\ge Q$ is just the right-hand expression of the corresponding estimate; as $B\le Z$ we can slightly decrease $\bar{Q}=Q$.

\pagebreak
Recall that the same estimate holds also for difference between upper and lower bounds for $\E_N$ (with $\Tr ((H_W-\nu)^-)+\nu N$ not replaced by its semiclassical approximation).

\begin{remark}\label{rem-26-7-1}
All arguments and conclusions of Subsection~\ref{book_new-sect-25-5-1} up to but excluding estimate (\ref{book_new-25-5-31}) are not related to the Schr\"odinger operator and remain true.
\end{remark}

So we need to calculate both the semiclassical errors and the principal parts. Note that \emph{all semiclassical errors for $W_\varepsilon$ do not exceed those obtained for $W$ we selected\/}. Consider approximations errors in the principal part, namely
\begin{gather}
\D\bigl(P'(W_\varepsilon +\nu)-P'(W +\nu),
P'(W_\varepsilon +\nu)-P'(W +\nu)\bigr)\label{26-7-2}\\
\shortintertext{and}
\D(\rho_\varepsilon -\rho, \rho_\varepsilon -\rho)\label{26-7-3}
\end{gather}
since we already estimated terms
$\D\bigl(P'(W +\nu)-\rho^\TF_B, P'(W +\nu)-\rho^\TF_B)\bigr)$ and
$\D(\rho -\rho_B^\TF, \rho -\rho_B^\TF)$ by $\bar{Q}$.

Note that
\begin{equation}
|W-W_\varepsilon|\le C (1+\ell\varepsilon^{-1})^{-2}\zeta^2
\label{26-7-4}
\end{equation}
and
\begin{multline}
|P'(W_\varepsilon +\nu)-P'(W +\nu)|\le \\
C (1+\ell\varepsilon^{-1})^{-2}\zeta^3+
C(1+\ell\varepsilon^{-1})^{-1} \zeta B
\label{26-7-5}
\end{multline}
and therefore expression (\ref{26-7-2}) does not exceed
$C\bigl(Z^3\varepsilon ^2 + ZB^2 \varepsilon^2\bar{r}^2\bigr)$ and it does not exceed $C\max(Z^{\frac{5}{3}},B^{\frac{4}{5}}Z^{\frac{3}{5}})$ for
$\varepsilon = \min (Z^{-\frac{2}{3}}, Z^{\frac{2}{5}}B^{-\frac{4}{5}})$ and this does not exceed $C\bar{Q}$.

Further, consider expression (\ref{26-7-3}); it is equal to
$4\pi |(W_\varepsilon-W, \rho_\varepsilon -\rho)|$ and one can prove easily the same estimate for it.

Furthermore, under this restriction an error in the principal part of asymptotics of $\int e(x,x,\lambda)\,dx$, namely
$|\int \bigl(P'(W_\varepsilon +\nu)-P'(W +\nu)\bigr)\,dx|$,  does not exceed $C\bigl(Z^{\frac{3}{2}}\varepsilon ^{\frac{3}{2}}+ Z^{\frac{1}{2}}B\varepsilon \bar{r}^{\frac{3}{2}}\bigr)$, which is less than the semiclassical error. Then $S\le C\bar{Q}$ with $S$ defined by (\ref{book_new-25-5-22}).

So, the following proposition is proven:

\begin{proposition-foot}\label{prop-26-7-2}\footnotetext{\label{foot-26-41} Cf. Proposition~\ref{book_new-prop-25-5-1}.}
If $\theta ,\chi $ are as in Subsection~\ref{book_new-sect-25-5-2}, then estimate \textup{(\ref{book_new-25-5-33})} holds, namely,
\begin{multline}
\cJ=|\int \Bigl(\rho _\Psi ^{(2)}(x,y)-
\rho(y)\rho _\Psi (x)\Bigr)\theta (x)\chi (x,y)\,dxdy |\le \\[3pt]
C\sup_x \|\nabla _y\chi _x\|_{\sL ^2(\bR ^3)}
\Bigl((\bar{Q}+\varepsilon ^{-1}N+T)^{\frac{1}{2}}\Theta + P^{\frac{1}{2}}\Theta^{\frac{1}{2}}\Bigr) +
C \varepsilon N\|\nabla _y\chi \|_{\sL ^\infty }\Theta
\label{26-7-6}
\end{multline}
with $\Theta=\Theta_\Psi$ defined by \textup{(\ref{book_new-25-5-15})} and $T,P$ defined by \textup{(\ref{book_new-25-5-23})}, \textup{(\ref{book_new-25-5-25})} and arbitrary
$\varepsilon \le \min(Z^{-\frac{2}{3}},Z^{\frac{2}{5}}b^{-\frac{4}{5}})$.
\end{proposition-foot}

Recall that  $\rho ^{(2)}_\Psi (x,y)$ defined by (\ref{book_new-25-5-13}) is the \emph{quantum correlation function\/}.\enlargethispage{\baselineskip}

\section{Excessive negative charge}
\label{sect-26-7-2}
Let us select $\theta=\theta_b$ according to (\ref{book_new-25-5-34}):
\begin{equation*}
\supp (\theta) \subset \{x\colon \ell (x)\ge b\}.
\tag{\ref*{book_new-25-5-34}}
\end{equation*}
Note that $\mathbf{H}_N\Psi =\E_N\Psi $ yields identity (\ref{book_new-25-5-35})
and isolating the contribution of $j$-th electron in $j$-th term we get inequality (\ref{book_new-25-5-36}):
\begin{multline}
-\I_N\int \rho _\Psi (x)\ell (x)\theta \,dx\ge
\\[3pt]
\sum_j\blangle \Psi ,\ell (x_j)\theta (x_j)
\Bigl(-V(x_j)+ \sum_{k: k\ne j}|x_j-x_k| ^{-1}\Bigr)\Psi \brangle -
\sum_j
\|\nabla \bigl(\theta^{\frac{1}{2}} (x_j) \ell (x_j)^{\frac{1}{2}}\bigr)\Psi\|^2
\tag{\ref{book_new-25-5-36}}
\end{multline}
due to the non-negativity of operator $\bigl((D_x-A(x))\cdot\boldupsigma\bigr)^2$.

Now let us select $b$ to be able to calculate the magnitude of $\Theta$. Note that inequality (\ref{book_new-25-5-37}) holds. Also (\ref{book_new-25-5-38}) holds
as long as
\begin{equation}
Z^{-\frac{1}{3}}\le b\le \epsilon \min \bigl((Z-N)_+^{-\frac{1}{3}}, B^{-\frac{1}{4}}\bigr)
\label{26-7-7}
\end{equation}
Using inequalities
\begin{gather*}
|\nabla (\theta _b(x) ^{\frac{1}{2}}\ell^{\frac{1}{2}})|\le
cb ^{-1}\theta _{(1-\epsilon )b}(x)\\
\shortintertext{and}
\int \rho _\Psi (x)\ell (x)\theta _b(x)\,dx \ge b\Theta _b
\end{gather*}
\pagebreak
(i.e. (\ref{book_new-25-5-43})) we conclude that
\begin{multline}
b\I_N\Theta _b\le \int \theta_b (x)V(x)\ell (x)\rho _\Psi (x)\,dx\\
\shoveright{-\int\rho _\Psi ^{(2)}(x,y)\ell (x)|x-y| ^{-1}\theta_b (x)\,dxdy +Cb ^{-1}\Theta_{b(1-\epsilon)}=\quad}\\
\begin{aligned}=&\int \theta_b (x)V(x)\ell (x)\rho _\Psi (x)\,dx\\
-&\int \rho _\Psi ^{(2)}(x,y)\ell (x)|x-y| ^{-1}
\bigl(1-\theta_b (y)\bigr)\theta_b (x)\,dxdy \\
-&\int \rho _\Psi ^{(2)}(x,y)\ell (x)|x-y| ^{-1}\theta_b (y)\theta_b (x)\,dxdy + Cb ^{-1}\Theta _{b(1-\epsilon )}
\end{aligned}
\label{26-7-8}
\end{multline}
(cf. (\ref{book_new-25-5-44})).
Denote by $\cI_1$, $\cI_2$, and $\cI_3$ the first, second and third terms in the right-hand expression of (\ref{26-7-8}) respectively. Symmetrizing $\cI_3$ with respect to $x$ and $y$
\begin{equation*}
\cI_3=-\frac{1}{2}
\int \rho _\Psi ^{(2)}(x,y)\bigl(\ell(x)+\ell(y)\bigr)
|x-y| ^{-1}\theta (y)\theta (x)\,dxdy
\end{equation*}
and using inequality
$\ell (x)+\ell (y)\ge \min_j(|x-\y_j|+|y-\y_j|)\ge |x-y|$ we conclude that this term does not exceed
\begin{multline}
-\frac{1}{2}\int \rho _\Psi ^{(2)}(x,y)\theta_b (y)\theta_b (x)\,dxdy=\\[5pt]
-\frac{1}{2}(N-1) \int \rho _\Psi (x)\theta_b (x)\,dx +
\frac{1}{2}
\int \rho _\Psi ^{(2)}(x,y)\bigl(1-\theta_b (y)\bigr) \theta_b (x)\, dxdy
\label{26-7-9}
\end{multline}
(cf. (\ref{book_new-25-5-45})).

Here the first term is exactly $-\frac{1}{2}(N-1)\Theta_b$; replacing
$\rho _\Psi ^{(2)}(x,y)$ by $\rho (y)\rho_\Psi(x)$ we get
\begin{gather}
\frac{1}{2}\int \bigl(1-\theta_b (y)\bigr)\rho(y)\,dy \times \Theta_b
\label{26-7-10}\\
\shortintertext{with an error}
\frac{1}{2}
\int \bigl(\rho _\Psi ^{(2)}(x,y)-\rho(y)\rho_\Psi (x)\bigr)
\bigl(1-\theta_b (y)\bigr) \theta_b (x)\, dxdy
\label{26-7-11}
\end{gather}
(cf. (\ref{book_new-25-5-46}), (\ref{book_new-25-5-47})). We estimate this expression using Proposition~\ref{prop-26-7-2} with $\chi(x,y) = 1-\theta_b (y)$. Then
$\|\nabla_y \chi_x\|_{\sL^2}\asymp b^{\frac{1}{2}}$,
$\|\nabla_y \chi\|_{\sL^\infty}\asymp b^{-1}$ and
$P\asymp b^{-1}\Theta_b$\,\footnote{\label{foot-26-42} Recall that
$P=\int |\nabla \theta^{\frac{1}{2}}|^2 \rho_\Psi \,dx$ and
$T= \sup _{\supp (\theta)} W$.}, while $T\lesssim b^{-4}$ as long as $B\le Z^{\frac{4}{3}}$ and
$b\le B^{-\frac{1}{4}}$.

To estimate the excessive negative charge we assume that $(N-Z)>0$ with
$\I_N> 0$. In this case the left-hand expression in (\ref{26-7-8}) should be positive.\enlargethispage{\baselineskip}

\begin{remark}\label{rem-26-7-3}
Recall that in Subsection~\ref{book_new-sect-25-5-2} we picked $b=Z^{-\frac{5}{21}}$ and it makes sense here as well as long as $b\le \bar{r}=B^{-\frac{1}{4}}$ i.e. as $B\le Z^{\frac{20}{21}}$. However for $B\ge Z^{\frac{20}{21}}$ we just pick up $b=C_0\bar{r}$ and then $T=0$ in our framework.
\end{remark}

Estimating (\ref{26-7-11}) we conclude that
\begin{gather}
\cI_3\le -\frac{1}{2} \Bigl(N-1-
\int \bigl(1-\theta _b(y)\bigr) \rho(y)\,dy\Bigr) \Theta_b + \cI _0 \label{26-7-12}\\
\shortintertext{with}
\cI _0=Cb ^{\frac{1}{2}}
\Bigl({\cS}\Theta _b+ Nb ^{-2}\Bigr) ^{\frac{1}{2}}\Theta _b ^{\frac{1}{2}} +
C\varepsilon Nb ^{-1}\Theta _b
\label{26-7-13}
\end{gather}
(cf. (\ref{book_new-25-5-48})).

On the other hand,
\begin{equation}
\cI_2\le - \int \rho _\Psi ^{(2)}(x,y)\ell (x)|x-y| ^{-1}
\bigl(1-\theta _{b(1-\epsilon)}(y) \bigr)\theta _b(x)\,dxdy
\label{26-7-14}
\end{equation}
and replacing $\rho _\Psi ^{(2)}(x,y)$ by $\rho (y)\rho_\Psi(x)$ and estimating an error due to Proposition~\ref{prop-26-7-2} again, we get
\begin{multline}
\cI_2\le - \int \rho (y)\rho _\Psi (x)\ell (x)|x-y| ^{-1}
\bigl(1-\theta _{b(1-\epsilon )}(y)\bigr)\theta _b(x)\,dxdy+\\
\shoveright{Cb ^{-\frac{1}{2}}
\bigl( \cS \Theta _b+Nb ^{-2}\bigr) ^{\frac{1}{2}}\Theta _b ^{\frac{1}{2}} +C\varepsilon Nb ^{-1}=}\\
\begin{aligned}
-&\int (V- W )(x)\ell (x)\theta _b(x)\,dx+\\
&\int \rho (y)\rho _\Psi (x)\ell (x)|x-y| ^{-1}
\theta _{b(1-\epsilon )}(y))\theta _b(x)\,dxdy+ \cI _0.
\end{aligned}
\label{26-7-15}
\end{multline}
So, we picked up
\begin{gather}
b=C\min ( Z ^{-\frac{5}{21}},\bar{r})=
\left\{\begin{aligned}
& Z ^{-\frac{5}{21}} &&\text{if\ \ } B\le Z^{\frac{20}{21}},\\
& B^{-\frac{1}{4}} &&\text{if\ \ } Z^{\frac{20}{21}}\le B\le Z^{\frac{4}{3}},\\
& B^{-\frac{2}{5}}Z^{\frac{1}{5}} &&\text{if\ \ } Z^{\frac{4}{3}}\le B\le Z^3
\end{aligned}\right.
\label{26-7-16}\\
\shortintertext{and}
\varepsilon = \min (Z^{-\frac{2}{3}}, B^{-\frac{4}{5}}Z^{\frac{2}{5}}).
\label{26-7-17}
\end{gather}
Then, preserving all the estimates one can take $W=\rho =0$ at
$\supp (\theta _{\frac{b}{2}})$\,\footnote{\label{foot-26-43} For $B\ge Z^{\frac{20}{21}}$ this is fulfilled automatically.} and then
\begin{multline}
\cI_1+\cI_2= \int \theta _b(x)W(x)\ell (x)\rho_\Psi (x)\,dx-\\
\int \Bigl(\rho _\Psi ^{(2)}(x,y)-\rho _\Psi (x)\rho (y)\Bigr)
\ell (x)|x-y| ^{-1}\bigl(1-\theta _b(y)\bigr)\theta _b(x)\,dxdy\le \cI_0.
\label{26-7-18}
\end{multline}
Further, since $\int \bigl(1-\theta _b(y)\bigr) \rho (y)\,dy \le Z$\,\footnote{\label{foot-26-44} Actually for $B\ge Z^{\frac{20}{21}}$ this is an equality.} we get from (\ref{26-7-8}) and estimate (\ref{26-7-12}) for $\cI_3$ that
\begin{equation}
(N-Z)\le Cb ^{\frac{1}{2}}\cS ^{\frac{1}{2}}+
C\Theta _b ^{-\frac{1}{2}}N^{\frac{1}{2}}b ^{-1}+
Cb ^{-1}\Theta _{b(1-\epsilon )}\Theta _b ^{-1}
\label{26-7-19}
\end{equation}
because then $\varepsilon N ^{\frac{1}{2}}b ^{-1}$ does not exceed
$Cb ^{\frac{1}{2}}\bar{Q} ^{\frac{1}{2}}$.

Let us assume that estimate (\ref{26-7-20}) below does not hold. Then
$\Theta _b=N-\int \bigl(1-\theta _b(y)\bigr)\rho _\Psi (y)\,dy$ and due to Theorem~\ref{thm-26-6-16}
\begin{equation*}
|\Theta _b-N-Z|\le Cb ^{\frac{1}{2}}\bar{Q} ^{\frac{1}{2}}\le \frac{1}{2}(N-Z)
\end{equation*}
and the same is true for $\Theta _{b(1-\epsilon )}$. Then (\ref{26-7-19}) yields (\ref{26-7-20}). So, (\ref{26-7-20}) has been proven.

Thus we proved the following theorem:\enlargethispage{\baselineskip}

\begin{theorem}\label{thm-26-7-4}
Let condition \textup{(\ref{26-2-28})} be fulfilled. In the fixed nuclei model let $\I_N>0$.
\begin{enumerate}[label=(\roman*), wide, labelindent=0pt]
\item\label{thm-26-7-4-i}
Then
\begin{equation}
(N-Z)_+ \le C\left\{\begin{aligned}
&Z ^{\frac{5}{7}}\qquad \qquad
&&\text{if\ \ } B\le Z ^{\frac{20}{21}},\\[3pt]
&Z ^{\frac{5}{6}}B ^{-\frac{1}{8}}+B ^{\frac{1}{2}}L
&&\text{if\ \ } Z ^{\frac{20}{21}}\le B\le Z ^{\frac{4}{3}}L \\[3pt]
&Z ^{\frac{2}{5}}B ^{\frac{1}{5}}L \qquad \qquad
&&\text{if\ \ } Z ^{\frac{4}{3}}\le B\le Z ^3
\end{aligned}\right.
\label{26-7-20}
\end{equation}
where $L=|\log (Z^{-3}B)|$.
\item\label{thm-26-7-4-ii}
For $M=1$ the same estimate holds with $L=1$:
\begin{equation}
(N-Z)_+ \le C\left\{\begin{aligned}
&Z ^{\frac{5}{7}}\qquad \qquad
&&\text{if\ \ } B\le Z ^{\frac{20}{21}},\\[3pt]
&Z ^{\frac{5}{6}}B ^{-\frac{1}{8}}
&&\text{if\ \ } Z ^{\frac{20}{21}}\le B\le Z ^{\frac{4}{3}}L \\[3pt]
&Z ^{\frac{2}{5}}B ^{\frac{1}{5}} \qquad \qquad
&&\text{if\ \ } Z ^{\frac{4}{3}}\le B\le Z ^3.
\end{aligned}\right.
\label{26-7-21}
\end{equation}
\end{enumerate}
\end{theorem}

Furthermore, for $B\le Z$ one can use a slightly sharper estimate for $\bar{Q}$:

\begin{theorem}\label{thm-26-7-5}
Let condition \textup{(\ref{26-2-28})} be fulfilled. In the fixed nuclei model let $\I_N>0$. Then for a single atom and for molecule with
$B\le Z$ and $a\ge Z^{-\frac{1}{3}+\delta_1}$
\begin{equation}
(N-Z)_+ \le C\left\{\begin{aligned}
&Z ^{\frac{5}{7}-\delta}\qquad \qquad
&&\text{if\ \ } B\le Z ^{\frac{20}{21}},\\[3pt]
&Z ^{\frac{5}{6}-\delta}B ^{-\frac{1}{8}+\delta}
&&\text{if\ \ } Z ^{\frac{20}{21}}\le B\le Z
\end{aligned}\right.
\label{26-7-22}
\end{equation}
\end{theorem}

Results for a free nuclei model follow from the above results and an estimate of $a$ from below (see Subsubsection~\emph{\ref{sect-26-8-4-4}.4. \nameref{sect-26-8-4-4}\/}).

\begin{theorem}\label{thm-26-7-6}
Let condition \textup{(\ref{26-2-28})} be fulfilled. In the free nuclei model let $\hat{\I}_N>0$. Then
\begin{enumerate}[label=(\roman*), wide, labelindent=0pt]
\item\label{thm-26-7-6-i}
Estimate \textup{(\ref{26-7-20})} holds.

\item\label{thm-26-7-6-ii}
For $B\le Z$ estimate \textup{(\ref{26-7-22})} holds.
\end{enumerate}
\end{theorem}

\section{Estimate for ionization energy}
\label{sect-26-7-3}
Finally, let us estimate the ionization energy, assuming that

\begin{claim}\label{26-7-23}
$(Z-N)_+$ does not exceed the right-hand expression of (\ref{26-7-20})\footnote{\label{foot-26-45} Or (\ref{26-7-21}), or (\ref{26-7-22}) in the framework of the corresponding theorem.}.
\end{claim}

Few cases are possible:\enlargethispage{2\baselineskip}

\begin{enumerate}[label=(\roman*), wide, labelindent=0pt]
\item\label{sect-26-7-3-i}
$B\le Z^{\frac{20}{21}}$ and $(Z-N)_+\le C_0 Z^{\frac{5}{7}}$. In this case we act exactly as in Subsection~\ref{book_new-sect-25-5-2}: we pick up
$b=\epsilon Z ^{-\frac{5}{21}}$ with a small enough constant
$\epsilon '>0$; then
\begin{gather}
|\int \theta_b (x)\bigl(\rho _\Psi -\rho\bigr)\,dx|\le
Cb^{\frac{1}{2}}Q^{\frac{1}{2}},
\label{26-7-24}\\
\shortintertext{while}
\int \theta_b (x)\rho\,dx \asymp b^{-3}
\label{26-7-25}
\shortintertext{and therefore}
\Theta \coloneqq \int \theta_b (x)\rho_\Psi \,dx \asymp b^{-3}
\label{26-7-26}
\shortintertext{and}
|\int \theta (x)\bigl(\rho _\Psi -\rho\bigr)\,dx|\le
\epsilon ''\Theta.
\label{26-7-27}
\end{gather}
Then (\ref{26-7-8}), (\ref{26-7-12}), (\ref{26-7-15}) yield that
$\I_N\le CZ ^{\frac{20}{21}}$; so estimate (\ref{26-7-37}) below in this case is recovered.

\smallskip
In all other cases one needs to replace $\theta _b$ by a function which is not $b$-admissible.

\item\label{sect-26-7-3-ii}
Let $Z ^{\frac{20}{21}}\le B\le Z^3$ and $M=1$. Let here $\bar{r}$ be the exact radius of $\supp (\rho)$, $\rho=\rho^\TF_B$ and $W=W^\TF_B$, which were obtained in the Thomas-Fermi theory with $\nu=0$. Recall that
$\bar{r}\asymp \max (B^{-\frac{1}{4}};\,B^{-\frac{2}{5}}Z^{\frac{1}{5}})$ and
$\bar{Q}\asymp \max (Z^{\frac{5}{3}};\, B^{\frac{4}{5}}Z^{\frac{3}{5}})$. Also recall that $W\asymp G t^4$ and $\rho \asymp B G^{\frac{1}{2}}$ for $r=(1-t)\bar{r}$ with $1-\epsilon \le t\le 1$, where  $G\coloneqq \min (B;\, B^{\frac{2}{5}}Z^{\frac{4}{5}})$.

We take in this case $\bar{r}t$-admissible function $\theta$, equal $0$ for
$|x-\y|\le \bar{r}(1-t)$ and equal $1$ for $|x-\y|\ge \bar{r}(1-\frac{1}{2}t)$.

\begin{claim}\label{26-7-28}
In all the above estimates one needs to replace
$Cb^{-1}\Theta _{b(1-\epsilon )}$ by $C\bar{r} ^{-1}t ^{-1}\Theta '$ with $\Theta '$ defined by $\theta '$ which is also $\bar{r}t$-admissible and equal $1$ in $\epsilon \bar{r}t$-vicinity of
$\supp (\theta )$.
\end{claim}

Then (\ref{26-7-24})--(\ref{26-7-27}) are replaced by
\begin{gather}
|\int \theta (x)\bigl(\rho _\Psi -\rho\bigr)\,dx|\le CQ^{\frac{1}{2}}\times \|\nabla \theta\| \asymp
Ct^{-\frac{1}{2}}\bar{r}^{\frac{1}{2}}Q^{\frac{1}{2}}
\label{26-7-29}\\
\shortintertext{while}
\int \theta (x)\rho\,dx \asymp B G^{\frac{1}{2}}\bar{r}^3 t^3
\label{26-7-30}\\
\shortintertext{and therefore}
\Theta \coloneqq \int \theta (x)\rho_\Psi \,dx \asymp
B G^{\frac{1}{2}}\bar{r}^3 t^3.
\label{26-7-31}
\end{gather}
Then (\ref{26-7-27}) holds provided the right-hand expression of (\ref{26-7-29}) does not exceed the right-hand expression of (\ref{26-7-31}), multiplied by $\epsilon$:
\begin{equation}
t= t_*\coloneqq C_0 B^{-\frac{2}{7}}G^{-\frac{1}{7}}\bar{r}^{-\frac{5}{7}}Q^{\frac{1}{7}}=
C_1 \max (B^{-\frac{1}{4}}Z^{\frac{5}{21}};\, B^{\frac{2}{35}}Z^{-\frac{6}{35}})
\label{26-7-32}
\end{equation}
where we picked up the smallest possible value of $t$. Note that
\begin{claim}\label{26-7-33}
$t\asymp 1$ as either $B\asymp Z^{\frac{20}{21}}$ or $B\asymp Z^3$.
\end{claim}

\pagebreak
Further, let us estimate from above
\begin{multline}
\cI'=
-\int \Bigl(\rho _\Psi ^{(2)}(x,y)-\rho_\Psi (x)\rho (y)\Bigr)
\ell (x)|x-y| ^{-1}\theta (x)\,dxdy\le \\
-\int \Bigl(\rho _\Psi ^{(2)}(x,y)-\rho _\Psi (x) \rho (y)\Bigr)
\bigl(1-\omega_\tau (x,y)\bigr)\ell (x)|x-y| ^{-1}\theta (x)\,dxdy +\\
\int \rho _\Psi (x) \rho (y)\omega _\tau (x,y) \ell (x)|x-y| ^{-1}
\theta (x)\,dxdy
\label{26-7-34}
\end{multline}
with $\omega =0$ as $|x-y|\ge 2\tau \bar{r}$ and $\omega =1$ as
$|x-y|\le \tau \bar{r}$, with $\tau \in ( t,1)$.

Then due to Proposition~\ref{prop-26-7-2} with
$\chi(x,y)= \bigl(1-\omega_\tau (x,y)\bigr)|\x-y|^{-1}$ the first term in the right-hand expression does not exceed
$C\bar{r}^{\frac{1}{2}}\tau^{-\frac{1}{2}} Q^{\frac{1}{2}}\Theta$ since
$\|\nabla_y\chi_x\|_{\sL^2(\bR^3)}\asymp (\bar{r}\tau)^{-\frac{1}{2}}$ and also one can prove easily that all other terms in
$\Bigl((Q+\varepsilon ^{-1}N+T)^{\frac{1}{2}}\Theta + P^{\frac{1}{2}}\Theta^{\frac{1}{2}}\Bigr)$ do not exceed $CQ \Theta$.

Meanwhile, the second term in in the right-hand expression of (\ref{26-7-34}) does not exceed
$CB G ^{\frac{1}{2}}\tau ^2\times \bar{r} ^3\tau ^2\times \Theta $ because
$ \rho(y) \le CBG ^{\frac{1}{2}}\tau ^2$ if $|x-y|\le 2\tau \bar{r}$,
$x\in \supp (\theta )$ and therefore
$\int \rho(y)\omega_\tau(x,y)\,dy \le CBG ^{\frac{1}{2}}\tau ^4$.

Minimizing their sum
\begin{gather*}
C\Bigl(\bar{r}^{\frac{1}{2}}\tau^{-\frac{1}{2}} Q^{\frac{1}{2}}+
B G ^{\frac{1}{2}}\bar{r} ^3\tau ^4 \Bigr)\Theta
\intertext{with respect to $\tau \ge t$\,\footnotemark, we arrive to estimate}
\cI '\le C\bar{r} ^{\frac{7}{9}}Q ^{\frac{4}{9}}B ^{\frac{1}{9}}
G ^{\frac{1}{18}}\Theta .
\end{gather*}
\footnotetext{\label{foot-26-46} One can see easily that minimum is achieved as $\tau\asymp t^{\frac{7}{9}}$.}
Then exactly as in the proof of Theorem~\ref{book_new-thm-25-5-3} we have inequality
\begin{equation}
\bar{r}\,\I_N \le C(Z-N)_+
+ C\bar{r} ^{\frac{7}{9}}Q ^{\frac{4}{9}}B ^{\frac{1}{9}} G ^{\frac{1}{18}},
\label{26-7-35}
\end{equation}
and therefore for
$(Z-N)_+\le
C\bar{r} ^{\frac{7}{9}}Q ^{\frac{4}{9}}B ^{\frac{1}{9}} G ^{\frac{1}{18}}$
we arrive to the estimate
$\I_N\le
C\bar{r} ^{-\frac{2}{9}}Q ^{\frac{4}{9}}B ^{\frac{1}{9}} G ^{\frac{1}{18}}$.
\end{enumerate}

Thus we have proven estimate~(\ref{26-7-37}) of Theorem~\ref{thm-26-7-7} below, at least as $N\ge Z$. Further, estimate~(\ref{26-7-39}) under the same assumption
$N\ge Z$ is due to the fact that for $B\le Z$ one can use
$\bar{Q}= Z^{\frac{5}{3}}(B^{\delta}Z^{-\delta}+Z^{-\delta})$ instead of $Q$.

\begin{theorem}\label{thm-26-7-7}
Let $M=1$.
\begin{enumerate}[label=(\roman*), wide, labelindent=0pt]
\item\label{thm-26-7-7-i} Then for $B\le Z^{3}$ and
\begin{align}
(Z-N)_+\le C_0&\left\{\begin{aligned}
&Z^{\frac{5}{7}}
&&\text{if\ \ } B\le Z^{\frac{20}{21}},\\
& B^{-\frac{1}{8}} Z^{\frac{5}{6}}
&&\text{if\ \ } Z^{\frac{20}{21}}\le B\le Z^{\frac{4}{3}},\\
&B^{\frac{1}{5}}Z^{\frac{2}{5}}\qquad
&&\text{if\ \ } Z ^{\frac{4}{3}}\le B\le Z ^3
\end{aligned}\right.\label{26-7-36}\\
\intertext{the following estimate holds}
\I_N\le C&\left\{\begin{aligned}
&Z^{\frac{20}{21}}
&&\text{if\ \ } B\le Z^{\frac{20}{21}},\\
&B^{\frac{2}{9}}Z ^{\frac{20}{27}}\qquad
&&\text{if\ \ } Z ^{\frac{20}{21}}\le B\le Z ^{\frac{4}{3}}\\
&B ^{\frac{26}{45}}Z ^{\frac{4}{15}}\qquad
&&\text{if\ \ } Z ^{\frac{4}{3}}\le B\le Z ^3.
\end{aligned}\right.
\label{26-7-37}
\end{align}
\item\label{thm-26-7-7-ii}
Furthermore for $B\le Z$ and
\begin{align}
(Z-N)_+\le C_0&\left\{\begin{aligned}
&Z^{\frac{5}{7}-\delta}
&&\text{if\ \ } B\le Z^{\frac{20}{21}},\\
& B^{-\frac{1}{8}+\delta} Z^{\frac{5}{6}-\delta}
&&\text{if\ \ } Z^{\frac{20}{21}}\le B\le Z\\
\end{aligned}\right.\label{26-7-38}\\
\intertext{the following estimate holds}
\I_N\le C&\left\{\begin{aligned}
&Z^{\frac{20}{21}-\delta'}
&&\text{if\ \ } B\le Z^{\frac{20}{21}},\\
&B^{\frac{2}{9}+\delta'}Z ^{\frac{20}{27}-\delta'}\qquad
&&\text{if\ \ } Z ^{\frac{20}{21}}\le B\le Z.
\end{aligned}\right.
\label{26-7-39}
\end{align}
\end{enumerate}
\end{theorem}

\begin{proof}[Proof in the general settings] To prove estimates (\ref{26-7-37}) and (\ref{26-7-39}) in the general settings (i.e. without assumption $N\ge Z$) observe that for $N<Z$
\begin{multline}
\D (\rho ^\TF _N- \rho^\TF_Z,\, \rho ^\TF _N - \rho^\TF_Z)\le
C(Z-N)^2 \bar{r}^{-1}\asymp \\
C\max
\bigl((Z-N)^{\frac{7}{3}};\, C(Z-N)^2 B^{\frac{1}{4}};\, C(Z-N)^2 B^{\frac{2}{5}} Z^{-\frac{1}{5}}\bigr)
\label{26-7-40}
\end{multline}
(where subscript here denotes the number of electrons rather than the intensity of the magnetic field) because the same estimate holds for $\cE^\TF_N-\cE^\TF_Z$:
\begin{gather}
0\le \cE^\TF_N-\cE^\TF_Z\le C(Z-N)^2 \bar{r}^{-1},
\label{26-7-41}\\
\intertext{which itself follows from }
\frac{\partial \cE^\TF}{\partial N}=\nu \asymp (Z-N)\bar{r}^{-1}.
\label{26-7-42}
\end{gather}

Therefore to preserve our estimates we need to assume that the right-hand expression of (\ref{26-7-40}) does not exceed $Q$; this assumption is equivalent to
$(Z-N)_+\le \min ( Z^{\frac{5}{7}};\,Z^{\frac{5}{6}}B^{-\frac{1}{8}})$ for
$B\le Z^{\frac{4}{3}}$ which is exactly the first and the second cases in (\ref{26-7-36}) (and these cases in (\ref{26-7-38}) appear in the same way), and to $(Z-N)_+\le CB^{\frac{1}{5}}Z^{\frac{2}{5}}$ for
$Z^{\frac{4}{3}}\le B\le Z^3$, which is exactly the third case in (\ref{26-7-40}).

Also there is a term $C(Z-N)_+\bar{r}^{-1}$ in the estimate of $\I_N$. However, under assumption (\ref{26-7-40}) this term does not exceed the right hand expression of (\ref{26-7-40}) or (\ref{26-7-42}), in fact coincides with it only in the first case. \end{proof}

Consider now $M\ge 2$. Assume that $B\ge Z^{\frac{20}{21}}$ since the opposite case has been analyzed already.\enlargethispage{\baselineskip}

Let us pick up $\bar{r}t$-admissible function $\theta$ such that $\theta =1$ if $W\le C_0 G t ^4$ and $\theta =0$ if $W\ge 2C_0G t ^4$. In this case ($M\ge 2$) we can claim only that
$\|\nabla \theta\|\le Ct^{-\frac{1}{2}}\bar{r}^{\frac{1}{2}}
|\log t|^{\frac{1}{2}}$ and
therefore
\begin{gather}
|\int \theta (x)\bigl(\rho _\Psi -\rho\bigr)\,dx|\le
Ct^{-\frac{1}{2}}|\log t|^{\frac{1}{2}}\bar{r}^{\frac{1}{2}}\bar{Q}^{\frac{1}{2}}
\tag*{$\textup{(\ref*{26-7-29})}'$}\label{26-7-29-'},\\
\shortintertext{while}
B G^{\frac{1}{2}}\bar{r}^3 t^3 \lesssim
\int \theta (x)\rho\,dx \lesssim B G^{\frac{1}{2}}|\log t|\bar{r}^3 t^3
\tag*{$\textup{(\ref*{26-7-30})}'$}\label{26-7-30-'}\\
\shortintertext{and therefore}
\Theta \coloneqq \int \theta (x)\rho_\Psi \,dx \gtrsim
B G^{\frac{1}{2}}\bar{r}^3 t^3
\tag*{$\textup{(\ref*{26-7-31})}'$}\label{26-7-31-'}
\shortintertext{for}
t\ge t_*\coloneqq C_0B^{-\frac{2}{7}}G^{-\frac{1}{7}}\bar{r}^{-\frac{5}{7}}\bar{Q}^{\frac{1}{7}}
|\log t|^{\frac{2}{7}}.
\tag*{$\textup{(\ref*{26-7-32})}'$}\label{26-7-32-'}
\end{gather}

Now we need to look more carefully at $\bar{Q}$, especially because while it may contain ``rogue'' factor $L$ or $L^2$, it can also be large as $(Z-N)_+$ is large. Fortunately, this is not the case in the current framework:

\begin{proposition}\label{prop-26-7-8}
\begin{enumerate}[label=(\roman*), wide, labelindent=0pt]
\item\label{prop-26-7-8-i}
Under condition \textup{(\ref{26-7-46})} below $\bar{Q}$ is as in the case $N=Z$ i.e.
\begin{equation}
\bar{Q}=\left\{\begin{aligned}
&Z^{\frac{5}{3}}+ B^{\frac{5}{4}}L^2 \qquad
&&\text{if\ \ }B\le Z^{\frac{4}{3}},\\
&B^{\frac{4}{5}}Z^{\frac{3}{5}}L^2 \qquad
&&\text{if\ \ } Z^{\frac{4}{3}}\le B\le Z^3.
\end{aligned}\right.
\label{26-7-43}
\end{equation}
\item\label{prop-26-7-8-ii}
Furthermore, if $B\le Z$ and $a\ge Z^{-\frac{1}{3}}$ under condition \textup{(\ref{26-7-48})} below $\bar{Q}$ is exactly as in the case $N=Z$, i.e.
\begin{equation}
\bar{Q}= Z^{\frac{5}{3}}
\bigl( Z^{-\delta}+ (aZ^{\frac{1}{3}})^{-\delta}+(BZ^{-1})^\delta\bigr).
\label{26-7-44}
\end{equation}
\end{enumerate}
\end{proposition}

\begin{proof}
One can either derive it from the existing estimates or just repeat estimates with $\nu=0$ adding $(Z-N)_+^2\bar{r}^{-1}$ to $\bar{Q}$. We leave easy details to the reader.\end{proof}

Therefore all the above arguments could be repeated with this new expression $\bar{Q}$ which also acquires factor $|\log t|$ (due to this factor in the estimate of
$\|\nabla \theta\|$ and this factor boils to $L_1^{\frac{1}{2}}$ with
\begin{equation}
L_1=\left\{\begin{aligned}
&|\log B Z^{-\frac{20}{21}}|+1 && Z^{\frac{20}{21}}\le B\le Z^{\frac{4}{3}},\\
&|\log B Z^{-3}|+1 && Z^{\frac{4}{3}}\le B\le Z^3.
\end{aligned}\right.
\label{26-7-45}
\end{equation}

Therefore we arrive to

\begin{theorem}\label{thm-26-7-9}
Let $M\ge 2$. Then
\begin{enumerate}[label=(\roman*), wide, labelindent=0pt]
\item\label{thm-26-7-9-i}
For
\begin{align}
(Z-N)_+\le C_0&\left\{\begin{aligned}
& Z^{\frac{5}{7}}
\qquad&&\text{if\ \ } B\le Z^{\frac{20}{21}},\\
& Z^{\frac{5}{6}}B^{-\frac{1}{8}}+B^{\frac{1}{2}}L
\qquad&&\text{if\ \ } Z^{\frac{20}{21}}\le B\le Z^{\frac{4}{3}},\\
&B^{\frac{1}{5}}Z^{\frac{2}{5}}L
\qquad&&\text{if\ \ }Z^{\frac{4}{3}}\le B\le Z^3
\end{aligned}\right.
\label{26-7-46}\\
\intertext{the following estimate holds}
\I_N\le CL_1^{\frac{2}{9}}&\left\{\begin{aligned}
&Z^{\frac{20}{21}}
&&\text{if\ \ } B\le Z^{\frac{20}{21}},\\
&Z ^{\frac{20}{27}}B^{\frac{2}{9}}+B ^{\frac{7}{9}}L ^{\frac{8}{9}}\qquad
&&\text{if\ \ } Z ^{\frac{20}{21}}\le B\le Z ^{\frac{4}{3}}\\
&Z ^{\frac{4}{15}}B ^{\frac{26}{45}}L ^{\frac{8}{9}}\qquad
&&\text{if\ \ } Z ^{\frac{4}{3}}\le B\le Z ^3.
\end{aligned}\right.
\label{26-7-47}
\end{align}
\item\label{thm-26-7-9-ii}
Furthermore, for $B\le Z$, $a\ge Z^{-\frac{1}{3}}$ and
\begin{align}
(Z-N)_+\le C_0\varsigma^\delta&\left\{\begin{aligned}
& Z^{\frac{5}{7}}
\qquad&&\text{if\ \ } B\le Z^{\frac{20}{21}},\\
& Z^{\frac{5}{6}}B^{-\frac{1}{8}}
\qquad&&\text{if\ \ } Z^{\frac{20}{21}}\le B\le Z ,\\
\end{aligned}\right.
\label{26-7-48}\\
\shortintertext{with}
&\varsigma = Z^{-1} + BZ^{-1} +a ^{-1}Z^{\frac{1}{3}}
\label{26-7-49}\\
\intertext{the following estimate holds}
\I_N\le CL_1^{\frac{2}{9}}\varsigma^{\delta'}
&\left\{\begin{aligned}
&Z^{\frac{20}{21}}
&&\text{if\ \ } B\le Z^{\frac{20}{21}},\\
&Z ^{\frac{20}{27}}B^{\frac{2}{9}}\qquad
&&\text{if\ \ } Z ^{\frac{20}{21}}\le B\le Z.
\end{aligned}\right.
\label{26-7-50}
\end{align}
\end{enumerate}
\end{theorem}

\chapter{Positively charged systems}
\label{sect-26-8}

Now let us estimate from above and below the ionization energy in the case when $N<Z$ and condition (\ref{26-7-36}) (if $M=1$) or (\ref{26-7-46}) (if $M\ge 2$) fails. We also estimate excessive the positive charge in the case of $M\ge 2$ and free nuclei model. We will follow arguments of the corresponding three subsections of Section~\ref{book_new-sect-25-6}.

\section{Upper estimate for ionization energy: $M=1$}
\label{sect-26-8-1}

Consider first the case of $M=1$. Then for $B=0$ arguments are well-known (see Section~\ref{book_new-sect-25-6}) but we repeat them for $B>0$: we pick up $\beta$-admissible function $\theta$ such that $\theta=1$ if
$|x-\y_1|\ge \bar{r}-\beta$ and $\theta=0$ if $|x-\y_1|\le \bar{r}-2\beta$ where $\bar{r}$ is an exact radius of support of $\rho^\TF$ (see the very beginning of Subsection~\ref{book_new-sect-25-6-1}) and $\beta \ll \bar{r}$. Recall that
\begin{equation}
\bar{r}\asymp \left\{\begin{aligned}
&(Z-N)^{-\frac{1}{3}} &&\text{if\ \ } B\le Z^{\frac{20}{21}},\\
&\min\bigl((Z-N)^{-\frac{1}{3}}, B^{-\frac{1}{4}}\bigr)
&&\text{if\ \ } Z^{\frac{20}{21}}\le B\le Z^{\frac{4}{3}},\\
&B^{-\frac{2}{5}}Z^{\frac{1}{5}}
&&\text{if\ \ }Z^{\frac{4}{3}}\le B\le Z^3,
\end{aligned}\right.
\label{26-8-1}
\end{equation}
where in the first case we used that $Z-N\ge Z^{\frac{5}{7}}$ while
in the second case both subcases $(Z-N)^{-\frac{1}{3}}\gtrless B^{-\frac{1}{4}}$ are possible.

We can assume without any loss of the generality that $\y_1=0$. Now in the spirit of Subsection~\ref{book_new-sect-25-6-1} we need to select as we did in Subsection~\ref{sect-26-7-3} the smallest $\beta$ such that
\begin{gather}
\Theta^\TF \coloneqq \int \theta (x)\rho^\TF(x)\, dx \ge C\beta^{-\frac{1}{2}}\bar{r} \bar{Q}^{\frac{1}{2}}
\label{26-8-2}
\intertext{implying that}
\Theta_\Psi \coloneqq \int \theta (x)\rho_\Psi(x)\, dx \asymp \Theta^\TF,
\label{26-8-3}
\end{gather}
where the right-hand expression of (\ref{26-8-2}) estimates
$|\int \theta (x)(\rho^\TF-\rho_\Psi)\,dx|$ (recall that it does not exceed
$\|\nabla \theta\|\cdot \D(\rho^\TF-\rho_\Psi,\rho^\TF-\rho_\Psi)^{\frac{1}{2}}$). Again as in Subsection~\ref{book_new-sect-25-6-1} $\rho^\TF=\rho^\TF_N$ is calculated for the actual value of $N<Z$.

Then, following Subsubsection~\ref{book_new-sect-25-6-1}, eventually we arrive to estimate (\ref{book_new-25-6-8}), namely:
\pagebreak
\begin{multline}
\I_N \int \ell(x) \rho_\Psi (x)\theta (x)\,dx \le
\int \theta (x)V(x)\ell (x)\rho _\Psi (x)\,dx \\
\begin{aligned}
-&\int \Bigl(\rho _\Psi ^{(2)}(x,y)-\rho _\Psi (x)\rho (y)\Bigr)
\ell (x)|x-y| ^{-1}\theta (x)\, dxdy \\
-&\int \rho _\Psi (x)\rho (y)\ell (x)|x-y| ^{-1}\theta (x)\, dxdy + C\beta^{-2}\bar{r}\Theta ,
\end{aligned}
\label{26-8-4}
\end{multline}
and then estimate from above the second term in the right-hand expression
\begin{multline}
-\int \Bigl(\rho _\Psi ^{(2)}(x,y)-\rho _\Psi(x)\rho(y)\Bigr)
\ell (x)|x-y| ^{-1}\theta (x)\, dxdy\le \\
-\int \Bigl(\rho _\Psi ^{(2)}(x,y)-\rho _\Psi (x)\rho (y)\Bigr)
\bigl(1-\omega (x,y)\bigr)\ell (x)|x-y| ^{-1}\theta (x)\,dxdy \\
+\int \rho_\Psi (x)\rho(y)\omega (x,y) \ell(x) |x-y|^{-1}\theta (x)\,dxdy
\label{26-8-5}
\end{multline}
with $\omega=\omega_\gamma$: $\omega =0$ if $|x-y|\ge 2\gamma$ and $\omega =1$ if $|x-y|\le \gamma$, $\gamma \ge \beta$ (see (\ref{book_new-25-6-9})).

To estimate the first term in the right-hand expression of (\ref{26-8-5}) one can apply Proposition~\ref{book_new-prop-25-5-1}. In this case
$\|\nabla_y\chi \|_{\sL^2} \asymp C\bar{r}\gamma^{-\frac{1}{2}}$,
$\|\nabla_y\chi \|_{\sL^\infty} \asymp \bar{r} \gamma^{-2}$
and plugging $P=\beta^{-2}\Theta$ and $T=|\nu|$, $\varepsilon =Z^{-\frac{2}{3}}$ we conclude that this term does not exceed (\ref{book_new-25-6-10})
\begin{equation}
C\bar{r}\bigl(\gamma^{-\frac{1}{2}}Q^{\frac{1}{2}} +
Z^{\frac{1}{3}} \gamma^{-2}\bigr) \Theta
\label{26-8-6}
\end{equation}
(if $Q\ge Z^{\frac{5}{3}}$; otherwise here we should reset here
$Q\coloneqq Z^{\frac{5}{3}}$).

Note that if $0\le \bar{r}-|x|\asymp \beta$
\begin{gather}
W+\nu \asymp \upsilon\coloneqq
\max \Bigl\{ \bigl(\frac{|\nu| \beta }{\bar{r}} \bigr);\,
G \bigl(\frac{\beta}{\bar{r}}\bigr)^4\Bigr\},
\label{26-8-7}\\
\intertext{with $G$ defined by (\ref{26-2-41}) and therefore}
\rho \asymp
\max \Bigl\{ \bigl(\frac{|\nu| \beta }{\bar{r}} \bigr)^{\frac{3}{2}};\,
B \bigl(\frac{|\nu |\beta }{\bar{r}} \bigr)^{\frac{1}{2}};\,
B G^{\frac{1}{2}} \bigl(\frac{\beta}{\bar{r}}\bigr)^2\Bigr\}
\label{26-8-8}
\end{gather}
where the first and the second clauses are forks of the first clause in (\ref{26-8-7}) since in the second clause automatically
$W+\nu \le B$ for $0\le \bar{r}-|x|\lesssim \beta$; therefore
\begin{gather}
\int \rho (x)\theta(x)\,dx \asymp
\max \Bigl\{ \bigl(\frac{|\nu| \beta }{\bar{r}} \bigr)^{\frac{3}{2}};\,
\ B \bigl(\frac{|\nu |\beta }{\bar{r}} \bigr)^{\frac{1}{2}};\,
B G^{\frac{1}{2}} \bigl(\frac{\beta}{\bar{r}}\bigr)^2\Bigr\} \beta\bar{r}^2,
\label{26-8-9}\\
\intertext{and therefore (\ref{26-8-2}) holds if and only if}
\max \Bigl\{ \bigl(\frac{|\nu| }{\bar{r}} \bigr)^{\frac{3}{2}}\beta ^3;\,
\ B \bigl(\frac{|\nu | }{\bar{r}} \bigr)^{\frac{1}{2}}\beta^2;\,
B G^{\frac{1}{2}} \bigl(\frac{1}{\bar{r}}\bigr)^2\beta^{\frac{7}{2}}\Bigr\}
\bar{r}\ge CQ^{\frac{1}{2}};
\label{26-8-10}\\
\shortintertext{then}
\beta=\min \Bigl\{ Q^{\frac{1}{6}}|\nu|^{-\frac{1}{2}}\bar{r}^{\frac{1}{6}};\,
B^{-\frac{1}{2}} Q^{\frac{1}{4}}|\nu|^{-\frac{1}{4}}\bar{r}^{-\frac{1}{4}};\,
B ^{-\frac{2}{7}} G^{-\frac{1}{7}} Q^{\frac{1}{7}}\bar{r}^{\frac{2}{7}}\Bigr\}
\label{26-8-11}\\
\intertext{and in the corresponding cases}
\upsilon=
\Bigl\{ Q^{\frac{1}{6}}|\nu|^{\frac{1}{2}}\bar{r}^{-\frac{5}{6}};\,
B^{-\frac{1}{2}} Q^{\frac{1}{4}}|\nu|^{\frac{3}{4}}\bar{r}^{-\frac{5}{4}};\,
B ^{-\frac{8}{7}}G^{\frac{3}{7}} Q^{\frac{4}{7}}\bar{r}^{-\frac{20}{7}}
\Bigr\}.
\label{26-8-12}
\end{gather}

Observe, however, that for $B\lesssim Q^{\frac{4}{7}}$ and
$|\nu|\lesssim Q^{\frac{4}{7}}$ we do not need these arguments; simpler arguments of Subsection~\ref{book_new-sect-25-5-3} show that in this case
$|\I_N|\le CQ^{\frac{4}{7}}$.

On the other hand, for $B\lesssim Q^{\frac{4}{7}}$ but
$|\nu|\gtrsim Q^{\frac{4}{7}}$, we pick
$\gamma= Q^{\frac{1}{8}}|\nu|^{-\frac{15}{32}}$, like in Subsection~\ref{book_new-sect-25-6-1}, and observe that $|\nu|\bar{r}^{-1}\gamma \gtrsim B$ and therefore we conclude that
$\I_N+\nu \le C Q^{\frac{1}{6}}|\nu|^{\frac{17}{24}}$, exactly like in that subsection. Therefore we arrive to

\begin{proposition}\label{prop-26-8-1}
Let $B\le C_0Z^{\frac{20}{21}}$. Then
\begin{enumerate}[label=(\roman*), wide, labelindent=0pt]
\item\label{prop-26-8-1-i}
If $|\nu|\le C_0Z^{\frac{20}{21}}$, then  estimate \ $\I_N \le C Z^{\frac{20}{21}}$ holds like in the case $B=0$.

\item\label{prop-26-8-1-ii}
If $|\nu|\ge C_0 Z^{\frac{20}{21}}$, then  estimate \
$\I_N +\nu \le C Z^{\frac{5}{18}}|\nu|^{\frac{17}{24}}$ holds like in the case $B=0$.
\end{enumerate}
\end{proposition}
\enlargethispage{2\baselineskip}

Therefore in what follows we assume that $B\ge Q^{\frac{4}{7}}$. One can see easily that then $\beta \le \bar{r}$.

Meanwhile, the same arguments imply that the second term in the right-hand expression of (\ref{26-8-5}) is of magnitude
\begin{gather*}
\max \Bigl\{ \bigl(\frac{|\nu| \gamma }{\bar{r}} \bigr)^{\frac{3}{2}};\,
\ B \bigl(\frac{|\nu |\gamma }{\bar{r}} \bigr)^{\frac{1}{2}};\,
B G^{\frac{1}{2}} \bigl(\frac{\gamma}{\bar{r}}\bigr)^2\Bigr\} \gamma^2
\intertext{and we need to minimize}
\gamma^{-\frac{1}{2}}Q^{\frac{1}{2}}+
\max \Bigl\{ \bigl(\frac{|\nu| \gamma }{\bar{r}} \bigr)^{\frac{3}{2}},\
\ B \bigl(\frac{|\nu |\gamma }{\bar{r}} \bigr)^{\frac{1}{2}};\,
B G^{\frac{1}{2}} \bigl(\frac{\gamma}{\bar{r}}\bigr)^2\Bigr\} \gamma^2,
\intertext{which is achieved when}
\gamma^{-\frac{1}{2}}Q^{\frac{1}{2}}\asymp
\max \Bigl\{ \bigl(\frac{|\nu| \gamma }{\bar{r}} \bigr)^{\frac{3}{2}};\,
\ B \bigl(\frac{|\nu |\gamma }{\bar{r}} \bigr)^{\frac{1}{2}};\,
B G^{\frac{1}{2}} \bigl(\frac{\gamma}{\bar{r}}\bigr)^2\Bigr\} \gamma^2.
\end{gather*}
Let us compare this equation with equation to $\beta$. It is the same albeit with factor $\bar{r}^2$ rather than $\gamma^2$. Therefore if $\gamma \ge \bar{r}$ then $\gamma \le \beta\le \bar{r}$ which is a contradiction. Thus $\gamma\le \bar{r}$ but then $\gamma\ge \beta$.\pagebreak

Therefore we conclude that this term does not exceed
\begin{equation}
\varsigma\coloneqq \max\Bigl\{ Q^{\frac{7}{16}} \bigl(\frac{|\nu|}{\bar{r}}\bigr)^{\frac{3}{16}};\,
Q^{\frac{5}{12}}B^{\frac{1}{6}}\bigl(\frac{|\nu|}{\bar{r}}\bigr)^{\frac{1}{12}};\,
Q^{\frac{4}{9}}B^{\frac{1}{9}}G^{\frac{1}{18}}\bar{r}^{-\frac{2}{9}} \Bigr\},
\label{26-8-13}
\end{equation}
and to estimate $\I_N+\nu$ we need just to compute its sum with $\upsilon$ defined by (\ref{26-8-12}).

Therefore we conclude that
\begin{equation}
\I_N+\nu \le C(\upsilon + \varsigma).
\label{26-8-14}
\end{equation}

\begin{remark}\label{rem-26-8-2}
Observe that
\begin{align}
&\upsilon (Z, B,|\nu|)= Z^{\frac{20}{21}}
\upsilon (1, Z^{-\frac{20}{21}}B, |\nu|Z^{-\frac{20}{21}}|\nu|) &&\text{if\ \ } Z^{\frac{20}{21}}\le B\le Z^{\frac{4}{3}}
\label{26-8-15}\\
\shortintertext{and}
&\upsilon (Z, BZ^{-3},|\nu|)= Z^2 \upsilon (1, Z^{-3}B, |\nu|Z^{-2}|\nu|) &&\text{if\ \ } Z^{\frac{4}{3}}\le B\le Z^3,
\label{26-8-16}
\end{align}
and $\varsigma$ has the same scaling properties.
\end{remark}

Therefore we can make all calculations with $Z=1$ and then scale. Leaving easy calculations to the reader, we arrive to

\begin{proposition}\label{prop-26-8-3}
\begin{enumerate}[label=(\roman*),wide, labelindent=0pt]
\item\label{prop-26-8-3-i}
For $Z^{\frac{20}{21}}\le B\le Z^{\frac{4}{3}}$
\begin{equation}
I_N+\nu \le C\left\{\begin{aligned}
& Z^{\frac{5}{18}} |\nu|^{\frac{17}{24}} \qquad
&&\text{if\ \ } |\nu|\ge Z^{-\frac{20}{51}} B^{\frac{24}{17}},\\
&Z^{\frac{5}{12}} B^{-\frac{1}{2}}|\nu|^{\frac{17}{16}}\qquad
&&\text{if\ \ } B\le |\nu|\le Z^{-\frac{20}{51}} B^{\frac{24}{17}},\\
&Z^{\frac{5}{48}} B^{-\frac{3}{16}} |\nu|^{\frac{3}{4}} \qquad
&&\text{if\ \ } Z^{\frac{5}{12}} B^{\frac{9}{16}}\le |\nu|\le B,\\
&Z^{\frac{25}{36}} B^{\frac{3}{16}}|\nu|^{\frac{1}{12}}\qquad
&&\text{if\ \ } Z^{\frac{5}{9}} B^{\frac{5}{12}}\le |\nu|\le
Z^{\frac{5}{9}} B^{\frac{9}{16}},\\
&Z^{\frac{20}{27}} B^{\frac{2}{9}}
&&\text{if\ \ } |\nu|\le Z^{\frac{5}{9}} B^{\frac{5}{12}}.
\end{aligned}\right.
\label{26-8-17}
\end{equation}
\item\label{prop-26-8-3-ii}
In particular,
\begin{equation}
\qquad I_N\le CZ^{\frac{20}{27}} B^{\frac{2}{9}} \qquad \text{if\ \ }
|\nu|\le Z^{\frac{20}{27}} B^{\frac{2}{9}}.
\label{26-8-18}
\end{equation}
\item\label{prop-26-8-3-iii}
For $Z^{\frac{4}{3}}\le B\le Z^3$
\begin{equation}
I_N+\nu \le C\left\{\begin{aligned}
&Z^{\frac{7}{30}} B^{\frac{8}{15}}|\nu|^{\frac{1}{12}} \qquad
&&\text{if\ \ } |\nu|\ge Z^{\frac{2}{5}} B^{\frac{8}{15}},\\
& Z^{\frac{4}{15}} B^{\frac{26}{45}} \qquad
&&\text{if\ \ } |\nu|\le Z^{\frac{2}{5}} B^{\frac{8}{15}}.
\end{aligned}\right.
\label{26-8-19}
\end{equation}
\item\label{prop-26-8-3-iv}
In particular,
\begin{equation}
\qquad I_N\le CZ^{\frac{4}{15}} B^{\frac{26}{45}}\qquad \text{if\ \ }
|\nu|\le Z^{\frac{4}{15}} B^{\frac{26}{45}}.
\label{26-8-20}
\end{equation}
\end{enumerate}
\end{proposition}

\begin{remark}\label{rem-26-8-4}
Recall that $Q=Z^{\frac{5}{3}}\bigl(B^{\delta} + 1\bigr)Z^{-\delta}$
if $B\le Z$; therefore we can add factor
$\bigl(B^{\delta'} + 1\bigr)Z^{-\delta'}$ in all estimates of Propositions~\ref{prop-26-8-1} and~\ref{prop-26-8-3}.
\end{remark}

\section{Lower estimate for ionization energy: $M=1$}
\label{sect-26-8-2}

Now let us derive an estimate $\I_N +\nu$ from below. Let
$\Psi =\Psi _N(x_1,\ldots, x_N)$ be the ground state for $N$ electrons, $\|\Psi\|=1$; consider an antisymmetric \emph{test function\/}
\begin{multline}
\tilde{\Psi}=\tilde{\Psi}(x_1,\ldots, x_{N+1})=
\Psi (x_1,\ldots, x_N)u(x_{n+1})-\\
\sum_{1\le j\le N}
\Psi (x_1,\ldots, x_{j-1},x_{N+1},x_{j+1},\ldots, x_N)u(x_j)
\label{26-8-21}
\end{multline}
Then exactly as in Subsection~\ref{book_new-sect-25-6-2}
\begin{multline*}
\E_{N+1}\|\tilde{\Psi}\| ^2\le
\blangle \mathbf{H}_{N+1}\tilde{\Psi},\tilde{\Psi} \brangle =
N\blangle \mathbf{H}_{N+1}\Psi u,\tilde{\Psi}\brangle =\\[5pt]
N\blangle \mathbf{H}_N\Psi u,\tilde{\Psi}\brangle +
N\blangle H_{V,x_{N+1}}\Psi u,\tilde{\Psi}\brangle +
N\blangle \sum_{1\le i\le N}|x_i-x_{N+1}| ^{-1}\Psi u, \tilde{\Psi}\brangle =\\
\shoveleft{\quad (\E_N-\nu' )\|\tilde{\Psi}\| ^2 +
N\blangle H_{W+\nu' ,x_{N+1}}\Psi u,\tilde{\Psi}\brangle}\\[5pt]
+N\blangle
\bigl(\sum_{1\le i\le N}|x_i-x_{N+1}| ^{-1}- (V-W)(x_{N+1})\bigr) \Psi u, \tilde{\Psi}\brangle
\end{multline*}
and therefore
\begin{multline}
N ^{-1}(\I_{N+1}+\nu')\|\tilde{\Psi}\| ^2\ge
-\blangle H_{W+\nu',x_{N+1}}\Psi u,\tilde{\Psi}\brangle\\
-\blangle \bigl(\sum_{1\le i\le N}|x_i-x_{N+1}| ^{-1} - (V-W)(x_{N+1})\bigr)
\Psi u, \tilde{\Psi}\brangle
\label{26-8-22}
\end{multline}
and
\begin{multline}
N ^{-1}\|\tilde{\Psi}\| ^2=\|\Psi \| ^2\cdot \|u\| ^2-\\
N\int \Psi (x_1,\ldots,x_{N-1},x)\Psi ^\dag (x_1,\ldots,x_{N-1},y)
u(y)u^\dag(x) \, dx_1\cdots dx_{N-1}\, dxdy
\label{26-8-23}
\end{multline}
as in (\ref{book_new-25-6-14}) and (\ref{book_new-25-6-15}) respectively where $ ^\dag$ means a complex or Hermitian conjugation and $\nu'\ge \nu $ to be chosen later.

Note that every term in the right-hand expression in (\ref{26-8-22}) is the sum of two terms: one with $\tilde{\Psi}$ replaced by
$\Psi (x_1,\ldots,x_N)u(x_{N+1})$ and another with $\tilde{\Psi}$ replaced by $-N\Psi (x_1,\ldots,x_{N-1},x_{N+1})u(x_N)$. We call these terms, as in Subsection~\ref{book_new-sect-25-6-2}, \emph{direct\/} and \emph{indirect\/} respectively.

Obviously, in the direct and indirect terms $u$ appears as
$|u(x)| ^2\,dx$ and as $u(x) u ^\dag (y)\,dxdy$ respectively multiplied by some kernels.

Recall that $u$ is an arbitrary function. Let us take
$u(x)=\theta^{\frac{1}{2}} (x)\phi_j(x)$ where $\phi _j$ are orthonormal eigenfunctions of $H_{W+\nu }$ and $\theta (x)$ is $\beta$-admissible function which is supported in $\{x\colon - \upsilon \ge W(x)+\nu \ge \frac{2}{3} \nu \}$ and equal $1$ in $\{x\colon -2\upsilon \ge W(x)+\nu \ge \frac{1}{2}\nu \}$, satisfying (\ref{book_new-25-5-11}), and $\upsilon$ is related to $\beta$ as in the previous Section~\ref{sect-26-7}:
\begin{equation}
\upsilon = C\max (\nu \bar{r}^{-1}\beta;\, G\bar{r}^{-4}\beta^4).
\label{26-8-24}
\end{equation}

Let us substitute it into (\ref{26-8-22}), multiply by
$\varphi (\lambda _jL ^{-1})$ and take the sum with respect to $j$; then we get the same expressions with $|u(x)| ^2\,dx$ and $u(x)u ^\dag(y)\,dxdy$ replaced by $F(x,x)\,dx$ and $F(x,y)\,dxdy$ respectively with
\begin{equation}
F(x,y)=\int \varphi (\lambda L ^{-1})\,d_\lambda e(x,y,\lambda ).
\label{26-8-25}
\end{equation}
Here $\varphi (\tau )$ is a fixed $\sC ^\infty $ non-negative function equal to $1$ for $\tau \le \frac{1}{2}$ and equal to $0$ for $\tau \ge 1$ and
$L =\nu'-\nu =6\upsilon$.

Under described construction and procedures the direct term generated by
$N ^{-1}\|\tilde{\Psi}\| ^2$ is
\begin{gather}
\int \theta (x) \varphi (\lambda L ^{-1})\,d_\lambda e(x,x,\lambda )\,dx.
\label{26-8-26}\\
\intertext{Then, applying semiclassical approximation, we get}
\Theta_\Psi
\coloneqq \int \varphi (\lambda L ^{-1})\,d_\lambda P'_B(W+\nu -\lambda )\,dx.
\label{26-8-27}
\end{gather}
Consider the remainder estimate. Assume that $M=1$ (case $M\ge 2$ will be considered later). Then since $L= C_1\upsilon $ the remainder does not exceed
\begin{gather}
Ch ^{s} (\mu h +1) \beta^{-2}\bar{r}^2,
\label{26-8-28}\\
\shortintertext{where}
h = 1/(\upsilon ^{\frac{1}{2}}\beta)\label{26-8-29}\\
\shortintertext{and}
\mu =B\beta \upsilon^{-\frac{1}{2}};\label{26-8-30}
\end{gather}
one can prove it easily by partition of unity on $\supp (\theta)$ and applying semiclassical asymptotics with effective semiclassical parameter $h $ and magnetic parameter $\mu$.

On the other hand, the indirect term generated by $N ^{-1}\|\tilde{\Psi}\| ^2$ is
\begin{multline}
-N\int \theta^{\frac{1}{2}} (x)\theta^{\frac{1}{2}} (y)
\Psi (x_1,\ldots,x_{N-1},x) \Psi ^\dag (x_1,\ldots,x_{N-1},y) \times \\
F(x,y)\, dxdydx_1\cdots dx_{N-1},
\label{26-8-31}
\end{multline}
and since the operator norm of $F(.,.,.)$ is $1$, the absolute value of this term does not exceed
\begin{multline}
N\int\theta(x) |\Psi (x_1,\ldots,x_{N-1},x)| ^2\,dx =
\int \theta(x) \rho _\Psi (x)\,dx\le\\
\int \theta(x) \rho^\TF (x)\,dx + CQ^{\frac{1}{2}} \|\nabla \theta^{\frac{1}{2}}\|
\label{26-8-32}
\end{multline}
where $\rho^\TF=0$ on $\supp (\theta)$ and
$\|\nabla \theta^{\frac{1}{2}} \|\asymp \beta^{-\frac{1}{2}}\bar{r}$.

Recall that $P' (W ^\TF+\nu )=\rho ^\TF$. We will take $\nu'=\nu + L$ to keep $\Theta_\Psi$ larger than all the remainders including those due to replacement $W$ by $W ^\TF$ and $\rho$ by $\rho ^\TF$ in the expression above. One can observe easily that then $\beta$ should satisfy (\ref{26-8-10}); let us define $\beta$ and then $\upsilon$ by (\ref{26-8-11}) and (\ref{26-8-12}) respectively. Then
\begin{equation}
\Theta_\Psi \asymp \bigl(\upsilon^{\frac{3}{2}} + B\upsilon^{\frac{1}{2}}\bigr) \beta \bar{r}^2.
\label{26-8-33}
\end{equation}
Therefore

\begin{claim}\label{26-8-34}
Let $h \le \epsilon_0$ (i.e. $\upsilon^{\frac{1}{2}} \beta\ge C_0$), and $\beta,\upsilon$ be defined by (\ref{26-8-11}) and (\ref{26-8-12}) respectively. Then expression (\ref{26-8-33}) is larger than
$C_0 \beta^{-\frac{1}{2}} Q^{\frac{1}{2}}$ and
the total expression generated by $N ^{-1}\|\tilde{\Psi}\| ^2$ is greater than $\epsilon \Theta $ with $\Theta=\Theta_\Psi$ defined by (\ref{26-8-33}).
\end{claim}

Now let us consider the direct terms in the right-hand expression of (\ref{26-8-22}). The first of them is like in (\ref{book_new-25-6-23})
\begin{multline}
-\int \theta^{\frac{1}{2}} (x)\varphi (\lambda L ^{-1})\,
d_\lambda \bigl(H_{W+\nu',x}\theta^{\frac{1}{2}}(x)
e(x,y,\lambda )\bigr)_{y=x}\,dx=\\
\shoveleft{-\int \theta (x)\varphi (\lambda L ^{-1})\,d_\lambda
\bigl(H_{W+\nu',x}e(x,y,\lambda )\bigr)_{y=x}\,dx}\\
\shoveright{-\frac{1}{2}\int \varphi (\lambda L ^{-1})
[[H_W,\theta^{\frac{1}{2}} ],\theta ^{\frac{1}{2}} ] \,
d_\lambda e(x,x,\lambda )\ge }\\
\int \theta (x) (\nu'-\nu -\lambda ) \varphi (\lambda L ^{-1})
\,d_\lambda e(x,x,\lambda )\,dx
-C\int |\nabla \theta^{\frac{1}{2}} |^2 e(x,x,\nu')dx.
\label{26-8-35}
\end{multline}
Observe that the absolute value of last term in the right-hand expression of (\ref{26-8-35}) does not exceed
$C\beta^{-1}\bar{r}^2 \bigl(\upsilon^{\frac{3}{2}} + B\upsilon^{\frac{1}{2}}\bigr)\asymp \beta^{-2}\Theta $.

The second direct term in the right-hand expression of (\ref{26-8-22}) is like in (\ref{book_new-25-6-24})
\begin{multline}
-\int \theta (x)\Bigl(\rho _\Psi * |x| ^{-1}-(V-W)(x)\Bigr)F(x,x)\,dx=\\ -\D\bigl(\rho _\Psi - \bar{\rho}, \theta (x)F(x,x)\bigr) \ge \\
-C\D\bigl(\rho _\Psi - \rho,\rho _\Psi -\rho\bigr)^{\frac{1}{2}} \cdot \D\Bigl(\theta^{\frac{1}{2}} F(x,x ),
\theta^{\frac{1}{2}} F(x,x ))\Bigr) ^{\frac{1}{2}}\ge
-C Q ^{\frac{1}{2}}\bar{r}^{-\frac{1}{2}}\Theta ,
\label{26-8-36}
\end{multline}
provided $V-W=|x| ^{-1}*\rho$ with
$\D(\rho-\rho ^\TF,\rho-\rho ^\TF)\le C Q$.

Further, the first indirect term in the right-hand expression of (\ref{26-8-22})~is like in (\ref{book_new-25-6-25})
\begin{multline}
-N\int \theta^{\frac{1}{2}} (y)\Psi (x_1,\ldots,x_{N-1},x) \Psi ^\dag (x_1,\ldots,x_{N-1},y) \times \\
\shoveright{\varphi (\lambda L ^{-1})\,d_\lambda
\bigl(H_{W+\nu',x}\theta^{\frac{1}{2}} (x)e(x,y,\lambda )\bigr)\,
dxdydx_1\cdots dx_{N-1}=}\\
\shoveleft{-N\int \theta^{\frac{1}{2}}(y)\theta^{\frac{1}{2}}(x)
\Psi (x_1,\ldots,x_{N-1},x)
\Psi ^\dag (x_1,\ldots,x_{N-1},y) \times }\\
\shoveright{\varphi (\lambda L ^{-1})(\nu'-\nu -\lambda )\,d_\lambda e(x,y,\lambda )\, dxdydx_1\cdots dx_{N-1}}\\
\shoveleft{-N\int \theta^{\frac{1}{2}}(y) \Psi (x_1,\ldots,x_{N-1},x)
\Psi ^\dag (x_1,\ldots,x_{N-1},y) \times }\\
\varphi (\lambda L ^{-1})[H_{W,x},\theta^{\frac{1}{2}}(x)]\,
d_\lambda e(x,y,\lambda )\,dxdydx_1\cdots dx_{N-1}.
\label{26-8-37}
\end{multline}
Observe that one can rewrite the sum of the first terms in the right-hand expressions in (\ref{26-8-35}) and (\ref{26-8-37}) as
$\sum_j \varphi (\lambda _jL ^{-1}) (\nu '-\nu -\lambda_j)\|\hat{\Psi}_j\| ^2$\linebreak
with
\begin{equation*}
\hat{\Psi}_j(x_1,\ldots,x_{N-1})\coloneqq
\int \Psi(x_1,\ldots,x_{N-1},x) \theta^{\frac{1}{2}}(x)\phi_j(x)\,dx
\end{equation*}
and therefore this sum is non-negative.
\enlargethispage{2\baselineskip}

One can see easily that the absolute value of the second term in the right-hand expression of (\ref{26-8-37}) does not exceed
\begin{multline*}
\int\rho _\Psi (y)\theta^{\frac{1}{2}} (y)\,dy \times \beta^{-1} \int \theta_1 (x) e(x,x,\nu')\,dx\asymp
C\Theta \times C \bigl(\upsilon^{\frac{3}{2}} + B\upsilon\bigr) \bar{r}^2\asymp \\
C \beta^{-\frac{3}{2}}\bar{r}Q^{\frac{1}{2}}\Theta
\end{multline*}
due the choice of $\beta$. This is larger than the absolute value of the right-hand expression in (\ref{26-8-36}).
Therefore (cf. \ref{book_new-25-6-26}) we conclude that
\begin{claim}\label{26-8-38}
The sum of the first direct and indirect terms in the right-hand expression of (\ref{26-8-22}) is greater than
$-C \beta^{-\frac{3}{2}}\bar{r}Q^{\frac{1}{2}}\Theta$.
\end{claim}

Finally, we need to consider the second indirect term generated by the right-hand expression of (\ref{26-8-22}):
\begin{multline}
-\int \Bigl(\sum_{1\le i\le N}|y-x_i| ^{-1}-(V-W)(y)\Bigr)\times\\[5pt]
\shoveright{\Psi (x_1,\ldots, x_N)\Psi ^\dag(x_1,\ldots,x_{N-1},y)
\theta^{\frac{1}{2}} (x_N)\theta^{\frac{1}{2}} (y)F(x_N,y)
\,dx_1\cdots dx_Ndy=}\\[5pt]
\shoveleft{ -\int \Bigl(|y| ^{-1}*\varrho _{\underline{x}}(y) -(V-W)(y)\Bigr)
\Psi (x_1,\ldots, x_N)\Psi ^\dag(x_1,\ldots,x_{N-1},y) \times}\\
\shoveright{\theta^{\frac{1}{2}} (x_N)\theta^{\frac{1}{2}} (y) F(x_N,y)\,dx_1\cdots dx_Ndy}\\[5pt]
\shoveleft{-\int \Bigl(\sum_{1\le i\le N}|y-x_i| ^{-1}-
|y| ^{-1}*\varrho _{\underline{x}}(y)\Bigr)
\Psi (x_1,\ldots, x_N)\Psi ^\dag(x_1,\ldots,x_{N-1},y) \times}\\
\theta^{\frac{1}{2}} (x_N)\theta^{\frac{1}{2}} (y)F(x_N,y)\,dx_1\cdots\, dx_Ndy;
\label{26-8-39}
\end{multline}
recall that $\varrho _{\underline{x}}$ is a smeared density, $\underline{x}=(x_1,\ldots, x_N)$.

Since $|y| ^{-1}*\varrho _{\underline{x}}(y) -(V-W)(y)=
|y|^{-1}*(\varrho _{\underline{x}} -\rho)$, the first term in the right-hand expression is equal to
\begin{multline}
\int \theta^{\frac{1}{2}} (x_N)\Psi (x_1,\ldots, x_N)\times \\
\D_y\Bigl( \varrho _{\underline{x}} (y) -\rho (y) , F(x_N,y,\lambda )
\theta^{\frac{1}{2}} (y) \Psi (x_1,\ldots,x_{N-1},y) \Bigr)\,dx_1\cdots dx_N
\label{26-8-40}
\end{multline}
and its absolute value does not exceed
\begin{multline}
\biggl(N\int \D\bigl(\varrho _{\underline{x}} (\cdot)-\rho(\cdot),
\varrho _{\underline{x}}(\cdot)- \rho(\cdot)\bigr) |\Psi (x_1,\ldots, x_N)|^2
\theta (x_N)\,dx_1\cdots dx_N\biggr) ^{\frac{1}{2}}\times \\
\shoveleft{N^{-\frac{1}{2}}
\biggl(\D_y\Bigl(F(x_N,y,\lambda )\theta^{\frac{1}{2}} (y)\Psi(x_1,\ldots,x_{N-1},y),} \\
F(x_N,y,\lambda ) \theta^{\frac{1}{2}} (y)\Psi (x_1,\ldots,x_{N-1},y) \Bigr)
\,dx_1\cdots dx_N \biggr) ^{\frac{1}{2}}.
\label{26-8-41}
\end{multline}

Recall that the first factor is equivalently defined by (\ref{book_new-25-5-4}) and therefore due to estimate (\ref{book_new-25-5-24}) it does not exceed
$\bigl((Q+T+\varepsilon^{-1}N)\Theta +P\bigr)^{\frac{1}{2}}$,
where we assume that $\varepsilon \le Z^{-\frac{2}{3}}$ and
$\Theta\asymp \beta \bigl( \upsilon^{\frac{3}{2}}+B\upsilon^{\frac{1}{2}}\bigr) \bar{r}^2\beta \asymp \beta^{-\frac{1}{2}}\bar{r}Q^{\frac{1}{2}}$ is now an upper estimate for $\int \theta (y)\rho_\Psi (y)\,dy$-like expressions.

Then, according to (\ref{book_new-25-5-25}), $P\asymp C\beta ^{-2}\Theta\ll Q\Theta $
and, according to (\ref{book_new-25-5-23}), $T\ll Q$ and therefore in all such inequalities we may skip $P$ and $T$ terms; so we get
$C(Q+\varepsilon^{-1}N)^{\frac{1}{2}}\Theta^{\frac{1}{2}}$.

Meanwhile, the second factor in (\ref{26-8-41}) (without square root) is equal to
\begin{multline*}
N^{-1}\int L^{-2}\varphi'(\lambda L^{-1})\varphi'(\lambda' L^{-1})
|y-z|^{-1} \underbracket{e(x_N,y,\lambda)}\theta^{\frac{1}{2}} (y)\Psi(x_1,\ldots,x_{N-1},y)\times\\
\underbracket{e(x_N,z,\lambda')}\theta^{\frac{1}{2}} (z)\Psi^\dag (x_1,\ldots,x_{N-1},z)
\,dy dz\,dx_1\cdots dx_{N-1} \underbracket{dx_N}\, d\lambda d\lambda';
\end{multline*}
after integration with respect to $x_N$ we get instead of the marked terms
$e(y,z,\lambda)$ (recall that $e(.,.,.)$ is the Schwartz kernel of the projector and we keep $\lambda<\lambda'$) and then, integrating with respect to $\lambda'$ we arrive to
\begin{multline*}
N^{-1}\int
|y-z|^{-1} F(y,z) \theta^{\frac{1}{2}} (y)\Psi(x_1,\ldots,x_{N-1},y)\times \\
\theta^{\frac{1}{2}} (z)\Psi^\dag (x_1,\ldots,x_{N-1},z)
\,dy dz\,dx_1\cdots dx_{N-1},
\end{multline*}
where now $F$ is defined by (\ref{26-8-25}) albeit with $\varphi^2$ instead of $\varphi$. This latter expression does not exceed
\begin{multline}
N^{-1}\iint
|y-z|^{-1} |F(y,z)| \theta^{\frac{1}{2}} (y)|\Psi(x_1,\ldots,x_{N-1},y)|^2 \times\\
\,dy dz\,dx_1\cdots dx_{N-1}.
\label{26-8-42}
\end{multline}
Then due to Proposition~\ref{prop-26-A-6} expression
$\int |y-z|^{-1}|F(y,z)|\,dz$ does not exceed $C\beta^{-1}(h ^{-1}+\mu)\asymp \upsilon^{\frac{1}{2}}+B\upsilon^{-\frac{1}{2}}$, and thus expression (\ref{26-8-42}) does not exceed
$CZ^{-2} \bigl(\upsilon^{\frac{1}{2}}+B\upsilon^{-\frac{1}{2}}\bigr) \Theta$. Therefore the second factor in (\ref{26-8-41}) does not exceed  $CN^{-1} \bigl(\upsilon^{\frac{1}{4}}+B^{\frac{1}{2}}\upsilon^{-\frac{1}{4}}\bigr) \Theta^{\frac{1}{2}}$ and the whole expression
(\ref{26-8-41}) does not exceed
\begin{multline*}
C(Q+\varepsilon^{-1}N)^{\frac{1}{2}}\Theta^{\frac{1}{2}} \times
N^{-1}\bigl(\upsilon^{\frac{1}{4}}+B^{\frac{1}{2}}\upsilon^{-\frac{1}{4}}\bigr)
\Theta^{\frac{1}{2}}=\\
CN^{-1}(Q+\varepsilon^{-1}N)^{\frac{1}{2}}
\bigl(\upsilon^{\frac{1}{4}}+B^{\frac{1}{2}}\upsilon^{-\frac{1}{4}}\bigr)\Theta.
\end{multline*}
Finally we arrive to

\begin{proposition-foot}\label{prop-26-8-5}
\footnotetext{\label{foot-26-47} Cf. claim (\ref{book_new-25-6-31}).}
Let
\begin{gather}
\upsilon \ge \max
\bigl( Z^{-\frac{4}{3}}Q^{\frac{2}{3}};\,
Z^{-\frac{4}{5}}Q^{\frac{2}{5}} B^{\frac{2}{5}}\bigr)
\label{26-8-43}\\
\shortintertext{and}
\varepsilon \ge Z^{-1} \max (\upsilon^{-\frac{3}{2}}, B\upsilon^{-\frac{5}{2}}\bigr).
\label{26-8-44}
\end{gather}
Then the first term in the right-hand expression of \textup{(\ref{26-8-39})} does not exceed $C\upsilon \Theta$.
\end{proposition-foot}

Further, we need to estimate the second term in the right-hand expression of (\ref{26-8-39}). It can be rewritten in the form
\begin{multline}
\sum_{1\le i\le N} \int U(x_i,y)
\Psi (x_1,\ldots, x_N)\Psi ^\dag(x_1,\ldots,x_{N-1},y)
\theta^{\frac{1}{2}} (x_N)\theta^{\frac{1}{2}} (y)\times\\
F(x_N,y)\,dx_1\cdots dx_N dy ,
\label{26-8-45}
\end{multline}
where $U(x_i,y)$ is the difference between two potentials, one generated by the charge $\updelta (x-x_i)$ and another by the same charge smeared; note that $U(x_i,y)$ is supported in $\{(x_i,y)\colon |x_i-y|\le \varepsilon \}$. Let us estimate the $i$-th term in this sum with $i<N$ first. Multiplied by $N(N-1)$, it does not exceed
\begin{multline}
N\biggl(\int |U(x_i,y)| ^2 |\Psi (x_1,\ldots, x_N)| ^2
\theta^{\frac{1}{2}} (x_N)\theta^{\frac{1}{2}} (y) |F(x_N,y)|\,
dx_1\cdots dx_Ndy\biggr) ^{\frac{1}{2}}\times \\
N\biggl(\int \omega (x_i,y) |\Psi (x_1,\ldots, x_{N-1},y)| ^2
\theta^{\frac{1}{2}} (x_N)\theta^{\frac{1}{2}} (y)|F(x_N,y)|
\,dx_1\cdots dx_Ndy\biggr) ^{\frac{1}{2}}
\label{26-8-46}
\end{multline}
here $\omega $ is $\varepsilon $-admissible and supported in
$\{(x_i,y)\colon |x_i-y|\le 2\varepsilon \}$ function.\pagebreak

Due to Proposition~\ref{prop-26-A-6} in the second factor
\begin{equation*}
\int \theta^{\frac{1}{2}} (x_N)|F(x_N,y)|\,dx_N\le C (1+ \mu h ) \asymp
C(1+B \upsilon^{-1})
\end{equation*}
and therefore the whole second factor does not exceed
\begin{equation}
C\Bigl(\underbracket{\int\theta^{\frac{1}{2}} (x)\omega (x,y)
\varrho ^{(2)}_\Psi (x,y)\,dx dy} \Bigr) ^{\frac{1}{2}}
(1+B \upsilon^{-1})^{\frac{1}{2}},
\label{26-8-47}
\end{equation}
where we replaced $x_i$ by $x$. According to Proposition~\ref{book_new-prop-25-5-1} in the selected expression one can replace $\rho^{(2)}_\Psi (x,y)$ by $\rho_\Psi(x)\rho(y) $, with an error which does not exceed
\begin{equation*}
C\Bigl(\sup_x \|\nabla _y\chi _x\|_{\sL^2(\bR ^3)}
\bigl( Q +\varepsilon ^{-1}N \bigr)^{\frac{1}{2}}
+C\varepsilon N\|\nabla _y\chi \|_{\sL ^\infty }\Bigr)\Theta.
\end{equation*}
When we plug
$\sup_x \|\nabla _y\chi _x\|_{\sL^2(\bR ^3)}\asymp \varepsilon^{\frac{1}{2}}$, $\|\nabla _y\chi \|_{\sL ^\infty }\asymp \varepsilon^{-1}$ this expression becomes $CN\Theta$.

Meanwhile, consider
\begin{equation}
\int |U(x_i,y)|^2 \theta ^{\frac{1}{2}}(y) |F(x_N,y)|\,dy.
\label{26-8-48}
\end{equation}
Again, due to Proposition~\ref{prop-26-A-6}, it does not exceed
\begin{equation*}
C\bigl(\upsilon^{\frac{3}{2}}+B\upsilon^{\frac{1}{2}}\bigr)
\int |U(x_i,y)|^2 \theta ^{\frac{1}{2}}(y)
\bigl(|x_N-y|\upsilon^{\frac{1}{2}}+1\bigr)^{-s}\,dy
\end{equation*}
and this integral should be taken over $B(x_i,\varepsilon)$, with
$|U(x_i,y)|\le |x_i-y|^{-1}$, so (\ref{26-8-48}) does not exceed
\begin{equation*}
C\varepsilon \bigl(\upsilon^{\frac{3}{2}}+B\upsilon^{\frac{1}{2}}\bigr)\omega' (x_i,x_N)
\end{equation*}
with $\omega '(x,y)=\bigl(1+\upsilon^{\frac{1}{2}}|x-y|\bigr) ^{-s}$ (provided $\varepsilon \le \upsilon^{-\frac{1}{2}}$ which will be the case).
Therefore the first factor in (\ref{26-8-46}) does not exceed
\begin{equation}
C \varepsilon^{\frac{1}{2}} \bigl(\upsilon^{\frac{3}{4}}+B^{\frac{1}{2}}\upsilon^{\frac{1}{4}}\bigr)
\Bigl(\underbracket{\int\theta^{\frac{1}{2}} (x)\omega '(x,y)
\rho ^{(2)}_\Psi (x,y)\,dx dy }\Bigr) ^{\frac{1}{2}}.
\label{26-8-49}
\end{equation}
Therefore in the selected expression one can replace $\rho^{(2)}_\Psi (x,y)$ by $\rho_\Psi(x)\rho(y) $ with an error which does not exceed what we got before but with $\varepsilon$ replaced by $\upsilon^{-\frac{1}{2}}$, i.e. also $CN\Theta$.

However, in both selected expressions, (\ref{26-8-47}) and (\ref{26-8-49}), replacing $\rho^{(2)}_\Psi (x,y)$ by $\rho_\Psi(x)\rho(y) $ we get just $0$. Therefore expression (\ref{26-8-46}) does not exceed
$C\varepsilon ^{\frac{1}{2}} \bigl(\upsilon^{\frac{3}{4}}+B^{\frac{1}{2}}\upsilon^{\frac{1}{4}}\bigr)Z\Theta$, which, in turn, does not exceed $C\upsilon \Theta$ provided
$\varepsilon \le C\upsilon^{\frac{1}{2}}
\bigl(1+B\upsilon^{-1}\bigr)^{-1} Z^{-2}$.

So, we have two restriction to $\varepsilon$ from above: the last one and $\varepsilon\le Z^{-\frac{2}{3}}$ and one can see easily that both of them are compatible with with restriction to $\varepsilon$ in (\ref{26-8-43}); also we can see easily that condition (\ref{26-8-43}) is weaker than
$\upsilon \ge \bigl\{ Z^{\frac{20}{21}}\colon Z^{\frac{20}{27}}B^{\frac{2}{9}}\colon
Z^{\frac{4}{15}}B^{\frac{26}{45}}\bigr\}$ if
$\bigl\{ B\le Z^{\frac{20}{21}};\,
Z^{\frac{20}{21}} \le B \le Z^{\frac{4}{3}};\,
Z^{\frac{4}{3}}\le B\le Z^3\bigr\}$ respectively.

Finally, consider term in (\ref{26-8-45}) with $i=N$ (multiplied by $N$):
\begin{equation}
N \int U(x_N,y) |\Psi (x_1,\ldots, x_N)|^2
\theta^{\frac{1}{2}} (x_N)\theta^{\frac{1}{2}} (y)
F(x_N,y)\,dx_1\cdots dx_Ndy
\label{26-8-50}
\end{equation}
due to Cauchy inequality it does not exceed
\begin{multline}
N \Bigl(\int |x_N-y |^{-2}|\Psi (x_1,\ldots, x_N)|^2
\theta^{\frac{1}{2}} (x_N)\theta^{\frac{1}{2}} (y)\,\,dx_1\cdots dx_Ndy\Bigr)^{\frac{1}{2}}\times \\
N \Bigl(\int |F(x_N,y)|^2|\Psi (x_1,\ldots, x_N)|^2
\theta^{\frac{1}{2}} (x_N)\theta^{\frac{1}{2}} (y)\,\,dx_1\cdots dx_Ndy\Bigr)^{\frac{1}{2}}
\label{26-8-51}
\end{multline}
where both integrals are taken over $\{|x_N-y|\le \varepsilon\}$. Integrating with respect to $y$ there we get that it that it does not exceed
\begin{equation*}
C\varepsilon ^{\frac{1}{2}}\Theta^{\frac{1}{2}} \times \bigl(\upsilon^{\frac{3}{4}}+B\upsilon^{\frac{1}{4}}\bigr)
\varepsilon^{\frac{3}{2}}\Theta^{\frac{1}{2}}=
C\bigl(\upsilon^{\frac{3}{4}}+B\upsilon^{\frac{1}{4}}\bigr) \varepsilon^2 \Theta \ll \upsilon \Theta.
\end{equation*}
Therefore the right-hand expression in (\ref{26-8-22}) is
$\ge -C \upsilon \Theta$ and recalling that $\nu'-\nu = O(\upsilon)$ we recover a lower estimate $\I_N +\nu \ge -C\upsilon$ in Theorem~\ref{thm-26-8-6} below. Here $\upsilon$ must be found from (\ref{26-8-11})--(\ref{26-8-12}) and must satisfy $\upsilon \le |\nu|$.

Combining this estimate with the estimate from the above, derived in Proposition~\ref{prop-26-8-3} we arrive to

\begin{theorem}\label{thm-26-8-6}
Let $M=1$. Let condition \textup{(\ref{26-2-28})} be fulfilled. Then
\begin{enumerate}[label=(\roman*),wide, labelindent=0pt]
\item\label{thm-26-8-6-i}
For $B\le Z^{\frac{20}{21}}$ and $|\nu|\ge Z^{\frac{20}{21}}$
\begin{equation}
|\I_N+\nu|\le CZ^{\frac{5}{18}} |\nu|^{\frac{17}{24}}.
\label{26-8-52}
\end{equation}
\item\label{thm-26-8-6-ii}
For $Z^{\frac{20}{21}}\le B\le Z^{\frac{4}{3}}$ and
$|\nu|\ge Z^{\frac{20}{27}}B^{\frac{2}{9}}$  estimate \textup{(\ref{26-8-17})} from above and estimate
\pagebreak
\begin{equation}
\I_N+\nu \ge -C\left\{\begin{aligned}
&Z^{\frac{5}{18}} |\nu|^{\frac{17}{24}} \qquad
&&\text{if\ \ } B \le Z^{\frac{5}{18}}|\nu|^{\frac{17}{24}},\\
&Z^{\frac{5}{12}}B^{-\frac{1}{2}}|\nu|^{\frac{17}{16}} \qquad
&&\text{if\ \ } Z^{\frac{5}{18}} |\nu|^{\frac{17}{24}} \le B \le |\nu|,\\
&Z^{\frac{5}{12}}B^{-\frac{3}{16}}|\nu|^{\frac{3}{4}} \qquad
&&\text{if\ \ } |\nu| \le B \le Z^{-\frac{20}{7}}|\nu|^{4},\\
& Z^{\frac{20}{21}}
&&\text{if\ \ } Z^{-\frac{20}{7}}|\nu|^{4}\le B
\end{aligned}\right.
\label{26-8-53}
\end{equation}
from below hold.

\item\label{thm-26-8-6-iii}
For $Z^{\frac{20}{21}}\le B\le Z^3$ and
$|\nu|\ge Z^{\frac{4}{15}}B^{\frac{26}{45}}$ estimate \textup{(\ref{26-8-19})} from above and estimate
\begin{equation}
\upsilon=\left\{\begin{aligned}
&Z^{-\frac{1}{10}} B^{\frac{1}{5}} |\nu|^{\frac{3}{4}} \qquad
&&\text{if\ \ } B \le Z^{-\frac{1}{2}}|\nu|^{\frac{1}{4}}\\
&Z^{\frac{4}{35}}B^{\frac{22}{35}} \qquad
&&\text{if\ \ } Z^{-\frac{1}{2}}|\nu|^{\frac{1}{4}}\le B
\end{aligned}\right.
\label{26-8-54}
\end{equation}
from below hold.
\end{enumerate}
\end{theorem}

\begin{remark}\label{rem-26-8-7}
Recall that $Q=Z^{\frac{5}{3}}\bigl(B^{\delta} + 1\bigr)Z^{-\delta}$
as $B\le Z$; therefore we can add factor
$\bigl(B^{\delta'} + 1\bigr)Z^{-\delta'}$ in all estimates of Theorem~\ref{thm-26-8-6}.
\end{remark}

\section{Estimates for ionization energy: $M\ge 2$}
\label{sect-26-8-3}

Recall that for $M\ge 2$ we have only estimate (\ref{26-6-73}):
\begin{equation*}
\D(\rho_\Psi-\rho,\, \rho_\Psi-\rho \le\bar{Q})\coloneqq \textup{(\ref{26-5-40})}+\textup{(\ref{26-6-69})}.
\end{equation*}
Then exactly the same arguments lead us to the following (we leave all details to the reader):

\begin{theorem}\label{thm-26-8-8}
Let $M\ge 2$. Then
\begin{enumerate}[label=(\roman*), wide, labelindent=0pt]
\item\label{thm-26-8-8-i}
Estimate $\I_N+\nu \le C(\upsilon+\varsigma)$ holds with $\upsilon$ and $\varsigma$ defined by \textup{(\ref{26-8-11})}--\textup{(\ref{26-8-12})} and
\textup{(\ref{26-8-13})} albeit with $Q$ replaced by $\bar{Q}$.

\item\label{thm-26-8-8-ii}
Estimate $\I_N +\nu \ge -C\upsilon$ holds with $\upsilon$ defined by \textup{(\ref{26-8-11})}--\textup{(\ref{26-8-12})} albeit with $Q$ replaced by $\bar{Q}$.
\end{enumerate}
\end{theorem}

Therefore the case when $\bar{Q}\le Q$ is not affected. One can see easily that it happens for sure as $B\le Z^{\frac{20}{17}}L^{-\kappa}$ where $\kappa>0$ is some exponent.

We leave to the reader

\begin{Problem}\label{Problem-26-8-9}
\begin{enumerate}[label=(\roman*), wide, labelindent=0pt]
\item\label{Problem-26-8-9-i}
Find explicit formula for $\upsilon+\varsigma$ and $\upsilon$.

\item\label{Problem-26-8-9-ii}
Find $\nu^*=\nu^* (Z,B)$ and $\nu_*=\nu_*(Z,B)$ such that
$|\nu |\lesssim \upsilon+\varsigma$ iff and only if $\nu\lesssim \nu^*$ and
$|\nu |\lesssim \upsilon$ iff and only if $\nu\lesssim \nu_*$.
\end{enumerate}
\end{Problem}

\section{Free nuclei model}
\label{sect-26-8-4}

In this subsection we consider two extra problems appearing in the free nuclei model--estimate the minimal distance between nuclei and the maximal excessive positive charge when system does not break apart. We also slightly improve estimates for the maximal negative charge and for the ionization energy.

\subsection{Preliminary arguments}
\label{sect-26-8-4-1}

Recall that we assume that
\begin{gather}
\Q\coloneqq \hat{\E} -\sum_{1\le m\le M} \E_m<0
\label{26-8-55}\\
\shortintertext{where}
\hat{\E}=\E+\sum_{1\le m<m'\le M} \frac{Z_mZ_{m'}}{|\y_m-\y_{m'}|}.
\label{26-8-56}
\end{gather}
We apply estimate from below for $\hat{\E}$ delivered by Proposition~\ref{prop-26-6-1}\ref{prop-26-6-1-ii}, and estimates from above for $\E_m$, delivered by Theorem~\ref{thm-26-6-6}\ref{thm-26-6-6-ii}; then
\begin{gather}
\hat{\cE}^\TF + \Scott
-\sum_{1\le m\le M}\Bigl(\cE^\TF_m - \Scott_m\Bigr)\le CQ+CZ^{\frac{4}{3}}B^{\frac{1}{3}}+C a^{-\frac{1}{2}}Z^{\frac{3}{2}} \notag\\
\intertext{or, equivalently, due to equality $\Scott=\sum_{1\le m\le M}\Scott_m$ and non-binding theorem}
0\le \cQ \coloneqq \hat{\cE}^\TF -\sum_{1\le m\le M} \cE^\TF_m \le
CQ+CZ^{\frac{4}{3}}B^{\frac{1}{3}}
\color{gray}{+C a^{-\frac{1}{2}}Z^{\frac{3}{2}}}.
\label{26-8-57}
\end{gather}
Assume that assumption (\ref{26-2-28}) is fulfilled. Then $\cQ\asymp \cE^\TF$ for $a\le \epsilon r^*$ with
$r^* =\min(Z^{-\frac{1}{3}};\, B^{-\frac{2}{5}}Z^{\frac{1}{5}})$ and therefore \begin{claim}
In the free nuclei model $a\ge \epsilon r^*$.
\label{26-8-58}
\end{claim}
Then the last term in (\ref{26-8-57}) is not needed.

\begin{remark}\label{rem-26-8-10}
\begin{enumerate}[label=(\roman*), wide, labelindent=0pt]
\item\label{rem-26-8-10-i}
Obviously, the second term $CZ^{\frac{4}{3}}B^{\frac{1}{3}}$ in the right-hand expression of (\ref{26-8-57}) matters only if $Z\le B\le Z^{\frac{11}{7}}$; however we will show that it could be skipped even in this case.

\item\label{rem-26-8-10-ii}
If $B\le Z$ we can replace the right-hand expression of (\ref{26-8-57}) by
$Z^{\frac{5}{3}-\delta}(1+B^\delta)$.

\item\label{rem-26-8-10-iii}
All these estimates hold also for
$\D(\rho_\Psi-\rho^\TF,\, \rho_\Psi-\rho^\TF)$ because this term is present in the estimate from below.
\end{enumerate}
\end{remark}

\subsection{Minimal distance}
\label{sect-26-8-4-2}

We are going to improve (\ref{26-8-58}). Consider the case $B\le Z^{\frac{4}{3}}$ first. Then since $\cQ\ge \epsilon_0 a^{-7}$ for
$\epsilon r^* \le a\le \epsilon \bar{r}$ (where in this case $r^*=Z^{-\frac{1}{3}}\le \bar{r}=B^{-\frac{1}{4}}$), we conclude that
$a\gtrsim Z^{-\frac{5}{21}}$ provided $B\le Z^{\frac{20}{21}}$.

Furthermore, then we can apply improved remainder estimate $O(Z^{\frac{5}{3}-\delta})$, since the difference between Dirac--Schwinger terms for a molecule and the sum of these terms for the atoms is also $O(Z^{\frac{5}{3}-\delta})$ as long as $a\ge Z^{-\frac{1}{3}+\delta_1}$, which is the case. Then we conclude that $a\ge Z^{-\frac{5}{21}-\delta'}$ as long as it is less than $\epsilon \bar{r}$ and we arrive to Statement~\ref{prop-26-8-11-i} of Proposition~\ref{prop-26-8-11}:

\begin{proposition}\label{prop-26-8-11}
Let condition \textup{(\ref{26-2-28})} be fulfilled. Then in the free nuclei model
\begin{enumerate}[label=(\roman*), wide, labelindent=0pt]
\item\label{prop-26-8-11-i}
For $B\le Z^{\frac{20}{21}}$ the minimal distance satisfies
\begin{equation}
a\ge \min(Z^{-\frac{5}{21}-\delta},\epsilon B^{-\frac{1}{4}}).
\label{26-8-59}
\end{equation}
\item\label{prop-26-8-11-ii}
For $Z^{\frac{20}{21}}\le B\le Z^3$ the distances satisfy
\begin{equation}
|\y_m-\y_{m'}|\ge \bar{r}_m + \bar{r}_{m'} - \epsilon \bar{r}\qquad
\forall m\ne m'
\label{26-8-60}
\end{equation}
with arbitrarily small constant $\epsilon>0$ where $\bar{r}_m$ denote the exact radii of $\supp (\rho_m^\TF)$.
\end{enumerate}
\end{proposition}

\begin{proof}
We need to prove Statement~\ref{prop-26-8-11-ii}. Observe that it also follows from the arguments above in the case $Z^{\frac{20}{21}}\le B\le Z$.

For $Z\le B\le C Z^{\frac{4}{3}}$ the remainder estimate is $O(Z^{\frac{4}{3}}B^{\frac{1}{3}})$ and the same arguments imply that
$|\y_m-\y_{m'}| \ge \epsilon Z^{-\frac{4}{21}}B^{-\frac{1}{21}}$ unless
$a \ge \epsilon \bar{r}$ and since the latter is weaker, it must be satisfied. Therefore if (\ref{26-8-60}) fails, then in virtue of Theorem~\ref{thm-26-2-17}
$\cQ\ge \epsilon_1 B^{\frac{7}{4}}$, which is larger than the remainder estimate $CZ^{\frac{4}{3}}B^{\frac{1}{3}}$.

Finally, case $C_\epsilon Z^{\frac{4}{3}}\le B\le Z^3$ follows from the fact that if (\ref{26-8-60}) fails then in virtue of Theorem~\ref{thm-26-2-17}
$\cQ\ge \epsilon_1 Z^{\frac{9}{5}}B^{\frac{2}{5}}$. \end{proof}

\begin{proposition}\label{prop-26-8-12}
Let condition \textup{(\ref{26-2-28})} be fulfilled. Then in the free nuclei model
\begin{equation}
\cQ + \D(\rho_\Psi -\rho^\TF,\, \rho_\Psi -\rho^\TF)\le CQ.
\label{26-8-61}
\end{equation}
\end{proposition}

\begin{proof}
We need to cover only case $Z\le B\le Z^{\frac{11}{7}}$, since only in this case term $CZ^{\frac{4}{3}}B^{\frac{1}{3}}$ matters.

We apply now estimate from below for $\hat{\E}$ delivered by Proposition~\ref{prop-26-6-1}\ref{prop-26-6-1-i}, and estimates from above for $\E_m$, delivered by Theorem~\ref{thm-26-6-6}\ref{thm-26-6-6-i}; then we do not have term $CZ^{\frac{4}{3}}B^{\frac{1}{3}}$ but instead of equal to $0$ difference of the Scott correction terms, we get
\begin{multline}
\Bigl(\Tr ((H_{A,W}-\nu)^-) +\int P_B(W^\TF +\nu) \,dx\Bigr)-\\
\sum_{1\le m\le M}
\Bigl(\Tr ((H_{A,W_m}-\nu')^-) +\int P_B(W^\TF +\nu') \,dx\Bigr),
\label{26-8-62}
\end{multline}
where we know that $\nu'=\nu_1=\ldots=\nu_M$.

Let us use partition of unity $\phi_0+\phi_1+\ldots+\phi_M=1$ where
$\phi_m=1$ in $B(\y_m, \epsilon \bar{r}_m)$ and is supported in
$B(\y_m,2\epsilon\bar{r}_m)$. Then our standard methods imply that the absolute values of
\begin{gather}
 \Tr \bigl((H_{A,W}-\nu)^-\phi _0 \bigr) +
\int P_B(W^\TF +\nu)\phi_0(x)\,dx ,
\label{26-8-63}\\
\shortintertext{and}
\Tr \bigl((H_{A,W_m}-\nu')^-\phi_{m'} \bigr) +
\int P_B(W_m^\TF +\nu')\phi_{m'}\,dx
\label{26-8-64}
\end{gather}
with $m=1,\ldots, M$, $m'=0,1,\ldots, M$, $m'\ne m$ do not exceed $CQ$\,\footnote{\label{foot-26-48} Recall that $W$ and $W_m$ are approximations to $W^\TF$ and $W^\TF_m$.}. Therefore we need to estimate an absolute value of
\begin{multline}
\Tr \bigl(\bigl[(H_{A,W}-\nu)^- - (H_{A,W_m}-\nu')^-\bigr]\phi _m \bigr)+\\
\int \bigl(P_B(W^\TF +\nu)-P_B(W_m^\TF +\nu'))\phi_m\,dx.
\label{26-8-65}
\end{multline}
Due to Proposition~\ref{prop-26-8-11} $B(\y_m ,3\epsilon\bar{r})$ does not intersect $B(\y_{m'} ,\bar{r}_{m'})$ and then in $B(\y_m ,3\epsilon\bar{r})$
$W_{m'}\le C(Z-N)\bar{r}^{-1}$. Using this inequality and
\begin{gather}
\D(\rho-\rho_1-\ldots-\rho_M,\, \rho-\rho_1-\ldots-\rho_M)\le C\cQ,
\label{26-8-66}\\
\intertext{one can prove easily that there also}
|W-W_m|\le CT\coloneqq C\cQ^{\frac{1}{2}}\bar{r}^{-\frac{1}{2}}+C(Z-N)\bar{r}^{-1}
\label{26-8-67}\\
\intertext{and, moreover,}
|\nabla(W-W_m)|\le CT\bar{r}^{-1}=
C\cQ^{\frac{1}{2}}\bar{r}^{-\frac{3}{2}} + C(Z-N)\bar{r}^{-2},
\label{26-8-68}\\
|\nabla^2 (W-W_m)|\le CT\bar{r}^{-2}=
C\cQ^{\frac{1}{2}}\bar{r}^{-\frac{5}{2}} + C(Z-N)\bar{r}^{-3}.
\label{26-8-69}
\end{gather}
Then using our standard methods one can prove easily that an absolute value of expression (\ref{26-8-65}) with $\phi_m$ replaced by $\ell$-admissible function $\psi_m$ does not exceed
\begin{gather}
CTh ^{-2}(1+\mu h )
\label{26-8-70}\\
\intertext{with our standard}
h = Z^{-\frac{1}{2}}r^{-\frac{1}{2}}, \qquad
\mu = BZ^{-\frac{1}{2}}r^{\frac{3}{2}}
\label{26-8-71}\\
\intertext{if \underline{either} $B\le Z^{\frac{4}{3}}$,
$r\le r^* Z^{-\frac{1}{3}}$ \underline{or} $Z^{\frac{4}{3}}\le B\le Z^3$,
$r\le \bar{r}$ and}
h = r, \qquad \mu = Br^3
\label{26-8-72}
\end{gather}
if $B\le Z^{\frac{4}{3}}$.
Plugging (\ref{26-8-71}) and (\ref{26-8-72}) and summing over partition we arrive to $CT Z^{\frac{2}{3}}$ as
$Z^{\frac{20}{21}}\le B\le Z^{\frac{4}{3}}$ and $CTZ^{\frac{2}{5}}B^{\frac{1}{5}}$ as $Z^{\frac{4}{3}}\le B\le Z^3$.

Plugging $T=(Z-N)\bar{r}^{-1}$ we get expressions which are much smaller than
$\epsilon (Z-N)^2\bar{r}^{-1}$ due to (\ref{26-8-61}); plugging $T=\cQ^{\frac{1}{2}}\bar{r}^{-\frac{1}{2}}$ we get terms smaller than
$\epsilon'\cQ + C(\epsilon')B^{\frac{1}{4}}Z^{\frac{4}{3}}$ if
$B\le Z^{\frac{4}{3}}$ and
$\epsilon'\cQ + C(\epsilon')Z^{\frac{3}{5}}B^{\frac{4}{5}}$
if $Z^{\frac{4}{3}}\le B\le Z^3$; here $\epsilon'>0$ is arbitrarily small and thus term (\ref{26-8-65}) does not make any difference.
\end{proof}

Since $\cQ \ge \epsilon a^{-1}(Z-N)^2$ we arrive to

\begin{corollary}\label{cor-26-8-13}
Let condition \textup{(\ref{26-2-28})} be fulfilled. Then
\begin{enumerate}[label=(\roman*), wide, labelindent=0pt]
\item\label{cor-26-8-13-i}
If $(Z-N)\ge C (Qa)^{\frac{1}{2}}$ where $Q$ is our remainder estimate in the ground state energy, then in free nuclei model minimal distance between nuclei must be at least $a$.

\item\label{cor-26-8-13-ii}
In particular, if $(Z-N)\ge C_1 (Q\bar{r})^{\frac{1}{2}}$ then in free nuclei model minimal distance between nuclei must be at least $C_0\bar{r}$ and molecule consists of separate atoms.
\end{enumerate}
\end{corollary}
\vglue-5pt
We leave to the reader

\begin{Problem}\label{Problem-26-8-14}
Using Theorem~\ref{thm-26-2-17} and the arguments used in the proof of Proposition~\ref{prop-26-8-11}, estimate overlapping of balls $B(\y_m,\bar{r}_m)$ if $Z^{-\frac{5}{21}-\delta}\ge \epsilon B^{-\frac{1}{4}}$ in the free nuclei model with $N=Z$ and prove that
\begin{multline}
(\bar{r}_m+\bar{r}_{m'}-|\y_m-\y_{m'}|)\le C\bar{r}(K^{-2}\bar{r}^{-1}Q)^{\frac{1}{2}}=\\
C\left\{\begin{aligned}
& B^{-\frac{1}{4}} \bigl(B^{-\frac{7}{4}}Z^{\frac{5}{3}}+B^{-\frac{1}{2}}L\bigr)^{\frac{1}{12}}
\qquad &&\text{if\ \ } Z^{\frac{20}{21}}\le B\le Z^{\frac{4}{3}},\\
&B^{-\frac{7}{15}}Z^{\frac{1}{10}}L^{\frac{1}{6}} \qquad
&&\text{if\ \ } Z^{\frac{4}{3}}\le B\le Z^3.
\end{aligned}\right.
\label{26-8-73}
\end{multline}
\end{Problem}

\subsection{Estimate of excessive positive charge}
\label{sect-26-8-4-3}

To estimate excessive positive charge when molecules can still exist in free nuclei model we apply arguments of section 5 of B.~Ruskai and J.~P.~Solovej~\cite{ruskai:solovej}. In view of Corollary~\ref{cor-26-8-13} for $(Z-N)$ violating (\ref{26-8-76}) below it is sufficient to assume that (\ref{book_new-25-6-41}) is satisfied:
\begin{equation}
a=\min_{j<k} |\y_j-\y_k|\ge C_0\bar{r}
\label{26-8-74}
\end{equation}
i.e. in Thomas-Fermi theory $\rho^\TF$ is supported in the separate ``atoms''. Really, it is the case if $C_0 Z^{\frac{20}{21}}\le B\le Z^3$ but also it is so if $B\le C_0Z^{\frac{20}{21}}$ and $(Z-N)_+\ge C_1 Z^{\frac{5}{7}}$ since then $\bar{r}\asymp (Z-N)_+^{-\frac{1}{3}}$.

Like in Subsection~\ref{book_new-sect-25-6-3} consider $a$-admissible functions
$\theta _m(x)$, supported in $B(\y_m, \frac{1}{3}a)$ as $m=1,\ldots , M$ and in
$\{ |x-\y_{m'}|\ge \frac{1}{4} a\ \ \forall m'=1,\ldots , M\}$ as $m=0$, such that
\begin{equation}
\theta _0 ^2+\ldots +\theta _M ^2=1.
\label{26-8-75}
\end{equation}
Then for the ground state $\Psi $ equality (\ref{book_new-25-6-43}) holds with \index{Hamiltonian!intercluster}\emph{cluster Hamiltonians\/} $\mathbf{H}_{\alpha_m}$ defined by (\ref{book_new-25-6-44}) and satisfying (\ref{book_new-25-6-45}) and with the
\index{Hamiltonian!intercluster}\emph{intercluster Hamiltonian\/} $J_\alpha$ defined by (\ref{book_new-25-6-46}) and satisfying (\ref{book_new-25-6-47}) with $J_{ml}$ defined by (\ref{book_new-25-6-48})--(\ref{book_new-25-6-49}). Furthermore, equality (\ref{book_new-25-6-50}) holds.

Applying Proposition~\ref{book_new-prop-25-5-1} and estimate (\ref{book_new-25-4-55}) (replacing first $\theta _k$ with $k=1,\ldots, M$ by $\tilde{\theta}_k$, supported in $B(\y_k, c\bar{r})$, and estimating the resulting error), we conclude that (\ref{book_new-25-6-51})--(\ref{book_new-25-6-54}) hold with $Y=Q^{\frac{1}{2}}\bar{r}^{\frac{1}{2}}$ since
$\D (\rho_\Psi -\rho^\TF,\,\rho_\Psi -\rho^\TF)\le CQ$.

The last term in (\ref{book_new-25-6-51}) is estimated by Proposition~\ref{book_new-prop-25-5-1} and estimate (\ref{26-8-61}) instead of (\ref{book_new-25-4-55}) and the same replacement trick; so we arrive to (\ref{book_new-25-6-55}) and repeating the same trick we get that it is larger than (\ref{book_new-25-6-56}).

Again let us note that the absolute value of the last term in the right-hand expression of (\ref{book_new-25-6-43}) does not exceed $Ca ^{-2}Y$ due to (\ref{book_new-25-6-52}). Now stability condition yields that (\ref{26-6-61}) must be fulfilled.

Then we conclude that (\ref{book_new-25-6-57}) and (\ref{book_new-25-6-59}) hold with $J_{ml}$ defined by (\ref{book_new-25-6-58}) provided (\ref{book_new-25-6-60}) is fulfilled as $|x-\y_k|\ge C\bar{r}$.

This inequality, (\ref{26-8-74}) and Proposition~\ref{book_new-prop-25-6-6} (which is the special case of Theorem~\ref{26-2-13}) yield that
$Z-N\le CY = C \bar{r}^{\frac{1}{2}} Q^{\frac{1}{2}}$. Now we need to consider two cases:

\begin{enumerate}[label=(\alph*), wide, labelindent=0pt]

\item\label{sect-26-8-4-3-a}
$B\le (Z-N)^{\frac{4}{3}}$; then $\bar{r}\asymp (Z-N)^{-\frac{1}{3}}$ and we conclude that $(Z-N)\le C Q^{\frac{3}{5}}$ exactly like in Subsubsection~\ref{book_new-sect-25-6-3}.

\item\label{sect-26-8-4-3-b}
$(Z-N)^{\frac{4}{3}} \le B\le Z^3$; then plugging $\bar{r}$ and $Q$ we arrive to two other cases of (\ref{26-8-76}).
\end{enumerate}

Then we arrive to Statement~\ref{thm-26-8-15-i} below; Statement~\ref{thm-26-8-15-ii} follows from Remark~\ref{rem-26-8-10}\ref{rem-26-8-10-ii}.

\begin{theorem-foot}\label{thm-26-8-15}\footnotetext{\label{foot-26-49} Cf. Theorem~\ref{book_new-thm-25-6-4}.}
Let condition \textup{(\ref{26-2-28})} be fulfilled.

\begin{enumerate}[label=(\roman*), wide, labelindent=0pt]
\item\label{thm-26-8-15-i}
Then in the framework of the free nuclei model with $M\ge 2$ the stable molecule does not exist unless
\begin{equation}
(Z-N)_+\le C_1\left\{\begin{aligned}
&Z^{\frac{20}{21}} \qquad
&&\text{if\ \ } B\le Z^{\frac{20}{21}},\\
&Z^{\frac{5}{6}}B^{-\frac{1}{8}}\qquad
&&\text{if\ \ }Z^{\frac{20}{21}}\le B\le Z^{\frac{4}{3}},\\
&Z^{\frac{2}{5}}B^{\frac{1}{5}} \qquad
&&\text{if\ \ }Z^{\frac{4}{3}}\le B\le Z^3.
\label{26-8-76}
\end{aligned}\right.
\end{equation}
\item\label{thm-26-8-15-ii}
Furthermore, for $B\le Z$ in the framework of the free nuclei model with
$M\ge 2$ the stable molecule does not exist unless
\begin{equation}
(Z-N)_+\le C_1\left\{\begin{aligned}
&Z^{\frac{20}{21}-\delta} \qquad
&&\text{if\ \ } B\le Z^{\frac{20}{21}},\\
&Z^{\frac{5}{6}-\delta} B^{-\frac{1}{8}+\delta}\qquad
&&\text{if\ \ }Z^{\frac{20}{21}}\le B\le Z.
\label{26-8-77}
\end{aligned}\right.
\end{equation}
\end{enumerate}
\end{theorem-foot}

\subsection{Estimate of excessive negative charge and ionization energy}
\label{sect-26-8-4-4}

Estimate (\ref{26-8-61}) and Remark~\ref{rem-26-8-10} immediately imply
\pagebreak
\begin{theorem}\label{thm-26-8-16}
Let condition \textup{(\ref{26-2-28})} be fulfilled.
\begin{enumerate}[label=(\roman*), wide, labelindent=0pt]
\item\label{thm-26-8-16-i}
Then in the framework of the free nuclei model with $M\ge 2$ estimates \textup{(\ref{26-7-21})} for the excessive negative charge and \textup{(\ref{26-7-37})} for the ionization energy
$\hat{\I}_N=-\hat{\E}_N+\hat{\E}_{N-1}$ hold.

\item\label{thm-26-8-16-ii}
Furthermore, if $B\le Z$ estimates \textup{(\ref{26-7-22})} for the excessive negative charge and \textup{(\ref{26-7-39})} for the ionization energy
$\hat{\I}_N$ hold.
\end{enumerate}
\end{theorem}

\begin{subappendices}
\chapter{Appendices}
\section{Electrostatic inequalities}
\label{sect-26-A-1}

There are two kinds of electrostatic inequalities: those which hold for any fermionic state $\Psi$ and those which hold only for the ground-state (or near ground state) $\Psi$. Inequalities of the first kind do not depend on the quantum Hamiltonian and they are (\ref{book_new-25-2-1}) repeated here:
\begin{multline}
\sum_{1\le j<k\le N}
\int |x_j-x_k| ^{-1}|\Psi (x_1,\dots ,x_N)| ^2\, dx_1 \cdots dx_N\ge\\
\frac{1}{2}\D(\rho_\Psi ,\rho_\Psi )-C\int \rho_\Psi ^{\frac{4}{3}}(x)\, dx
\label{26-A-1}
\end{multline}
and (\ref{26-A-5}) below.

Inequalities of the second kind are for $B=0$:
\begin{multline}
\sum_{1\le j<k\le N}
\int |x_j-x_k| ^{-1}|\Psi (x_1,\dots ,x_N)| ^2\, dx_1 \cdots dx_N\ge\\
\frac{1}{2}\D(\rho_\Psi ,\rho_\Psi )-CZ^{\frac{5}{3}},
\label{26-A-2}
\end{multline}
and more precise one (\ref{26-A-26}) below.

For $\vec{B}=\const$ there is an inequality established in E.~Lieb, J.~P.~Solovej and J.~Yngvarsson~\cite{LSY2} (p. 122):

\begin{theorem}\label{thm-26-A-1}
Let $\vec{B}=\const $. Then for the ground state $\Psi $
\begin{gather}
\int \rho _\Psi ^{\frac{4}{3}}\, dx \le
CZ ^{\frac{5}{6}} N ^{\frac{1}{2}} (Z +N )^{\frac{1}{3}}
\bigl(1+ B Z ^{-\frac{4}{3}} \bigr)^{\frac{2}{5}};
\label{26-A-3}
\end{gather}

In particular, for $c ^{-1}N\le Z\le cN$ the right-hand expression does not exceed\pagebreak
\begin{equation}
CZ ^{\frac{5}{3}} \bigl(1+B Z ^{-\frac{4}{3}}\bigr)^{\frac{2}{5}}\asymp
C\left\{\begin{aligned}
&Z ^{\frac{5}{3}} \qquad&&\text{if\ \ } B\le Z^{\frac{4}{3}},\\
&Z ^{\frac{17}{15} } B^{\frac{2}{5}} \qquad&&\text{if\ \ } B\ge Z^{\frac{4}{3}}.
\end{aligned}\right.
\label{26-A-4}
\end{equation}
\end{theorem}

We want to establish inequality, similar to (\ref{book_new-25-A-2}), but in the magnetic case. We will use for this the following

\begin{theorem-foot}\label{thm-26-A-2}\footnotetext{\label{foot-26-50} Lemma 6 of G.~Graf and J.~P.~Solovej~\cite{graf:solovej}.}
Fix $0<\delta\le 1/6$. Then for any density matrix $F$ and any
density $\rho_0(x)\ge 0$ the following inequality holds
\begin{multline}
\sum_{1\le j<k\le N}
\int |x_j-x_k| ^{-1}|\Psi (x_1,\dots ,x_N)| ^2\, dx_1 \cdots dx_N\ge\\[2pt] \D(\rho_0,\rho_\gamma)-\frac{1}{2}\D(\rho_0,\rho_0)
-\frac{1}{2} \sum_{\varsigma,\varsigma'}
\iint |F(x,\varsigma;y,\varsigma')|^2|x-y| ^{-1}\,dxdy\\
-C \|\rho\|_{5/3}^{5/6}\cdot \|\rho\|_1^{1/6+\delta} ,\upsilon(\gamma,F)^{1/3-\delta}
\label{26-A-5}
\end{multline}
where $\rho=\rho_0+\rho_F+\rho_\Psi$,
$\upsilon(\gamma ,F)\coloneqq \Tr (\gamma (I-F))$ and
\begin{equation}
\gamma =\gamma _\Psi(x,y)=N\int \Psi (x,x_2,\ldots,x_N)
\Psi^\dag (y,x_2,\ldots,x_N)\,dx_2\cdots x_N
\label{26-A-6}
\end{equation}
is two-point one particle density.
\end{theorem-foot}

Recall that $\|.\|_p$ denotes $\sL^p$-norm.

There is a connection between (\ref{26-A-1}) and (\ref{26-A-5}): if we set $F=0$, we get $\varsigma= \|\rho\|_1$ and the last term in (\ref{26-A-5}) becomes $\|\rho\|_{5/3}^{5/6}\cdot \|\rho\|_1^{1/2}$. On the other hand,
$\|\rho\|^{4/3}_{4/3}\le \|\rho\|_{5/3}^{5/6}\cdot \|\rho\|_1^{1/2}$, so (\ref{26-A-5}) is slightly deteriorated (\ref{26-A-1}) with $F=0$ but with ``free'' $\rho_0$.

Let us follow G.~Graf and J.~P.~Solovej~\cite{graf:solovej} further albeit in the case of magnetic field. Let us estimate first $\|\rho\|_{5/3}$.

If $\rho=\rho^\TF$ direct calculations show that for $N\asymp Z$
\begin{gather}
\int \rho^\TF \,dx =\min (Z,N),\label{26-A-7}\\
\int (\rho^\TF)^{\frac{4}{3}} \,dx \asymp C\rho^{*\,\frac{4}{3}}r^{*\,3}=
CZ^{\frac{5}{3}}\bigl(1+BZ^{-\frac{4}{3}}\bigr)^{\frac{2}{5}},\label{26-A-8}\\
\int (\rho^\TF)^{\frac{5}{3}} \,dx \asymp C\rho^{*\,\frac{5}{3}}r^{*\,3}=
CZ^{\frac{7}{3}}\bigl(1+BZ^{-\frac{4}{3}}\bigr)^{\frac{4}{5}}\label{26-A-9}
\shortintertext{with}
r^*= \min (Z^{-\frac{1}{3}}, B^{-\frac{2}{5}}Z^{\frac{1}{5}})\asymp
Z^{-\frac{1}{3}}\bigl(1+BZ^{-\frac{4}{3}}\bigr)^{-\frac{2}{5}},
\label{26-A-10}\\[3pt]
\rho^{*}= \min (N,Z) r^{*\,-3}\label{26-A-11}
\end{gather}
and we use
$\|\rho \|_{5/3}^{5/6}\cdot \|\rho\|_1^{1/2} \asymp \|\rho\|_{4/3}^{4/3}$ for $\rho=\rho^\TF$.

If $\rho=\rho_\Psi$ we use magnetic Lieb--Thirring inequality (see f.e. Theorem~2.2 in L.~Erd\"os~\cite{erdos:magnetic})
\begin{equation}
\Tr (H_{A,W}^-) \ge -C \int P_B (W)\,dx
\label{26-A-12}
\end{equation}
and therefore
\begin{multline}
\blangle\mathbf{H} \Psi,\Psi\brangle\ge \Tr((H_{A,W})^-) + \int W\rho_\Psi \,dx\\
-\int V\rho_\Psi \,dx + \frac{1}{2}\D(\rho_\Psi,\rho_\Psi)
- C\|\rho_\Psi\|_{4/3}^{4/3},
\label{26-A-13}
\end{multline}
which due to (\ref{26-A-12}) is greater than
\begin{multline}
\int \bigl(-CP_B(W)+ W\rho_\Psi \bigr)\,dx -\\
\shoveright{\int V\rho_\Psi \,dx +
\frac{1}{2}\D(\rho_\Psi,\rho_\Psi)- C\|\rho_\Psi\|_{4/3}^{4/3}\ge}\\
3\epsilon_0 \int \tau_B (\rho_\Psi)\,dx -\int V\rho_\Psi \,dx +
\frac{1}{2}\D(\rho_\Psi,\rho_\Psi)- C\|\rho_\Psi\|_{4/3}^{4/3}
\label{26-A-14}
\end{multline}
where we picked up $W: C P_B'(W)=\rho_\Psi$.

The first two terms in the right-hand expression are estimated from below by
\begin{equation*}
2\epsilon_0 \int \tau_B (\rho_\Psi)\,dx - C\int P_B (V)\phi \,dx - C\int V\rho_\Psi(1-\phi)\,dx,
\end{equation*}
where $\supp(\phi)\subset \{x\colon \ell (x)\le 2r^*\}$ and
$\supp (1-\phi) \subset \{x\colon \ell (x)\ge r^*\}$.

One can see easily that the absolute value of the second term is
$\asymp Z^{\frac{7}{3}}\bigl(1+BZ^{-\frac{4}{3}}\bigr)^{\frac{2}{5}}$, while the absolute value of the third term does not exceed $CZ\int V(1-\phi)\,dx \asymp CZ^2r^{*\,-1}$ which does not exceed the same expression
$Z^{\frac{7}{3}}\bigl(1+BZ^{-\frac{4}{3}}\bigr)^{\frac{2}{5}}$. Therefore
\pagebreak
\begin{multline}
\blangle\mathbf{H} \Psi,\Psi\brangle+ C_1Z^{\frac{7}{3}}\bigl(1+BZ^{-\frac{4}{3}}\bigr)^{\frac{2}{5}}\ge\\
2\epsilon_0 \int \tau_B(\rho_\Psi)\,dx + \frac{1}{2}\D(\rho_\Psi,\rho_\Psi)
- C\|\rho_\Psi\|_{4/3}^{4/3} .
\label{26-A-15}
\end{multline}
Note that $\|\rho_\Psi\|^{4/3}_{4/3}$, calculated over domain
$\{x\colon \rho_\Psi(x)\ge B^{\frac{4}{3}}\}$, does not exceed
$C\|\rho_\Psi \|_{5/3}^{5/6}\cdot \|\rho\|_1^{\frac{1}{2}}$ with norms, calculated over the same domain, which does not exceed
$CT^{\frac{1}{2}}Z^{\frac{1}{2}}$ with $T=\int \tau_B (\rho_\Psi)\,dx$.

Meanwhile, $\|\rho_\Psi\|^{4/3}_{4/3}$, calculated over domain
$\{x\colon \rho_\Psi(x)\le B^{\frac{4}{3}}\}$, does not exceed
$C\|\rho\|_3^{\frac{1}{2}}\cdot \|\rho\|_1^{\frac{5}{6}}$ with norms, calculated over the same domain, which does not exceed
$CZ^{\frac{1}{2}} B^{\frac{1}{3}}T^{\frac{1}{6}}$.

Therefore
\begin{gather}
\|\rho_\Psi\|_{4/3}^{4/3}\le CT^{\frac{1}{2}}Z^{\frac{1}{2}} +
CZ^{\frac{5}{6}} B^{\frac{1}{3}}T^{\frac{1}{6}}
\label{26-A-16}\\
\intertext{and therefore (\ref{26-A-15}) implies that if}
\blangle\mathbf{H} \Psi,\Psi\brangle \le C_1Z^{\frac{7}{3}}\bigl(1+BZ^{-\frac{4}{3}}\bigr)^{\frac{2}{5}},
\label{26-A-17}\\
\shortintertext{then}
T=\int \tau_B (\rho_\Psi)\,dx\le C_2Z^{\frac{7}{3}}\bigl(1+BZ^{-\frac{4}{3}}\bigr)^{\frac{2}{5}}
\label{26-A-18}\\
\shortintertext{and}
\|\rho_\Psi\|_{4/3}^{4/3} \le
CZ^{5/3}\bigl(1+BZ^{-\frac{4}{3}}\bigr)^{\frac{2}{5}};
\label{26-A-19}
\end{gather}
taking $F=0$ we arrive to (\ref{26-A-3}) if $N\asymp Z$.

However, on our preparatory step we need to estimate also $\|\rho_\Psi\|_{5/3}^{5/3}$ and due to (\ref{26-A-18}) we need to consider only norms over $\{x\colon\rho_\Psi(x)\ge B^{\frac{4}{3}}\}$. Then
$\|\rho_\Psi\|_{5/3}^{5/3}\le C\|\rho_\Psi\|_{4/3}^{16/15}\cdot \|\rho_\Psi\|_{3}^{3/5}$ and plugging the same estimates (\ref{26-A-19}), (\ref{26-A-19}) we conclude that
\begin{equation}
\|\rho_\Psi\|_{5/3}^{5/3} \le
CZ^{7/3}\bigl(1+BZ^{-\frac{4}{3}}\bigr)^{\frac{4}{5}}.
\label{26-A-20}
\end{equation}

Now we assume that $B\le Z^3$, take $F=e(x,y,\nu)$, where $e(x,y,\nu)$ is the Schwartz kernel of spectral projector for potential $W$, approximating $W^\TF$ and $\nu\le 0$ is a chemical potential. One can prove easily that $\|\rho_F\|_{5/3}^{5/3}$ satisfies the same estimate and we need to estimate
$\Tr \bigl(\gamma _\Psi(I-E(\mu))\bigr)$.

Consider
\begin{multline}
N\blangle H_{A,W (x_1)} \Psi,\Psi\brangle - \Tr (H E(\nu) )-
\alpha \Tr \bigl(\gamma_\Psi (I-E(\nu))\bigr)\\
\begin{aligned}
\ge & \int_{\beta<0} \beta \,d_\beta \Tr E(\beta)-\int_{\beta\le \nu} (\beta -\nu+\alpha)\,d_\beta \Tr E(\beta)\\
=-&\int _{\nu-\alpha<\beta < \nu} (\beta-\nu+\alpha) d_\beta E(\beta)\\
=-& \alpha E(\nu) + \int _{\nu-\alpha<\beta < \nu} E(\beta)\,d\beta.
\end{aligned}
\label{26-A-21}
\end{multline}
We can replace $E(\beta)$ by $\int P'(W+\beta)\,dx$ with a resulting error $O(Z\alpha \hslash^\delta)$, $\hslash\coloneqq BZ^{-3}$. Then the right-hand expression becomes
\begin{multline}
-L(\alpha)\coloneqq \int \Bigl(-\alpha P_B'(W+\nu) + \int _0^\alpha P_B'(W+\nu-\beta)\,d\beta\Bigr) \,dx=\\
-\int _0^\alpha (\alpha-\beta)\Bigl(\int P_B ''(W+\nu-\beta)\,dx\Bigr)\,d\beta.
\label{26-A-22}
\end{multline}
Therefore
\begin{multline}
\alpha \Bigl(\Tr \bigl(\gamma _\Psi(I-E(\mu))\bigr)-CZ\hslash^\delta\Bigr)\le \\ \underbracket{N \blangle H_{A,W (x_1)} \Psi,\Psi\brangle - \Tr (H E(\nu) )} +L(\alpha).
\label{26-A-23}
\end{multline}
Note that adding to the selected terms $-\frac{1}{2}\D(\rho^\TF,\rho^\TF)$  we obtain exactly the snippet, occurring in the lower estimate of $\E_N$, but in virtue of the upper estimate it should not exceed
$Q=CZ^{\frac{5}{3}}\bigl(1+BZ^{-\frac{4}{3}})^{\frac{2}{5}}\le
CZ^{\frac{5}{3}}+ Ch ^2 Z^{\frac{7}{3}}$, and therefore, plugging
$\alpha= Z^{\frac{4}{3}}\hslash ^\delta$, we conclude that
\begin{gather}
\Tr \bigl(\gamma _\Psi(I-E(\mu)) \le Z\hslash ^\delta
\label{26-A-24}\\
\intertext{provided we prove that}
L(\alpha)\le Q\qquad \text{for \ } \alpha= Z^{\frac{4}{3}}\hslash ^\delta .
\label{26-A-25}
\end{gather}
Therefore modulo proof of (\ref{26-A-25}) we arrive to the estimate (\ref{26-A-26}) below:

\begin{theorem}\label{thm-26-A-3}
Let $N\asymp Z$ and $B\le Z$. Then for the ground state energy
\begin{multline}
\sum_{1\le j<k\le N}
\int |x_j-x_k| ^{-1}|\Psi (x_1,\dots ,x_N)| ^2\, dx_1 \cdots dx_N\ge\\
\frac{1}{2}\D(\rho^\TF,\rho^\TF) +\Dirac -
CZ^{\frac{5}{3}-\delta}(1 + B^\delta)
\label{26-A-26}
\end{multline}
\end{theorem}

To prove (\ref{26-A-25}) we note that $0\le P''_B(w) \asymp
w^{\frac{1}{2}}+ Bw^{-\frac{1}{2}}$. One can prove easily then that
$L(\alpha)\le C\alpha^{\frac{7}{4}}+ CB^{\frac{1}{4}}\alpha ^{\frac{3}{2}}$,
which obviously implies (\ref{26-A-25}).

\section{Very strong magnetic field case}
\label{sect-26-A-2}
Let us consider now case $Z^2\le B\le Z^3$.

\begin{proposition}\label{prop-26-A-4}
Consider the Schr\"odinger operator $H_{A,W}$ with a constant magnetic field of intensity $B$ and potential $W$: $W\le Z|x|^{-1}$. Let $\phi(x)\coloneqq \phi _r(x)$ be $r$-admissible function. Then if $Z^2\lesssim B\lesssim Z^3$ and
$r\asymp Z^{-1}$
\begin{equation}
|e(x,y,0)|\le CZ B \qquad \text{in\ \ } B(0,r)
\label{26-A-27}
\end{equation}
and
\begin{claim}\label{26-A-28}
All eigenvalues are $\ge -CZ^2$.
\end{claim}
\end{proposition}

\begin{proof}
Without any loss of the generality one can assume that
\begin{equation}
H_{A,W} = D_3^2 + D_2^2 + (D_1-Bx_2)^2 -W.
\label{26-A-29}
\end{equation}
Consider $f\in \sL^2$; then $\| E (\lambda) f\| \le \|f\|$ and then one can prove easily (\ref{26-A-28}) and inequality
\begin{equation}
\| H_{A,0} E (\lambda) f \|\le (CZ^2+\lambda_+) \|f\|.
\label{26-A-30}
\end{equation}
Indeed, $\frac{1}{2} D_3^2 + C Z^2 \ge W$ in the operator sense.

Then (\ref{26-A-30}) implies that in $B(0,r) \times B(y',r')$ with $r'=B^{-\frac{1}{2}}$
\begin{equation*}
|P^\alpha E(\lambda)f |\le C Z^{\alpha_3} B^{\frac{1}{2}|\alpha'|} \qquad \forall \alpha:|\alpha|\le 2\ \forall \lambda\le Z^2
\end{equation*}
with $P=(D_1-Bx_2, D_2, D_3)$ and therefore
$\|E(\lambda)f\|_{\sC} \le CZ^{\frac{1}{2}}B^{\frac{1}{2}}\|f\|$. Then
$\|E(x,.,\lambda)\|_{\sL^2_y}\|\le CZ^{\frac{1}{2}}B^{\frac{1}{2}}$.

Repeating the same arguments with respect to $y$ we arrive to estimate (\ref{26-A-27}).
\end{proof}

The following corollary follows immediately:
\pagebreak

\begin{corollary}\label{cor-26-A-5}
In the framework of Proposition~\ref{prop-26-A-4} with
$\phi\in \sL^\infty (B(0,r))$, $\|\phi\|_{\sL^\infty}\le 1$
\begin{gather}
|\int \phi (x)e(x,x,0)\,dx|\le CZ^{-2}B,\label{26-A-31}\\[3pt]
\D \bigl( \phi (x)e(x,x,0), \phi (x)e(x,x,0)\bigr)\le C Z^{-3}B^2
\label{26-A-32}\\
\shortintertext{and}
|\int^0_{-\infty} \int \phi(x)e(x,x,\tau)\,d\tau dx|\le CB.
\label{26-A-33}
\end{gather}
\end{corollary}

\section{Riemann sums and integrals}
\label{sect-26-A-3}
If $f\in \sC^\infty (\bR^+)$ and fast decays at $+\infty$, then
\begin{align}
f(0)h + \sum_{n\ge 1} 2f(2nh) h & \sim
\int_0^\infty f(t)\,dt+\sum_{m\ge 1} \kappa_{m}f^{(2m-1)}(0)h^{2m},
\label{26-A-34}\\[3pt]
\sum_{n\ge 0} 2f((2n+1)h) h & \sim
\int_0^\infty f(t)\,dt+\sum_{m\ge 1} \kappa'_{m}f^{(2m-1)}(0)h^{2m}
\label{26-A-35}
\end{align}
as $h\to +0$. The proofs of both formulae follow from the Taylor's decomposition and observation that the odd powers of $h$ should disappear. Taking $f(t)=e^{-tz/h}$ with $\Re z>0$ we arrive to
\begin{align}
1-\frac{\cosh(z)}{\sinh(z)}z\sim \sum_{m\ge 1}\kappa_m z^{2m},
\label{26-A-36}\\[3pt]
1-\frac{1}{\sinh(z)}z\sim \sum_{m\ge 1}\kappa'_m z^{2m}
\label{26-A-37}
\end{align}
for $|z|\ll 1$. In particular, $\kappa_1=\frac{1}{3}$ and $\kappa'_1=-\frac{1}{6}$.

\section{Some spectral function estimates}
\label{sect-26-A-4}

\begin{proposition}\label{prop-26-A-6}
For the Schr\"odinger operator with $A, W\in \sC ^\infty $ and for
$\phi \in \sC_0 ^\infty ([-1,1])$ the following estimate holds for any $s$:
\begin{gather}
|F(x,y)|\le C (\mu h+1) h ^{-3}\bigl(1+h ^{-1}|x-y|\bigr) ^{-s}\label{26-A-38}\\
\shortintertext{where}
F(x,y)\coloneqq \int \phi (\lambda )\,d_\lambda e(x,y,\lambda ).
\label{26-A-39}
\end{gather}
\end{proposition}

\begin{proof} Let
$u(x,y,t)=\int e ^{-ih ^{-1}t\lambda }\,d_\lambda e(x,y,\lambda )$ be the
Schwartz's kernel of $e ^{-ih ^{-1}Ht}$.

Let us fix $y$. Note first that $\sL ^2$-norm\footnote{\label{foot-26-51} With respect to $x,t$ here and below.} of $\phi (hD_t)\chi (t)\omega (x)u(x,y,t)$ is less than $Ch ^s$ for $\chi \in \sC_0 ^\infty ([-\epsilon ,\epsilon ])$ and
$\omega \in \sC ^\infty $
supported in $\{x\colon |x-y|\ge \epsilon _1\}$ (with $\epsilon _1=C\epsilon $) due to the finite speed of propagation of singularities.

We conclude then that $\sL ^2$-norm of $\phi (hD_t)\chi (t)\omega (x)u(x,y,t)$ does not exceed $C(\mu h+1) h ^s$ for $\omega \in \sC ^\infty $ supported in
$\{x\colon |x-y|\ge C\}$.

Then $\sL ^2$-norm of
$\partial _t^l \nabla ^\alpha \phi (hD_t)\chi (t)\omega (x)u$ does not exceed
$C(\mu h+1) h ^s$. Therefore due to imbedding inequality $\sL ^\infty $-norm of
$\phi (hD_t)\chi (t)\omega (x)u$ also does not exceed $C (\mu h+1) h ^s$. Setting $t=0$ and using this inequality and estimate
$ |F(x,y)|\le C(\mu h+1)h ^{-3}$ (due to Chapter~\ref{book_new-sect-7}), we conclude that $ |F(x,y)|\le C(\mu h+1)h^s$ for $|x-y| \ge \epsilon _0$.

Now let us consider general $x$ with $|x-y|=r\ge Ch$. Then rescaling
$(x-y) \mapsto (x-y)r ^{-1}$ we need also to rescale $h\mapsto hr ^{-1}$,
$\mu\mapsto \mu r $ and rescaling the above inequality and keeping in mind that $F(x,y)$ is a density with respect to $x$, we conclude that
$|F(x,y)|\le C h ^sr ^{-3-s}$ which is equivalent to (\ref{26-A-38})--(\ref{26-A-39}). \end{proof}

\section{Zhislin's theorem for constant magnetic field}
\label{sect-26-A-5}

We provide just a scheme to prove Zhislin's theorem in the case of the constant magnetic field. In this analysis $\underline{Z}$, $\underline{y}$, $N$ and $B$ are constant.

\begin{proposition}\label{prop-26-A-7}
Let $\Psi=\Psi_N$ be the ground state with the energy $\E_N<\E_{N-1}$. Then

\begin{enumerate}[label=(\roman*), wide, labelindent=0pt]
\item\label{prop-26-A-7-i}
$\Psi\in \sC^1$ and $\Psi= O(e^{-\epsilon |\underline{x}|})$ as $|\underline{x}|\to \infty$.

\item\label{prop-26-A-7-ii}
Let $N<Z$. Then $V_\Psi -V\in \sC^2$ and $V_\Psi = (Z-N)|x|^{-1} +O(|x|^{-2})$,
$\nabla{V}_\Psi = (Z-N)|x|^{-2} +O(|x|^{-3})$ as $|x|\to \infty$.
\end{enumerate}
\end{proposition}
\begin{proof}
Obvious proof is left to the reader.
\end{proof}

\begin{theorem}[Zhislin's theorem]\label{thm-26-A-8}
$\E_{N+1}< \E_N$ for $N<Z$.
\end{theorem}

\begin{proof}
We can assume that $\E_N<0$ and the ground state energy exists. Really, it is true for some $N<Z$ and if we prove that then automatically $\E_{N+1}<\E_N$,  then it would be true for $(N+1)$ as well, so we may go by induction.
\pagebreak

Consider $\Psi=\Psi_N(x_1,\ldots, x_N)$ and also $\tilde{\Psi}_{N+1}$, which is an antisymmetrized $\Psi_N(x_1,\ldots, x_N)u(x_{N+1})$ (cf. (\ref{26-8-21})):
\begin{multline}
\tilde{\Psi}=\tilde{\Psi}(x_1,\ldots, x_{N+1})=
\Psi (x_1,\ldots, x_N)u(x_{n+1})-\\
\sum_{1\le j\le N}
\Psi (x_1,\ldots, x_{j-1},x_{N+1},x_{j+1},\ldots, x_N)u(x_j).
\label{26-A-40}
\end{multline}
Then like in the estimate of the ionization energy
(cf. (\ref{26-8-22})--(\ref{26-8-23})):
\begin{equation}
N ^{-1} \I_{N+1} \|\tilde{\Psi}\| ^2\ge
-\blangle H_{V,x_{N+1}}\Psi u,\tilde{\Psi}\brangle
-\blangle \sum_{1\le i\le N}|x_i-x_{N+1}| ^{-1}
\Psi u, \tilde{\Psi}\brangle
\label{26-A-41}
\end{equation}
and
\begin{multline}
N ^{-1}\|\tilde{\Psi}\| ^2=\|\Psi \| ^2\cdot \|u\| ^2-\\
N\int \Psi (x_1,\ldots,x_{N-1},x)\Psi ^\dag (x_1,\ldots,x_{N-1},y)
u(y)u^\dag(x) \, dx_1\cdots dx_{N-1}\, dxdy .
\label{26-A-42}
\end{multline}
Now let us consider $u$ supported in $\{x\colon \frac{1}{2}a\le |x|\le 3a\}$ with $a$ to be chosen later. Then in virtue of Proposition~\ref{prop-26-A-7}\ref{prop-26-A-7-i} modulo $O(e^{-\epsilon_1a})$ we can replace in the right-hand expressions
$\tilde{\Psi}$ by $\Psi_N(x_1,\ldots, x_N)u(x_{N+1})$ resulting in
$-\blangle H_{W}u,u\brangle$ and $\|u\|^2$ respectively with $W=V_\Psi$ defined in Proposition~\ref{prop-26-A-7}\ref{prop-26-A-7-ii}.

Therefore all we need  to prove this theorem is to be able to select $u$ with $\|u\|\asymp 1$, supported in $\{x\colon  a\le |x|\le 3a\}$  and with
$\blangle H_{W}u,u\brangle \le - \epsilon_0 a^{-1}$.

In virtue of Proposition~\ref{prop-26-A-7}\ref{prop-26-A-7-ii} $
V_\Psi \ge \epsilon_0 a^{-1}$ in $\{x\colon a\le |x|\le 3a\}$ and therefore we can replace $W$ by $\epsilon_0 a^{-1}$. Without any loss of the generality one can assume that $A= (Bx_2, 0, 0)$. Recall that for the linear vector-potential $\vec{A}$ operator $H_0=((i\nabla -A)\cdot\boldupsigma)^2$ is a direct sum of
$H_0^+=(i\nabla -A)^2+B$ and $H_0^-=(i\nabla -A)^2-B$; so we can consider only the latter. Note that $H_0^-=(i\partial_1-Bx_2)^2 -\partial_2^2 -\partial_3^2 $ and $H_0^- v=0$ with $v=\exp (-\frac{1}{2}B(x_2-a)^2+iBax_1)$.

Then $u= v(x) \chi (r^{-1}(x-\bar{x}))$ with $\chi\in \sC^\infty (B(0,1))$, $\chi=1$ in $B(0,\frac{1}{2})$, $\bar{x}=(0,2a,0)$, $r=\frac{1}{3}a$ is a required function. \end{proof}

\end{subappendices}

\chapter*{Comments}

We already mentioned papers E.~H.~Lieb, J.~P.~Solovej and J.~Yngvarsson \cite{LSY2, LSY1} where asymptotics of the ground state energy were derived in the cases $B\ll Z^3$ and $B\gg Z^3$ respectively. Intermediate case $B\sim Z^3$ was covered also in \cite{LSY1}. Even without remainder estimates certain results concerning ionization energy and maximal possible positive and negative charges were also derived.

Remainder estimates in the case $B\ll Z^3$ were derived by V.~Ivrii in \cite{ivrii:MQT1, ivrii:MQT2}. Unfortunately there are gaps in the proofs of the second paper in the case of $M\ge 2$ and large $Z-N>0$ which I was unable to fill.; so our results in this case are not as sharp as they supposed to be.

 %\printindex

\end{document}